\documentclass{article}
\usepackage[margin=1in]{geometry}
\usepackage{tikz}
\usetikzlibrary{shapes.geometric}
\usetikzlibrary{hobby}  
\usepackage{graphicx} 
\usepackage{subcaption}
\usepackage{url}
\usepackage{enumitem, amsmath, amsthm, amssymb,mathtools, bbm}
\usepackage{algorithm}
\usepackage{algpseudocode}
\usepackage{titling}
\usepackage{comment}
\usepackage{todonotes}

\usepackage{hyperref}
\usepackage{cleveref}
\usepackage{thmtools,thm-restate}

\newtheorem{thm}{Theorem}[section]
\newtheorem{lemma}[thm]{Lemma}
\newtheorem{cor}[thm]{Corollary}
\newtheorem{remark}[thm]{Remark}
\newtheorem{defn}[thm]{Definition}
\newtheorem{claim}[thm]{Claim}

\def\polylog{\operatorname{polylog}}
\def\poly{\operatorname{poly}}

\DeclareMathOperator*{\Tr}{Tr}
\DeclareMathOperator*{\Var}{Var}
\DeclareMathOperator*{\coeff}{coeff}

\DeclareMathOperator*{\argmax}{arg\,max}
\DeclareMathOperator*{\E}{\mathbb{E}}
\newcommand{\R}{\mathbb{R}}
\newcommand{\1}{\mathbbm{1}}
\newcommand{\e}{\epsilon}
\newcommand{\sketch}{\textsc{Sketch}}
\newcommand{\sk}{\textsc{Sk}}
\newcommand{\dotproduct}{\textsc{IsCovered}}
\newcommand{\findcluster}{\textsc{FindNeighbors}}
\newcommand{\findclustermeans}{\textsc{Preprocessing}}
\newcommand{\randomwalks}{\textsc{RandomWalks}}
\newcommand{\findk}{\textsc{ApproxK}}

\newcommand{\tdelta}{20 \cdot \frac{\epsilon}{\varphi^2}\cdot \log n + 2 \log(\varphi^2/\epsilon)}

\title{Spectral Clustering in Birthday Paradox Time}
\author{Michael Kapralov\\EPFL \and Ekaterina Kochetkova \\EPFL \and Weronika Wrzos-Kaminska \\EPFL}
\date{}

\begin{document}
\begin{titlingpage}

\maketitle
\begin{abstract}
\noindent
    Given a vertex in a $(k, \varphi, \epsilon)$-clusterable graph, i.e.\,a graph whose vertex set can be partitioned into a disjoint union of $\varphi$-expanders of size $\approx n/k$ with outer conductance bounded by $\epsilon$, can one quickly tell which cluster it belongs to? This is a classical question going back to the expansion testing problem of Goldreich and Ron'11 (the case of $k=2$) that has received a lot of attention in the literature. For $k=2$  a sample of $\approx n^{1/2+O(\epsilon/\varphi^2)}$ logarithmic length walks from a given vertex approximately determines its cluster membership by the birthday paradox: two vertices whose random walk samples are `close' are likely in the same cluster, and otherwise in different clusters.   

    The study of the general case $k>2$ was initiated by Czumaj, Peng and Sohler [STOC'15], and the works of Chiplunkar et al. [FOCS'18], Gluch et al. [SODA'21] showed that $\approx \text{poly}(k)\cdot n^{1/2+O(\epsilon/\varphi^2)}$ random walk samples, and the same query time, suffices for general $k$. This matches the $k=2$ result up to polynomial factors in $k$, but one notices a conceptual inconsistency: if the birthday paradox is indeed the phenomenon guiding the query complexity here, then the query complexity should {\em decrease}, as opposed to increase, with the number of clusters $k$! Since clusters have size $\approx n/k$, we expect to need $\approx (n/k)^{1/2+O(\epsilon/\varphi^2)}$ random walk samples, which gets smaller when $k$ gets larger, and reduces to constant when $k\approx n$. The currently best known query time (of Gluch et al. [SODA'21]), however, increases with $k$ due to computationally heavy linear-algebraic post-processing of random walk samples. 

     In this paper we design a novel representation of vertices in a $(k, \varphi, \epsilon)$-clusterable graph by a mixture of samples of logarithmic length walks. This representation not only uses the optimal $\approx (n/k)^{1/2+O(\epsilon/\varphi^2)}$ number of walks per vertex, but also allows for fast nearest neighbor search: given a collection of $\approx k$ vertices representing the clusters, and a query vertex $x$, we can find the cluster of $x$, using nearly linear time in the representation size of $x$. This gives a spectral clustering oracle with query time $\approx (n/k)^{1/2+O(\epsilon/\varphi^2)}$ and space complexity $k\cdot (n/k)^{1/2+O(\epsilon/\varphi^2)}$, matching the birthday paradox bound.  
\end{abstract}
\end{titlingpage}

\setcounter{page}{0}
\tableofcontents

\newpage

\section{Introduction}
 Graph clustering is a central problem in data analysis, and its applications span  a wide range of domains, from data mining to social science, statistics and more. The goal of graph clustering is to partition
the vertex set of the graph into disjoint ``well-connected'' subgraphs which are sparsely connected to each other. In this paper we focus on a popular version of the problem where the input graph admits a good planted clustering, and the algorithm's task is to recover this clustering efficiently. Formally, we assume that the input graph $G=(V, E)$ admits a partitioning of the vertex set $V$ into clusters $C_1, C_2, \ldots, C_k$ of size $\approx n/k$ such that individual clusters $C_i$ induce well-connected subgraphs (i.e., $\varphi$-expanders) and have a sparse boundary (i.e., the cuts $\partial C_i$ are $\epsilon$-sparse).  Such graphs are referred to as $(k, \varphi, \epsilon)$-clusterable, and are a natural worst-case analog of the stochastic block model \cite{DBLP:journals/jmlr/Abbe17}. Our goal in this paper is to design efficient algorithms that, given access to a $(k, \varphi, \epsilon)$-clusterable graph $G$, recover an approximate clustering $\widehat{C}_1, \widehat{C}_2,\ldots, \widehat{C}_k$ that approximates $C_1, C_2, \ldots, C_k$ well. Formally, we would like to have small misclassification rate, i.e.\,there must exist a permutation $\pi:[k]\to [k]$ such that 
\begin{equation}\label{eq:ms-rate}
\sum_{i\in [k]} |C_i\triangle \widehat{C}_{\pi(i)}| \approx \epsilon' \cdot n
\end{equation}
for some small $\epsilon'\approx \epsilon$\footnote{One cannot achieve a better than $O(\epsilon)$ misclassification rate since a given $(k, \varphi, \epsilon)$-clusterable graph may admit clusterings that disagree on an $\Omega(\epsilon)$ fraction of vertices. The best known result, due to \cite{GKLMS21},  achieves misclassification rate $O(\epsilon)$ while running in exponential time in $k$, the work of \cite{DBLP:conf/soda/Sinop16} achieves $O(\sqrt{\epsilon})$  misclassification rate while running in polynomial time in $k$ and $n$ (where we assume that $\varphi=\Omega(1)$ for simplicity).}. The classical approach to achieving \eqref{eq:ms-rate} is spectral clustering. In spectral clustering one first computes, for every $x\in V$, the {\em spectral embedding} $f_x\in \mathbb{R}^k$, thereby reducing graph clustering to the more manageable problem of clustering the embedded vertices $\{f_x\}_{x\in V}$  in $\mathbb{R}^k$. Indeed, one can show that for most pairs of vertices $x, y\in V$ that belong to the same cluster the dot product $\langle f_x, f_y\rangle$ is large and for most pairs that belong to different clusters the dot product $\langle f_x, f_y\rangle$ is low (see \Cref{sec:tech-overview} for more details). The embedding $f_x$ is exactly the $x$-th column of the $k\times n$ matrix of $k$ smallest eigenvectors of the Laplacian, and can therefore be found using an SVD. This is a global computation, however, which is prohibitively expensive for large graphs, and in this paper we would like to design a sublinear time algorithm for clustering. The central goal of this paper is:
\vspace{0.05in}

\fbox{
\parbox{0.9\textwidth}{
\begin{center}
Design a {\em locally computable} version $\tilde f_x$ of the spectral embedding $f_x$ of $x\in V$.
\end{center}
}
}
\vspace{0.05in}

What are the requirements that such an embedding $\{\widetilde{f}_x\}_{x\in V}$ should satisfy? In a nutshell, we would like to ensure that $\{\widetilde{f}_x\}_{x\in V}$ enable the following very simple (and idealized) algorithm for clustering $(k, \varphi, \epsilon)$-clusterable graphs. 

\vspace{-0.1in}

\begin{minipage}[t]{0.45\textwidth}
\begin{algorithm}[H]  \caption{$\textsc{IdealizedPreprocessing}(G)$}\label{alg:ideal_prepro}
    \begin{algorithmic}
        \State $S \gets$ sample one vertex from each cluster 
        \State 
        \State Compute $\widetilde{f}_y$ for all $y \in S$
        \State \Return $\{ \widetilde{f}_y\}_{y \in S}$
    \end{algorithmic}
\end{algorithm}
\end{minipage}
\hspace{0.07\textwidth}
\begin{minipage}[t]{0.45\textwidth}
\begin{algorithm}[H]
    \caption{$\textsc{IdealizedQuery}(x,\{\widetilde{f}_y\}_{y \in S} )$}\label{alg:ideal_query}
    \begin{algorithmic}
        \State \phantom{$S \gets$ sample one vertex from each cluster }
        \State 
        \State Compute $\widetilde{f}_x$ 
        \State \Return $\textsc{NearestNeighbor}(\widetilde{f}_x,\{\widetilde{f}_y\}_{y \in S})$
    \end{algorithmic}
\end{algorithm}
\end{minipage}

\vspace{0.1in}

Our idealized preprocessing algorithm (Algorithm \ref{alg:ideal_prepro}) simply samples a vertex from every cluster\footnote{Of course an actual algorithm cannot sample a vertex per cluster, since it does not know the clusters! Our actual preprocessing algorithm gets around this issue by using a filtering step -- see Algorithm \ref{alg:find_cluster_means}.}, computes the embedding $\widetilde{f}_y$ for every $y$ in the sample $S$ and stores these embeddings. Thus, the total space used by the data structure is $k$ times the size of a single embedding. Our idealized query primitive (Algorithm \ref{alg:ideal_query}), given a query $x$, simply returns the closest point $y$ in the sample $S$ as the answer. In order for queries to be efficient it is crucial that our embeddings admit a fast nearest neighbor search primitive $\textsc{NearestNeighbor}(\widetilde{f}_x,\{\widetilde{f}_y\}_{y \in S} )$ that quickly returns 
\begin{equation}
    \argmax_{y \in S} \left \langle\widetilde{f}_x,  \widetilde{f}_y \right \rangle.
\end{equation}

Classical works on expansion testing \cite{KaleS08,NachmiasS10,DBLP:books/sp/goldreich2011/GoldreichR11,DBLP:journals/cpc/CzumajS10,KalePS13} can be viewed as instantiating the above framework with 
\begin{equation}\label{eq:wtfx-2}
\widetilde{f}_x=\textsc{Sample}\left(M^t\mathbbm{1}_x,  n^{1/2+O(\epsilon/\varphi^2)}\right),
\end{equation}

where $M$ is the transition matrix of the lazy random walk on $G$, the number of steps $t$ satisfies $t=(C/\varphi^2)\log n$ for a sufficiently large constant $C$, and for a distribution $p$ and an integer $s$ the function $\textsc{Sample}(p, s)$ outputs the empirical distribution obtained by taking $s$ independent samples from $p$. 

\paragraph{The size of $\widetilde{f}_x$ for $k>2$.} The embeddings $f_x$ defined by \eqref{eq:wtfx-2} have size $n^{1/2+O(\epsilon/\varphi^2)}$. Here the $n^{1/2}$ factor is due to the tightness of the birthday paradox: since clusters $C_1, C_2$ have size $\approx n/2$, a significantly smaller than $\sqrt{n/2}$ number of random walks will mostly result in $\langle \widetilde{f}_x, \widetilde{f}_y\rangle=0$ for most pairs $x, y$, thereby giving no information about cluster structure. In the case of $k>2$ clusters the same argument based on the birthday paradox suggests that one should expect $\widetilde{f}_x$ to have size $\approx (n/k)^{1/2+O(\epsilon/\varphi^2)}$, the main factor being the square root of cluster sizes. 
A construction of embeddings $\{\widetilde{f}_x\}_{x\in V}$ that yield constant misclassification rate for constant $
\epsilon$ as per~\eqref{eq:ms-rate} has been constructed in the literature \cite{GKLMS21}, their size in fact {\bf increases} with $k$, whereas birthday paradox type reasoning above suggests that it should {\bf decrease}! This is due to a key inefficiency of the techniques of \cite{GKLMS21} that we elaborate on in \Cref{sec:tech-overview}. This leads us to the question:

\vspace{0.05in}

\fbox{
\parbox{0.9\textwidth}{
\begin{center}
Can one design embeddings $\{\widetilde{f}_x\}_{x\in V}$ of size $\approx (n/k)^{1/2+O(\epsilon/\varphi^2)}$?
\end{center}
}
}
\vspace{0.05in}

While the birthday paradox suggests that the leading factor in the size of the embedding should be $(n/k)^{1/2}$,
the additional $(n/k)^{O(\epsilon/\varphi^2)}$ factors are necessary as per \cite{chiplunkar2018testing}, where a lower bound of $n^{1/2+\Omega(\epsilon)}$ on the query complexity of distinguishing between an expander (i.e., a $1$-clusterable graph) and a $(2, \Omega(1),\epsilon)$-clusterable graph is shown. Adapting the lower bound to the setting of $k$ clusters results in  $(n/k)^{O(\epsilon/\varphi^2)}$ -- see \Cref{sec:approx-k} for more details.  

\paragraph{Efficient nearest neighbor search?} How efficient can the \textsc{NearestNeighbor} procedure from Algorithm \ref{alg:ideal_query} be? Supposing that the size of embeddings is $\approx (n/k)^{1/2+O(\epsilon/\varphi^2)}$, it takes $\approx (n/k)^{1/2+O(\epsilon/\varphi^2)}$ time to evaluate the dot product $\langle \widetilde{f}_x, \widetilde{f}_y\rangle$. So a trivial linear scan of the points in $S$ leads to $k\cdot (n/k)^{1/2+O(\epsilon/\varphi^2)}$ runtime for \textsc{NearestNeighbor}, whereas the minimum required time is $(n/k)^{1/2+O(\epsilon/\varphi^2)}$, the size of the representation of $\widetilde{f}_x$. 

\vspace{0.05in}

\fbox{
\parbox{0.9\textwidth}{
\begin{center}
Can  \textsc{NearestNeighbor} be made to run in time linear in the size of $\widetilde{f}_x$, i.e. $\approx (n/k)^{1/2+O(\epsilon/\varphi^2)}$ time? 
\end{center}
}
}

\paragraph{Our contribution.} Our main contribution is an affirmative answer to both of the questions above: we construct  locally computable embeddings $\{\widetilde{f}_x\}$  of size $\approx (n/k)^{1/2}\cdot n^{O(\epsilon/\varphi^2 \log (1/\epsilon))}$, and design a \textsc{NearestNeighbor} primitive running in time $\approx (n/k)^{1/2}\cdot n^{O(\epsilon/\varphi^2 \log(1/\epsilon))}$ that achieves misclassification rate $\epsilon'=\epsilon^{\Omega(1)}$ as per \eqref{eq:ms-rate} when used in the framework above (Algorithm \ref{alg:ideal_prepro} and Algorithm \ref{alg:ideal_query}). 

A very simple instantiation of our techniques yields a nearly query optimal algorithm for approximating the number $k$ of clusters: 
\begin{thm}[Approximating the number of clusters; informal version of \Cref{thm:find_k}]\label{thm:informal_k}
The number $k$ of clusters in a $(k, \varphi, \epsilon)$-clusterable graph can be $(1+\epsilon^{\Omega(1)})$-approximated in time $\approx (n/k)^{1/2} \cdot n^{O(\epsilon/
\varphi^2 \log(1/\epsilon))}$. 
\end{thm}
All prior works that were able to approximate $k$ even when $\epsilon$ is merely a small constant (as opposed to, say, vanishing with $n$) rely on expensive linear algebraic processing, and therefore incur at least $k^2\cdot (n/k)^{1/2+\Omega(\epsilon)}$ runtime -- Section~\ref{sec:construct-embeddings} for more details. 

Furthermore, for $k\leq n^{0.999}$, our result in Theorem~\ref{thm:informal_k} above matches, up to the $O(\log(1/\epsilon))$ factor in the exponent,  a lower bound of $(n/k)^{1/2+\Omega(\epsilon)}$ that follows by an application of the result of~ \cite{chiplunkar2018testing} -- see Section~\ref{sec:approx-k} for more details. 

\begin{thm}[Clustering oracle; informal version of \Cref{thm:main}]\label{thm:informal_main}
    There exists a clustering oracle with  misclassification rate $\epsilon^{\Omega(1)}$ that has
    \begin{itemize}
        \item preprocessing time and space complexity $\approx (nk)^{1/2}\cdot n^{O(\epsilon/
\varphi^2 \log(1/\epsilon))}$
        \item query time $\approx (n/k)^{1/2}\cdot n^{O(\epsilon/
\varphi^2 \log(1/\epsilon))}$
    \end{itemize}
\end{thm}
More generally, we can obtain a tradeoff between the query time and the space complexity:
\begin{restatable}[Space/query tradeoffs for clustering oracle; informal]{thm}{tradeoffs}\label{thm:tradeoff}
    For every $\delta \in [0,1]$, there exists a clustering oracle with misclassification rate $\epsilon^{\Omega(1)}$ that 
    \begin{itemize}
        \item has preprocessing time $\approx \left(n^{1/2}k^{1/2}+ k \cdot \left( \frac{n}{k}\right)^{1-\delta}\right)\cdot n^{O(\epsilon/
\varphi^2 \log(1/\epsilon))}.$
    \item computes a data structure of size $\approx k \cdot \left( \frac{n}{k}\right)^{1-\delta}\cdot n^{O(\epsilon/
\varphi^2 \log(1/\epsilon))}$
\item has query time $\approx \left( \frac{n}{k}\right)^{\delta}\cdot n^{O(\epsilon/
\varphi^2 \log(1/\epsilon))}$
    \end{itemize}

Thus, the product of the size of the data structure and the query time is $\approx n^{1+O(\epsilon/\varphi^2 \log(1/\epsilon))}$, {\bf independent of the number of clusters $k$}.
\end{restatable}

\begin{remark}
The best previous tradeoff of this nature is due to \cite{GKLMS21} (Theorem 3), which requires the product of the size of the data structure and query time to be increasing polynomially in $k$ (see \Cref{sec:tech-overview} for a more detailed discussion of the cause of this inefficiency). Our tradeoffs only require the product of the size of the data structure and query time to be $\approx n$, {\bf independent of $k$}!
\end{remark}

\begin{remark}
    \Cref{thm:tradeoff} is stated informally to illustrate the space–query tradeoffs achievable by our techniques. We do not include a full formal statement or proof, since the argument follows the same principles as our main analysis. See \Cref{sec:tradeoffs} for more details, and \Cref{sec:tradeoff_sketch} for a proof sketch and discussion of the required modifications.
\end{remark}
\paragraph{Prior work.} The classical work of Goldreich and Ron \cite{DBLP:books/sp/goldreich2011/GoldreichR11} on expansion testing essentially considered the setting $\e=0$ and used 
\eqref{eq:wtfx-2}. Follow-up works \cite{KaleS08,NachmiasS10,DBLP:books/sp/goldreich2011/GoldreichR11,DBLP:journals/cpc/CzumajS10,KalePS13} resulted in clustering oracles that work for $k=2$ and any small constant $\epsilon>0$ using \eqref{eq:wtfx-2}. The work of \cite{czumaj2015testing} was the first to consider the case of general $k$, but only for $\epsilon\ll \frac1{\text{poly}(k) \log n}$. They showed that the same choice 
\eqref{eq:wtfx-2} works as long as the number of samples is increased by a factor polynomial in $k$. The case of general $\epsilon$ and general $k$ required new techniques: dot product, or distance, based classifiers on the embedding \eqref{eq:wtfx-2} fail to provide enough information when the number of clusters $k$ is larger than $2$ \cite{chiplunkar2018testing}. To address this difficulty, \cite{GKLMS21} introduced a collection of linear algebraic tools in the post-processing of random walk distributions. Unfortunately, this comes at a cost of polynomial in $k$ blowup in runtime. For example, for the embeddings $(\widetilde{f}_x)_{x\in V}$ of \cite{GKLMS21}  even computing $\|\widetilde{f}_x\|^2_2$ requires at least quadratic overhead in $k$ -- see \Cref{sec:tech-overview} for more details. 

\paragraph{Comparison to \cite{GKLMS21}.} The work of \cite{GKLMS21} achieves the following tradeoffs. They show that for every $\delta \in (0,1/2]$ there exists a spectral clustering oracle with misclassification rate $O(\epsilon \log k)$, preprocessing time $2^{\frac{1}{\epsilon}k^{O(1)}}n^{1-\delta+O(\epsilon/\varphi^2)}$, space complexity $\approx k^{O(1)}n^{1-\delta+O(\epsilon/\varphi^2)}$ and query time $\approx k^{O(1)}n^{\delta + O(\epsilon/\varphi^2)}$. In particular, the product of space complexity and query time is $\approx k^{O(1)} \cdot n$. The $k^{O(1)}$ overhead is inherent to their approach, as their spectral embeddings are obtained from a ``global" computation (see \Cref{sec:tech-overview} for more details). Therefore, we cannot use their tools and we need to develop a new approach. 

\paragraph{Matrix polynomials.} Our embeddings are based on the idea of applying a polynomial to the random walk transition matrix to turn the random walk embedding \eqref{eq:wtfx-2} into essentially the spectral embedding $f_x$. The polynomial is carefully chosen to ensure that coefficients are small, namely, $n^{O(\epsilon/\varphi^2\log(1/\epsilon))}$, which allows to estimate dot products between the resulting embeddings efficiently. The idea of using matrix polynomials to modify eigenvalues of matrices is classical. It is typical in numerical linear algebra, however, to operate with polynomial precision, and therefore the magnitude of coefficients being polynomial is typically enough\footnote{Rather, one often designs recursive evaluation primitives and shows that these are numerically stable.}.  To the best of our knowledge, however, our work is the first to apply these techniques in the context of sublinear clustering algorithms. In particular, unlike in typical numerical linear algebra applications, the magnitude of coefficients of our polynomial is crucial: it translates directly into the overhead relative to the main $(n/k)^{1/2}$ factor, and nearly matches (up to a $\log(1/\epsilon)$ factor in the exponent) the information theoretically optimal bound.

\paragraph{Additional related work.} Recent works \cite{jha_et_al:LIPIcs.ICALP.2024.91}, \cite{DBLP:conf/soda/0001Y23} used the clustering oracle of \cite{GKLMS21} to design testers for MAX-CUT as well as more general constraint satisfaction problems (CSPs). The recent work of \cite{DBLP:conf/nips/Shen023} improved the preprocessing time of \cite{GKLMS21} from exponential in $k$ to polynomial in $k$, but their algorithm only applies when $\epsilon$ is polynomially small in $k$. Robust versions of the $(k, \varphi, \epsilon)$-clusterable model are studied in \cite{DBLP:conf/focs/KalePS08} and, more recently,~\cite{Peng20}.

\paragraph{Other related work.} Multiway Cheeger's inequalities (e.g.,~\cite{lee2014multiway}, \cite{KwokMultiwaySpectrumGap}) can be used to recover the clusters to within a good precision, but they run in at least linear time in the size of the graph. Another closely related area is local clustering. In local clustering one is interested in finding the entire cluster around a node $v$ in time proportional to the size of the cluster. Several algorithms are known for this problem~\cite{DBLP:journals/im/AndersenCL08, DBLP:journals/jacm/AndersenGPT16, DBLP:conf/soda/OrecchiaZ14, DBLP:journals/siammax/SpielmanT14, DBLP:conf/icml/ZhuLM13} but unfortunately they cannot be applied to solve our problem because when the clusters have linear size they take linear time (in addition, the output clusters may overlap). In this paper instead we focus on solving the problem using strictly sublinear time.

\section{Technical overview} \label{sec:tech-overview}
We begin by describing the graph access model. 
\paragraph{Graph access model:} We work with the \emph{bounded degree graph model}, in which we can specify a vertex $ x \in  V$ and a number $i \in [d]$, and access the $i$-th neighbor of $x$ in constant time. 
\vspace{1em}

Next, we define the notion of conductance, which characterizes the quality of a cluster. 

\begin{defn}[Inner and outer conductance] \label{def:conductance}
Let $G = (V, E) $ be a graph. For a set $C\subseteq V$ and a set $S\subseteq C$, let $E(S,C\setminus S)$ be the set of edges with one endpoint in $S$ and the other in $C\setminus S$. The \textit{conductance of $S$ within $C$} is $ \phi^G_C(S)=\frac{|E(S, C\setminus S)|}{d|S|} $. The \textit{inner conductance} of $C\subseteq V$ is defined to be
$$\phi_{\text{in}}^G(C)=\min_{S\subseteq C\text{,} 0<|S|\leq \frac{|C|}{2}}\phi^G_C(S).$$   The \textit{outer conductance} of $C$ is defined to be $\phi^G_{\text{out}}(C)=\phi_V^G(C)=\frac{|E(C,V\setminus C)|}{d|S|}\text{.}$ 
\end{defn}

\begin{remark}
We present our proofs for $d$-regular graphs, but the result also holds for $d$-bounded graphs, with the same definition of conductance as in \Cref{def:conductance}, i.e. normalized by $d|S|$ instead of the volume of $S$. This is because we can convert a $d$-bounded graph $G$ into a $d$-regular graph $G^{reg}$ by adding $d-deg(x)$ self-loops to each vertex $x \in V$. Then, the random walks on $G$ are equivalent to lazy random walks on $G^{reg}$.
\end{remark}
We work with the following notion of a clusterable graph:
\begin{defn}[$(k,\varphi,\epsilon)$-clusterable graph]\label{def:clusterable}
Let $G=(V,E)$ be a $d$-regular graph. A $(k,\varphi,\epsilon)$-clustering of $G$ is a partition of vertices $V$ into disjoint subsets $C_1, \ldots, C_k$ such that for all $i\in [k]$, $\phi_{\text{in}}^G(C_i)\geq\varphi$, $\phi_{\text{out}}^G(C_i)\leq\epsilon$ and for all $i,j \in [k]$ one has $\frac{|C_i|}{|C_j|} \in O(1).$ 
A graph $G$ is called $(k,\varphi,\epsilon)$-clusterable if there exists a $(k,\varphi,\epsilon)$-clustering of $G$.
\end{defn}
Given a clustering of $G$ and a vertex $x \in V$, we use $i(x) \in [k]$ to denote the label such that $x \in C_{i(x)}$. We refer to $C_{i(x)}$ as the cluster of $x$. We define $\eta \coloneqq \max_{i,j} \frac{|C_i|}{|C_j|}$ to be the maximum ratio between cluster sizes. Note that by the assumption that $\frac{|C_i|}{|C_j|} \in O(1)$ for all $i,j$, we have $\eta = O(1)$, and for all $i \in [k]$, it holds that $\frac{1}{\eta} \cdot \frac{n}{k} \leq |C_i| \leq \eta \cdot \frac{n}{k}$. We note that this assumption is common in prior work, e.g., in \cite{GKLMS21}. 

\begin{defn}[Spectral Clustering Oracle]
A randomized algorithm $\mathcal{O}$ is a $(k, \varphi, \epsilon)$-clustering oracle if, when given query access to a $d$-regular graph $G = (V, E)$ that admits a $(k, \varphi, \epsilon)$-clustering $C_1, \ldots, C_k$, 
the algorithm $\mathcal{O}$ provides consistent query access to a partition $\widehat{P} = (\widehat{C}_1, \dots,  \widehat{C}_k)$ of $V$. Moreover, with probability at least $9/10$ over the random bits of $\mathcal{O}$,  the partition $\widehat{P}$ has the following property: for some permutation $\pi: [k] \to [k]$, 
$$ \sum_{i\in [k]} |C_i\triangle \widehat{C}_{\pi(i)}| \leq \left(\frac{\epsilon}{\varphi^2} \right)^{\Omega(1)}n.$$
\end{defn}

Next, we formally define the spectral embedding. 

\begin{defn}[Spectral embedding]\label{def:spec-emb}
Let $U_{[k]} \in \mathbb{R}^{n\times k}$ denote the matrix whose columns are the bottom $k$ eigenvectors of the Laplacian $\mathcal{L}$. For every vertex $x \in V$, the spectral embedding of $x$, denoted $f_x \in \mathbb{R}^k$ is the $x$-th row of $U_{[k]}$, i.e.  $f_x\coloneqq U_{[k]}^\top   \1_x $. 
\end{defn}
\begin{remark}
We note that the spectral embeddings $f_x$ are not uniquely defined. However, the dot products $\langle f_x, f_y \rangle $ are uniquely defined for any $(k, \varphi, \epsilon)$-clusterable graph with $\epsilon/\varphi^2$ smaller than an absolute constant. (See \Cref{remark:spectral_unique}).
\end{remark}

One of the key properties of the spectral embedding is that it concentrates well around its cluster mean:

\begin{equation}\label{eq:muis-def}
\mu_i \coloneqq \frac{1}{|C_i|} \sum_{x \in C_i} f_x
\end{equation}

For almost all $x \in V$ (see \Cref{lemma:close_to_clutermean} for the formal quantification), the embedding satisfies the following concentration bound:
\begin{equation}\label{eq:B_delta_informal}
    \|f_x - \mu_{i(x)}\|^2_2 \leq \left(\frac{\epsilon}{\varphi^2}\right)^{\Omega(1)}\|\mu_{i(x)}\|_2^2. 
\end{equation}

The cluster means~\eqref{eq:muis-def} of $(k, \varphi, \epsilon)$-clusterable graphs exhibit several useful structural properties. Notably, their squared $\ell_2$ norms are close to the inverse of the cluster size, and, after appropriate scaling, they form an almost orthonormal basis in $\mathbb{R}^k$:

\begin{lemma}[Lemma 7 from \cite{GKLMS21}]\label{lemma:clustermeans} Let $k \geq 2$ be an integer, $\varphi \in (0,1)$, and $\epsilon \in (0,1)$. Let $G = (V,E)$ be a $d$-regular graph that admits $(k, \varphi, \epsilon)$-clustering $C_1, \dots, C_k$. Then we have

\begin{enumerate}[label=(\textbf{\arabic*})]
    \item \textit{for all $i \in [k]$,}
    \[
    \left| \|\mu_i\|_2^2 - \frac{1}{|C_i|} \right| \leq \frac{4 \sqrt{\epsilon}}{\varphi} \frac{1}{|C_i|}\label{bulletpt:mu_norm}
    \]
    \item \textit{for all $i \neq j \in [k]$,}
    \[
    |\langle \mu_i, \mu_j \rangle| \leq \frac{8 \sqrt{\epsilon}}{\varphi} \frac{1}{\sqrt{|C_i||C_j|}}\label{bulletpt:mu_i_dot_mu_j}.
    \]
\end{enumerate}
\end{lemma}

\begin{remark}\label{rem:||mu_i||}
    Note that {\bf \ref{bulletpt:mu_norm}} implies that $\|\mu_i\|^2_2 = \Theta(k/n)$ for all $i \in [k]$.
\end{remark}

\begin{lemma}[Lemma 9 from \cite{GKLMS21}]\label{lemma:mu_identity}
    For all $\alpha \in \mathbb{R}^k$, it holds that
    \begin{equation*}
    \left| \alpha^\top  \left( \sum_{i=1}^k |C_i| \mu_i \mu_i^\top  -I\right)\alpha\right| \leq \frac{4 \sqrt{\epsilon}}{\varphi} \|\alpha\|^2_2.
    \end{equation*}
\end{lemma}
Consequently, by \Cref{eq:B_delta_informal} and \Cref{lemma:clustermeans}, the inner products of spectral embeddings can be used to test if two vertices belong to the same cluster. For vertices $x, y$ in the same cluster $C_i$, we have 
\begin{equation}\label{eq:tech_samecluster}
    \langle f_x , f_y \rangle \approx \langle \mu_i , \mu_i\rangle = \|\mu_i\|^2_2 \approx \frac{k}{n}.
\end{equation}

Conversely, for vertices belonging to different clusters, say $x \in C_i$ and $y \in C_j$, $i\neq j$, we have
\begin{equation}\label{eq:tech_differentcluster}
    \langle f_x , f_y \rangle  \approx \langle \mu_i, \mu_j \rangle \leq \left(\frac{\epsilon}{\varphi^2} \right)^{\Omega(1)}\|\mu_i\|_2 \|\mu_j\|_2 \approx \left(\frac{\epsilon}{\varphi^2} \right)^{\Omega(1)}\frac{k}{n} \ll \frac{k}{n}.
\end{equation} 
Equations~\eqref{eq:tech_samecluster} and \eqref{eq:tech_differentcluster} show that vertices in the same cluster have large dot products, while those in different clusters have small ones. This key property of spectral embeddings justifies our idealized framework (Algorithms \ref{alg:ideal_prepro} and \ref{alg:ideal_query}).

\subsection{Construction of embeddings $\widetilde{f}_x$}\label{sec:construct-embeddings}
Recall that our goal is to construct efficiently computable embeddings $\widetilde{f}_x$ that can be used to classify vertices using a version of \textsc{IdealizedPreprocessing} (Algorithm~\ref{alg:ideal_prepro}) and \textsc{IdealizedQuery} (Algorithm~\ref{alg:ideal_query}). We now present our construction and outline the main ideas that underlie its analysis.

We design our embeddings so that they can be used to approximate the dot products of the true spectral embeddings to within an additive error  
 of $\approx \frac{k}{n}$, i.e. for (almost) all $x,y \in V$, they must satisfy 
\begin{equation}\label{eq:dotproduct}
    \left| \langle\widetilde{f}_x, \widetilde{f}_y \rangle- \langle f_x,f_y \rangle  \right| \leq \xi \cdot \frac{k}{n},
\end{equation}
for some sufficiently small constant $\xi$. A data structure that approximates dot products in the spectral embedding space as in~\eqref{eq:dotproduct} was introduced in \cite{GKLMS21} under the name of a \emph{spectral dot product oracle}. The spectral dot product oracle of~\cite{GKLMS21} takes $k^{O(1)} n^{1/2 + O(\epsilon/\varphi^2)}$ time and queries to approximate dot products to precision sufficient for clustering. 
The approach in \cite{GKLMS21} involves  computing a proxy for the eigendecomposition of the lazy random walk matrix $M$. To achieve this, they sample a set $S \subseteq V$ of size $\poly(k)$ and compute the eigendecomposition of the Gram matrix $\mathcal G$ of random walk distributions starting from $S$. This inherently incurs a \( \poly(k) \) runtime overhead, as the sample size \( |S| \) must be at least \( k \) just to hit all of the clusters. Therefore, we cannot use their techniques and need to take a different approach. We design new embeddings that give a spectral dot product oracle using only $(n/k)^{1/2 + O(\epsilon/\varphi^2)}$ time and queries, which is close to information theoretically optimal -- see discussion in Section~\ref{sec:approx-k}. We now give the construction of our embedding $\widetilde{f}_x$.
\paragraph{Our embeddings.}
Unlike previous work, which used (a few samples from) the distribution of a $t$-step walk started at $x\in V$ as the embedding of $x$ (see~\eqref{eq:wtfx-2}), we use a carefully designed mixture of distributions of $t$-step walks for $t$ in a range. Specifically, we select a minimum walk length $t_{\min}\approx \frac1{\varphi^2}\log n$, i.e. the mixing time of the clusters, and a parameter $\Delta\approx \epsilon\cdot t_{\min}$. Our embedding adds samples from these distributions with carefully chosen and, importantly, small coefficients. Formally,
\begin{equation}\label{eq:wtfx-general}
\widetilde{f}_x=\bigoplus_{t_{\min}\leq t\leq t_{\min}+\Delta} c_t\cdot \textsc{Sample}\left(M^t\mathbbm{1}_x,  (n/k)^{1/2+O(\epsilon/\varphi^2 \log(1/\epsilon))}\right),
\end{equation}
where $c_t$ are coefficients satisfying $|c_t|=n^{O(\epsilon/\varphi^2 \log(1/\epsilon))}$ and $\bigoplus$ stands for addition of (sparse) vectors of length $n$.  The coefficients $c_t$ are coefficients of a polynomial $p$ such that the application of $p$ to the lazy random walk matrix $M$ of $G$ is essentially equivalent to projecting $M$ onto its eigenspace corresponding to eigenvalues greater than $1-\epsilon$. 

\paragraph{Constructing the polynomial $p$.} Let $M = U \Sigma U^\top $ denote the eigendecomposition of the lazy random walk matrix $M$. Recall (from Definition~\ref{def:spec-emb}) that the columns of $U_{[k]}$ are eigenvectors with eigenvalues at least $1-\epsilon$ of $M$, while the remaining $n-k$ eigenvectors of $M$ have eigenvalues at most $1-\varphi^2/4$ (see \Cref{remark:Meigengap}). Thus, if we have a polynomial 
$$
p(x) = \sum c_t x^t
$$ 
mapping the interval $\left[0,1-\varphi^2/4\right]$ close to $0$ and the interval $[1-\epsilon,1]$ close to $1$, then the matrix $$
p(M) = U p(\Sigma) U^\top $$ approximates the projection on the space spanned by the columns of $U_{[k]}$. Consequently, we obtain 
\begin{equation}\label{eqn:tech_poly}
    \langle f_x, f_y\rangle = \1_x^\top U_{[k]}U^\top _{[k]}\1_y  \approx \1_x^\top  U p(\Sigma)^2U^\top  \1_y  = \left \langle p(M)\1_x, p(M)\1_y\right  \rangle  = \left \langle  \sum_t c_t M^t \1_x ,  \sum_t c_t M^t \1_y\right \rangle. 
\end{equation}
Furthermore, if the coefficients of $p$ are small and only nonzero for $t$ larger than $\approx \frac1{\varphi^2}\log n$, the mixing time of the clusters, then the right hand side of \Cref{eqn:tech_poly} can be efficiently approximated using a few random walks as per the birthday paradox. 

Our first contribution is to construct such a polynomial $p$ with the desired properties.   

\begin{restatable}{thm}{standardbasis}\label{thm:standardbasis}
For every $\epsilon, \varphi \geq 0$ satisfying $ \epsilon/\varphi^2 \log(1/\epsilon) \leq c_1$, $\log(1/\epsilon) \log(\varphi^2/\epsilon) \leq c_2 \log n$
 for sufficiently small absolute constants $c_1, c_2$, and for every $t \geq 20 \frac{\log n}{\varphi^2}$, there exists a polynomial $p$ of the 
     form $$p (x) = x^t \cdot q(x),$$ where $q$ is a polynomial of degree $\deg{q} = O(\epsilon t + \log(\varphi^2/\epsilon))$, such that the coefficients of $p$ are bounded by $(1/\epsilon)^{O(\epsilon t + \log(\varphi^2/\epsilon))}$ in absolute value and
\begin{itemize}
    \item $|p(x)-1| \leq \epsilon/\varphi^2$ for $x\in [1-\epsilon, 1]$
    \item $|p(x)|\leq  n^{-4}$  for $x \in [0, 1 - \varphi^2/4]$. 
\end{itemize}
\end{restatable}

We present the proof of \Cref{thm:standardbasis} in Appendix \ref{sec:chebyshev}.
The main idea is to start with the polynomial $x^t$, which has the required behavior on $[0, 1 - \varphi^2/4]$ and at $x = 1$, but drops too fast as $x$ moves away from $1$, and flatten it out on the interval $[1-\epsilon,1]$. We achieve this by multiplying it by a polynomial $q$ that almost cancels $x^t$ out on this interval. We construct such a polynomial by using the Chebyshev approximation to the function $(1-\epsilon x)^{-t}$. Our analysis of the Chebyshev approximation of $(1-\epsilon x)^{-t}$ is heavily inspired by \cite{AA22}, where the authors study the Chebyshev approximation of $e^x$, and overall seems to follow the standard approach for such Chebyshev approximations. Note that it is crucial that the coefficients of $p$ are small. This is because we approximate  
\begin{equation}\label{eq:tech_expandpoly}
\left \langle p(M)^\top \1_x, p(M)^\top \1_y\right  \rangle  = \left \langle  \sum_t c_t M^t \1_x ,  \sum_t c_t M^t \1_y\right \rangle = \sum_{t,t'} c_t c_{t'}\1_x^\top  M^{t+t'}\1_y  
\end{equation}
by sampling, with a sampling rate that depends polynomially on the coefficient size.  In particular, our polynomial $p$  has coefficients bounded by $(1/\epsilon)^{O(\epsilon t + \log(\varphi^2/\epsilon))} = n^{O(\epsilon/\varphi^2\log(1/\epsilon))}$ for $\epsilon \geq 2^{-o(\sqrt{\log n})}$. We note that similar factors appeared in the query complexity of~\cite{GKLMS21} because of the need to stably invert the Gram matrix of $\approx k$ distributions of $t=t_{\min}$-step random walks started at uniformly random vertices of the graph, whose condition number can only be bounded by $n^{O(\epsilon/\varphi^2)}$, and are necessary due to a lower bound of~\cite{chiplunkar2018testing}. We also note here that the need to form the Gram matrix and apply its inverse imply that the techniques of~\cite{GKLMS21} cannot yield query time faster than $k^2 \cdot \sqrt{n/k}$.

\paragraph{Runtime of our dot product oracle.}
We show that each of the terms  $\1_x^\top  M^{t+t'}\1_y$ in \Cref{eq:tech_expandpoly} can be approximated to within an additive error of $\approx \frac{k}{n}$ by running $\approx \sqrt{n/k} \cdot n^{ O(\epsilon/\varphi^2\log (1/\epsilon))}$, random walks, as expected from the birthday paradox (see \Cref{lemma:variance_calc} for the precise statement). 
This gives a dot product oracle satisfying
~\eqref{eq:dotproduct} with runtime  $\approx \sqrt{n/k} \cdot n^{ O(\epsilon/\varphi^2\log (1/\epsilon))}$ (up to factors polylogarithmic in $n$), which we note is at most $(n/k)^{1/2+O(\epsilon/\varphi^2\log (1/\epsilon))}$  for $ k \leq n^{0.99}$. This removes all the polynomial in $k$ loss factors from~\cite{GKLMS21}, in particular results in runtime that decreases, as opposed increases, as in~\cite{GKLMS21}, with $k$. 

\subsection{An efficient implementation of \textsc{NearestNeighbor}} 

Given the embeddings $\widetilde{f}_x, \{ \widetilde{f}_y\}_{y \in S}$, the idealized query algorithm (Algorithm \ref{alg:ideal_query}) requires an efficient nearest neighbor search operation that quickly returns
\begin{equation*}
    \argmax_{y \in S} \left \langle\widetilde{f}_x,  \widetilde{f}_y \right \rangle.
\end{equation*}

A natural approach to the nearest neighbor search problem would be to implement Locality-Sensitive Hashing (LSH) \cite{DBLP:conf/stoc/IndykM98} in the embedding space. Given the factor $\epsilon^{\Omega(1)}$ gap between the distances of `typical' points from different clusters and distances between `typical' points within the same cluster (see~\eqref{eq:tech_samecluster} and \eqref{eq:tech_differentcluster}), one would expect LSH to give runtime on the order of $k^{\epsilon^{\Omega(1)}}$ times the size of the embeddings. However, it seems difficult to implement Locality-Sensitive Hashing in embedded space efficiently. Instead, we take a group testing approach, which relies on our ability to efficiently compute inner products $\langle \widetilde{f}_x, \widetilde{f}_y \rangle$ (and, in fact, significantly refines this ability -- see below). Our approach results in runtime nearly linear in the size of the embeddings.

\paragraph{Solving a decision version of the problem.} Consider the sub-problem in which, given a vertex $x \in V$ and a set $S \subseteq V$ containing at most one vertex per cluster, our task is to efficiently determine whether $S$ contains a vertex from the same cluster as $x$. For simplicity, let us first consider an idealized scenario in which we are given the true spectral embeddings $f_x$ and  $\{f_y\}_{y \in S}$. Our approach is to combine the embeddings $f_y$ with random signs $\sigma_y \sim \{-1,1\}$ drawn independently uniformly at random and consider the combined dot product 
$ \left \langle f_x, \sum_y \sigma_y f_y \right \rangle$. We show that for almost all $x \in V$, with a good probability over the randomness of $\sigma$, it holds that
\begin{equation}\label{eq:cases}
     \left| \left \langle f_x, \sum_{y\in S} \sigma_y f_y \right \rangle \right| \approx \begin{dcases} \frac{k}{n},  &\text{if $S$ contains a vertex from the cluster of $x$} \\
     \left(\frac{\epsilon}{\varphi^2}\right)^{\Omega(1)}\frac{k}{n} \ll \frac{k}{n} , & \text{otherwise}
     \end{dcases}
\end{equation}
(see \Cref{lemma:grouptest_exact} for the formal statement). In particular, this can be used to determine whether $S$ contains a vertex from the cluster of $x$. We now sketch the proof of  \eqref{eq:cases}, which relies on an interplay between the spectral structure of $(k, \varphi, \epsilon)$-clusterable graphs and birthday paradox type collision testing techniques. The full proof of \Cref{lemma:grouptest_exact} is presented in \Cref{sec:sketch}. 

Let $C_i$ be the cluster of $x$. We first consider the case where $S$ \emph{does not} contain a vertex from $C_i$. 
Define the random variable $X \coloneqq \left \langle f_x, \sum_{y \in S} \sigma_y f_y \right \rangle $ over the randomness of $\sigma$. Then $\E[X] = 0 $, and using \Cref{eq:B_delta_informal}, we can bound the variance as
\begin{equation*}
    \Var[X]\leq \E[X^2] = \sum_{y \in S} \langle f_x, f_y\rangle^2 \approx \sum_{y \in S}\langle \mu_i, \mu_{i(y)}\rangle^2 \leq \sum_{j \in [k] \setminus \{i\} } \langle \mu_i, \mu_j \rangle^2, 
\end{equation*}
where the last inequality follows by the assumption that $S$ contains at most one vertex per cluster and that $S$ does not contain any vertices from the cluster of $x$. To upper bound  $\sum_{j \in [k] \setminus \{i\} }  \langle \mu_i, \mu_j \rangle^2$, we exploit the spectral structure of $(k, \varphi, \epsilon)$-clusterable graphs. Using Lemma  \ref{lemma:mu_identity} and \Cref{lemma:clustermeans}, we get 
\begin{align*}
    \sum_{j \in [k] \setminus \{ i\} }\langle \mu_j , \mu_i\rangle^2 &  
     \approx \frac{k}{n} \mu_i^\top  \left(\sum_{j \in [k]\setminus \{i\}}|C_j| \mu_j \mu_j^\top  \right)\mu_i \\
     & =  \frac{k}{n} \mu_i^\top  \left(\sum_{j \in [k]}|C_j| \mu_j \mu_j^\top  
-  |C_i| \mu_i \mu_i^\top  \right)\mu_i \\
    & = \frac{k}{n} \mu_i^\top  \left(\sum_{j \in [k]}|C_j| \mu_j \mu_j^\top   - I\right)\mu_i +  \frac{k}{n} \mu_i^\top  \left(I - |C_i| \mu_i \mu_i^\top  \right) \mu_i\\
    & \lesssim \frac{\sqrt{\epsilon}}{\varphi^2}\frac{k}{n} \|\mu_i\|^2_2  + \frac{k}{n}\|\mu_i\|^2_2(1-|C_i| \|\mu_i\|^2_2)&& \text{by \Cref{lemma:mu_identity}} \\
    &\approx  \frac{\sqrt{\epsilon}}{\varphi^2}\frac{k^2}{n^2} && \text{by \Cref{lemma:clustermeans},}
\end{align*}
 which upper-bounds the variance. So by Chebyshev's inequality, with a good constant probability, we get
\begin{equation*}
   \left| \left \langle f_x, \sum_{y\in S} \sigma f_y \right \rangle \right| \approx    \left(\frac{\epsilon}{\varphi^2}\right)^{\Omega(1)}\frac{k}{n}.
\end{equation*}
Now consider the case where $S$ \emph{does} contain a vertex the same cluster as $x$, call this vertex $x^* \in C_i$. Then $ \left \langle f_x, \sum_y \sigma_y f_y \right \rangle$ is dominated by the term $\langle f_x, f_{x^*}\rangle \approx \frac{k}{n}$. An identical argument to the one above shows that with a good constant probability, $ \left| \left \langle f_x, \sum_{y \in S \setminus C_i} \sigma_{y}f_{y}\right   \rangle\right|  \approx \left(\frac{\epsilon}{\varphi^2}\right)^{\Omega(1)}\frac{k}{n} \ll \frac{k}{n}$, which gives
\begin{equation*}
   \left| \left \langle f_x, \sum_y \sigma_y f_y \right \rangle \right| \geq \langle f_x, f_{x^*} \rangle - \left| \left \langle f_x, \sum_{y \in S \setminus C_i} \sigma_{y}f_{y}\right   \rangle\right| \approx \frac{k}{n},
\end{equation*}
as required for \Cref{eq:cases}. 

Note that the random signs $\sigma$ are crucial. Consider for example the case where $f_x = \mu_i + \frac{\sqrt{\epsilon}/\varphi}{\sqrt{k}}\sum_{j\in [k] \setminus \{i\}} \mu_j$ and $S = \left\{ \mu_j : j \in [k] \setminus \{i\}\right\}$, which is perfectly consistent with the spectral bounds that we have on the input graph. Then 
$\|f_x - \mu_i\|^2_2 \approx \frac{\epsilon}{\varphi^2}\frac{k}{n}$ but $\langle f_x, \sum_{y \in S}f_y \rangle \approx \sqrt{k}\frac{\sqrt{\epsilon}}{\varphi}\frac{k}{n} \gg \frac{k}{n}$. In other words, $S$ does not contain any vertices from the cluster of $x$, but the inner product $\langle f_x,\sum_{y \in S}f_y\rangle $ without random signs  fails to capture this. 

\paragraph{Approximating $\left \langle f_x, \sum_{y \in S} \sigma_y f_y \right \rangle$.} By \Cref{eq:cases}, our task reduces to approximating $\left \langle f_x, \sum_{y \in S} \sigma_y f_y \right \rangle$ to within an additive  error of $\approx \frac{k}{n}$. Ideally, we would like to directly use our spectral dot product oracle from \Cref{eq:dotproduct}. However, we cannot apply it as a black box to each term individually, since for a set $S$ of size $|S| = k$, this only gives an error bound of 
\begin{equation*}
\begin{split}
    \left|\left \langle \widetilde{f}_x, \sum_{y\in S} \sigma_y \widetilde{f}_y\right \rangle  - \left \langle f_x, \sum_{y\in S} \sigma_y f_y\right \rangle \right| & = \left| \sum_{y \in S} \sigma_y \left( \left \langle \widetilde{f}_x,  \widetilde{f}_y\right\rangle - \left \langle f_x,  f_y\right\rangle\right)\right|
    \leq \left|\sum_{y \in S}\sigma_y\right|  \cdot \xi \cdot \frac{k}{n}
     \approx \xi \cdot \frac{k \sqrt{k}}{n},
\end{split}
\end{equation*}
which is a factor $\sqrt{k}$ more than we can tolerate. 
Therefore, instead of computing each dot product separately, we approximate the entire sum at once. 
Our core primitive to achieve this is the \textsc{Sketch} procedure presented below. First, for every $y \in S$, we generate samples from the distribution of the $t$-step lazy random walk in $G$ started at every $x\in S$, for $t$ in a range $[t_{\min}, t_{\min}+t_\Delta]$, and recombine these samples with coefficients $c_t$ of the polynomial from \Cref{thm:standardbasis}. This is essentially a form of generating samples from $p(M)\mathbbm{1}_x$. Second, \textsc{Sketch} combines these empirical distributions over $x\in S$ with random signs, where the sign is chosen independently for every $x\in S$.

 \begin{algorithm}[H]
\caption{\sketch$(S)$}
\label{alg:spectralsketch}
    \begin{algorithmic}[1]
    \State $\sigma \sim \text{Unif}\left(\{-1,1\}^{S}\right)$
\State $t_{\min} \gets 20\log n/\varphi^2$  \Comment{shortest walk length}
\State $t_{\Delta} \gets \tdelta$ \Comment{number of different walk lengths}
\State $r \gets O^*\left(\sqrt{\frac{n}{k}} \cdot n^{O(\epsilon/\varphi^2\log(1/\epsilon))}\left(1/\epsilon\right)^{O(\log(\varphi^2/\epsilon))}\right) $  \Comment{number of walks}\label{line:sketch_r}
\For{$ x \in S$}
\For{$t =t_{\min} \text{ to }t_{\min}+t_{\Delta}$}
\State $\widehat{p}^t_x \gets$ \randomwalks$(r, t, x)$ \Comment{vector $\widehat{p}^t_x$ has support at most $r$ }
\EndFor
\EndFor 
\Comment{ $c_t$ is the coefficient of $x^t$ in the polynomial $p(x)$ from Theorem~\ref{thm:standardbasis}}
\end{algorithmic}
\end{algorithm}

The \textsc{RandomWalks} primitive simply generates samples from the distribution of the $t$-step lazy random walk started at the input vertex $x$:

\begin{algorithm}[H]
    \caption{$\randomwalks(r, t, x)$}\label{alg:randwalks}
    \begin{algorithmic}
        \State Run $r$ lazy random walks of length $t$ starting from $x$
        \State Let $\widehat{p}_{x}(y)$  be the fraction of random walks that ends at $y$  \Comment{vector $\widehat{p}_x$ has support at most $r$ 
        \State \Return{} $\widehat{p}_{x}$}
    \end{algorithmic}
\end{algorithm}
We show that for almost all vertices $x \in V$, we can use the \sketch{} procedure to correctly determine whether $S$ contains a vertex from the cluster of $x$.  
\begin{restatable}{lemma}{sketchguarantee}\label{lemma:spectralsketch}
  The procedure \sketch{} (Algorithm~\ref{alg:spectralsketch}) has the following properties.
It runs in time
\[
O^*\!\left(
|S|\cdot \sqrt{\frac{n}{k}}\cdot
n^{O(\epsilon/\varphi^2 \log(1/\epsilon))}\cdot
(1/\epsilon)^{O(\log(\varphi^2/\epsilon))}
\right),
\]
and outputs a vector whose support size is at most
\[
O^*\!\left(
|S|\cdot \sqrt{\frac{n}{k}}\cdot
n^{O(\epsilon/\varphi^2 \log(1/\epsilon))}\cdot
(1/\epsilon)^{O(\log(\varphi^2/\epsilon))}
\right).
\]

Moreover, for every set $S \subseteq V$ and every vertex
$x \in V \setminus B_\delta$ that is typical with respect to $S$
(see Definition~\ref{def:good_pt_wrt_S}) the following holds.
\begin{itemize}
    \item If $C_{i(x)} \cap S \neq \emptyset$, then
    \[
    \Pr\!\left[
    \big|\langle \sketch(x), \sketch(S)\rangle\big|
    \;\geq\;
    0.8\,\big|\langle \sketch(x), \sketch(x)\rangle\big|
    \right]
    \;\geq\; 0.4.
    \]
    \item If $C_{i(x)} \cap S = \emptyset$, then
    \[
    \Pr\!\left[
    \big|\langle \sketch(x), \sketch(S)\rangle\big|
    \;\leq\;
    0.2\,\big|\langle \sketch(x), \sketch(x)\rangle\big|
    \right]
    \;\geq\; 0.999.
    \]
\end{itemize}
In both cases, the probability is over the internal randomness of \sketch{},
and all invocations of \sketch{} are independent.
\end{restatable}

Here, we use $O^*$-notation to suppress $\poly(\varphi^2/\epsilon)$ and $\polylog n$ factors. We formally define the set $B_{\delta}$ and bound its size in \Cref{sec:prelims}, and we formally define the notion of a \emph{typical vertex} and prove that almost all vertices are typical in \Cref{sec:def_well_spread}. We give the proof of \Cref{lemma:spectralsketch} in \Cref{sec:sketch}.

\paragraph{Finding the nearest neighbor.} Using the $\sketch$ primitive, we can efficiently implement a nearest neighbor search. 
Given a vertex $x$ and a set $S$, we simply perform a binary search on $S$ by using queries to $\langle \sketch(x), \sketch(S)\rangle$ (and $\langle \sketch(x), \sketch(S')\rangle$ for subsets $S' \subseteq S$ encountered during the binary search). We show for every well-behaved set $S$, for all but an $\approx \left(\epsilon/\varphi \right)^{1/3}$ fraction of vertices $x$, our nearest neighbor search (Algorithm \ref{alg:find_cluster}) correctly returns the set of vertices in $S$ that belong to the cluster of $x$ (see \Cref{thm:find_cluster} for the precise guarantee). Our nearest neighbor search can be implemented with query time $\approx \sqrt{n/k} \cdot n^{O(\epsilon/\varphi^2\log (1/\epsilon))}$, space complexity $\approx \sqrt{nk} \cdot n^{O(\epsilon/\varphi^2\log (1/\epsilon))}$ and preprocessing time $\approx \sqrt{nk} \cdot n^{O(\epsilon/\varphi^2\log (1/\epsilon))}$ by pre-computing and storing the vectors $\sketch(S)$ (and $\sketch(S')$  for subsets $S' \subseteq S$ in the binary search tree of $S'$). We note that our nearest neighbor search works for all values of $k$. If $k\leq n^{0.99}$, then the query time can be upper bounded by $(n/k)^{1/2+O(\epsilon/\varphi^2\log (1/\epsilon))}$.  We formally describe the algorithm (Algorithm \ref{alg:find_cluster}) in \Cref{sec:NNsearch}, and we prove its guarantee (\Cref{thm:find_cluster}) in Section~\ref{sec:NNanalysis}.  

\subsection{Space vs. query complexity tradeoffs}\label{sec:tradeoffs}
Our birthday paradox-style collision-counting analysis allows for tradeoffs in the number of random walks (see \Cref{thm:tradeoff}). Our analysis shows that the inner products $\langle \sketch(x), \sketch(S)\rangle $ approximate the required inner product $\langle f_x, \sum_{y \in S} \sigma_y f_y \rangle $ as long as the product of the number of random walks from the query vertex $x$ and the number of random walks from each $y \in S$ is at least $\approx\frac{n}{k}$ (see \Cref{lemma:variance_calc}). In particular, for any $\delta \in [0,1]$, we can run $\approx \left( \frac{n}{k}\right)^{\delta}$ random walks from $x$ and $ \approx \left( \frac{n}{k}\right)^{1-\delta}$ random walks from each  $y \in S$, resulting in a query time of $\approx \left( \frac{n}{k}\right)^{\delta}$ and a data structure of size $\approx |S| \cdot  \left( \frac{n}{k}\right)^{1-\delta} \approx  k \cdot \left( \frac{n}{k}\right)^{1-\delta}$, which gives the tradeoff stated in \Cref{thm:tradeoff}. 

Note that when $\delta \leq \frac{1}{2}$, running only $\approx \left(\frac{n}{k}\right)^{\delta}$ random walks from the query vertex $x$ does not suffice to approximate $\langle \sketch(x), \sketch(x)\rangle$ to a sufficiently high precision. In this regime, the algorithm should therefore not compute the inner products $\langle \sketch(x), \sketch(x) \rangle $, but instead should compare \(|\langle \sketch(x), \sketch(S)\rangle|\) against the threshold \(\frac{1}{\eta}\cdot\frac{k}{n}\). 
We omit the details for brevity, as it follows from the same principles as our main analysis. A proof sketch and a discussion of the required modifications appear in \Cref{sec:tradeoff_sketch}.

\subsection{Obtaining cluster representatives (preprocessing)}
So far, we have assumed access to a set $S$ containing one representative vertex per cluster. We use the procedure $\findclustermeans{}$ (Algorithm \ref{alg:find_cluster_means}) to create such a set. 

\paragraph{Algorithm overview:} We first sample a set of vertices $S$, which may include multiple vertices from the same cluster. Using our nearest neighbor algorithm (Algorithm \ref{alg:find_cluster}), we construct a directed similarity graph on $S$ (lines \ref{line:preproc_forstart} - \ref{line:preproc_forend} in Algorithm \ref{alg:find_cluster_means}). We then convert it to an undirected similarity graph $H$ (line \ref{line:preproc_H} in Algorithm \ref{alg:find_cluster_means}) by adding an undirected edge only if both corresponding directed edges exist. This ensures that we don't include edges that were incorrectly formed due to the nearest neighbor search failing for one of the endpoints. 

Ideally, we expect the sampled vertices from each cluster to form a single connected component. Therefore, we create an initial candidate set of representative vertices, denoted $R_{\mathrm{cand}}$, by selecting one representative vertex from each connected component in $H$ (lines \ref{line:preproc_L_candidate} - \ref{line:preproc:endforC_i}). 

However, the candidate set $R_{\mathrm{cand}}$ may be too large, so we need to refine it further. For example, we may have multiple outlier vertices with no near neighbors, each forming their own singleton component. To refine the set of representative vertices, we sample a random set of vertices, $S_{\mathrm{test}}$, and perform a nearest neighbor search for each vertex in $S_{\mathrm{test}}$ to identify the vertices in $R_{\mathrm{cand}}$ that most frequently get selected as nearest neighbors (lines \ref{line:preproc_S_test} - \ref{line:preproc:endfor2}). Our final set of representative vertices,  denoted $R$, consists of the vertices in $R_{\mathrm{cand}}$ that were selected as the nearest neighbor at least once (line \ref{line:R}). We formally analyze \findclustermeans{} (Algorithm \ref{alg:find_cluster_means}) in \Cref{sec:prepro_analysis}. 
\begin{algorithm}[H]
\caption{$\findclustermeans(\hat{k})$ \\
\# $\hat{k}$ is a constant factor approximation to the number of clusters}\label{alg:find_cluster_means}
    \begin{algorithmic}[1]   
        \State $S \gets$ (multi) set of $ O( k \cdot \log(\varphi^2/\epsilon))$ vertices sampled independently uniformly at random   \label{line:preproc_S}
        \State $\mathcal{T}_S \gets$ tree of sketches of $S$ \Comment{(see \Cref{def:tree})} \label{line:TS}
        \State \Comment{construct a directed graph based on asymmetric similarity tests}
        \State $A \gets \emptyset$
        \For{$x \in S$}\label{line:preproc_forstart}
    \State $C \gets \findcluster(x,\mathcal{T}_S.\mathrm{root})$ \Comment{find $S \cap C_{i(x)}$ using Algorithm \ref{alg:find_cluster}} \label{line:preproc_find_nbhrs}
        \State $ A \gets A \cup \{(x,y) : y \in C \} $ \label{line:preproc_A}
        \EndFor \label{line:preproc_forend}
        \State
        \State \Comment{create undirected correlation graph $H$ on the samples $S$}
        \State  $H =(S,\{ \{x,y\}  : (x,y) \in A$ and $ (y,x) \in A\} )$ \Comment{add undirected edge if both the directed edges exist \label{line:preproc_H}}
      \State $\mathcal C \gets$ the set of connected components of $H$ \label{line:preproc_components}
        \State 
        \State \Comment{pick one candidate representative vertex from each connected component of $H$} 
        \State 
        \State $R_{\mathrm{cand}} \gets \emptyset$ \Comment{store candidates for representative vertices \label{line:preproc_L_candidate}}
        \For{ $\widehat{C}_i \in  \mathcal C $}\label{line:preproc:forC_i}
        \State Select an arbitrary $x_i \in \widehat{C}_i$
        \State $R_{\mathrm{cand}}\gets  R_{\mathrm{cand}} \cup \{ x_i\} $ 
        \EndFor \label{line:preproc:endforC_i}
        \State $\mathcal{T}_{\mathrm{cand}} \gets$ tree of sketches of $R_{\mathrm{cand}}$ \Comment{(see \Cref{def:tree})}\label{line:Tcand}
        \State
        \State \Comment{reduce number of representative vertices} \label{line:phase2start}
        \State $ S_{\mathrm{test}} \gets$ (multi) set of $ O( k \cdot \log(\varphi^2/\epsilon))$ vertices sampled independently uniformly at random   \label{line:preproc_S_test}
        \State $c_y\gets 0$ for each $y \in R_{\mathrm{cand}}$ \Comment{counter for number of times $\findcluster$ returns each $y \in R_{\mathrm{cand}}$}
        \For{$x \in S_{\mathrm{test}}$}\label{line:preproc:startfor2}
        \State $N \gets \findcluster( x, \mathcal{T}_{\mathrm{cand}}.\mathrm{root})$ \Comment{find all neighbors of $x$ in $R_{\mathrm{cand}}$ using Algorithm \ref{alg:find_cluster}}
        \label{line:preproc:L_candidate_findnbhrs}
        \If{$|N| = 1$}  \Comment{increment counter if exactly one neighbor was found}
        \State Let $y$ be the unique element of $N$
        \State $c_y \gets c_y +1$ \Comment{counters are indexed by elements of $R_{\mathrm{cand}}$}
    \EndIf
\EndFor\label{line:preproc:endfor2}
\State \Comment{keep vertices that are a unique nearest neighbor of some $x \in S_{\mathrm{test}}$}
\State $R \gets \{y \in R_{\mathrm{cand}} : c_y >0\}$  
\label{line:R} 
    \State \Return  tree of sketches of $R$  \Comment{(see \Cref{def:tree})}\label{line:preproc_return}
        \end{algorithmic}
\end{algorithm}
 
\paragraph{Obtaining the candidate representatives $R_{\mathrm{cand}}$:} We construct our similarity graph by sampling a set $S$ of $O(k \log(\varphi^2/\epsilon))$ vertices, ensuring that $S$ intersects all but an $O(\epsilon/\varphi^2)$ fraction of clusters. For each $x \in S$, we call our nearest neighbor search $\findcluster(x,S)$ (Algorithm \ref{alg:find_cluster}).  This gives a directed similarity graph $A$. If the nearest neighbor search succeeds, then $\findcluster(x,S)$ correctly identifies all vertices in $S$ that belong to the same cluster as $x$, and we add the edges $\{ (x,y)  : y \in C_{i(x)} \cap S\}$ to the directed similarity graph $A$. Thus, the outgoing neighborhood of $x$ in $A$ is exactly $N^+A(x) = S \cap C_{i(x)}$. So if $\findcluster(x,S)$ succeeds for \emph{every} $x \in S \cap C_i$, then $S \cap C_i$ forms its own connected component in the undirected similarity graph constructed in line \ref{line:preproc_H} of Algorithm \ref{alg:find_cluster_means}. If this happens, then exactly one vertex from the cluster $C_i$ gets included in $R_{\mathrm{cand}}$, which is what we want.  Since $\findcluster$ succeeds for all but an $\approx (\epsilon/\varphi^2)^{1/3}$ fraction of vertices, we can show at most $\approx (\epsilon/\varphi^2)^{1/3}k$ cluster contain a vertex $x$ for which $\findcluster(x, S)$ fails, and all the remaining clusters that have a non-empty intersection with the sample $S$, have a unique cluster representative in $R_{\mathrm{cand}}$. 

\paragraph{Refining the set of representative vertices:} The clusters for which $\findcluster$ fails, may get split into multiple connected components in the similarity graph, leading to multiple representative vertices in $R_{\mathrm{cand}}$. To remove redundant representatives, we sample another set $S_{\mathrm{test}}$ of $O(k \log(\varphi^2/\epsilon))$ vertices and run our nearest neighbor search $\findcluster(x,R_{\mathrm{cand}})$ (Algorithm \ref{alg:find_cluster}) for each $x \in S_{\mathrm{test}}$.

If $\findcluster(x,R_{\mathrm{cand}})$ succeeds for a vertex $x \in C_i \cap S_{\mathrm{test}}$, and $C_i$ is one of the $\approx k(1 - (\epsilon/\varphi^2)^{1/3})$ clusters with a unique representative $y_i$ in $R_{\mathrm{cand}}$, then $\findcluster(x,R_{\mathrm{cand}})$ returns exactly the set $\{y_i\}$, in which case $y_i$ gets included in the refined set $R$. Thus, we can show that all but $\approx (\epsilon/\varphi^2)^{1/3}k$ clusters still have a unique cluster representative in the refined set $R$. On the other hand, since $\findcluster(x,R_{\mathrm{cand}})$ fails for at most $\approx (\epsilon/\varphi^2)^{1/3}|S_{\mathrm{test}}| \approx (\epsilon/\varphi^2)^{1/3}k$ vertices in $|S_{\mathrm{test}}|$, at most $\approx (\epsilon/\varphi^2)^{1/3}k$ additional vertices from $R_{\mathrm{cand}}$ get included in $R$.

\subsection{Approximating $k$ in sublinear time}\label{sec:approx-k}

As a simple application of our techniques we give a nearly query optimal algorithm that achieves a $\left(1+\varepsilon\right)$-approximation to $k$, for $\varepsilon=\left( \frac{\epsilon}{\varphi^2}\right)^{\Omega(1)}$.  More formally, assume that we are given a constant factor approximation to $k$, and refine it to a $\left(1+\varepsilon\right)$-approximation.

\begin{restatable}{thm}{findK}\label{thm:find_k}
Let $C$ be a sufficiently large constant. Assume that $\epsilon, \varphi$ satisfy the standard assumption as per \Cref{rem:param_assumptions}. For every  $\varepsilon \geq  C\cdot\left(\epsilon/\varphi^2\right)^{1/4} $ , 
    \findk$(\varepsilon)$ (Algorithm \ref{alg:find_k}) runs in time \[O^*\left(\left(\frac{n}{k}\right)^{1/2+O(\epsilon/\varphi^2 \log(1/\epsilon))}\cdot \left(1/\epsilon\right)^{O(\log(\varphi^2/\epsilon))}\cdot\frac{1}{\varepsilon^2}\right)\] and finds a $(1+\varepsilon)$-multiplicative approximation to $k$ with probability at least $0.99$. 
\end{restatable}

\paragraph{Proof sketch:} The idea behind Algorithm \findk \ is as follows. Suppose we had access to \textit{the inverse of cluster size oracle} which for every $x \in V$ returns value $\frac{1}{|C_{i(x)}|}$. Then, we could simply sample a set $F \coloneqq \{x_1, \ldots, x_L\}$ uniformly at random from $V$, get values $\frac{1}{|C_{i(x_1)}|}, \ldots, \frac{1}{|C_{i(x_L)}|}$ and approximate $\frac{k}{n}$ with an average of the obtained values. Since, 

\[\E\left[\frac{1}{L}\sum_{l \in [L]}\frac{1}{|C_{i(x_l)}|}\right] = \sum_{i \in k}\cdot\frac{|C_i|}{n}\cdot\frac{1}{|C_i|} = \frac{k}{n},\]
 we have that $\sum_{l \in [L]}\frac{1}{|C_{i(x_l)}|}$ would be an unbiased estimator for $\frac{k}{n}$. Moreover, since for all $i \in [k]$, we have $|C_i| = \Theta(k/n)$, by Hoeffding's inequality, 

\[\Pr\left[\left|\frac{1}{L}\sum_{l \in [L]}\frac{1}{|C_{x_l}|} - \frac{k}{n}\right| \geq \varepsilon\frac{k}{n}\right] \leq e^{-\Omega(\varepsilon^2\cdot L)},\] so it would suffice to take $L = O\left(\frac{1}{\varepsilon^2}\right).$ If $\varepsilon = \poly\left(\frac{\epsilon}{\varphi^2}\right)$, this would only give a $O^*(1)$ overhead to the time needed to access the inverse cluster size oracle.

Even though we do not have an access to the inverse cluster size oracle, we have a procedure for approximating $\frac{1}{|C_{i(v)}|}$ for almost all $v \in V$ ~-- \sketch. More particularly, note that for almost all $v \in V$,  the quentity $\|f_v\|^2_2$ is close to $\frac{1}{|C_{i(v)}|}$ (see \Cref{lemma:good_pts_properties} and \Cref{remark:norm} for a formal statement), and \sketch \ can be used to obtain approximations to $\|f_v\|^2_2$ for all except for a small fraction of $v \in V$ (see \Cref{lemma:sketchx} for a formal statement). These ideas are formalized in \Cref{lem:apx_k/n_||f_x||} and Corollary \ref{cor:adasketch} in \Cref{sec:appx_sqrtnk}.

An application of the lower bound in \cite{chiplunkar2018testing} shows that this is nearly optimal for constant $\varphi$ and $\epsilon=2^{-o(\sqrt{\log n})}$: 
\begin{restatable}{thm}{lowerbnd}\label{thm:lowerbnd}
    Any algorithm that, given access to a $(k, \varphi, \epsilon)$-clusterable graph, approximates $k$ to within a $2-\Omega(1)$ factor with probability at least $2/3$ must make at
least $\left( \frac{n}{k}\right)^{1/2 + \Omega(\epsilon)}$ queries. 
\end{restatable}

\newpage

\section{Preliminaries and Notation}\label{sec:prelims}
Let $G = (V,E)$ be a $d$-regular graph and let $n\coloneqq |V|$. Given a vertex $x \in V$, we write $\1_x \in \mathbb{R}^n$ for the indicator of $x$, i.e. the vector with entry $1$ at index $x$ and $0$ everywhere else. We use notation $\1 \in \R^n$ for the vector with all entries equal to $1$.

Given a clustering $C_1, \dots,  C_k$ of $G$ and a vertex $x \in V$, we use $i(x) \in [k]$ to denote the label such that $x \in C_{i(x)}$. We refer to $C_{i(x)}$ as the cluster of $x$. We define $\eta \coloneqq \max_{i,j} |C_i|/|C_j|$ to be the maximum ratio between cluster sizes. Note that by the assumption that $|C_i|/|C_j| \in O(1)$ for all $i,j$, we have $\eta = O(1)$, and for all $i \in [k]$, it holds that $\frac{1}{\eta} \cdot \frac{n}{k} \leq |C_i| \leq \eta \cdot \frac{n}{k}$.

We let $A$ denote the adjacency matrix of $G$, and $\mathcal{L} \coloneqq I - \frac{1}{d}A$ denote the normalized Laplacian of $G$. We denote the eigenvalues of $\mathcal{L}$ by $0 \leq \lambda_1 \leq \dots \lambda_n \leq 2$, and we denote the corresponding eigenvectors by $u_1, \dots , u_n \in \mathbb{R}^n$. Furthermore, we let $U \in \mathbb{R}^{n \times n}$ denote the matrix whose $i$-th column is the $i$-th eigenvector $u_i$ of $\mathcal{L}$. We use $M$ to denote the \emph{lazy random walk matrix} $M \coloneqq \frac{1}{2}I + \frac{1}{2d}A$. Note that for all $i \in [n]$, $u_i$ is an eigenvector of $M$ with eigenvalue $\left(1-\frac{\lambda_i}{2}\right) \in [0,1]$. We let $\Sigma$ denote the diagonal matrix of eigenvalues of $M$ arranged in \emph{descending} order. Then $M$ has eigendecomposition $M = U \Sigma U^\top $. 

It is a standard result that the Laplacian of a $(k,\varphi,\epsilon)$-clusterable graph admits an eigengap between its $k$-th and $(k+1)$-st eigenvalues \cite{lee2014multiway} \cite{chiplunkar2018testing}: 

\begin{lemma}[Lemma 3 from \cite{GKLMS21}]
\label{lem:bnd-lambda}
Let $G=(V,E)$ be a $d$-regular graph that admits a $(k,\varphi,\epsilon)$-clustering. 
Then $\lambda_k\leq 2\epsilon$ and $\lambda_{k+1}\geq \frac{\varphi^2}{2}$. 
\end{lemma}
\begin{remark}\label{remark:Meigengap}
    As a corollary, the eigenvalues $\left(1 - \frac{\lambda_i}{2}\right) $ of the lazy random walk matrix $M$ satisfy $\left(1 - \frac{\lambda_k}{2}\right) \geq 1-\epsilon$ and $\left(1-\frac{\lambda_{k+1}}{2}\right)\leq 1-\frac{\varphi^2}{4}.$
\end{remark}

\begin{remark}\label{remark:spectral_unique}
If $G$ is $(k,\varphi,\epsilon)$-clusterable with $\epsilon/\varphi^2$ smaller than a constant, then
	it follows from Lemma~\ref{lem:bnd-lambda} that the space spanned by the bottom $k$ eigenvectors of $\mathcal{L}$ is uniquely defined, i.e. the choice of $U_{[k]}$ is unique up to multiplication by an orthonormal matrix $R\in \mathbb{R}^{k \times k}$ on the right. We note 
	that while the choice of $f_x$ for $x \in V$ is not unique,  the dot product 
	between the spectral embedding of $x\in V$ and $y\in V$ is well defined, since for every orthonormal 
	$R\in \mathbb{R}^{k\times k}$ one has  
	\[\langle Rf_x, Rf_y\rangle=(Rf_x)^\top (Rf_y)=\left(f_x\right)^\top  (R^\top R) \left(f_y\right)=\left(f_x\right)^\top \left(f_y\right)\text{.}\]
\end{remark}
\noindent For a cluster $C_i$ we define its \emph{spectral embedding} $\mu_i$ to be 
$$ \mu_i \coloneqq \frac{1}{|C_i|}\sum_{x \in C_i} f_x. $$
A key property of the spectral embeddings of a $(k,\varphi,\epsilon)$-clusterable graph is that they concentrate well around their respective cluster means. More concretely, we bound the number of vertices $x \in V$ that deviate non-trivially from their cluster mean below. 
\begin{defn}\label{def:B_delta}
    Given $\delta > 0$, define 
    \begin{equation*}
        B_{\delta} \coloneqq \left \{x \in V :\|f_x - \mu_{i(x)} \|_2^2 > \frac{\delta}{|C_{i(x)}|}  \right \} 
    \end{equation*}
    to be the set of vertices whose spectral embeddings deviate from their cluster-mean. We sometimes refer to the vertices $x \in B_{\delta}$ as \emph{bad vertices}. 
\end{defn}

It is not hard to show that $B_{\delta}$ constitutes at most an $O\left(\frac{\epsilon}{\varphi^2}\cdot \frac{1}{\delta} \right)$-fraction of the vertices. The proof is simple, but the result is a key fact we will frequently use, so we include it here. The proof relies on the following lemma from \cite{GKLMS21}.

\begin{lemma}[Lemma 6 from \cite{GKLMS21}; ``variance bounds"]\label{lemma:variancebound}
Let $k \geq 2$ be an integer, $\varphi \in (0,1)$, and $\epsilon \in (0,1)$. Let $G = (V,E)$ be a $d$-regular graph that admits $(k,\varphi, \epsilon)$-clustering $C_1, \ldots, C_k$. Then for all $\alpha \in \mathbb{R}^k$ with $\|\alpha\|_2=1$ we have 
\[ \sum_{i =1}^k \sum_{x \in C_i}\left \langle f_x - \mu_i, \alpha\right \rangle^2 \leq \frac{4\epsilon}{\varphi^2}.\]
\end{lemma}

\begin{restatable}{lemma}{Bdeltasize}\label{lemma:close_to_clutermean}For every $\delta >0$, it holds that
    \begin{equation*}
        |B_{\delta}| \leq O\left(n \cdot  \frac{\epsilon}{\varphi^2} \cdot \frac{1}{\delta}\right). 
    \end{equation*}
\end{restatable}
\begin{proof}
Let  $\alpha_1, \ldots, \alpha_k \in \mathbb{R}^k$ be an orthonormal basis of $\mathbb{R}^k$. Applying \Cref{lemma:variancebound} to $\alpha_1, \dots ,\alpha_k$ and summing, we obtain 

\begin{equation}\label{eq:Bdelta1}
\sum_{i = 1}^k\sum_{x \in C_i}\|f_x - \mu_i\|^2_2 = \sum_{j=1}^k \sum_{i=1}^k \sum_{x \in C_i}\left \langle f_x - \mu_i, \alpha_j\right \rangle^2 \leq \frac{4\epsilon k}{\varphi^2}.
\end{equation}

On the other hand, by summing over $x \in B_{\delta}$ and using the definition of $B_{\delta}$, we get 

\begin{equation}\label{eq:Bdelta2}
\sum_{x \in B_{\delta}} \|f_x - \mu_{i(x)} \|_2^2 \geq  \sum_{x \in B_{\delta}} \frac{\delta}{|C_{i(x)}|} \geq |B_{\delta}| \cdot \frac{\delta}{\max_{i \in [k]} |C_i|} \geq |B_{\delta}| \cdot \frac{\delta}{\eta} \cdot \frac{k}{n},
\end{equation}
where the last inequality follows from the definition $\eta = \max_{i,j} |C_i|/|C_j|$. 
Combining Equations \eqref{eq:Bdelta1} and \eqref{eq:Bdelta2}, we get 
\[|B_{\delta}| \cdot \frac{\delta}{\eta} \cdot \frac{k}{n} \leq \sum_{i = 1}^k \sum_{x \in C_i }\|f_x - \mu_i\|_2^2 \leq \frac{4 \epsilon k}{\varphi^2},\]  

which rearranges to 

\[ |B_{\delta} | \leq  \frac{4 \epsilon}{\varphi^2}\cdot \frac{\eta}{\delta} \cdot n = O\left( n \cdot \frac{\epsilon }{\varphi^2} \cdot \frac{1}{\delta} \right). \]    
\end{proof}

Our first observation about the set $V \setminus B_{\delta}$ is that any two points $x, y \in V\setminus B_{\delta}$ which belong to the same cluster are close. Formally, we have

\begin{restatable}{lemma}{goodptsproperties}\label{lemma:good_pts_properties}

For all $x,y \in V\setminus B_{\delta}$, if $x$ and $y$ belong to the same cluster $C_i$, then

\[\left|\langle f_x, f_y \rangle - \frac{1}{|C_i|} \right|\leq \left(\frac{4 \sqrt{\epsilon}}{\varphi}+4\sqrt{\delta}\right)\frac{1}{|C_i|}.\]  \label{bulletpt:fx_dot_fy}

\end{restatable}
\begin{remark}\label{remark:norm}
    \Cref{lemma:good_pts_properties} implies 
    $\left| \|f_x\|_2^2- \frac{1}{|C_{i(x)}|} \right| \leq \left( \frac{4 \sqrt{\epsilon}}{\varphi} +4\sqrt{\delta}\right)\cdot \frac{1}{|C_{i(x)}|}$ \label{bulletpt:f_norm} for all $x \in V \setminus B_{\delta}$ by choosing $x = y$. In particular, if $\delta$ and $\epsilon/\varphi^2$ are smaller than a sufficiently small constant, then  $$\|f_x\|^2_2 = \Theta\left(\frac{1}{|C_{i(x)}|}\right)=  \Theta\left( \frac{k}{n} \right)$$ for all $x \in V \setminus B_{\delta}$. We will use this fact repeatedly. 
\end{remark}
The proof of \Cref{lemma:good_pts_properties} is presented in \Cref{sec:spectral_facts}. 

\paragraph{Notation.}
We use $O^*, \Omega^*$-notation to suppress $\poly(\varphi^2/\epsilon)$ and $\polylog n$-factors. 

\noindent Throughout this paper, when referring to the size of a multiset, we count elements with their multiplicities.

\noindent For an integer $l$, we write $I_l$ for the identity matrix on $\R^l$. We drop the subscript whenever $l$ is clear from the context. 

\noindent For matrices $A,B$, we write $A \bullet B$ for the Frobenius inner product $A \bullet B = \sum_{ij}A_{ij}B_{ij}$. We will often write expressions of the form $\langle f_x , f_y\rangle ^2 $ as  $\langle f_x , f_y\rangle ^2 = f_x f_x^\top  \bullet f_y f_y^\top .  $ This is because many of our proofs use averaging arguments over $f_x$, in which case it is more convenient to view these expressions as linear functions of second tensor powers of $f_x$, as opposed to a quadratic function of $f_x$.

\section{Main Result}

In this section, we prove the guarantees for our clustering oracle. 
\begin{restatable}[Formal version of \Cref{thm:informal_main}]{thm}{thmmain}\label{thm:main}
    For every $\epsilon, \varphi>0$, and every $k\geq 1$, there exists a $(k, \varphi, \epsilon)$-clustering oracle with  misclassification rate $O\left(\left(\epsilon/\varphi^2\right)^{1/3}\log^2(\varphi^2/\epsilon)\right)$ that has
    \begin{itemize}
        \item preprocessing time and space complexity $O^*\left(\sqrt{nk} \cdot n^{O(\epsilon/\varphi^2 \log(1/\epsilon))}\cdot \left(1/\epsilon\right)^{O(\log(\varphi^2/\epsilon))}\right)$
        \item query time $O^*\left(\sqrt{\frac{n}{k}} \cdot n^{O(\epsilon/\varphi^2 \log(1/\epsilon))}\cdot \left(1/\epsilon\right)^{O(\log(\varphi^2/\epsilon))}\right)$.
    \end{itemize}
\end{restatable}

\begin{remark}[On parameter assumptions]\label{rem:param_assumptions}
Our construction of the polynomial in Theorem~\ref{thm:standardbasis}, and consequently the guarantees of our clustering oracle, rely on the parameter assumptions
\[
\frac{\epsilon}{\varphi^2}\log(1/\epsilon) \le c_1
\quad\text{and}\quad
\log(1/\epsilon)\,\log(\varphi^2/\epsilon) \le c_2 \log n
\]
for sufficiently small absolute constants \(c_1, c_2\).
These conditions ensure that the polynomial approximation used in the analysis has sublinear (in \(n\)) coefficients, which is crucial for achieving sublinear time and space complexity.

In the regime when either \(\frac{\epsilon}{\varphi^2}\log(1/\epsilon) > c_1\) or \(\log(1/\epsilon)\log(\varphi^2/\epsilon) > c_2\log n\), the factors 
 \(n^{O(\epsilon/\varphi^2 \log(1/\epsilon))}\) or $(1/\epsilon)^{O(\log(\varphi^2/\epsilon))}$ 
appearing in our runtime bounds become polynomial in \(n\), so the overall running time is no longer sublinear. Since a polynomial-time clustering oracle is already known~\cite{DBLP:conf/soda/Sinop16}, one can instead apply the algorithm of~\cite{DBLP:conf/soda/Sinop16} in this regime. 
\end{remark}

Throughout the remainder of this section, we will assume that $G=(V,E)$ admits a $(k,\varphi, \epsilon)$-clustering $C_1, \dots C_k$ for $\epsilon,\varphi$ satisfying the assumptions in \Cref{rem:param_assumptions} above. Additionally, we set the following parameters. 

\begin{defn}[Parameter settings]\label{def:params}

Throughout this section, we fix the following parameters.
Set
\[
\delta= c \cdot \left(\frac{\epsilon}{\varphi^2}\right)^{2/3},
\] where \(c>0\) is a sufficiently small absolute constant.
Accordingly, whenever we write \(B_{\delta}\), we refer to the set defined in \Cref{def:B_delta} with this fixed choice of \(\delta\).
We also set
\(t_{\min} = 20 \log n / \varphi^2\) (the minimum random walk length in \Cref{alg:spectralsketch}),
\(t_{\Delta} = \tdelta\) (the number of distinct walk lengths in \Cref{alg:spectralsketch}),
and
\(J = 6 \cdot 10^4\) (the number of independent \(\sketch(S_u)\) vectors computed for each node \(u\) in in the tree of sketches defined in \Cref{def:tree} and computed by \Cref{alg:find_cluster_means}).
\end{defn}

\begin{remark}
\Cref{lemma:close_to_clutermean} and \Cref{lemma:good_pts_properties} hold for arbitrary \(\delta>0\).
We instantiate \(\delta\) as above to optimize the final misclassification rate.
\end{remark}

The rest of the section is structured as follows. In \Cref{sec:NNsearch}, we describe our nearest neighbor procedure, \findcluster{} (Algorithm \ref{alg:find_cluster}). Then, in \Cref{sec:prepro_describe}, we describe the preprocessing procedure \findclustermeans{} (Algorithm \ref{alg:find_cluster_means}), and we state its main guarantee. 
\Cref{sec:def_well_spread} introduces several key definitions that we need for the analysis. In \Cref{sec:sketch}, we prove the correctness of \sketch{} (\Cref{alg:spectralsketch}). Then, in \Cref{sec:NNanalysis}, we prove the correctness of the nearest neighbor procedure \findcluster{} (\Cref{alg:find_cluster}).  \Cref{sec:prepro_analysis} proves the correctness of \findclustermeans{} (\Cref{alg:find_cluster_means}). Finally, in \Cref{sec:main_wrapping} we put everything together and wrap the proof of \Cref{thm:main}.

\subsection{Fast nearest neighbor search}\label{sec:NNsearch}

Our \sketch{} algorithm (Algorithm \ref{alg:spectralsketch}) allows us to efficiently decide if a set $S$ contains a vertex from the cluster of $x$. We will use this primitive to obtain a fast nearest neighbor search algorithm. 

As a first step, we boost the success probability of $\sketch{}$. This is accomplished by the procedure \dotproduct{} (Algorithm \ref{alg:dotproduct}) below, which compares $O(\log n)$  independent trials of $|\langle \sketch(x), \sketch(S)\rangle|$ to $O(\log n)$  independent trials of  of $|\langle \sketch(x), \sketch(x)\rangle|$, and returns a (biased) majority vote.
That way, \dotproduct{} returns \textsc{True} with a high probability if $S$ contains a vertex from the cluster of $x$, and returns \textsc{False} with a high probability otherwise. 

To ensure an efficient query time, \dotproduct{} avoids computing $\sketch(S)$ during each query. Instead, it takes as input a list $\{\sk_j\}_{j \in [J]}$ of $J = O(\log n)$ previously computed $\sketch(S)$ vectors. This avoids the computational cost of recomputing $\sketch(S)$ each time $\dotproduct{}$ is called.By Lemma~\ref{lemma:spectralsketch}, computing $\sketch(S)$ from scratch takes time $\approx |S| \cdot \sqrt{n/k}$, which is $\approx \sqrt{nk}$ when $|S| \approx k$, and is therefore too expensive to perform for each query. 
For this reason, these sketches are computed once during the preprocessing procedure
(Algorithm~\ref{alg:find_cluster_means}).
As a result, \dotproduct{} only needs to compute the $\sketch(x)$ vectors and take inner prodcuts during a query, leading to a runtime of $\approx \sqrt{n/k}$.

Using the procedure \dotproduct{}, the algorithm \findcluster{} (Algorithm~\ref{alg:find_cluster})
performs a binary search over the set $S$ (see lines~\ref{line:bin_search_start}--\ref{line:bin_search_end})
to identify the elements of $S$ that are near neighbors of $x$, namely those whose spectral embeddings have a large inner product with $f_x$. To support this binary search efficiently, we introduce the data structure \emph{tree of sketches},
which stores precomputed $\sketch(S')$ vectors for all subsets $S' \subseteq S$ that arise in the binary search tree.
This guarantees that all $\sketch$ vectors required by \dotproduct{} are available at query time.

\begin{defn}[Tree of sketches of $S$]\label{def:tree}
Let \( S = (s_1,\dots,s_{|S|}) \) be a set represented as an ordered list.
A \emph{tree of sketches of $S$} is a complete binary tree in which each node \( u \)
is associated with a subset \( S_u \subseteq S \). 
The root node is associated with the full set \( S \).
For any internal node \( u \) associated with a subset \( S_u \),
its left child \( u.\text{left} \) is associated with the first half of \( S_u \)
and its right child \( u.\text{right} \) with the second half of \( S_u \),
according to the ordering of \( S \).
The leaves of the tree correspond to singleton subsets \( \{s_i\} \),
and hence there is a one-to-one correspondence between the leaves and the elements of \( S \).

Each node \( u \) stores two attributes:
\( u.\text{label} = S_u \), and
\( u.\text{sketch} \), which consists of \( O(\log n) \) independent
\(\sketch(S_u)\) vectors precomputed for the subset \( S_u \).

The tree additionally stores a pointer \( \text{root} \) to its root node.
\end{defn}

To ensure a sufficiently small runtime, we cap the maximum number of leaves that the procedure \findcluster{} (Algorithm \ref{alg:find_cluster}) is allowed to visit at $O^*(1)$ (see line \ref{line:cap} in Algorithm \ref{alg:find_cluster}). In the analysis, we show that this does not affect the misclassification rate (see \Cref{thm:find_cluster}).
\begin{algorithm}[H]
\caption{$\dotproduct(x, \{\sk_j\}_{j\in [J]})$
\newline\textbf{Input:} vertex $x$ and $J$ precomputed sketches $\sk_1,\dots,\sk_1$ of the same set $S$).
\newline\textbf{Output:} \textsc{True} if $S$ contains a vertex from $C_{i(x)}$,  otherwise \textsc{False} }\label{alg:dotproduct}
\begin{algorithmic}[1]
\State $c\gets 0$ \Comment{counter for number of \textsc{True} votes}
\For{$j=1$ to $J$} \Comment{independent trials of \sketch}
    \If{$|\langle \sketch(x), \sk_j\rangle| \geq 0.5\,|\langle \sketch(x), \sketch(x)\rangle|$} \label{line:dotproducttest}
        \State $c\gets c+1$ \Comment{increment if estimated correlation of $x$ with $S$ is high}
    \EndIf
\EndFor
\If{$c\geq 0.3J$}
    \State \Return \textsc{True}
\Else
    \State \Return \textsc{False}
\EndIf
\end{algorithmic}
\end{algorithm}

\begin{algorithm}[H]
\caption{\findcluster$(x,r)$
\newline\textbf{Input:} a vertex $x$ in $G$, and a node $r$ in a tree $\mathcal{T}$ of sketches of $S$ (see Definition~\ref{def:tree})
\newline\textbf{Output:} all vertices in $S$ that belong to $C_{i(x)}$}\label{alg:find_cluster}
\begin{algorithmic}[1]
\State $N \gets \emptyset$
\If{number of leaves visited $> O^*(1)$} \label{line:cap}
    \Comment{cap the maximum number of visited leaves at $O^*(1)$}
    \State \Return $\emptyset$
\EndIf

\Comment{binary search over $S$}
\If{$r$ is a leaf of $\mathcal{T}$} \label{line:if|S|=1}
    \State \Return $r.\text{label}$
    \Comment{return the singleton set associated with the leaf node $r$}
\Else
    \If{$\dotproduct(x, r.\text{left}.\text{sketches})$} \label{line:bin_search_start}
        \Comment{test if left subtree contains vertices from $C_{i(x)}$}
        \State $N \gets N \cup \findcluster(x, r.\text{left})$
        \Comment{recurse on left subtree}
    \EndIf
    \If{$\dotproduct(x, r.\text{right}.\text{sketches})$}
        \Comment{test if right subtree contains vertices from $C_{i(x)}$}
        \State $N \gets N \cup \findcluster(x, r.\text{right})$
        \Comment{recurse on right subtree}
    \EndIf \label{line:bin_search_end}
    \State \Return $N$
\EndIf
\end{algorithmic}
\end{algorithm}

We state and prove the guarantees of \dotproduct{} (see \Cref{lemma:dotproduct} and \findcluster{} (see \Cref{thm:find_cluster}) in \Cref{sec:NNanalysis}

\subsection{Preprocessing}\label{sec:prepro_describe}
In the previous section, we introduced the $\findcluster$ algorithm, which finds all near neighbors of a vertex $x$ in a set $S$. If we are given a set
$R = \{y_1, \dots, y_k\}$ of representative vertices with
$y_i \in C_i \setminus B_\delta$ for each $i \in [k]$, together with a tree of
sketches $\mathcal{T}$ for $R$ (as per Definition~\ref{def:tree}) with  root $r$, then $\findcluster(\cdot, r)$ naturally gives a clustering oracle. Specifically, given a query vertex $x$, we run $\findcluster(x,r)$ on the tree of
sketches $\mathcal T$ of $R$ and assign $x$ to the cluster corresponding to the returned
representative. This is similar to the idealized query procedure is formalized in
Algorithm~\ref{alg:ideal_query}.

We now formalize the notion of representative vertices.
\begin{defn}[Representative vertex]\label{def:rep_vertex}
    Given a set $R \subseteq V$, we say that a vertex $y \in R$ is \emph{a representative vertex for cluster $C_i$ in $R$} if $y \in C_i$, and no other element of $R$ belongs to $C_i$. 

    We say that cluster $C_i$ \emph{has a representative in $R$} if $R$ contains a representative vertex for $C_i$. 
\end{defn}
Ideally, we would like to compute a set $R$ that contains a representative vertex for every cluster.
However, we do not know how to achieve this goal in sublinear time.
Instead, our preprocessing algorithm will return a set $R$ that satisfies a relaxed guarantee:
almost all vertices belong to clusters that have a representative in $R$, and for such vertices,
the corresponding representative can be recovered using \findcluster{}.

The following theorem formalizes this guarantee.
\begin{restatable}[Correctness of \findclustermeans{}]{thm}{thmpreprocessing}\label{thm:preprocessing}
The procedure $\findclustermeans(\hat{k})$ (\Cref{alg:find_cluster_means}) runs in time
\[
O^*\!\left( \sqrt{nk} \cdot n^{O\!\left(\frac{\epsilon}{\varphi^2}\log\frac{1}{\epsilon}\right)}
\cdot \left(\frac{1}{\epsilon}\right)^{O\!\left(\log\frac{\varphi^2}{\epsilon}\right)} \right)
\]
and outputs a tree of sketches $\mathcal{T}$ constructed on a set $R \subseteq V$.

With probability at least $0.999$ over the internal randomness of $\findclustermeans$,
there exists a set $B \subseteq V$ of size
\[
|B| = O\!\left(n \cdot \left(\frac{\epsilon}{\varphi^2}\right)^{1/3}
\cdot \log\!\left(\frac{\varphi^2}{\epsilon}\right)\right)
\]
such that the following holds.

For every vertex $x \in V \setminus B$, the cluster $C_{i(x)}$ of $x$ has a unique
representative vertex $y_{i(x)} \in R$, and $\findcluster(x, \mathcal{T}.\mathrm{root})$ returns $y_{i(x)}$
with probability at least $1 - n^{-50}$, i.e.,
\[
\Pr\!\left[\findcluster(x, \mathcal{T}.\mathrm{root}) = y_{i(x)}\right]
\ge 1 - n^{-50},
\]
where the probability is over the internal randomness of $\findcluster$.

\end{restatable}
\begin{restatable}{remark}{remarkprepromain}
\label{remark:prepromain}
The set $R$ returned by $\findclustermeans(\hat{k})$ implicitly defines a clustering oracle:
given a query vertex $x$, we assign $x$ to the cluster corresponding to the
representative returned by $\findcluster(x, \mathcal{T}.\mathrm{root})$.
By Theorem~\ref{thm:preprocessing}, this oracle misclassifies at most
\[
O\!\left( \left(\frac{\epsilon}{\varphi^2}\right)^{1/3}
\cdot \log^2\!\left(\frac{\varphi^2}{\epsilon}\right) \right)
\]
fraction of vertices.
\end{restatable}

We prove \Cref{thm:preprocessing} in \Cref{sec:prepro_analysis}. We now briefly describe the \findclustermeans{} algorithm. The algorithm has two main phases.

\paragraph{Phase 1: Find candidates for representative vertices.} In this phase (lines \ref{line:preproc_S}-\ref{line:Tcand} in Algorithm \ref{alg:find_cluster_means}), we compute a (possibly too large) set of candidate representative vertices, which we denote $R_{\mathrm{cand}}$. We sample a set $S$ of random vertices and use \findcluster{} to build a similarity graph on $S$ (lines  \ref{line:preproc_S} - \ref{line:preproc_H}). This idea is similar to \cite{czumaj2015testing}. However, note that in \cite{czumaj2015testing}, the similarity graph is based on Euclidean distance, a symmetric notion of similarity, whereas our notion of similarity is asymmetric. For a vertex $x$, we draw directed edges to the vertices returned by $\findcluster(x,\mathcal{T}_S.\mathrm{root})$, where $\mathcal{T}_S$ is the tree of sketches of $S$ (as per \Cref{def:tree}). This results in a directed graph, which we later convert to an undirected similraity graph $H$ (line \ref{line:preproc_H} in Algorithm \ref{alg:find_cluster_means}) by adding an undirected edge only if both corresponding directed edges exist. This ensures robustness. If $\findcluster(x,\mathcal{T}_S.\mathrm{root})$ succeeds for a vertex $x$, then the neighborhood of $x$ in $H$ \emph{only} contains the true near neighbors of $x$, even if $\findcluster(y,\mathcal{T}_S.\mathrm{root})$ incorrectly returns $x$ for some other vertex $y \in S$. 

Once we have the similarity graph $H$, we select one candidate representative vertex from each connected component of $H$ (lines \ref{line:preproc_components}- \ref{line:Tcand}). Intuitively, vertices from the same cluster should be connected in the similarity graph, so each connected component should correspond to a cluster. Therefore, a natural way to obtain candidates for the representative vertices is to select one vertex from each connected component. The only issue is that the obtained set of candidate representative vertices might be too large. For example, it is possible that the sampled bad vertices $x \in S \cap B_{\delta}$ have no near neighbors, each forming their own singleton component, in which case each would contribute its own representative vertex. We handle this in phase 2 of the algorithm. 

\paragraph{Phase 2: Refining the set of representative vertices.} In lines \ref{line:phase2start} - \ref{line:preproc_return}, we filter out the unnecessary representative vertices from $R_{\mathrm{cand}}$. We sample another set of vertices $S_{\mathrm{test}}$ and run $\findcluster(x,\mathcal{T}_{\mathrm{cand}}.\mathrm{root})$ for each $x \in S_{\mathrm{test}}$ (where $\mathcal{T}_{\mathrm{cand}}.\mathrm{root}$ denotes the root of a tree of sketches on $R_{\mathrm{cand}}$, as per \Cref{def:tree}). This allows us to identify the vertices in $R_{\mathrm{cand}}$ that appear most frequently as nearest neighbors. Our final set of representative vertices, denoted $R$, consists of the vertices from $R_{\mathrm{cand}}$ that appear at least once as near neighbors of $S_{\mathrm{test}}$ (see line \ref{line:R} in Algorithm \ref{alg:find_cluster_means}). 
The intuition behind this step is as follows. A vertex may enter $R_{\mathrm{cand}}$ even if it is a ``bad'' vertex in
$V \setminus B_\delta$, for example because it has no other near neighbors and therefore
appears as an isolated singleton component.
However, when we sample a fresh set $S_{\mathrm{test}}$, it is unlikely that such a vertex
will be the nearest neighbor of any vertex in $S_{\mathrm{test}}$.
By retaining only those vertices of $R_{\mathrm{cand}}$ that are returned by
\findcluster{} on some vertex in $S_{\mathrm{test}}$, we effectively eliminate these unnecessary representatives.

\subsection{Well-spread set $S$ and typical vertices with respect to $S$}\label{sec:def_well_spread}
Before we prove the guarantee of $\sketch{}$ (\Cref{lemma:spectralsketch}), we need to set up a few definitions.  

For $S\subseteq V$ and integer $t\geq 0$ define
\begin{equation}\label{eq:p-s}
p_S^t=\mathbb{E}_{y\sim S}[M^t \1_y]   
\end{equation}
to be the distribution of the $t$-step walk started at a uniformly random vertex in $S$. 

The collision probability $p_S^t$ is one of our measures for how well-spread a set $S$ is (see \Cref{def:good_S}). This is a quite natural notion. For intuition, consider the case where $\epsilon = 0$. In this setting, the graph $G$ consists of $k$ disjoint expander, each of size $\approx \frac{n}{k}$. If $S$ consists of $k$ vertices from a single cluster, then $\|p_S^t \|^2_2 \approx \frac{k}{n}$. On the other hand, if $S$ consists of $k$ vertices from different clusters, then $\|p_S^t \|^2_2\approx \frac{1}{n}$.

We use $(p_S^t)^2$ to denote the vector given by $$(p_S^t)^2(v) \coloneqq (p^t_S(v))^2$$
for all $v \in V$, where $p^t_S(v)$ denotes the $v^{th}$ entry of $p^t_S$.  The vector $(p_S^t)^2$ plays a key role in a variance calculation underpinning the birthday paradox. It appears in \Cref{def:good_pt_wrt_S}, where we say that a vertex is typical if it satisfies the variance conditions  $\left \langle p_{x}^{t} , \left(p_S^{l}\right)^2 \right \rangle \leq O^*\left(\frac1{n^{3/2}} \cdot \frac{k^{3/2}}{|S|^2}\right)$ and $\left \langle \left( p_x^t\right)^2, p_S^l\right \rangle \leq O^*\left(\frac1{n^{3/2}} \cdot \frac{k^{3/2}}{|S|}\right) $. 

Our nearest neighbor search procedures, namely \sketch{}, \dotproduct{}{} and \findcluster{}, only exhibit good performance on sets $S$ of candidate vertices that represent cluster structure well — we refer to these sets as {\em well-spread} sets. The list of conditions, given below in \Cref{def:good_S}, includes basic properties such as low overlap with outlier vertices in the graph $G$ (condition {\bf \ref{con:B_delta}} below), a condition upper bounding the self-collision probability of $p_S^t$ (condition {\bf \ref{con:rw}} below) as well as spectral bounds akin to the spectral embeddings of vertices in $S$ being spectrally close to the identity (condition {\bf \ref{con:isotropic}}). 

As we show below (see \Cref{lemma:random_good_set}), a random set $S$ satisfies these conditions with high probability.
\begin{defn}[Well-spread set $S$]\label{def:good_S}
    Say that a set $S \subseteq V$of size $O(k \log(\varphi^2/\epsilon))$ is \emph{well-spread} if it satisfies the following properties: 
    \begin{enumerate}[label=(\textbf{\arabic*})]
         \item (Small overlap with $B_{\delta}$). The intersection of $S$ with $B_{\delta}$  (as per \Cref{def:B_delta}) is bounded by $|S \cap B_{\delta}| = O\left( \frac{\epsilon/\varphi^2}{\delta}\cdot  k\cdot \log(\varphi^2/\epsilon)\right)$.
\label{con:B_delta}\label{con:good_S1}
      \item (Low self-collision probability) The self-collision probability $p^t_S$ satisfies $ \| p^t_S\|_2^2=O^*\left(\frac{1}{n} \cdot \frac{k^2}{|S|^2}\right)$ for all $t\in [t_{\min}, t_{\min}+t_{\Delta}]$, i.e. it is comparable to the self-collision probability of the uniform distribution. \label{con:rw} \label{con:good_S2}
      \item (Nearly isotropic in the embedding space) The spectral embeddings $f_y$ of $y \in S$ satisfy $\sum_{y \in S} \|f_y\|^2_2 \leq O\left(\frac{k}{n}\cdot k \log(\varphi^2/\epsilon)\right)$. Additionally, for all but $O((\epsilon/\varphi^2)^{1/3}\log(\varphi^2/\epsilon) k)$ 
      clusters, it holds that 
      $$ \mu_i \mu_i^\top  \bullet \sum_{y \in S\setminus C_{i}} f_y f_y^\top  \leq 10^{-10}\cdot \|\mu_i\|^4_2 $$  and $$\mu_i \mu_i^\top  \bullet \sum_{y \in S}f_y f_y^\top   \leq O\left(\left(\varphi^2/\epsilon\right)^{1/3}\right)\|\mu_i\|_2^4.$$
      We refer to clusters that violate this condition as \emph{bad clusters with respect to $S$}. Note that this condition  is very similar to asking that 
$\sum_{y\in S} f_y f_y^\top $ spectrally approximates $I$. However, our condition {\bf \ref{con:isotropic}} is distinctly weaker (see \Cref{rem:spectral_approx}). \label{con:isotropic} \label{con:good_S3}
    \end{enumerate}
\end{defn}
\begin{remark}\label{rem:spectral_approx}
If we sample a set $S$ of $\Omega( k \log k)$ vertices, then we can impose the condition that $\sum_{y \in S} f_y f_y^\top $ spectrally approximates an appropriately scaled identity matrix. However, if we were to do that, then our analysis would give a misclassification rate of $\left(\epsilon/\varphi^2\right)^{1/3}\polylog k$, as opposed to $O\left( \left(\epsilon/\varphi^2\right)^{1/3}\polylog (\varphi^2/\epsilon )\right)$. However, this requires $\epsilon < 1/\log k$, which can be as small as $1/\log n$ for large values of $k$. As a result, we would not be able to handle constant $\epsilon$, which is a central part of our paper.
\end{remark}

  It is straightforward to verify that if $S$ is a well-spread set, then every subset $S' \subseteq S$ is also a well-spread set. For completeness, we include the proof in \Cref{sec:well-spread_sets}. 
\begin{restatable}{lemma}{subsetremark}\label{claim:good_set_subset}
 If $S$ is a well-spread set, then every subset $S' \subseteq S$ is also a well-spread set.
\end{restatable}

A random set $S$ of size linear in $k$ is well-spread with high probability:
\begin{restatable}{lemma}{randomgoodset}\label{lemma:random_good_set}
    If $S$ is a (multi) set of at most $O(k \log(\varphi^2/\epsilon))$ vertices sampled independently uniformly at random, then with probability at least $0.9999$, the set $S$ is well-spread. 
\end{restatable}
\Cref{lemma:random_good_set} follows from multiple applications of Markov's inequality. The proof is presented in Appendix~\ref{sec:well-spread_sets}. 

\Cref{def:good_S} ensures that $S$ represents the cluster structure well. Our procedures \sketch{}, \dotproduct{}{} and \findcluster{} also require that the queried vertex $x$ interacts with $S$ in a way that reflects its interaction with $V$. We formalize this below in \Cref{def:good_pt_wrt_S}, where we require that the cluster of $x$ does not contain too many representatives in $S$ (Condition \ref{con:ScapC_i}), that these representatives are not bad points (Condition \ref{con:bad_cluster}), that the random walk distribution from $x$ is not too correlated with the distributions from $S$ (Conditions  \ref{con:rw1} and \ref{con:rw2}) and finally that the deviation of the spectral embedding $f_x$ from the cluster mean $\mu_{i(x)}$ does not align too much with the spectral embeddings of the vertices in $S$.  We refer to vertices that satisfy these conditions as \emph{typical with respect to $S$} (Condition \ref{con:rw3}).

We show in \Cref{lemma:good_pts_wrt_S}, that for a well-spread set $S$, almost all $x \in V$ satisfy these conditions.

\vspace{0.05in}
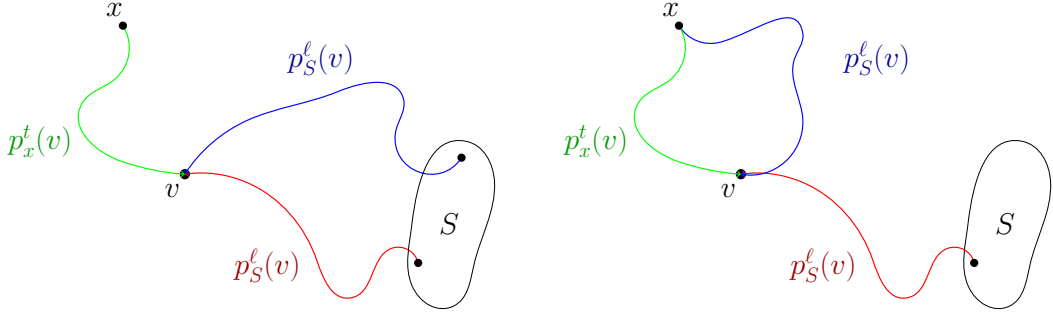
\begin{figure}
	\resizebox{0.9\textwidth}{!}{ 
        \centering
    \begin{subfigure}[b]{0.7\textwidth}
	\begin{center}
\resizebox{300pt}{!}{
\begin{tikzpicture}
    
    % Important vertex v (filled black)
    \node[circle, fill=black,inner sep=2.5pt] (v) at (-6.5,3.2) {};

    % Smooth colorful paths (now perfectly touching the black dots)
    \path[draw, thick, red, ->, use Hobby shortcut] 
    (-0.85,1.05) .. (-1.5,1.4) .. (-2.5,0.2) .. (-3.3,1.1) .. (-5,3) .. (v);

    \path[draw, thick, blue, ->, use Hobby shortcut] 
    (0.2,3.6) .. (-1.4,4.2) .. (-1.2,5) .. (-2.8,5.2) .. (-5,4.5) .. (v);

    \path[draw, thick, green, ->, use Hobby shortcut] 
    (-8,6.8) .. (-8.2,5.5) .. (-9,4.9) .. (-8,3.5) .. (v);

    % The region S
    \path[draw,use Hobby shortcut,closed=true]
    (0,0) .. (.5,1) .. (1,3) .. (.3,4) .. (-1,2) .. (-1,.5);
    \node at (-0.1,2) {\huge$S$};

    % Adjusted position for x (moved left and down)
    \node (T) at (-8,6.8) {}; % Green path start
    \node[circle, fill=black,inner sep=2pt] (Tdot) at (T) {}; % Green path actual start

    % Path ultimate starting points (adjusted to align with paths)
    \node[circle, fill=black,inner sep=2pt] (A) at (-0.85,1.05) {}; % Red path start
    \node[circle, fill=black,inner sep=2pt] (K) at (0.2,3.6) {}; % Blue path start

    % Labels
    \node[] at (-6.8,2.8) {\huge$v$};
    \node[] at (-8.2,7.2) {\huge$x$}; % Adjusted label for x

    % Path Labels
    \node[] at (-10,4) {\textcolor{green!60!black}{\huge$p_x^t(v)$}};
    \node[] at (-3.2,6) {\textcolor{blue!60!black}{\huge$p_S^\ell(v)$}};
    \node[] at (-4.5,1) {\textcolor{red!60!black}{\huge$p_S^\ell(v)$}};

\end{tikzpicture}
}
\end{center}
        \centering
	\label{fig:random_walks}
    \end{subfigure}
    	\hfill
            \centering
        \begin{subfigure}[b]{0.7\textwidth}
        \begin{center}
\resizebox{300pt}{!}{
\begin{tikzpicture}
    % The region S
    \path[draw,use Hobby shortcut,closed=true]
    (0,0) .. (.5,1) .. (1,3) .. (.3,4) .. (-1,2) .. (-1,.5);
    \node at (-0.1,2) {\huge$S$};

    % Important vertex v (filled black)
    \node[circle, fill=black,inner sep=2.5pt] (v) at (-6.5,3.2) {};

        % Smooth colorful paths (now perfectly touching the black dots)
    \path[draw, thick, red, ->, use Hobby shortcut] 
    (-0.85,1.05) .. (-1.5,1.4) .. (-2.5,0.2) .. (-3.3,1.1) .. (-5,3) .. (v);

    \path[draw, thick, blue, ->, use Hobby shortcut] 
    (-8,6.8) .. (-7,6.4) .. (-5,6.6) .. (-5.2,6) .. (-5,4.5) .. (v);

    \path[draw, thick, green, ->, use Hobby shortcut] 
    (-8,6.8) .. (-8.2,5.5) .. (-9,4.9) .. (-8,3.5) .. (v);

    % Position for x 
    \node (T) at (-8,6.8) {}; % Green path start
    \node[circle, fill=black,inner sep=2pt] (Tdot) at (T) {}; % Green path actual start

    % Path starting points 
    \node[circle, fill=black,inner sep=2pt] (A) at (-0.85,1.05) {}; % Red path start

    % Labels
    \node[] at (-6.8,2.8) {\huge$v$};
    \node[] at (-8.2,7.2) {\huge$x$}; % Adjusted label for x

    % Path Labels
    \node[] at (-10,4) {\textcolor{green!60!black}{\huge$p_x^t(v)$}};
    \node[] at (-3.2,6) {\textcolor{blue!60!black}{\huge$p_S^\ell(v)$}};
    \node[] at (-4.5,1) {\textcolor{red!60!black}{\huge$p_S^\ell(v)$}};

\end{tikzpicture}
}
\end{center}
        \label{fig:collision2}
        \end{subfigure}
        }
         \caption{Illustration of $\left \langle p_{x}^{t} , \left(p_S^{\ell}\right)^2 \right \rangle$ (Left) and $\left \langle \left( p_x^t\right)^2, p_S^{\ell}\right \rangle $ (Right). 
    Intuitively, we can think of $\left \langle p_{x}^{t} , \left(p_S^{\ell}\right)^2 \right \rangle $ and $\left \langle \left( p_x^t\right)^2, p_S^{\ell}\right \rangle$ as the expected three-way collision rates between random walks started from $x$ and random walks started from $S$.}
         \label{fig:collision}
\end{figure}
\begin{defn}[Typical vertex with respect to $S$]\label{def:good_pt_wrt_S}
    Given a set $S \subseteq V$ and a vertex $x \in V$,  say that $x$ is \emph{typical with respect to $S$} if it satisfies the following: 
    \begin{enumerate}[label=(\textbf{\arabic*})]
        \item (Avoids bad points) The set $S$ does not contain any bad points from the cluster of $x$, i.e.,  
        
        $$S \cap B_{\delta} \cap C_{i(x)} = \emptyset.$$ \label{con:bad_cluster}
     \item (Not over-represented in $S$) The number of vertices in $S$ from the cluster of $x$ is bounded by 
     
     $$|S \cap C_{i(x)}| \leq O\left(\left(\frac{\varphi^2}{\epsilon}\right)^{1/3}\right).$$ \label{con:ScapC_i}
    \item (Typical under the distribution $p^t_S$) The probability mass assigned to $x$ under $p^t_S$ satisfies 
    
    $$p^{t}_S(x)\leq O^*\left(
    \frac{1}{n} \cdot \frac{k}{|S|}\right)$$ 
    for all $t  \in [2t_{\min}, 2t_{\min} + 2t_{\Delta}]$.
    \label{con:rw1}
       \item (Low correlation with $p_S^t$) The $t$ step distribution of $x$ is weakly correlated with $p_S^t$, that is 
       $$
       \left \langle p_{x}^{t} , \left(p_S^{l}\right)^2 \right \rangle \leq O^*\left(\frac{1}{|S|^2} \cdot \frac{k^{3/2}}{n^{3/2}}\right) \text{ and } \left \langle \left( p_x^t\right)^2, p_S^l\right \rangle \leq O^*\left(\frac{1}{|S|} \cdot \frac{k^{3/2}}{n^{3/2}}\right) $$
       for all $t, l  \in [t_{\min}, t_{\min} + t_{\Delta}]$ (see \Cref{def:params} for $t_{min}$ and $t_{\Delta}$).   Intuitively, can think of $\left \langle p_{x}^{t} , \left(p_S^{l}\right)^2 \right \rangle $ and $\left \langle \left( p_x^t\right)^2, p_S^l\right \rangle$ as the expected three-way collision rates between random walks started from $x$ and random walks started from the set $S$ (see Figure \ref{fig:collision} for illustration). \label{con:rw2} 
     
       \item (Not part of a bad cluster and well-aligned with mean) The vertex $x$ does not belong to any bad cluster (as per \Cref{def:good_S}, {\bf \ref{con:isotropic}}). Additionally, its deviation from the cluster mean satisfies 
       
       $(f_x -\mu_{i(x)})(f_x -\mu_{i(x)})^\top  \bullet \sum_{y \in S }f_yf_y^\top  \leq 10^{-10}\|f_x\|^4_2$. \label{con:rw3}
    \end{enumerate}

\end{defn}
\begin{remark}
Condition {\bf \ref{con:rw3}} can be equivalently expressed as $\sum_{y \in S }\left \langle f_y, f_x -\mu_{i(x)} \right \rangle^2  \leq 10^{-10}\|f_x\|^4_2$. However, our proof showing that most vertices $x$ are typical with respect to $S$ (\Cref{lemma:good_pts_wrt_S} below) uses an averaging argument over  $f_x -\mu_{i(x)}$. For this reason, it is more convenient to view these expressions as linear functions of second tensor powers of $f_x-\mu_{i(x)}$, as opposed to a quadratic function of $f_x-\mu_{i(x)}$.
\end{remark}

Additionally, we introduce the notion of a strongly typical vertex. These are typical vertices, which additionally satisfy a stronger upper-bound in \ref{con:rw2}.  While this definition is not useful for the proof of our main result \Cref{thm:main},  it is important for the trade-off results \Cref{thm:tradeoff}. See \Cref{sec:tradeoff_sketch} for more details.

\begin{defn}[Strongly typical vertex]\label{rem:strongly_typical}
  Given a set $S \subseteq V$ and a vertex $x \in V$,  say that $x$ is \emph{strongly typical with respect to $S$} if it is typical with respect to $S$, and additionally, satisfies the stronger version of  {\bf \ref{con:rw2}}:
 $$
       \left \langle p_{x}^{t} , \left(p_S^{l}\right)^2 \right \rangle \leq O^*\left(\frac{1}{|S|^2} \cdot \frac{k^{2}}{n^{2}}\right) \text{ and } \left \langle \left( p_x^t\right)^2, p_S^l\right \rangle \leq O^*\left(\frac{1}{|S|} \cdot \frac{k^{2}}{n^{2}}\right) $$
for all $t, l  \in [t_{\min}, t_{\min} + t_{\Delta}].$  
\end{defn}

\begin{remark}\label{remark:good_pt}
   It is straightforward to verify that if $x$ is typical (strongly typical) with respect to a set $S$, then $x$ is typical (strongly typical) with respect to every subset $S' \subseteq S$. 
\end{remark}

The following lemma shows that for any well-spread set $S$, almost all $x \in V$ are typical with respect to $S$. 
\begin{restatable}{lemma}{goodptswrtS} \label{lemma:good_pts_wrt_S}
  For every well-spread set $S$, it holds that 
    \[ \left| \left\{ x \in V : \text{x is \emph{not} strongly typical with respect to  S}\right\}\right| \leq  O \left(n \cdot \left(\frac{\epsilon}{\varphi^2}\right)^{1/3}\cdot \log(\varphi^2/\epsilon)\right). \]
\end{restatable}

\Cref{lemma:good_pts_wrt_S} is stated in a form stronger than what is required for the proof of \Cref{thm:main}. For the main theorem, it suffices to bound the number of vertices that are not \emph{typical} with respect to $S$. The stronger statement, which bounds the number of vertices that are not \emph{strongly typical}, is used for obtaining the trade-off result in \Cref{thm:tradeoff}. The lemma follows by several applications of Markov's inequality. The proof is presented in \Cref{sec:well-spread_sets}. 

\subsection{Analysis of \sketch{} (Algorithm \ref{alg:spectralsketch})}\label{sec:sketch}
In this section, we prove the correctness guarantee of the \sketch{} algorithm (Algorithm \ref{alg:spectralsketch}), restated below for the convenience of the reader. 
\sketchguarantee*
The proof consists of two parts: First, we show that the quantity $\left \langle f_x, \sum_{y \in S} \sigma_y f_y\right \rangle $ correctly identifies whether $S$ contains any vertices from the cluster of $x$. This is achieved by \Cref{lemma:grouptest_exact}. Second, we show that $|\langle \sketch(S), \sketch(x) \rangle|$ approximates $\left \langle f_x, \sum_{y \in S} \sigma_y f_y\right \rangle $. This is achieved by \Cref{lemma:rw_to_embedding}. 

We start with the first part. More concretely, we show that with high constant probability (over the randomness of $\sigma$), the inner product $\left| \left \langle f_x, \sum_{y \in S} \sigma_y f_y\right \rangle\right| $ is large when $S$ contains a vertex from  $C_{i(x)}$, and small otherwise. We need the following technical lemma, which will be useful for bounding the variance $\operatorname{Var}_{\sigma}\left[ \left\langle f_x, \sum_{y \in S} \sigma_y f_y\right\rangle\right]$. 
\begin{restatable}{lemma}{boundingvar}\label{lemma:bounding_var}
    Let $S \subseteq V$ be a set, and suppose that $x$ is typical with respect to $S$ (as per \Cref{def:good_pt_wrt_S}). Then  
    $$f_x f_x^\top  \bullet \sum_{y \in S\setminus C_{i(x)}} f_y f_y^\top  \leq 10^{-6}\cdot \|f_x\|^4_2,$$
and 
    $$f_x f_x^\top  \bullet \sum_{y \in S}f_yf_y^\top  \leq (\varphi^2/\epsilon)^{1/3}\cdot \|f_x\|^4_2.$$
In other words, $f_x$ has negligible correlation with the spectral embeddings from other clusters, and a bounded correlation with the entire set $S$. 
\end{restatable}
\noindent The proof of \Cref{lemma:bounding_var} follows from \Cref{def:good_pt_wrt_S}, {\bf \ref{con:isotropic}} and is provided in Appendix \ref{sec:well-spread_sets}. 

We now show that with high constant probability (over the randomness of $\sigma$), the inner product \\ $\left| \left \langle f_x, \sum_{y \in S} \sigma_y f_y\right \rangle\right| $ is large when $S$ contains at least one vertex from  $C_{i(x)}$, and small otherwise. 
\begin{lemma}\label{lemma:grouptest_exact}  Let $S \subseteq V$, and let $x \in V \setminus B_{\delta}$ be typical with respect to $S$ (as per \Cref{def:good_pt_wrt_S}). Let $\sigma \sim Unif(\{-1,1\}^S)$. 
    Then the following holds.
    If $C_{i(x)} \cap S \neq \emptyset$, then 
    \[\Pr_{\sigma}\left[   \left|\left\langle f_x,  \sum_{y \in S} \sigma_y f_y \right\rangle\right| \geq 0.89 \|f_x\|^2_2\right] \geq 0.4999.\] 
    Otherwise, if $C_{i(x)} \cap S = \emptyset$, then 
        \[ \Pr_{\sigma}\left[  \left|\left\langle f_x,  \sum_{y \in S} \sigma_y f_y \right\rangle\right| \leq 0.1 \|f_x\|^2_2\right] \geq 0.9999.\]
\end{lemma}
\begin{proof}

We first consider the case of $C_{i(x)} \cap S = \emptyset$, and then the case of  $C_i(x) \cap S \neq \emptyset$. 
\paragraph{Case 1: $C_{i(x)} \cap S = \emptyset$.} In that case we have $S = S \setminus C_{i(x)}$.  We will show that for all $x \in V \setminus B_{\delta}$, it holds that

\begin{equation*}
   \Pr_{\sigma}\left[  \left|\left \langle f_x,  \sum_{y \in S\setminus C_{i(x)}} \sigma_y f_y \right \rangle\right| \leq 0.1 \|f_x\|^2_2\right] \geq 0.9999.
\end{equation*}
\noindent 
Define the random variable 
 $X \coloneqq \left \langle f_x,  \sum_{y \in S\setminus C_{i(x)}} \sigma_y f_y \right \rangle$ over the randomness of $\sigma$. 
We have 
\[ \mathbb{E}[X] = 0\]
and 
\begin{align*}
    \operatorname{Var}[X] & \leq \E[X^2]  = \sum_{y \in S\backslash C_{i(x)}}\left \langle f_x, f_y\right \rangle^2 \leq 10^{-6} \|f_x\|^4_2,
\end{align*}
where the last inequality holds by \Cref{lemma:bounding_var}. So by Chebyshev's inequality, we have 

$$ \Pr_{\sigma}\left[  \left|\left \langle f_x,  \sum_{y \in S\setminus C_{i(x)}} \sigma_y f_y \right \rangle\right| > 0.1 \|f_x\|^2_2\right] = \Pr_{\sigma}[|X| > 0.1\|f_x\|_2^2]\leq \frac{10^{-6}\|f_x\|^4_2}{0.01\|f_x\|^4_2} \leq 0.0001,$$
which concludes Case 1. 

\paragraph{Case 2: $C_{i(x)} \cap S \neq \emptyset$.} 
By triangle inequality, we have 

\[    \left|\left\langle f_x,  \sum_{y \in S} \sigma_y f_y \right\rangle\right| \geq \left|\left\langle f_x,  \sum_{y \in S \cap C_{i(x)}} \sigma_y f_y \right\rangle\right| -  \left|\left\langle f_x,  \sum_{y \in S \setminus C_{i(x)}} \sigma_y f_y \right\rangle\right| . \]
\noindent 
We bound each of the two terms separately. 
First, we upper bound the second term. By an identical proof as in Case 1, we have 
\begin{equation}\label{eqn:grouptest_exact_1}
\Pr_{\sigma}\left[  \left|\left \langle f_x,  \sum_{y \in S\setminus C_{i(x)}} \sigma_y f_y \right \rangle\right| \leq 0.1 \|f_x\|^2_2\right] \geq 0.9999.
\end{equation}
This upper bounds the second term. 

Now we just need to lower bound the first term. We first show that $\langle f_x, f_y\rangle $ is large for all $y \in S \cap C_{i(x)}$. Given $y \in C_i(x) \cap S$, it holds that $y \in C_{i(x)} \setminus B_{\delta}$. This is because $x$ is typical for $S$ (as per \Cref{def:good_pt_wrt_S}), so by \Cref{def:good_pt_wrt_S}, {\bf \ref{con:bad_cluster}}, the set $S$ contains no bad points from cluster $C_{i(x)}$, i.e., $S \cap B_{\delta} \cap C_{i(x)} = \emptyset$

Therefore, since both $x$ and $y$ belong to $C_{i(x)} \setminus B_{\delta}$, we can apply \Cref{lemma:good_pts_properties} and \Cref{bulletpt:f_norm} to $x$ and $y$ and obtain 
\begin{equation}\label{eq:grouptest_exact_fxfy}
    \begin{aligned}
        \left \langle f_x, f_y \right\rangle & \geq \left(1-\frac{4 \sqrt{\epsilon}}{\varphi}-4\sqrt{\delta}\right)\frac{1}{|C_{i(x)}|} \geq \left(1-\frac{12 \sqrt{\epsilon}}{\varphi}-12\sqrt{\delta}\right)\|f_x\|^2_2 ,  \\
        \left \langle f_x, f_y \right \rangle &\leq\left(1+\frac{4 \sqrt{\epsilon}}{\varphi}+4\sqrt{\delta}\right)\frac{1}{|C_{i(x)}|} \leq\left(1+\frac{12 \sqrt{\epsilon}}{\varphi}+12\sqrt{\delta}\right) \|f_x\|^2_2,
    \end{aligned}
\end{equation}
where the last inequalities follow by the assumption that $\delta = O\left(\left( \epsilon/\varphi^2\right)^{2/3}\right)$ as per \Cref{def:params}. 
Thus, each of the terms in $\sum_{y \in C_{i(x)} \cap S}\sigma_y \langle f_x, f_y\rangle$ has a relatively large absolute value. 
Next, we show that with sufficiently good probability, the sum does not cancel out. Intuitively, this is because cancellation only occurs if exactly half of the $\sigma_y$'s are positive and exactly half are negative, which does not happen with too high a probability. Formally, given $\sigma \in \{-1,1\}^S$, 
let 
$$t_+ \coloneqq |\{ y \in S \cap C_{i(x)} : \sigma_y = 1\}|$$ denote the number of positive signs assigned to $y$'s in  $S \cap C_{i(x)}$ and let 
$$t_- \coloneqq |\{ y \in S \cap C_{i(x)} : \sigma_y = -1\}| $$denote the number of negative signs assigned to  $y$'s in $S \cap C_{i(x)}$. 
Let $\mathcal{E}_{\mathrm{equal}}$ be the event that $t_+ = t_-$.

\begin{restatable}{claim}{Equal}\label{claim:Equal}$ \Pr_{\sigma} \left[\mathcal{E}_{\mathrm{equal}} \right] \leq 0.5. $
\end{restatable}
\noindent The proof follows from  Stirling's formula, and is included in Appendix \ref{sec:well-spread_sets}.
We now show that conditioned on the event that  $t_+ \neq t_-$ (i.e. $\bar{\mathcal{E}}_{equal}$),  the sum does not cancel out. Without loss of generality, suppose that $t_+ > t_-$ (the case when $t_+ < t_-$ is similar). Then 
\begin{align*}
\left|\left \langle f_x, \sum_{y\in C_{i(x)} \cap S}\sigma_yf_y\right \rangle\right| &\geq \left \langle f_x, \sum_{y\in C_{i(x)} \cap S}\sigma_yf_y\right \rangle \\
& \geq \left(1-\frac{12 \sqrt{\epsilon}}{\varphi}-12\sqrt{\delta}\right)\|f_x\|^2_2\cdot t_{+} - \left(1+\frac{12 \sqrt{\epsilon}}{\varphi}+12\sqrt{\delta}\right)\|f_x\|^2_2 \cdot t_- \\ 
& = (t_+ - t_-) \cdot \|f_x\|^2_2 - 
(t_+ + t_-)\left(\frac{12 \sqrt{\epsilon}}{\varphi}+12\sqrt{\delta}\right)\|f_x\|^2_2  \\& \geq \|f_x\|^2_2 - |S \cap C_{i(x)}|\left(\frac{12 \sqrt{\epsilon}}{\varphi}+12\sqrt{\delta}\right)\|f_x\|^2_2  \\
& \geq  0.99\|f_x\|^2_2.
\end{align*}
Here, the second transition follows by \Cref{eq:grouptest_exact_fxfy}. The third transition simply follows by grouping the $\|f_x\|^2_2$ together and the $\left(\frac{12 \sqrt{\epsilon}}{\varphi}+12\sqrt{\delta}\right)\|f_x\|^2_2$ terms together. The fourth transition uses the assumption that $t_+ > t_-$, and in particular $t_+ - t_- \geq 1$, and also that $t_+ + t_- = |S \cap C_{i(x)}|$ by definition of $t_+$ and $t_-$. Finally, the last transition follows since $|S \cap C_{i(x)}| \leq O\left((\varphi^2/\epsilon)^{1/3}\right)$ which is a property of $x$ as a typical point with respect to $S$ (\Cref{def:good_pt_wrt_S}, {\bf \ref{con:ScapC_i}}) together with the assumption that $\delta \leq c \cdot (\epsilon/\varphi^2)^{2/3}$ in the lemma statement.

In particular, we have 

\begin{equation}\label{eqn:grouptest_exact_2}
\Pr_{\sigma}\left[   \left |\left \langle f_x,  \sum_{y \in S\cap C_{i(x)}} \sigma_y f_y \right\rangle\right| \geq 0.99 \|f_x\|^2_2\right] \geq 1-\Pr[\mathcal{E}_{\mathrm{equal}}] \geq0.5.
\end{equation}
Finally, to conclude the case $S\cap C_{i(x)}\neq \emptyset$, take a union bound over failure events in Equations \eqref{eqn:grouptest_exact_1} and \Cref{eqn:grouptest_exact_2}, to get that with probability at least $1 - 0.5-0.0001 = 0.4999$, it holds that 

\[    \left|\left\langle f_x,  \sum_{y \in S} \sigma_y f_y \right\rangle\right| \geq \left|\left\langle f_x,  \sum_{y \in S \cap C_{i(x)}} \sigma_y f_y \right\rangle\right| -  \left|\left\langle f_x,  \sum_{y \in S \setminus C_{i(x)}} \sigma_y f_y \right\rangle\right| \geq 0.99\|f_x\|^2_2 - 0.1\|f_x\|^2_2 =0.89\|f_x\|^2_2,\] which give the result.  
\end{proof}

\Cref{lemma:grouptest_exact} shows that we can use the inner product of the spectral embeddings $\left \langle f_x, \sum_{y \in S}\sigma_y f_y \right \rangle $ to test whether $S$ contains a vertex from the cluster of $x$. Next, we show that $\langle \sketch(x), \sketch(S) \rangle $ approximates $\left \langle f_x, \sum_{y \in S}\sigma_y f_y \right \rangle $. Recall that $\sketch(S)$ (\Cref{alg:spectralsketch}) outputs $\sum_{y \in S} \sigma_y \sum_{t} c_t \cdot \widehat{p}^t_y$, where $\widehat{p}^t_y$ denotes the empirical distribution of the $t$-step lazy random walks from $y$, and $c_t$ is the coefficient of $x^t$ in the polynomial $p(x)$ from \Cref{thm:standardbasis}. The following lemma shows that if we run $\approx \sqrt{\frac{n}{k}}$ lazy random walks from each vertex, then $\left \langle \widehat{p}^{t_1}_x, \sum_{y \in S} \sigma_y \widehat{p}^{t_2}_y\right\rangle$ concentrates well around its mean $\left \langle M^{t_1} \1_x, \sum_{y \in S} \sigma_y M^{t_2} \1_{y}\right \rangle $.

Suppose that we run $Q$ lazy random walks of length $t_1$ from $x$ and $R$ lazy random walks of length $t_2$ from every $y \in S$. Let $\widehat{p}_x^{t_1}$ denote the empirical distribution of the random walks from $x$, and for each $y \in S$, let $\widehat{p}_y^{t_2}$ denote the empirical distribution of the random walks from $y$. Then we have the following: 

\begin{restatable}[Collision counting]{lemma}{collisioncounting} \label{lemma:variance_calc}
Let $\beta, \gamma \geq 1$ be parameters, let $x \in V,$  $S \subseteq V$ and $t_1, t_2 \geq 1$. Suppose that the following conditions hold: 
    \begin{enumerate}[label=(\textbf{\arabic*})]
    \item $p^{t_1+t_2}_S(x) \leq \beta \cdot \frac{1}{n}\cdot  \frac{k}{|S|}$, \label{con:p^{t1+t2}_S}
    \item   $\left \langle p_{x}^{t_1} , \left(p_S^{t_2}\right)^2 \right \rangle \leq \gamma \cdot \frac{1}{n^2} \cdot \frac{k^2}{|S|^2} \text{ and } \left \langle \left( p_x^{t_1}\right)^2, p_S^{t_2}\right \rangle \leq \gamma\cdot \frac{1}{n^2} \cdot \frac{k^2}{|S|}.$\label{con:p^{t1}_x, p^{t2}_S}
\end{enumerate}
For every $\rho >0$ (desired failure probability), $\xi > 0$ (desired precision), $\sigma \in \{-1,1\}^{S}$ (sign assignment), if 

\[ Q \cdot R \geq \frac{7\beta}{\rho \cdot  \xi^2} \cdot \frac{n}{k} \text{ and }  Q, R \geq  \frac{7(\gamma+\beta^2)}{\rho \cdot  \xi^2}\]
 then, 
$$\Pr_{\text{random walks}}\left[\left| \left \langle \widehat{p}_x^{t_1}, \sum_{y \in S}\sigma_y \widehat{p}_y^{t_2}\right \rangle  - \left\langle M^{t_1} \1_x, \sum_{y \in S} \sigma_y M^{t_2} \1_y \right \rangle\right| \leq \xi \cdot \frac{k}{n}\right] \geq 1-\rho.$$ 
\end{restatable}
\noindent 
The proof is obtained by carefully bounding the variance of the random walks, and is presented in \Cref{section:collision_counting}.
\begin{lemma}\label{lemma:rw_to_embedding}
    Suppose that $x \in V \setminus B_{\delta}$ is typical with respect to $S$ (as per \Cref{def:good_pt_wrt_S}). Let $\sigma \sim Unif(\{-1,1\}^S)$. 
    Then there exists $r =  O^*\left(\sqrt{\frac{n}{k}} \cdot   \left(\frac{1}{\epsilon}\right)^{ O(t_{\Delta})} \right)$
    such that if $\widehat p^t_x = \randomwalks(r,t,x)$, then 
    \begin{equation}\label{eq:rw_to_embedding_main}
         \Pr_{\substack{\sigma, \\\text{random walks}}}
    \left[\left| \left \langle \sum_t c_t \cdot \widehat{p}^t_x, \sum_{y \in S} \sigma_y \sum_{t} c_t \cdot \widehat{p}^t_y\right \rangle - \left \langle f_x, \sum_{y \in S}\sigma_y f_y\right \rangle \right| \geq 10^{-2} \|f_x\|^2_2 \right] \leq 10^{-5},
    \end{equation}
    where $c_t$ is the coefficient of $x^t$ in the polynomial $p(x)$ from \Cref{thm:standardbasis}.  
\end{lemma}\begin{proof}
   By the triangle inequality, we have 
    \begin{equation} \begin{aligned}\label{eq:rw_to_embedding_two_terms}
    & \left|\left \langle \sum_t c_t \cdot \widehat{p}^t_x,\sum_{\substack{ y \in S, t \geq 0} } c_t \cdot \sigma_y  \widehat{p}^t_y \right\rangle - \left\langle f_x, \sum_{y \in S}\sigma_y f_y \right\rangle \right| \leq A_1 + A_2
    \end{aligned}
    \end{equation}
    where 
    $$A_1\coloneqq \left| \left \langle \sum_t c_t \cdot \widehat{p}^t_x, \sum_{\substack{y \in S, t \geq 0} } c_t \cdot \sigma_y  \widehat{p}^t_y \right\rangle -   \left \langle p( M) \1_x, \sum_{y \in S}   \sigma_y p(M) \1_y\right \rangle\right|$$
    and
    $$ A_2 \coloneqq \left|\left \langle p( M) \1_x, \sum_{y \in S}   \sigma_y p(M) \1_y\right \rangle - \left \langle f_x, \sum_{y \in S}\sigma_y f_y\right\rangle \right|.$$
   Note that $A_1$ captures the error introduced by taking finite samples from random walk distributions, whereas $A_2$ reflects the errors in approximating the spectral embedding that the polynomial $p$ introduces. We now bound $A_1$ and $A_2$. 
    \paragraph{Bounding $A_1$:}
    We have
    \begin{align*}
        A_1&= \left| \left \langle \sum_t c_t \cdot \widehat{p}^t_x, \sum_{\substack{y \in S, t \geq 0} } c_t \cdot \sigma_y  \widehat{p}^t_y \right\rangle -   \left \langle p( M) \1_x, \sum_{y \in S}   \sigma_y p(M) \1_y\right \rangle\right| \\
        &= \left| \left \langle \sum_t c_t \cdot \widehat{p}^t_x, \sum_{\substack{y \in S, t \geq 0} } c_t \cdot \sigma_y  \widehat{p}^t_y \right\rangle -   \left \langle \sum_t c_t \cdot M^t \1_x, \sum_{\substack{y \in S, t \geq 0}}   c_t \cdot \sigma_y M^t \1_y\right \rangle\right| \\
        & \leq \sum_{t_1, t_2 \in [t_{\min}, t_{\min} + t_{\Delta}]} |c_{t_1}||c_{t_2}|  \left| \left\langle \widehat{p}_x^{t_1}, \sum_{y \in S}\sigma_y \widehat{p}_y^{t_2}\right \rangle  - \left\langle M^{t_1} \1_x, \sum_{y \in S} \sigma_y M^{t_2} \1_y \right\rangle\right|, 
    \end{align*}
    where the last inequality follows by the triangle inequality. 
    
For every  $t_1, t_2 \in [t_{\min}, t_{\min} + t_{\Delta}]$. Since $x$ is typical with respect to $S$ (as per  \Cref{def:good_pt_wrt_S}), we have  $p^{t_1+t_2}_S(x)\leq O^*\left(
    \frac{1}{n} \cdot \frac{k}{|S|}\right)$, 
    $\left \langle p_{x}^{t_1} , \left(p_S^{t_2}\right)^2 \right \rangle \leq O^*\left(\frac{1}{|S|^2} \cdot \frac{k^{3/2}}{n^{3/2}}\right),$ and $ \left \langle \left( p_x^{t_1}\right)^2, p_S^{t_2}\right \rangle  \leq O^*\left(\frac{1}{|S|} \cdot \frac{k^{3/2}}{n^{3/2}}\right)$ by \Cref{def:good_pt_wrt_S}, {\bf\ref{con:rw1}} and  {\bf\ref{con:rw2}}.  
    So we can apply \Cref{lemma:variance_calc} with $\beta = O^*(1)$, $\gamma = O^*(\sqrt{n/k})$, $\rho = 10^{-6}\frac{1}{(t_{\Delta}+1)^2} = \Omega^*(1)$ and $\xi = \frac{5 \cdot 10^{-3}}{\eta \cdot |c_{t_1}||c_{t_2}|(t_{\Delta}+1)^2} =\Omega^*\left( \frac{1}{(1/\epsilon)^{t_{\Delta}}}\right)$.
Here, the last equality holds by the bound on the coefficients of $p$ in \Cref{thm:standardbasis}.
 Setting $Q, R =   \frac{7(\gamma+\beta^2)}{\rho \cdot  \xi^2} =  O^*\left(\sqrt{\frac{n}{k}} \cdot \left(\frac{1}{\epsilon}\right)^{ O(t_{\Delta})}\right)$ in \Cref{lemma:variance_calc},
 and using the fact that $\|f_x\|^2_2 = \Theta(\frac{k}{n})$ for all $x \in  V \setminus B_{\delta}$ (see \Cref{remark:norm}),
 we obtain
$$\Pr_{\text{random walks}}\left[\left| \left\langle \widehat{p}_x^{t_1}, \sum_{y \in S}\sigma_y \widehat{p}_y^{t_2}\right \rangle  - \left\langle M^{t_1} \1_x, \sum_{y \in S} \sigma_y M^{t_2} \1_y \right \rangle\right| \leq \frac{5 \cdot 10^{-3}}{|c_{t_1}| |c_{t_2}|(t_{\Delta}+1)^2} \cdot \|f_x\|^2_2 \right] \geq 1- 10^{-6}\frac{1}{(t_{\Delta}+1)^2}.$$ 
By taking a union bound over the $(t_{\Delta} + 1)^2$ possible pairs of $t_1, t_2 \in [t_{\min}, t_{\min} + t_{\Delta}]$, we get that 

 \begin{equation}\label{eq:rw_to_embedding_term1}
 \begin{aligned}
&  \Pr_{\text{random walks}}\left [ A_1 > 5\cdot 10^{-3} \|f_x\|^2_2\right ]\\
=&  
     \Pr_{\text{random walks}} \left[ \left| \left \langle \sum_t c_t \cdot \widehat{p}^t_x, \sum_{\substack{y \in S, t \geq 0} } c_t \cdot \sigma_y   \widehat{p}^t_y \right \rangle -   \left\langle \sum_t c_t \cdot M^t \1_x, \sum_{\substack{y \in S, t \geq 0}}   c_t \cdot \sigma_y M^t \1_y\right \rangle\right| > 5 \cdot 10^{-3} \|f_x\|^2_2 \right] \\
     & \leq 10^{-6}.
\end{aligned}
 \end{equation}
 This completes the bound on $A_1$. 

 \paragraph{Bounding $A_2$: } We now turn to bounding $A_2$.
 Let $U\Sigma U^\top $ denote the eigendecomposition of $M$, where $\Sigma$ is the diagonal matrix of eigenvalues of $M$ arranged in \emph{descending} order.   We have $\sum_{t} c_t M^t = p(M)= Up(\Sigma)U^\top  $. Recall that $U_{[k]}$ is the submatrix of $U$ which consists of its first $k$ columns and $f_x = U^\top _{[k]}\1_x$ for all $x \in V$. Then,
\begin{align*}
    \left \langle f_x, \sum_{y \in S}\sigma_y f_{y}\right \rangle =  \1_x^\top  U_{[k]}U_{[k]}^\top \left(\sum_{y \in S} \sigma_y  \1_y \right) 
\end{align*}
 Denote by $\Sigma_{[k]}$ the $k\times k$ diagonal matrix whose diagonal elements are equal to the top $k$ elements of $\Sigma$, and by $\Sigma_{[-k]}$ the $(n-k) \times (n-k)$ matrix whose diagonal elements are equal to the bottom $n-k$ elements of $\Sigma$.  Furthermore, write $U_{[-k]}$ for the matrix consisting of last $n-k$ columns of $U$. Then 
\begin{equation}\label{id:Sigma}
    Up(\Sigma)^2U^\top  = U_{[k]}p(\Sigma_{[k]})^2U_{[k]}^\top  + U_{[-k]}p(\Sigma_{[-k]})^2U_{[-k]}^\top .
\end{equation}

\noindent Hence, we can rewrite $A_2$ as 
 
\begin{equation}\label{eq:rw_to_embedding_polynomial}
 \begin{aligned}
A_2 &= \left|\left \langle p(M)\1_x, p(M)\sum_{y \in S}\sigma_y\1_y\right \rangle  - \left\langle f_x, \sum_{y \in S}\sigma_y f_{y}\right \rangle\right| \\
& = \left|\mathbbm{1}^\top _xUp(\Sigma)^2U^\top \sum_{y\in S}\sigma_y\1_y - \mathbbm{1}^\top _xU_{[k]}U_{[k]}^\top \sum_{y\in S}\sigma_y\1_y\right|  \\
&= \left|\1_x^\top U_{[k]}\left(p(\Sigma_{[k]})^2 - I_k\right)U^\top _{[k]}\sum_{y \in S}\sigma_y\1_y + \1_x^\top U_{[-k]}p(\Sigma_{[-k]})^2U_{[-k]}^\top \sum_{y \in S}\sigma_y\1_y\right|, \qquad \qquad  \qquad \text{by  Equation} \eqref{id:Sigma}\\
& \leq \left|f_x^\top (p(\Sigma_{[k]})^2 - I_k)\sum_{y\in S}\sigma_y f_y\right| + \left|\mathbbm{1}^\top _xU_{[-k]}p(\Sigma_{[-k]})^2U_{[-k]}^\top \sum_{y\in S}\sigma_y\1_y\right|, \ \text{by triangle inequality and } f_x = U_{[k]}^\top \1_x.
 \end{aligned}
 \end{equation}

 \noindent Let us start by bounding the second summand. Since the matrix $U_{[-k]}p(\Sigma_{[-k]})^2U_{[-k]}^\top$ is PSD, we have 
 \begin{align*}
 \left|\mathbbm{1}^\top _xU_{[-k]}p(\Sigma_{[-k]})^2U_{[-k]}^\top \sum_{y\in S}\sigma_y\1_y\right| 
 &\leq \mathbbm{1}^\top _xU_{[-k]}p(\Sigma_{[-k]})^2U_{[-k]}^\top  \mathbbm{1}_x + \left( \sum_{y\in S} \sigma_y \mathbbm{1}_y\right)^\top U_{[-k]}p(\Sigma_{[-k]})^2U_{[-k]}^\top  \left( \sum_{y\in S} \sigma_y \mathbbm{1}_y\right) \\
 & \leq\max_{z \in [0, 1 - \varphi^2/4]}p(z)^2 \left \|U_{[-k]}^\top  \mathbbm{1}_x\right \|^2_2 + \max_{z \in [0, 1 - \varphi^2/4]}p(z)^2 \left\| U_{[-k]}^\top  \left( \sum_{y\in S} \sigma_y \mathbbm{1}_y\right)\right\|^2_2,
 \end{align*}
 where the last line follows from the fact that the bottom diagonal $n - k$ elements of $\Sigma$ are bounded by $1 - \varphi^2/4$ (see \Cref{remark:Meigengap}).  By \Cref{thm:standardbasis}, we have $\max_{z \in [0, 1 - \varphi^2/4]}p(z) \leq n^{-4}$. So we can bound the second summand in \Cref{eq:rw_to_embedding_polynomial} as 
 \begin{equation}\label{eq:rw_to_embedding_small}
 \begin{aligned}
  \left|\mathbbm{1}^\top _xU_{[-k]}p(\Sigma_{[-k]})^2U_{[-k]}^\top \sum_{y\in S}\sigma_y\1_y\right| 
   & \leq\max_{z \in [0, 1 - \varphi^2/4]}p(z)^2 \left \|U_{[-k]}^\top  \mathbbm{1}_x\right \|^2_2 + \max_{z \in [0, 1 - \varphi^2/4]}p(z)^2 \left\| U_{[-k]}^\top  \left( \sum_{y\in S} \sigma_y \mathbbm{1}_y\right)\right\|^2_2 \\
    &\leq n^{-8}  \left \|U_{[-k]}^\top  \mathbbm{1}_x\right\|^2_2 + n^{-8} \left\| \left(U_{[-k]}^\top  \sum_{y\in S} \sigma_y \mathbbm{1}_y\right)\right\|^2_2  \\
    & \leq n^{-8}  \left \| \mathbbm{1}_x\right\|^2_2 + n^{-8} \left\|  \sum_{y \in S} \sigma_y \mathbbm{1}_y\right\|^2_2 \\
    & \leq n^{-8}\cdot 1  + n^{-8}\cdot k^2 \\
    & \leq 10^{-3} \|f_x\|^2_2, 
\end{aligned}
 \end{equation}
where the last inequality holds for $n$ sufficiently large, using the fact that $1\leq k \leq n$ and $\|f_x\|^2_2 = \Omega(k/n)$, by \Cref{remark:norm}. 

It remains to bound the first summand in \Cref{eq:rw_to_embedding_polynomial}. We will use the randomness of $\sigma$. Define the random variable 
\[ X \coloneqq f_x^\top (p(\Sigma_{[k]})^2 - I_k)\sum_{y\in S}\sigma_y f_y\] over the randomness of $\sigma.$
We have 
\[ \E[X] = 0\]
and 
\begin{align*}
    \operatorname{Var}[X] & \leq \E[X^2] \\
    & = \sum_{y \in S} \left\langle (p(\Sigma_{[k]})^2 - I_k) f_x, f_y\right \rangle^2 \\ 
    & \leq O\left((\epsilon/\varphi^2)^2\right)\sum_{y \in S} \left\langle f_x,  f_y \right \rangle ^2  &&  \text{ by \Cref{thm:standardbasis} } \\
    & = O\left((\epsilon/\varphi^2)^2\right) \cdot f_x f_x^\top  \bullet \sum_{y \in S}f_yf_y^\top  \\
    & \leq  O\left((\epsilon/\varphi^2)^2 \cdot (\varphi^2/\epsilon)^{1/3}\right)\cdot \|f_x\|^4_2 &&  \text{by \Cref{lemma:bounding_var}, since $x$ is typical with respect to $S$} \\
    & \leq 16 \cdot 10^{-12} \|f_x\|^4_2 &&  
    \text{since } \epsilon/\varphi^2 \text{ is less than a small constant by \Cref{rem:param_assumptions}.}
\end{align*}
\noindent 
So by Chebyshev's inequality, we get
\begin{equation}\label{eq:rw_to_embedding_chebyshev}
\Pr_{\sigma}\left[\left|f_x^\top (p(\Sigma_{[k]})^2 - I_k)\sum_{y\in S}\sigma_y f_y\right| \geq 4 \cdot 10^{-3} \|f_x\|^2_2 \right] = \Pr_{\sigma}\left[|X| \geq 4 \cdot 10^{-3}\|f_x\|^2_2 )\right]\leq 10^{-6}.
\end{equation}
Combining \Cref{eq:rw_to_embedding_polynomial}, \Cref{eq:rw_to_embedding_small} and \Cref{eq:rw_to_embedding_chebyshev} , we obtain 
\begin{equation}\label{eq:rw_to_embedding_term2} 
 \Pr_{\sigma}\left[A_2 \geq 5 \cdot 10^{-3}\|f_x\|^2_2  \right] = 
    \Pr_{\sigma}\left[ \left|\left \langle p(M)\1_x, p(M)\sum_{y \in S}\sigma_y\1_y\right \rangle - \left\langle f_x, \sum_{y \in S}\sigma_y f_y\right \rangle 
 \right|  \geq 5 \cdot 10^{-3}\|f_x\|^2_2\right] \leq 10^{-6}. 
\end{equation}

The lemma now follows by taking a union bound over \Cref{eq:rw_to_embedding_term1} and \Cref{eq:rw_to_embedding_term2}  to bound both of the two terms in \Cref{eq:rw_to_embedding_two_terms}. 
\end{proof}

 \begin{remark}[Trade-offs for strongly typical vertices] \label{rem:tradeoff_modific}We can modify \Cref{lemma:rw_to_embedding} in order to achieve the trade-offs in \Cref{thm:tradeoff} as follows. 
 If we strengthen the assumption in \Cref{lemma:rw_to_embedding} to assume that $x$ is \emph{strongly typical} with respect to $S$ (as per \Cref{rem:strongly_typical}) instead of merely typical, then we can strengthen the conclusion of \Cref{lemma:rw_to_embedding} to: 
 
 For every $\delta' \in [0,1]$, there exists $q=O^*\left(\left(\frac{n}{k}\right)^{\delta'} \cdot \left( \frac{1}{\epsilon}\right)^{O(t_{\Delta})}\right)$ and $r = O^*\left(\left(\frac{n}{k}\right)^{1-\delta'} \cdot \left( \frac{1}{\epsilon}\right)^{O(t_{\Delta})}\right)$, such that if $\widehat p_x^t = \randomwalks(q,t,x)$ and $\widehat p_y^t = \randomwalks(r,t,y)$ for every $y \in S$, then \Cref{eq:rw_to_embedding_main} holds. 

This is because if $x$ is strongly typical with respect to $S$, then the two conditions in  \Cref{lemma:variance_calc} are satisfied with  $  \beta = O^*(1)$ and  $\gamma  = O^*(1)$. As a result, we can apply \Cref{lemma:variance_calc} with 
$$ Q = \left( \frac{n}{k}\right)^{\delta'} \cdot  \frac{7(\gamma+\beta^2)}{\rho \cdot  \xi^2} =  O^*\left(\left(\frac{n}{k}\right)^{\delta'} \cdot \left(\frac{1}{\epsilon}\right)^{ O(t_{\Delta})}\right)$$
and 
$$ R = \left( \frac{n}{k}\right)^{1-\delta'} \cdot  \frac{7(\gamma+\beta^2)}{\rho \cdot  \xi^2} =  O^*\left(\left(\frac{n}{k}\right)^{1-\delta'} \cdot \left(\frac{1}{\epsilon}\right)^{ O(t_{\Delta})}\right) $$
when bounding the term $A_1$ in the proof of \Cref{lemma:rw_to_embedding}. The bound on $A_2$ and its proof remain unchanged. 
 \end{remark}

Similarly, we can show that $| \langle \sketch(x), \sketch(x)\rangle | $ approximates $\|f_x\|^2_2$.
\begin{restatable}{lemma}{sketchxguarantee}\label{lemma:sketchx}
For every vertex $x \in V \setminus B_\delta$, it holds that 
    $$\Pr\left [\left|\langle \sketch(x), \sketch(x)\rangle  - \|f_x\|^2_2 \right| \leq 0.01 \|f_x\|^2_2\right] \leq 10^{-6}. $$
In the above, the probability is over the internal randomness of \sketch{},
and the two invocations of \sketch{} are independent.
\end{restatable}
 The proof is similar to \Cref{lemma:rw_to_embedding}, and is included in Appendix $\ref{sec:spectral_facts}$. 
We are now ready to prove the main guarantee of $\sketch{}$ (\Cref{lemma:spectralsketch}), restated below for the convenience of the reader. 
\sketchguarantee*
\begin{proof}
We proceed as follows.
We first apply Lemma~\ref{lemma:grouptest_exact} together with Lemma~\ref{lemma:rw_to_embedding} to the set $S$ to show that the ratio between
$|\langle \sketch(S), \sketch(x)\rangle|$ and the true norm $\|f_x\|_2^2$ correctly indicates whether $S$ contains a vertex from $C_{i(x)}$.  We then use Lemma~\ref{lemma:sketchx} to argue that $|\langle \sketch(x), \sketch(x)\rangle|$ provides a good approximation to $|f_x|_2^2$.

We now relate $|\langle \sketch(S), \sketch(x)\rangle|$  to $\|f_x\|^2_2$ . Recall that $\sketch(S)$ outputs $\sum_{y \in S} \sigma_y \sum_{t} c_t \cdot \widehat{p}^t_y$ and $\sketch(x)$ outputs $\sigma_x \sum_{t} c_t \cdot \widehat{p}^t_x$. 

Suppose that $C_{i(x)} \cap S \neq \emptyset$. Then, by taking a union bound over the failure events in \Cref{lemma:grouptest_exact} and \Cref{lemma:rw_to_embedding}, with probability at least $1- 10^{-5}-0.4999 \geq 0.49$, it holds that 
\begin{equation}\label{eq:sketchSX1}
\begin{aligned}
    \langle \sketch(x), \sketch(S) \rangle| & = \left|\left\langle \sum_{t} c_t \cdot \widehat{p}^t_x,   \sum_{y \in S} \sigma_y \sum_{t} c_t \cdot \widehat{p}^t_y \right\rangle\right| \\
    & \geq \left| \left \langle f_x, \sum_{y \in S} \sigma_y f_y \right \rangle \right|-  \left|\left \langle \sum_{t} c_t \cdot \widehat{p}^t_x,   \sum_{y \in S} \sigma_y \sum_{t} c_t \cdot \widehat{p}^t_y \right \rangle - \left\langle f_x, \sum_{y \in S} \sigma_y f_y \right \rangle \right|   \\
    & \geq 0.89 \|f_x\|^2_2 - 10^{-2}\|f_x\|^2_2  \qquad \qquad \qquad  \text{by \Cref{lemma:grouptest_exact} and \Cref{lemma:rw_to_embedding}} \\
    & \geq 0.88 \|f_x\|^2_2. 
\end{aligned}
\end{equation}
Similarly, if instead $C_{i(x)} \cap S = \emptyset$, then by taking a union bound over the failure events in \Cref{lemma:rw_to_embedding} and \Cref{lemma:grouptest_exact}, with probability at least $1- 10^{-5}-0.0001 \geq 0.9998$, it holds that
\begin{equation}\label{eq:sketchSX2}
\begin{aligned}
    | \langle \sketch(x), \sketch(S) \rangle| & = \left |\left\langle \sum_{t} c_t \cdot \widehat{p}^t_x,   \sum_{y \in S} \sigma_y \sum_{t} c_t \cdot \widehat{p}^t_y \right\rangle\right| \\
    & \leq \left| \left \langle f_x, \sum_{y \in S} \sigma_y f_y \right \rangle \right|+  \left|\left\langle \sum_{t} c_t \cdot \widehat{p}^t_x,   \sum_{y \in S} \sigma_y \sum_{t} c_t \cdot \widehat{p}^t_y \right \rangle - \left \langle f_x, \sum_{y \in S} \sigma_y f_y \right \rangle \right| \\
    & \leq 0.1 \|f_x\|^2_2 + 10^{-2}\|f_x\|^2_2 \qquad \qquad \qquad  \text{ by \Cref{lemma:grouptest_exact} and \Cref{lemma:rw_to_embedding}}\\
    & \leq 0.11 \|f_x\|^2_2. 
\end{aligned}
\end{equation}
\noindent
Finally, by \Cref{lemma:sketchx}, we have that with probability at least $1-10^{-6}$, it holds that 
\begin{equation}\label{eq:sketchxx}
0.99\|f_x\|^2_2 \leq |\langle \sketch(x), \sketch(x)\rangle| \leq  1.01\|f_x\|^2_2.
\end{equation}
 
Combining \Cref{eq:sketchSX1} and \Cref{eq:sketchxx}, by a union bound, if $C_{i(x)} \cap S \neq \emptyset$, then with probability at least $1-0.51-10^{-6} > 0.4$, it holds that
\begin{align*}
     \langle \sketch(x), \sketch(S) \rangle| &\geq 0.84 \|f_x\|^2_2  \\
     & \geq \frac{0.84}{0.99}  |\langle \sketch(x) , \sketch(x) \rangle| \\
    & > 0.8 |\langle \sketch(x) , \sketch(x) \rangle|. 
\end{align*}
Similarly, combining \Cref{eq:sketchSX2} and \Cref{eq:sketchxx}, by a union bound, if $C_{i(x)} \cap S \neq \emptyset$, then with probability at least $1-0.0002-10^{-6} > 0.999$, it holds that
\begin{align*}
     \langle \sketch(x), \sketch(S) \rangle| &\leq 0.11 \|f_x\|^2_2  \\
     & \leq \frac{0.11}{1.001}  |\langle \sketch(x) , \sketch(x) \rangle| \\
    & < 0.2 |\langle \sketch(x) , \sketch(x) \rangle|. 
\end{align*}
\paragraph{Running time and sparsity:} The running time of $\sketch{}$ is dominated by the $|S|$ calls to $\randomwalks(r,t,x)$ (\Cref{alg:randwalks}), each of which runs in time $O^*(r)$ and returns a vector $\widehat{p}^t_x$ of support at most $r$. So the running time of \sketch{} is $O^*(|S|\cdot r) =   O^*\left(|S|\sqrt{\frac{n}{k}} \cdot n^{O(\epsilon/\varphi^2\log(1/\epsilon))}\left(1/\epsilon\right)^{O(\log(\varphi^2/\epsilon))}\right)$. The output of $\sketch$ is a linear combination of $|S|\cdot (t_{\Delta}+1)$ vectors $\widehat{p}^t_x$, so it has support 
$$O(|S|\cdot (t_{\Delta}+1)\cdot r)=  O^*\left(|S|\sqrt{\frac{n}{k}} \cdot n^{O(\epsilon/\varphi^2\log(1/\epsilon))}\left(1/\epsilon\right)^{O(\log(\varphi^2/\epsilon))}\right).$$

\end{proof}
\subsection{Analysis of \dotproduct{} (Algorithm \ref{alg:dotproduct})  and \findcluster{} (Algorithm \ref{alg:find_cluster})}\label{sec:NNanalysis}
In this section, we prove the correctness guarantees of $\dotproduct{}$ and  \findcluster{}, which we state below.  

\begin{lemma}[Correctness of \dotproduct]\label{lemma:dotproduct}
The procedure \dotproduct{} (Algorithm~\ref{alg:dotproduct}) runs in time
\[
O^*\!\left(
\sqrt{\frac{n}{k}} \cdot
n^{O(\epsilon/\varphi^2 \log(1/\epsilon))} \cdot
(1/\epsilon)^{O(\log(\varphi^2/\epsilon))}
\right).
\]

For every set $S \subseteq V$ and every $x \in V \setminus B_\delta$
that is typical with respect to $S$ (as per \Cref{def:good_pt_wrt_S}), the following guarantee holds: 
Let $\{\sk_j\}_{j \in [J]}$ be the outputs of $J \ge 6 \cdot 10^4$
independent executions of $\sketch$ on input $S$. Then, with probability
at least $1 - n^{-102}$  (over the randomness of the
sketches and the internal randomness of \dotproduct),
\begin{itemize}
    \item If $C_{i(x)} \cap S \neq \emptyset$, then
    $\dotproduct(x,\{\sk_j\}_{j \in [J]})$ returns \textsc{True}.
    \item If $C_{i(x)} \cap S = \emptyset$, then
    $\dotproduct(x,\{\sk_j\}_{j \in [J]})$ returns \textsc{False}.
\end{itemize}

Hence,  \dotproduct{} correctly determines whether $S$ contains a vertex
from the cluster $C_{i(x)}$.
\end{lemma}

 \begin{lemma}[Correctness of \findcluster{}] \label{thm:find_cluster}
    The procedure \findcluster{} (Algorithm \ref{alg:find_cluster}) runs in time 
    $$O^*\left( \sqrt{\frac{n}{k}} \cdot n^{O(\epsilon/\varphi^2 \log(1/\epsilon))} \cdot \left(1/\epsilon\right)^{O(\log(\varphi^2/\epsilon))}\right)$$ and returns a set of size at most $O^*(1)$.  
    
    For every set $S \subseteq V$ and every $x \in V \setminus B_{\delta}$  that is typical with respect to $S$ (as per \Cref{def:good_pt_wrt_S}) the following guarantee holds: 
     Let $r$ be the root of a tree of sketches of $S$ (see \Cref{def:tree}). Then, with probability at least $1-n^{-100}$ (over the randomness of the tree of sketches of $S$ and the internal randomness of $\findcluster$),  the output of $\findcluster(x,r)$ is exactly
     \[
\findcluster(x,r) = S \cap C_{i(x)}.
\]
\end{lemma}

We start by proving \Cref{lemma:dotproduct}. 

\begin{proof}[Proof of \Cref{lemma:dotproduct}]
Suppose $x \in  V \setminus B_{\delta}$ is typical for $S$.  $\dotproduct(x, \{\sk_j\}_{\{j \in [J]\}})$ runs $J \geq 6 \cdot 10^4\log(n)$ independent trials of $\sketch$. Define $X_j$ to be the event of the $j$-th trial of $\sketch$ indicating  $|\langle \sketch(x), \sk_j \rangle| \geq 0.5 |\langle \sketch(x), \sketch(x) \rangle|.$ 

Suppose $C_{i(x)}\cap S \neq \emptyset$. Then, by \Cref{lemma:spectralsketch}, 
 \[\Pr[X_j = 1] \geq \Pr\left[   |\langle \sketch(x), \sk_j \rangle| \geq 0.8 \langle \sketch(x), \sketch(x) \rangle \right] \geq 0.4\] where the probability is over the internal randomness of  \dotproduct{} and the randomness of $\sk_j$. Therefore, by an application of Chernoff bounds
 \[\Pr \left[\sum_{i = j}^JX_j < 0.3J\right] \leq e^{-0.002J} \leq n^{-102}\]
where the last inequality holds by the choice of $J$. Finally, by the definition of the algorithm
 \[\Pr\left[\dotproduct(x,\{\sk_j\}_{\{j \in [J]\}}) = \textsc{True}  \right] = \Pr \left[\sum_{j = 1}^JX_j \geq 0.3T\right]  \geq 1-n^{-102}.\]
 
Suppose instead that $C_{i(x)}\cap S = \emptyset$. Define $Y_j$ to be the event of the $j$-th trial of $\sketch$ indicating  $|\langle \sketch(x), \sk_j \rangle| < 0.5 |\langle \sketch(x), \sketch(x) \rangle|.$ By \Cref{lemma:spectralsketch} 
 \[\Pr[Y_j = 1] \geq \Pr\left[   |\langle \sketch(x), \sk_j \rangle| \leq 0.2 \langle \sketch(x), \sketch(x) \rangle \right] \geq 0.999\] where the probability is over the internal randomness of  \dotproduct{} and the randomness of $\sk_j$. By an application of Chernoff bounds, 
 \[\Pr\left[\sum_{j = 1}^JY_j \leq 0.7J \right] \leq e^{-0.002J} \leq n^{-102}\] where the last inequality holds by the choice of $J$. Hence,

 \[\Pr\left[\dotproduct(x,\{\sk_j\}_{\{j \in [J]\}}) = \textsc{False}  \right] = \Pr \left[\sum_{j = 1}^JY_j > 0.7J\right]  \geq 1-n^{-102}.\]

 \paragraph{Running time:}By \Cref{lemma:spectralsketch}, $\sketch(x)$ runs in time $O^*\left(\sqrt{\frac{n}{k}} \cdot n^{O(\epsilon/\varphi^2 \log(1/\epsilon))}\cdot \left(1/\epsilon\right)^{O(\log(\varphi^2/\epsilon))}\right)$. We call it $3J = O^*(1)$ times, so the total running time from the calls to $\sketch$ is  \newline
 $O^*\left(\sqrt{\frac{n}{k}} \cdot n^{O(\epsilon/\varphi^2 \log(1/\epsilon))}\cdot \left(1/\epsilon\right)^{O(\log(\varphi^2/\epsilon))}\right)$. 
 Furthermore, by \Cref{lemma:spectralsketch}, $\sketch(x)$
is a vector of support $O^*\left( \sqrt{\frac{n}{k}} \cdot n^{O(\epsilon/\varphi^2 \log(1/\epsilon))} \cdot \left(1/\epsilon\right)^{O(\log(\varphi^2/\epsilon))}\right)$, so we can compute each  inner products $\langle \sketch(x), \sketch(x) \rangle $ and $\langle \sketch(x), \sk_j \rangle $ in time $O(|\text{support}(\sketch(x)|) = O^*\left( \sqrt{\frac{n}{k}} \cdot n^{O(\epsilon/\varphi^2 \log(1/\epsilon))} \cdot \left(1/\epsilon\right)^{O(\log(\varphi^2/\epsilon))}\right)$. We compute $O^*(1)$ inner products in total, so the total running time of all the inner product computations is 
$O^*\left( \sqrt{\frac{n}{k}} \cdot n^{O(\epsilon/\varphi^2 \log(1/\epsilon))} \cdot \left(1/\epsilon\right)^{O(\log(\varphi^2/\epsilon))}\right)$. 
\end{proof}

 We now prove the main correctness guarantee of \findcluster{} (\Cref{thm:find_cluster}). 
 
\begin{proof}[Proof of \Cref{thm:find_cluster}] 
Suppose that $x$ is typical with respect to $S$. Then by \Cref{remark:good_pt}, $x$ is also typical for every subset $S' \subseteq S$, and in particular, $x$ is typical with respect to the set $u.\mathrm{label}$ for every node $u$ in the tree of sketches of $S$.  So by \Cref{lemma:dotproduct}, each call to \dotproduct{} in \findcluster{} returns the correct answer with probability at least $1-n^{-102}$. 
Since there are at most $|S| \leq n$ different paths in the recursive tree and each path has length at most $\log n$, we can take a union bound over all the recursive calls to \dotproduct{} to get that with probability at least $1-n^{-102}\cdot n \log n  \geq 1-n^{-100}$, all the recursive calls output the correct answer. 

Condition on the event that all the recursive calls to  \dotproduct{} output the correct answer. 
Then, we only call $\findcluster{(x,u)}$ for nodes $u$ corresponding to sets $S_u \subseteq S$ that contain at least one element of $S \cap C_{i(x)}$. Since this applies also to the leaf nodes (for which $S_u$ is a singleton set),  $\findcluster{(x,r)}$ can only return vertices in $S \cap C_{i(x)}$, Therefore, $\findcluster{(x,r)} \subseteq S\cap C_{i(x)}$.

We now argue that $\findcluster{(x,r)}$ returns all the vertices in  $S \cap C_{i(x)}$, i.e.\, that  $\findcluster{(x,r)} \supseteq S\cap C_{i(x)}$. 
Since $x$ is typical with respect to $S$, by \Cref{def:good_cluster_wrt_S}, {\bf  \ref{con:ScapC_i}}, we have $|S  \cap C_{i(x)}| \leq O^*(1)$, so the maximum number of leaves visited does not get exceeded. As a result, conditioned on the success of all the calls to \dotproduct{}, the procedure $\findcluster(x,r)$ calls $\findcluster{(x,u)}$ for all the nodes $u$ in the binary search tree of $S$ that correspond to sets $S_u \subseteq $ that contain at least one element of $S \cap C_{i(x)}$. In particular, $\findcluster{(x,r)}$ outputs all the elements of $S \cap C_i$. 

Hence, with probability at least $1-n^{-100}$, $\findcluster(x,r)$ returns exactly the set $S \cap C_{i(x)}$, as required. 

\paragraph{Running time and output size:} The running time of 
is dominated by the calls to $\dotproduct.$ By \Cref{lemma:dotproduct}, each call to $\dotproduct$ takes time  $O^*\left( \sqrt{\frac{n}{k}} \cdot n^{O(\epsilon/\varphi^2 \log(1/\epsilon))} \cdot \left(1/\epsilon\right)^{O(\log(\varphi^2/\epsilon))}\right)$. Since the 
algorithm can visit at most $O^*(1)$ leaf nodes in the tree of sketches, this means that it visits at most $O^*(1) \cdot \log n = O^*(1)$ nodes in total. Hence, the total running time of $\findcluster$ is $O^*\left( \sqrt{\frac{n}{k}} \cdot n^{O(\epsilon/\varphi^2 \log(1/\epsilon))} \cdot \left(1/\epsilon\right)^{O(\log(\varphi^2/\epsilon))}\right)$. 
The guarantee on the output size is immediate since the algorithm visits at most $O^*(1)$ leaf nodes in the tree of sketches by construction of the algorithm.  
\end{proof}

\subsection{Analysis of \findclustermeans{} (Algorithm \ref{alg:find_cluster_means})}\label{sec:prepro_analysis}
In this section, we prove \Cref{thm:preprocessing}, restated below for the convenience of the reader. 
\thmpreprocessing*
\remarkprepromain*

First, we need to define the notion of clusters that are well-represented by $S$. These are the clusters that get ``discovered" by the sample $S$, and all the sampled vertices from them are typical with respect to $S$. Formally, 
\begin{defn}[Well-represented by $S$]\label{def:good_cluster_wrt_S}
Given a set $S \subseteq V$, say that a cluster $C_i$ is \emph{well-represented} by $S$ if the following conditions hold: 
\begin{enumerate}[label=(\textbf{\arabic*})]
    \item$|S \cap C_i| \neq \emptyset$ \label{con:nonempty}
    \item Every $x \in S \cap C_i$ is typical with respect to $S$ (as per \Cref{def:good_pt_wrt_S}).\label{con:alltypical}
\end{enumerate}
\end{defn}
For a random set $S$, almost all clusters are well-represented by $S$. 
\begin{restatable}{lemma}{randomgoodcluster}
    \label{lemma:random_good_cluster}
If $S$ is a (multi) set of $O(k \log(\varphi^2/\epsilon))$ vertices sampled independently uniformly at random from $V$, then with probability at least $0.9998$, it holds that 
    \[ \left| \{ i \in [k] :  C_i \text{ is \emph{not} well-represented by $S$}  \} \right| \leq O\left(k \cdot \left(\frac{\epsilon}{\varphi^2}\right)^{1/3}\cdot \log(\varphi^2/\epsilon)\right). \]
\end{restatable}
The proof follows by Markov bounds and is included in \Cref{sec:randomgoodcluster}. 
We need the following guarantees on the sets $S$ and $R_{\mathrm{cand}}$ computed in the procedure \findclustermeans{} (Algorithm \ref{alg:find_cluster_means}). 
\begin{restatable}{lemma}{preprocessingmain}\label{thm:preprocessing_main}
   Let $S$ be the set of vertices sampled in line \ref{line:preproc_S} of Algorithm \ref{alg:find_cluster}, and let $R_{\mathrm{cand}}$ be the set produced in line \ref{line:preproc_L_candidate} of Algorithm \ref{alg:find_cluster}. For every $i\in [k]$ such that $C_i$ is well-represented by $S$ (as per \Cref{def:good_cluster_wrt_S}) with probability at least $1-n^{-99}$ over the internal randomness of $\findclustermeans$, the set $R_{\mathrm{cand}}$ contains exactly \emph{one} vertex from the cluster $C_i$, ie. $|R_{\mathrm{cand}} \cap C_i| =1$.     
\end{restatable}

\begin{proof}
Let $i\in [k]$ be such that $C_i$ is well-represented by S (as per \Cref{def:good_cluster_wrt_S}). By \Cref{def:good_cluster_wrt_S}, {\bf \ref{con:nonempty}},  $S \cap C_i \neq \emptyset$. Furthermore, by \Cref{def:good_cluster_wrt_S}, {\bf \ref{con:alltypical}} every $x \in S \cap C_i$ is typical with respect to $S$ (as per \Cref{def:good_cluster_wrt_S}). By \Cref{def:good_pt_wrt_S}, {\bf \ref{con:bad_cluster}}, this implies that $S\cap C_i \cap B_{\delta} = \emptyset.$ So every $x \in C_i \cap S $ belongs to $V \setminus B_{\delta}$. 

Therefore, for every $x \in S \cap C_i$, $x$ is typical with respect to $S$ and $x \notin B_{\delta}$. So by \Cref{thm:find_cluster}, with probability at least $1-n^{-100}$, the call to $\findcluster(x,\mathcal{T}_S.\mathrm{root})$ in line \ref{line:preproc_find_nbhrs} of Algorithm \ref{alg:find_cluster_means} returns exactly the set $S \cap C_i$. In particular, the set of directed edges that gets added to $A$ in line \ref{line:preproc_A} of Algorithm \ref{alg:find_cluster_means} is exactly $\{ (x, y) : y \in S \cap C_i \} $. 

By a union bound over all $x \in S \cap C_i$, with probability at least $1-n^{-99}$, the set $A$ contains all edges between pairs of vertices in $C_i \cap S$, and no other outgoing edges from $S \cap C_i$. 

Conditioned on this, $S \cap C_i$ forms a clique in the graph $H$ computed in line \ref{line:preproc_H} of Algorithm \ref{alg:find_cluster_means}, and this clique is disconnected from the rest of $H$. 
In particular, $S \cap C_i$ forms one connected component in $H$. Hence, in lines \ref{line:preproc:forC_i}-\ref{line:preproc:endforC_i}, exactly one vertex from the cluster $C_i$ is selected. So $R_{\mathrm{cand}}$ contains exactly one vertex from $C_i.$
\end{proof}

\Cref{lemma:random_good_cluster} shows that most clusters are well-represented by $S$ (as per \Cref{def:good_cluster_wrt_S}), and
\Cref{thm:preprocessing_main} shows that if the cluster $C_i$ is well-represented by $S$, then $R_{\mathrm{cand}}$ contains a unique cluster representative for $C_i$. These are exactly the clusters for which we can guarantee a low misclassification rate. Therefore, it is important that their representatives are included in the refined set $R$ produced by \findclustermeans{} (Algorithm \ref{alg:find_cluster_means}). \Cref{lemma:count} below shows that almost all the cluster representatives of clusters that are well-represented by $S$ get included in $R$. It also shows that not too many other vertices get included. 

\begin{lemma}\label{lemma:count}Let $R_{\mathrm{cand}}$ and $R$ be the sets computed in lines \ref{line:preproc_L_candidate}  and \ref{line:R} of Algorithm \ref{alg:find_cluster_means}, respectively. With probability at least 0.994 over the internal randomness of $\findclustermeans(\hat{k})$, it holds that 
    \begin{enumerate}[label=(\textbf{\arabic*})]
        \item $\left|\left \{  y \in R_{\mathrm{cand}} : C_{i(y)} \text{ is well-represented by $S$  and } c_y > 0 \right \}\right| \geq k - O(k \cdot (\epsilon/\varphi)^{1/3} \log(\varphi^2/\epsilon))$ 
        \item  $\left|\left \{  y \in R_{\mathrm{cand}} : C_{i(y)} \text{ is \emph{not} well-represented by $S$  and } c_y > 0 \right \}\right| \leq (k \cdot (\epsilon/\varphi)^{1/3} \log^2(\varphi^2/\epsilon)).$ 
    \end{enumerate} 
    In particular, 
    $\{ y \in R : \text{$C_{i(y)}$ is well-represented by } S\} \geq k - O(k \cdot (\epsilon/\varphi)^{1/3} \log(\varphi^2/\epsilon))$. 
\end{lemma}
 As an immediate corollary, \Cref{lemma:count} implies that \findclustermeans{} can be used as an algorithm for approximating the number of clusters. 
\begin{remark}[Approximating $k$ in time $\sqrt{nk}$]\label{rem:k_approx}
    Let $k$ be the true number of clusters, let $\hat{k}$ be a constant factor approximation to $k$ and let $R$ be the set computed in line \ref{line:R} of $\findclustermeans(\hat{k})$ (Algorithm \ref{alg:find_cluster_means}).  Then, with probability at least $0.994$ (over the internal randomness of \findclustermeans{}), it holds that 
    $$ \left|k - |R|  \right|\leq   O\left(k \cdot \left(\frac{\epsilon}{\varphi^2}\right)^{1/3}\log^2(\varphi^2/\epsilon) \right).$$
    Therefore, \findclustermeans, as a byproduct, computes a $poly\left(\frac{\epsilon}{\varphi^2}\right)$-multiplicative approximation to the number of clusters $k$ in time $O^*\left(\sqrt{nk} \cdot n^{O(\epsilon/\varphi^2 \log(1/\epsilon))}\cdot \left(1/\epsilon\right)^{O(\log(\varphi^2/\epsilon))}\right)$. We remark that the methods developed in this paper may be used to obtain an algorithm which, assuming access to a lower bound $\widetilde{k}$ on $k$, computes a in time $O^*\left(\sqrt{\frac{n}{k}} \cdot n^{O(\epsilon/\varphi^2 \log(1/\epsilon))}\cdot \left(1/\epsilon\right)^{O(\log(\varphi^2/\epsilon))}\right)$. See section \ref{sec:appx_sqrtnk} for a more detailed discussion.
\end{remark}

\begin{proof}[Proof of \Cref{lemma:count}:]
We first prove that the first condition holds with probability at least $0.997$, and then we prove that the second condition also holds with probability at least $0.997$. The main statement of the lemma then follows by a union bound, and the  ``in particular" follows immediately from {\bf (1)} and the definition of $R$. 
\paragraph{ Proving (1):}
Let $y \in R_{\mathrm{cand}}$ be such that $C_{i(y)}$ is well-represented by $S$ (as per  \Cref{def:good_cluster_wrt_S}). By \Cref{thm:preprocessing_main}, since $C_{i(y)}$ is well-represented by $S$, with probability at least $1-n^{-99}$, it holds that $|C_{i(y)} \cap R_{\mathrm{cand}}| = 1$, i.e. $y$ is the unique element in $R_{\mathrm{cand}}$ that belongs to $C_{i(y)}$. In particular, if $S_{\mathrm{test}}$ contains at least one vertex from $x \in C_{i(y)} \setminus B_{\delta}$, then by \Cref{thm:find_cluster}, with probability at least $1-n^{-100}$, $\findcluster(x, \mathcal{T}_{\mathrm{cand}}.\mathrm{root})$ returns $y$, and then $c_y >0$. 

So our task reduces to proving that
$$ |\{ i \in [k] : C_i \text{ is well-represented by $S$ and } (S_{\mathrm{test}}  \cap C_i \setminus B_{\delta}) \neq \emptyset\}| \geq k -   O(k \cdot (\epsilon/\varphi)^{1/3} \log(\varphi^2/\epsilon)), $$
or equivalently that 
\begin{equation}\label{eq:rtp}
 |\{ i \in [k] : C_i \text{ is \emph{not} well-represented by $S$}\} \cup \{ i\in [k] : S_{\mathrm{test}}  \cap C_i  \setminus B_{\delta} = \emptyset \}| \leq  O(k \cdot (\epsilon/\varphi)^{1/3} \log(\varphi^2/\epsilon)).
 \end{equation}
 
 By \Cref{lemma:random_good_cluster}, with probability at least $0.9998$, it holds that $| \{ i \in [k] : C_i \text{ is \emph{not} well-represented by $S$}\}| \leq O(k \cdot (\epsilon/\varphi)^{1/3} \log(\varphi^2/\epsilon)). $  This bounds the size of the first set in \Cref{eq:rtp}. 
 
 We now bound the size of the second set, namely $ \{ i\in [k] : S_{\mathrm{test}}  \cap C_i  \setminus B_{\delta} = \emptyset \} $, in \Cref{eq:rtp}. We have 
$$ | \{ i\in [k] : S_{\mathrm{test}}  \cap C_i  \setminus B_{\delta} = \emptyset \} | \leq |\{ i \in [k] : S_{\mathrm{test}}\cap C_i = \emptyset \}|+ |S_{\mathrm{test}} \cap B_{\delta}|.$$

 Let $\alpha$ be the constant hidden in the $O$-notation in line \ref{line:preproc_S_test}, so that $|S_{\mathrm{test}}|= \alpha k \log(\varphi^2/\epsilon)$. 
Fix a cluster $C_i$. Then, for every cluster $C_i$, $\E [ |C_i \cap S_{\mathrm{test}}|] = \frac{|C_i|}{n} \alpha k \log(\varphi^2/\epsilon) \geq \frac{\alpha}{\eta}\log(\varphi^2/\epsilon)$. So by  Chernoff bounds, we get 
$\Pr_{S_{\mathrm{test}}}\left[|S_{\mathrm{test}} \cap C_i| \geq 1 \right] \geq 1-00.1 \epsilon/\varphi^2$, for $\alpha$ sufficienly large. 
Therefore,  by Markov's inequality, with probability at least $0.999$, it holds that $| \{i \in [k]:  C_i \cap S_{\mathrm{test}} = \emptyset \}| \leq O(\epsilon/\varphi^2).$

Finally, by \Cref{lemma:close_to_clutermean} and the setting $\delta = \Omega((\epsilon/\varphi^2)^{2/3})$ (as per \Cref{def:params}), we have $|B_{\delta}| = O(n \cdot (\epsilon/\varphi)^{1/3})$, so by Markov's inequality, with probability at least $0.999$, we have $|S_{\mathrm{test}} \cap B_{\delta}| \leq O((\epsilon/\varphi)^{1/3} \log(\varphi^2/\epsilon))$. 
 Putting everything together, by a union bound, with probability at least $1-n^{-99} - n^{-100} -0.0002 - 0.001-0.001 > 0.997$, it holds that 
 \begin{align*}
 & |\{ i \in [k] : C_i \text{ is \emph{not} well-represented by $S$}\} \cup \{ i\in [k] : S_{\mathrm{test}}  \cap C_i  \setminus B_{\delta} = \emptyset \}| \\
 &  \leq |\{ i \in [k] : C_i \text{ is \emph{not} well-represented by $S$}\} |  +  |\{ i \in [k] : S_{\mathrm{test}}\cap C_i = \emptyset \}|+ |S_{\mathrm{test}} \cap B_{\delta}| \\
  & \leq O(k \cdot (\epsilon/\varphi)^{1/3} \log(\varphi^2/\epsilon)),
\end{align*}
which gives \Cref{eq:rtp}, as required. 
\paragraph{ Proving (2):}
First, observe that with high probability, only elements of $S_{\mathrm{test}}$  that can contribute to increasing the count $c_y$ for $\{  y \in R_{\mathrm{cand}} : C_{i(y)} \text{ is \emph{not} well-represented by $S$}\} $ are $x \in B_{\delta}$ or $x$ such that $C_{i(x)}$ is not well-represented with by $S$. Indeed, suppose that $x \in S_{\mathrm{test}} \setminus B_{\delta}$ and $C_{i(x)}$ is well-represented by $S$. Then by \Cref{thm:find_cluster}, with probability at least $1-n^{-100}$,  $\findcluster(x,\mathcal{T}_{\mathrm{cand}}.\mathrm{root})$ returns exactly the set $C_{i(x)} \cap R_{\mathrm{cand}}$, and in particular, it increases the counter $c_y$ for an element $y$ whose cluster $\emph{is}$ well-represented by $S$. 
Therefore, by a union bound, with probability at least $1- |S_{\mathrm{test}}| \cdot n^{-100} \geq 1- n^{-98}$, it holds that
\begin{equation}\label{eq:count2}
\begin{aligned}
    & \left|\left \{  y \in R_{\mathrm{cand}} : C_{i(y)} \text{ is \emph{not} well-represented by $S$  and } c_y > 0 \right \}\right| \\
    &\leq |S_{\mathrm{test}} \cap B_{\delta}| + |\{ x \in S_{\mathrm{test}} \text{:  $C_{i(x)}$ is \emph{not} well-represented by $S$}\}|
\end{aligned}
\end{equation}
We now turn to bounding the right hand side of \Cref{eq:count2}.

By \Cref{lemma:close_to_clutermean} and the setting $\delta = c \cdot (\epsilon/\varphi^2)^{2/3}$ (as per \Cref{def:params}), we have $|B_{\delta}| = O(n \cdot (\epsilon/\varphi)^{1/3})$, so by Markov's inequality, with probability at least $0.999$, we have $|S_{\mathrm{test}} \cap B_{\delta}| \leq O((\epsilon/\varphi)^{1/3} \log(\varphi^2/\epsilon))$. This bounds the first term in \Cref{eq:count2}. We now bound the second term. 
Define 
$$\mathcal{F} \coloneqq \{ i \in [k] :  C_i \text{ is \emph{not} well-represented by $S$}\}.$$
By \Cref{lemma:random_good_cluster}, with probability at least $0.9998$, it holds that 
\begin{equation}\label{eq:Fbound}
    \left| \mathcal{F}  \right| \leq O\left(k \cdot \left(\frac{\epsilon}{\varphi^2}\right)^{1/3}\cdot \log(\varphi^2/\epsilon)\right).
\end{equation}
Condition on the event that \Cref{eq:Fbound} holds. Then the expected number of samples in $S_{\mathrm{test}}$ from clusters belonging to $\mathcal{F}$ is 
\begin{align*}
    \mathbb{E}[|\{ x \in S_{\mathrm{test}} : i(x) \in \mathcal{F}\}| ]& = \frac{|S_{\mathrm{test}}|}{n} |\{x \in V : i(x) \in \mathcal{F}\}| \\
    & \leq \frac{|S_{\mathrm{test}}|}{n} \max_i |C_i| \cdot |\mathcal{F}| \\
    & \leq \frac{O(k \log(\varphi^2/\epsilon))}{n}\cdot \eta \cdot \frac{n}{k} \cdot O\left(k \cdot \left(\epsilon/\varphi^2\right)^{1/3}\cdot \log(\varphi^2/\epsilon)\right) \\
    & = O \left( k \cdot \left(\epsilon/\varphi^2\right)^{1/3}\log^2(\varphi^2/\epsilon)\right).
\end{align*}
So by Markov's inequality, with probability at least $0.999$, it holds that $|\{ x \in S_{\mathrm{test}} : i(x) \in \mathcal{F}\}| \leq  O \left( k \cdot \left(\epsilon/\varphi^2\right)^{1/3}\log^2(\varphi^2/\epsilon)\right).$
Putting everything together, by a union bound, with probability at least $1-n^{-98}-0.001-0.0002 - 0.001 \geq 0.997$, it holds that 
\begin{align*}
    \left|\left \{  y \in R_{\mathrm{cand}} : C_{i(y)} \text{ is \emph{not} well-represented by $S$  and } c_y > 0 \right \}\right| &\leq |S_{\mathrm{test}} \cap B_{\delta}| + |\{ x \in S_{\mathrm{test}} :  i(x) \in \mathcal{F}|\} \\
    & \leq  O \left( k \cdot \left(\epsilon/\varphi^2\right)^{1/3}\log^2(\varphi^2/\epsilon)\right).
\end{align*}
\end{proof}

Finally, we are ready to prove \Cref{thm:preprocessing}. 
\begin{proof}[Proof of \Cref{thm:preprocessing}]
We start by defining the set $B$. Let $S$ be the set sampled in line \ref{line:preproc_S}of \Cref{alg:find_cluster_means}. 
First, define
$$\mathcal{F} \coloneqq \{ i \in [k]: C_i \text{ is not well-represented by $S$} \} \cup \{ i \in [k]:  C_i \cap R = \emptyset\}$$
to be the set of clusters that fail either because  $R_{\mathrm{cand}}$ did not compute a good representative for it, or because they were not discovered by $S_{\mathrm{test}}$. We now define the set $B$ from the lemma statement as
$$ B \coloneqq \{ x \in V: i(x) \in \mathcal{F}\}  \cup \{x\in V : x \text{ is \emph{not} typical with respect to $S$}\} \cup B_{\delta}.$$
First, we bound the size of $B$. Then, we will prove that the guarantee from the lemma holds for all $x \in V \setminus B$. 
\paragraph{Bounding the size of $B$:}
From the definition of $B$, we have 
\begin{align*}
    |B| &\leq |\{x : C_{i(x)} \text{ in \emph{not} well-represented by $S$}\}| + |\{x : C_{i(x)} \cap R = \emptyset\}| \\
    & + |\{x : x \text{ is \emph{not} typical with respect to $S$} \}| + |B_{\delta}|.
\end{align*}
We now bound each of the four terms. 
By \Cref{lemma:random_good_cluster}, with probability at least $0.999$, it holds that \newline
$ |\{i\in [k] : C_{i(x)} \text{ in \emph{not} well-represented by $S$}\}| \leq O(k \cdot (\epsilon/\varphi^2)^{1/3} \log(\varphi^2/\epsilon) ), $ and hence 

\begin{equation*}
\begin{aligned}
|\{x : C_{i(x)} \text{ in \emph{not} well-represented by $S$}\}|& \leq \max_i |C_i| \cdot   O(k \cdot (\epsilon/\varphi^2)^{1/3} \log(\varphi^2/\epsilon) ) \\
&\leq \eta \cdot \frac{n}{k} \cdot  O(k \cdot (\epsilon/\varphi^2)^{1/3} \\
&= O(n \cdot (\epsilon/\varphi^2)^{1/3}).
\end{aligned}
\end{equation*}

By \Cref{lemma:count}, with probability at least $0.994$, we have $\{y \in R  : C_{i(y)} \text{ well-represented by $S$}\} \geq k - O(k \cdot (\epsilon/\varphi^2)^{1/3} \log(\varphi^2/\epsilon) ) $, and by \Cref{thm:preprocessing_main}, with probability at least $1-k\cdot n^{-50}$, there is at most one vertex from each well-represented cluster in $R$. So by a union bound, with probability at least $1 - k\cdot n^{-50} - 0.006$, it holds that $ |\{i \in [k] : |R \cap C_i| \geq 1\}| \geq k - O(k \cdot (\epsilon/\varphi^2)^{1/3} \log(\varphi^2/\epsilon)$, or equivalently, 
 $ |\{i \in [k] : R \cap C_i = \emptyset \}| \leq O(k \cdot (\epsilon/\varphi^2)^{1/3} \log(\varphi^2/\epsilon)$. From this, we get 
 \begin{equation*}
 \begin{aligned}
 |\{ x \in V : R \cap C_{i(x)} = \emptyset \}| & \leq \max_i |C_i| \cdot |\{ i \in [k] :R \cap C_i = \emptyset \}| \\
 &\leq O(n \cdot (\epsilon/\varphi^2)^{1/3} \log(\varphi^2/\epsilon)).
  \end{aligned}
 \end{equation*}
 
Since $S$ is a set of vertices sampled independently uniformly at random, by \Cref{lemma:random_good_set}, with probability at least $0.9999$, $S$ is a well-spread set (as per \Cref{def:good_S}). If $S$ is a well-spread set, then by \Cref{lemma:good_pts_wrt_S}, we have
$$|\{x \in V: x \text{ is \emph{not} typical with respect to $S$}\}|\leq O(n \cdot (\epsilon/\varphi^2)^{1/3} \log(\varphi^2/\epsilon)).$$
Here,  we use the fact that every vertex that is \emph{strongly typical} with respect to $S$ (as per \Cref{rem:strongly_typical}) is also typical (as per \Cref{def:good_pt_wrt_S}). 
 
Finally, by \Cref{lemma:close_to_clutermean}, we have 
$$|B_{\delta}| \leq O(n \cdot (\epsilon/\varphi^2)^{1/3}).$$

Putting all of this together, by union bound, with probability at least $1-0.001-k \cdot n^{-50} - 0.006 - 0.0001 \geq 0.992$, it holds that 

\begin{align*}
    |B| &\leq |\{x : C_{i(x)} \text{ in \emph{not} well-represented by $S$}\}| + |\{x : C_{i(x)} \cap R = \emptyset\}| \\
    & + |\{x : x \text{ is \emph{not} typical with respect to $S$} \}| + |B_{\delta}| \\
    & \leq O(n \cdot  (\epsilon/\varphi^2)^{1/3} \log(\varphi^2\epsilon)). 
\end{align*}

\paragraph{Proving the guarantee for all $x \in V \setminus B$:} Suppose $x \in V \setminus B$. Then, by definition of $B$, the cluster $C_{i(x)}$ is well-represented by $S$ (as per \Cref{def:good_cluster_wrt_S}) and $R \cap C_{i(x)} \neq \emptyset$. By \Cref{thm:preprocessing_main}, since $C_{i(x)}$ is well-represented by $S$, we have $|R_{\mathrm{cand}} \cap C_{i(x)}|=1$. Therefore, since  $R \cap C_{i(x)} \neq \emptyset$ and $R \subseteq R_{\mathrm{cand}}$, it must hold that $|R \cap C_{i(x)}|=1 $, i.e. $C_{i(x)}$ has a unique representative vertex in $R$. 

By definition of $B$, we also have $x \notin B_{\delta}$ and $x$ is typical with respect to $S$. In particular,  since $R \subseteq S$, by \Cref{claim:good_set_subset}, $x$ is also typical with respect to $R$. Let $\mathcal{T}$ denote the tree of sketches of $R$ (as per \Cref{def:tree}). By \Cref{thm:find_cluster}, since $ x \notin B_{\delta}$ and since $x$ is typical with respect to $R$, with probability at least $1-n^{-100}$,  $\findcluster(x,\mathcal{T}.\mathrm{root})$ returns the unique element of $R \cap C_{i(x)}$, as claimed. 

\paragraph{Runtime:} 
First, \findcluster{} computes the tree $\mathcal{T}_S$ of sketches in line \ref{line:TS}. This requires running $\sketch$ for $O(\log |S|)$ sets of size at most $|S|$. By \Cref{lemma:spectralsketch}, this takes time  \newline $O(\log |S|) \cdot O^*\left( |S|\sqrt{\frac{n}{k}} \cdot n^{O(\epsilon/\varphi^2 \log(1/\epsilon))} \cdot \left(1/\epsilon\right)^{O(\log(\varphi^2/\epsilon))}\right)$ $= O^*\left(\sqrt{nk} \cdot n^{O(\epsilon/\varphi^2 \log(1/\epsilon))} \cdot \left(1/\epsilon\right)^{O(\log(\varphi^2/\epsilon))}\right).$
 
Given the tree of sketches $\mathcal{T}_S$, each call to $\findcluster(x,\mathcal{T}_S.\mathrm{root})$ in line \ref{line:preproc_find_nbhrs} takes time \newline
$O^*\left( \sqrt{\frac{n}{k}} \cdot n^{O(\epsilon/\varphi^2 \log(1/\epsilon))} \cdot \left(1/\epsilon\right)^{O(\log(\varphi^2/\epsilon))}\right)$, by \Cref{thm:find_cluster}. 
Since $\findcluster$ returns a set of size most $O^*(1)$, updating the set $A$ in line \ref{line:preproc_A} takes time $O^*(1)$. 
Thus, the total running time of lines \ref{line:preproc_forstart} - \ref{line:preproc_forend} is $|S| \cdot O^*\left( \sqrt{\frac{n}{k}} \cdot n^{O(\epsilon/\varphi^2 \log(1/\epsilon))} \cdot \left(1/\epsilon\right)^{O(\log(\varphi^2/\epsilon))}\right) = O^*\left( \sqrt{nk} \cdot n^{O(\epsilon/\varphi^2 \log(1/\epsilon))} \cdot \left(1/\epsilon\right)^{O(\log(\varphi^2/\epsilon))}\right)$.

In lines \ref{line:preproc_H}- \ref{line:preproc_components}, we define the graph $H$ and compute its connected components. The degree of $H$ is bounded by $O^*(1)$, the number of vertices is at most $|S| = O^*(k)$ vertices , the number of edges is at most $|S| \cdot O^*(1) =  O^*(k)$ edges, so finding all connected components takes time $O^*(k) \leq O^*(\sqrt{nk})$. 
In lines \ref{line:preproc:forC_i}- \ref{line:preproc:endforC_i}, we simply select one vertex from each connected component, which in total takes time at most $O^*(k)$. 

Similarly to computing $\mathcal{T}_S$, computing the tree of sketches for $R_{\mathrm{cand}}$ in line \ref{line:Tcand} takes time \newline $O^*\left(\sqrt{nk} \cdot n^{O(\epsilon/\varphi^2 \log(1/\epsilon))} \cdot \left(1/\epsilon\right)^{O(\log(\varphi^2/\epsilon))}\right).$
Given a tree of sketches $\mathcal{T}_{\mathrm{cand}}$, each call to  \newline $\findcluster(x,\mathcal{T}_{\mathrm{cand}}.\mathrm{root})$ in line \ref{line:preproc:L_candidate_findnbhrs} takes time $O^*\left( \sqrt{\frac{n}{k}} \cdot n^{O(\epsilon/\varphi^2 \log(1/\epsilon))} \cdot \left(1/\epsilon\right)^{O(\log(\varphi^2/\epsilon))}\right)$, by \Cref{thm:find_cluster}. So the total running time of lines \ref{line:preproc:startfor2} - \ref{line:preproc:endfor2} is $|S_{\mathrm{test}}| \cdot O^*\left( \sqrt{\frac{n}{k}} \cdot n^{O(\epsilon/\varphi^2 \log(1/\epsilon))} \cdot \left(1/\epsilon\right)^{O(\log(\varphi^2/\epsilon))}\right)= \\ O^*\left( \sqrt{nk} \cdot n^{O(\epsilon/\varphi^2 \log(1/\epsilon))} \cdot \left(1/\epsilon\right)^{O(\log(\varphi^2/\epsilon))}\right)$.
Finally, in line \ref{line:preproc_return}, we compute the tree of sketches for $R$. Similarly to computing $\mathcal{T}_S$ and $\mathcal{T}_{\mathrm{cand}}$, this takes time $O^*\left(\sqrt{nk} \cdot n^{O(\epsilon/\varphi^2 \log(1/\epsilon))} \cdot \left(1/\epsilon\right)^{O(\log(\varphi^2/\epsilon))}\right).$

Combining everything, the overall running time is $O^*\left( \sqrt{nk} \cdot n^{O(\epsilon/\varphi^2 \log(1/\epsilon))} \cdot \left(1/\epsilon\right)^{O(\log(\varphi^2/\epsilon))}\right)$ as claimed. 
\end{proof}

\begin{remark} The set $R$ produced by \ref{alg:find_cluster_means} implicitly defines the cluster labels. Fix a bijection $l: R \to [|R|]$, and let $\mathcal{T}$ be the tree of sketches of $R$ (as per \Cref{def:tree}). Upon query $x$, run  $\findcluster(x, \mathcal{T}.\mathrm{root})$, which returns a subset of $R$. If the output contains a single vertex $y$, then label $x$ with $l(y)$.
\end{remark}

\subsection{Proof of \Cref{thm:main}}\label{sec:main_wrapping}
Finally, we are ready to prove \Cref{thm:main}, restated below for the convenience of the reader. 
\thmmain*
\begin{proof}
By \Cref{thm:preprocessing}, with probability at least $0.999$, $\findclustermeans{}$ (\Cref{alg:find_cluster_means}) computes a tree of sketches $\mathcal{T}$ of a set $R \subseteq V$, such that the following holds: 
There exists a set $B$ of size $O\left(n \cdot \left(\frac{\epsilon}{\varphi^2}\right)^{1/3}\cdot \log(\varphi^2/\epsilon)\right)$ such that for every vertex $x \in V \setminus B$,
the cluster $C_{i(x)}$ of $x$ has a representative vertex
$y_{i(x)} \in R$ (see \Cref{def:rep_vertex}),
and
\[
\Pr\!\left[\findcluster(x,\mathcal{T}.\mathrm{root}) = y_{i(x)}\right]
\ge 1 - n^{-50},
\]
where the probability is over the internal randomness of $\findcluster$.

It follows that the procedure $\findcluster(\cdot,\mathcal{T}.\mathrm{root})$
defines a clustering oracle that misclassifies at most
\[
O\!\left( n \cdot \left(\frac{\epsilon}{\varphi^2}\right)^{1/3}
\cdot \log\!\left(\frac{\varphi^2}{\epsilon}\right) \right)
\]
vertices. 

By \Cref{thm:preprocessing} the preprocessing time is $O^*\left( \sqrt{nk} \cdot n^{O(\epsilon/\varphi^2 \log(1/\epsilon))} \cdot \left(1/\epsilon\right)^{O(\log(\varphi^2/\epsilon))}\right)$.

We store the tree of sketches for the set $R$, so we need to store $O(\log k)$ sketches for sets of size at most $O(k)$. By \Cref{lemma:spectralsketch}, the space complexity of this is $O^*\left( \sqrt{nk} \cdot n^{O(\epsilon/\varphi^2 \log(1/\epsilon))} \cdot \left(1/\epsilon\right)^{O(\log(\varphi^2/\epsilon))}\right)$.
Finally, by \Cref{thm:find_cluster}, the runtime of $\findcluster{}$ (\Cref{alg:find_cluster}) is $O^*\left( \sqrt{\frac{n}{k}} \cdot n^{O(\epsilon/\varphi^2 \log(1/\epsilon))} \cdot \left(1/\epsilon\right)^{O(\log(\varphi^2/\epsilon))}\right)$.

\paragraph{Random bits.} Note that the current version of the algorithm, which samples random walks using fresh randomness for every query, although is successful with high probability, is inconsistent. Using standard techniques, we can limit ourselves to using only a few random bits, sampled at the preprocessing stage, without hurting the probability of success of the data structure. With the updated sampling of random walks, our clustering oracle is completely deterministic after the preprocessing stage is complete. 

Indeed, observe that the randomness of \findcluster{} (Algorithm \ref{alg:find_cluster}) comes from the use of procedure \sketch{} (\Cref{alg:spectralsketch}) approximating pairwise spectral dot products. In the analysis of the performance of \sketch{}, we use that the random walks are 4-wise independent (in the proof of \Cref{lemma:variance_calc}) but no more independence is required. Therefore,  for every query we need $W = O^*\left( \sqrt{\frac{n}{k}} \cdot n^{O(\epsilon/\varphi^2 \log(1/\epsilon))} \cdot \left(1/\epsilon\right)^{O(\log(\varphi^2/\epsilon))}\right)$ many random walks of length $T = O(\log(n)/\varphi^2)$, we can generate the random walks in the following way:
\begin{itemize}
  \item Generate a table $H$ of size $s = T \times W$ where every entry is selected independently to be a 4-wise independent hash function $h: V \to [d]$;
  \item Select the $t$-th (out of $T$) step of the $w$-th (out of $W$) walk to be the $H_{t, w}$-th neighbor of the $t-1$-st step of the $w$-th walk. 
\end{itemize}
The number of random bits required to build the above table is \[O(\log d \cdot T\cdot W) = O^*\left( \sqrt{\frac{n}{k}} \cdot n^{O(\epsilon/\varphi^2 \log(1/\epsilon))} \cdot \left(1/\epsilon\right)^{O(\log(\varphi^2/\epsilon))}\right),\] so all of them may be sampled during the preprocessing stage and stored in the sublinear data structure. Thus, the query stage is completely deterministic.

As mentioned before, since all of the functions $h$ are chosen to be 4-wise independent, the analysis of the performance of \findcluster \ still holds, and therefore our deterministic classifier works with probability $1 - n^{-50}$ for each vertex outside of set $B$. By a union bound, with probability at least $1 - n^{-49}$, the algorithm works for all vertices outside $B$. 
\end{proof}

\bibliographystyle{alpha}
\bibliography{references}

@inproceedings{DBLP:conf/stoc/IndykM98,
  author       = {Piotr Indyk and
                  Rajeev Motwani},
  editor       = {Jeffrey Scott Vitter},
  title        = {Approximate Nearest Neighbors: Towards Removing the Curse of Dimensionality},
  booktitle    = {Proceedings of the Thirtieth Annual {ACM} Symposium on the Theory
                  of Computing, Dallas, Texas, USA, May 23-26, 1998},
  pages        = {604--613},
  publisher    = {{ACM}},
  year         = {1998},
  url          = {https://doi.org/10.1145/276698.276876},
  doi          = {10.1145/276698.276876},
  bibsource    = {dblp computer science bibliography, https://dblp.org}
}

@inproceedings{DBLP:conf/focs/KalePS08,
  author       = {Satyen Kale and
                  Yuval Peres and
                  C. Seshadhri},
  title        = {Noise Tolerance of Expanders and Sublinear Expander Reconstruction},
  booktitle    = {49th Annual {IEEE} Symposium on Foundations of Computer Science, {FOCS}
                  2008, October 25-28, 2008, Philadelphia, PA, {USA}},
  pages        = {719--728},
  publisher    = {{IEEE} Computer Society},
  year         = {2008},
  url          = {https://doi.org/10.1109/FOCS.2008.65},
  doi          = {10.1109/FOCS.2008.65},
  bibsource    = {dblp computer science bibliography, https://dblp.org}
}

@inproceedings{DBLP:conf/nips/Shen023,
  author       = {Ranran Shen and
                  Pan Peng},
  editor       = {Alice Oh and
                  Tristan Naumann and
                  Amir Globerson and
                  Kate Saenko and
                  Moritz Hardt and
                  Sergey Levine},
  title        = {A Sublinear-Time Spectral Clustering Oracle with Improved Preprocessing
                  Time},
  booktitle    = {Advances in Neural Information Processing Systems 36: Annual Conference
                  on Neural Information Processing Systems 2023, NeurIPS 2023, New Orleans,
                  LA, USA, December 10 - 16, 2023},
  year         = {2023},
  url          = {http://papers.nips.cc/paper\_files/paper/2023/hash/82aec8518602748540a42b783468c94d-Abstract-Conference.html},
  bibsource    = {dblp computer science bibliography, https://dblp.org}
}

@article{DBLP:journals/jmlr/Abbe17,
  author       = {Emmanuel Abbe},
  title        = {Community Detection and Stochastic Block Models: Recent Developments},
  journal      = {J. Mach. Learn. Res.},
  volume       = {18},
  pages        = {177:1--177:86},
  year         = {2017},
  url          = {https://jmlr.org/papers/v18/16-480.html},
  bibsource    = {dblp computer science bibliography, https://dblp.org}
}

@incollection{DBLP:books/sp/goldreich2011/GoldreichR11,
  author    = {Oded Goldreich and
               Dana Ron},
  title     = {On Testing Expansion in Bounded-Degree Graphs},
  booktitle = {Studies in Complexity and Cryptography. Miscellanea on the Interplay
               between Randomness and Computation - In Collaboration with Lidor Avigad,
               Mihir Bellare, Zvika Brakerski, Shafi Goldwasser, Shai Halevi, Tali
               Kaufman, Leonid Levin, Noam Nisan, Dana Ron, Madhu Sudan, Luca Trevisan,
               Salil Vadhan, Avi Wigderson, David Zuckerman},
  pages     = {68--75},
  year      = {2011},
  crossref  = {DBLP:books/sp/goldreich11},
  url       = {https://doi.org/10.1007/978-3-642-22670-0\_9},
  doi       = {10.1007/978-3-642-22670-0\_9},
  bibsource = {dblp computer science bibliography, https://dblp.org}
}

@book{DBLP:books/sp/goldreich11,
  editor    = {Oded Goldreich},
  title     = {Studies in Complexity and Cryptography. Miscellanea on the Interplay
               between Randomness and Computation - In Collaboration with Lidor Avigad,
               Mihir Bellare, Zvika Brakerski, Shafi Goldwasser, Shai Halevi, Tali
               Kaufman, Leonid Levin, Noam Nisan, Dana Ron, Madhu Sudan, Luca Trevisan,
               Salil Vadhan, Avi Wigderson, David Zuckerman},
  series    = {Lecture Notes in Computer Science},
  volume    = {6650},
  publisher = {Springer},
  year      = {2011},
  url       = {https://doi.org/10.1007/978-3-642-22670-0},
  doi       = {10.1007/978-3-642-22670-0},
  isbn      = {978-3-642-22669-4},
  bibsource = {dblp computer science bibliography, https://dblp.org}
}

@inproceedings{Peng20,
  author    = {Pan Peng},
  title     = {Robust Clustering Oracle and Local Reconstructor of Cluster Structure
               of Graphs},
  booktitle = {Proceedings of the 2020 {ACM-SIAM} Symposium on Discrete Algorithms,
               {SODA} 2020, Salt Lake City, UT, USA, January 5-8, 2020},
  pages     = {2953--2972},
  year      = {2020},
  crossref  = {DBLP:conf/soda/2020},
  url       = {https://doi.org/10.1137/1.9781611975994.179},
  doi       = {10.1137/1.9781611975994.179},
  bibsource = {dblp computer science bibliography, https://dblp.org}
}

@proceedings{DBLP:conf/soda/2020,
  editor    = {Shuchi Chawla},
  title     = {Proceedings of the 2020 {ACM-SIAM} Symposium on Discrete Algorithms,
               {SODA} 2020, Salt Lake City, UT, USA, January 5-8, 2020},
  publisher = {{SIAM}},
  year      = {2020},
  url       = {https://doi.org/10.1137/1.9781611975994},
  doi       = {10.1137/1.9781611975994},
  isbn      = {978-1-61197-599-4},
  bibsource = {dblp computer science bibliography, https://dblp.org}
}

@inproceedings{KaleS08,
  author    = {Satyen Kale and
               C. Seshadhri},
  title     = {An Expansion Tester for Bounded Degree Graphs},
  booktitle = {Automata, Languages and Programming, 35th International Colloquium,
               {ICALP} 2008, Reykjavik, Iceland, July 7-11, 2008, Proceedings, Part
               {I:} Tack {A:} Algorithms, Automata, Complexity, and Games},
  pages     = {527--538},
  year      = {2008},
  crossref  = {DBLP:conf/icalp/2008-1},
  url       = {https://doi.org/10.1007/978-3-540-70575-8\_43},
  doi       = {10.1007/978-3-540-70575-8\_43},
  bibsource = {dblp computer science bibliography, https://dblp.org}
}

@proceedings{DBLP:conf/icalp/2008-1,
  editor    = {Luca Aceto and
               Ivan Damg{\aa}rd and
               Leslie Ann Goldberg and
               Magn{\'{u}}s M. Halld{\'{o}}rsson and
               Anna Ing{\'{o}}lfsd{\'{o}}ttir and
               Igor Walukiewicz},
  title     = {Automata, Languages and Programming, 35th International Colloquium,
               {ICALP} 2008, Reykjavik, Iceland, July 7-11, 2008, Proceedings, Part
               {I:} Tack {A:} Algorithms, Automata, Complexity, and Games},
  series    = {Lecture Notes in Computer Science},
  volume    = {5125},
  publisher = {Springer},
  year      = {2008},
  url       = {https://doi.org/10.1007/978-3-540-70575-8},
  doi       = {10.1007/978-3-540-70575-8},
  isbn      = {978-3-540-70574-1},
  bibsource = {dblp computer science bibliography, https://dblp.org}
}

@article{KalePS13,
  author    = {Satyen Kale and
               Yuval Peres and
               C. Seshadhri},
  title     = {Noise Tolerance of Expanders and Sublinear Expansion Reconstruction},
  journal   = {{SIAM} J. Comput.},
  volume    = {42},
  number    = {1},
  pages     = {305--323},
  year      = {2013},
  url       = {https://doi.org/10.1137/110837863},
  doi       = {10.1137/110837863},
  bibsource = {dblp computer science bibliography, https://dblp.org}
}

@article{NachmiasS10,
  author    = {Asaf Nachmias and
               Asaf Shapira},
  title     = {Testing the expansion of a graph},
  journal   = {Inf. Comput.},
  volume    = {208},
  number    = {4},
  pages     = {309--314},
  year      = {2010},
  url       = {https://doi.org/10.1016/j.ic.2009.09.002},
  doi       = {10.1016/j.ic.2009.09.002},
  bibsource = {dblp computer science bibliography, https://dblp.org}
}

@article{lee2014multiway,
  title={Multiway spectral partitioning and higher-order cheeger inequalities},
  author={Lee, James R and Gharan, Shayan Oveis and Trevisan, Luca},
  journal={Journal of the ACM (JACM)},
  volume={61},
  number={6},
  pages={37},
  year={2014},
  publisher={ACM}
}

@inproceedings{chiplunkar2018testing,
  title={Testing Graph Clusterability: Algorithms and Lower Bounds},
  author={Chiplunkar, Ashish and Kapralov, Michael and Khanna, Sanjeev and Mousavifar, Aida and Peres, Yuval},
  booktitle={2018 IEEE 59th Annual Symposium on Foundations of Computer Science (FOCS)},
  pages={497--508},
  year={2018},
  organization={IEEE}
}

@article{DBLP:journals/cpc/CzumajS10,
  author    = {Artur Czumaj and
               Christian Sohler},
  title     = {Testing Expansion in Bounded-Degree Graphs},
  journal   = {Combinatorics, Probability {\&} Computing},
  volume    = {19},
  number    = {5-6},
  pages     = {693--709},
  year      = {2010},
  url       = {https://doi.org/10.1017/S096354831000012X},
  doi       = {10.1017/S096354831000012X},
  bibsource = {dblp computer science bibliography, https://dblp.org}
}

@inproceedings{DBLP:conf/soda/Sinop16,
  author    = {Ali Kemal Sinop},
  title     = {How to Round Subspaces: {A} New Spectral Clustering Algorithm},
  booktitle = {Proceedings of the Twenty-Seventh Annual {ACM-SIAM} Symposium on Discrete
               Algorithms, {SODA} 2016, Arlington, VA, USA, January 10-12, 2016},
  pages     = {1832--1847},
  year      = {2016},
  url       = {https://doi.org/10.1137/1.9781611974331.ch128},
  doi       = {10.1137/1.9781611974331.ch128},
  bibsource = {dblp computer science bibliography, https://dblp.org}
}

@inproceedings{DBLP:conf/soda/OrecchiaZ14,
  author    = {Lorenzo Orecchia and
               Zeyuan {Allen Zhu}},
  title     = {Flow-Based Algorithms for Local Graph Clustering},
  booktitle = {Proceedings of the Twenty-Fifth Annual {ACM-SIAM} Symposium on Discrete
               Algorithms, {SODA} 2014, Portland, Oregon, USA, January 5-7, 2014},
  pages     = {1267--1286},
  year      = {2014},
  url       = {https://doi.org/10.1137/1.9781611973402.94},
  doi       = {10.1137/1.9781611973402.94},
  bibsource = {dblp computer science bibliography, https://dblp.org}
}

@inproceedings{DBLP:conf/icml/ZhuLM13,
  author    = {Zeyuan {Allen Zhu} and
               Silvio Lattanzi and
               Vahab S. Mirrokni},
  title     = {A Local Algorithm for Finding Well-Connected Clusters},
  booktitle = {Proceedings of the 30th International Conference on Machine Learning,
               {ICML} 2013, Atlanta, GA, USA, 16-21 June 2013},
  pages     = {396--404},
  year      = {2013},
  url       = {http://jmlr.org/proceedings/papers/v28/allenzhu13.html},
  bibsource = {dblp computer science bibliography, https://dblp.org}
}

@article{DBLP:journals/jacm/AndersenGPT16,
  author    = {Reid Andersen and
               Shayan Oveis Gharan and
               Yuval Peres and
               Luca Trevisan},
  title     = {Almost Optimal Local Graph Clustering Using Evolving Sets},
  journal   = {J. {ACM}},
  volume    = {63},
  number    = {2},
  pages     = {15:1--15:31},
  year      = {2016},
  url       = {https://doi.org/10.1145/2856030},
  doi       = {10.1145/2856030},
  bibsource = {dblp computer science bibliography, https://dblp.org}
}

@article{DBLP:journals/im/AndersenCL08,
  author    = {Reid Andersen and
               Fan R. K. Chung and
               Kevin J. Lang},
  title     = {Local Partitioning for Directed Graphs Using PageRank},
  journal   = {Internet Mathematics},
  volume    = {5},
  number    = {1},
  pages     = {3--22},
  year      = {2008},
  url       = {https://doi.org/10.1080/15427951.2008.10129297},
  doi       = {10.1080/15427951.2008.10129297},
  bibsource = {dblp computer science bibliography, https://dblp.org}
}

@article{DBLP:journals/siammax/SpielmanT14,
  author    = {Daniel A. Spielman and
               Shang{-}Hua Teng},
  title     = {Nearly Linear Time Algorithms for Preconditioning and Solving Symmetric,
               Diagonally Dominant Linear Systems},
  journal   = {{SIAM} J. Matrix Analysis Applications},
  volume    = {35},
  number    = {3},
  pages     = {835--885},
  year      = {2014},
  url       = {https://doi.org/10.1137/090771430},
  doi       = {10.1137/090771430},
  bibsource = {dblp computer science bibliography, https://dblp.org}
}

@article{KwokMultiwaySpectrumGap,
author = {Chiu Kwok, Tsz and Chi Lau, Lap and Lee, Yin Tat and oveis gharan, Shayan and Trevisan, Luca},
year = {2013},
month = {01},
pages = {},
title = {Improved Cheeger's Inequality: Analysis of Spectral Partitioning
Algorithms through Higher Order Spectral Gap},
journal = {Proceedings of the Annual ACM Symposium on Theory of Computing},
doi = {10.1145/2488608.2488611}
}

@article{Robbins55,
  title={A Remark on Stirling’s Formula},
  author={Herbert E. Robbins},
  journal={American Mathematical Monthly},
  year={1955},
  volume={62},
  pages={402-405},
  url={https://api.semanticscholar.org/CorpusID:122180319}
}

@article{AA22,
  title={Optimal-degree polynomial approximations for exponentials and gaussian kernel density estimation},
  author={Amol Aggarwal and Josh Alman},
  journal={Proceedings of the 37th Computational Complexity Conference},
  year={2022},
  url={https://api.semanticscholar.org/CorpusID:248721692}
}

@inproceedings{GKLMS21,
  author       = {Grzegorz Gluch and
                  Michael Kapralov and
                  Silvio Lattanzi and
                  Aida Mousavifar and
                  Christian Sohler},
  title        = {Spectral Clustering Oracles in Sublinear Time},
  booktitle    = {Proceedings of the 2021 {ACM-SIAM} Symposium on Discrete Algorithms,
                  {SODA} 2021, Virtual Conference, January 10 - 13, 2021},
  pages        = {1598--1617},
  publisher    = {{SIAM}},
  year         = {2021}
}

@book{MH_chebyshev,
  title={Chebyshev polynomials},
  author={Mason, John C and Handscomb, David C},
  year={2002},
  publisher={Chapman and Hall/CRC}
}

@inproceedings{czumaj2015testing,
  title={Testing cluster structure of graphs},
  author={Czumaj, Artur and Peng, Pan and Sohler, Christian},
  booktitle={Proceedings of the forty-seventh annual ACM symposium on Theory of Computing},
  pages={723--732},
  year={2015}
}

@InProceedings{jha_et_al:LIPIcs.ICALP.2024.91,
  author =	{Jha, Agastya Vibhuti and Kumar, Akash},
  title =	{{A Sublinear Time Tester for Max-Cut on Clusterable Graphs}},
  booktitle =	{51st International Colloquium on Automata, Languages, and Programming (ICALP 2024)},
  pages =	{91:1--91:17},
  series =	{Leibniz International Proceedings in Informatics (LIPIcs)},
  ISBN =	{978-3-95977-322-5},
  ISSN =	{1868-8969},
  year =	{2024},
  volume =	{297},
  editor =	{Bringmann, Karl and Grohe, Martin and Puppis, Gabriele and Svensson, Ola},
  publisher =	{Schloss Dagstuhl -- Leibniz-Zentrum f{\"u}r Informatik},
  address =	{Dagstuhl, Germany},
  URL =		{https://drops.dagstuhl.de/entities/document/10.4230/LIPIcs.ICALP.2024.91},
  URN =		{urn:nbn:de:0030-drops-202344},
  doi =		{10.4230/LIPIcs.ICALP.2024.91}
}

@inproceedings{DBLP:conf/soda/0001Y23,
  author       = {Pan Peng and
                  Yuichi Yoshida},
  editor       = {Nikhil Bansal and
                  Viswanath Nagarajan},
  title        = {Sublinear-Time Algorithms for Max Cut, Max E2Lin(\emph{q}), and Unique
                  Label Cover on Expanders},
  booktitle    = {Proceedings of the 2023 {ACM-SIAM} Symposium on Discrete Algorithms,
                  {SODA} 2023, Florence, Italy, January 22-25, 2023},
  pages        = {4936--4965},
  publisher    = {{SIAM}},
  year         = {2023},
  url          = {https://doi.org/10.1137/1.9781611977554.ch180},
  doi          = {10.1137/1.9781611977554.CH180},
  bibsource    = {dblp computer science bibliography, https://dblp.org}
}

\newpage
\appendix

\section{Proof of Theorem  \ref{thm:standardbasis}}\label{sec:chebyshev}

In this section, we prove \Cref{thm:standardbasis}, restated here for the convenience of the reader: 
\standardbasis*
The main idea is to multiply $x^t$ by a polynomial $q$ that nearly cancels it on the interval $[1-\epsilon, 1]$. We construct $q$ by truncating the Chebyshev expansion of the function $(1 - \epsilon x)^{-t}$. In particular, we prove \Cref{thm:Chebyshevapprox}, which provides the necessary guarantees for this approximation.

First, we formally define the Chebyshev polynomials $Q_d(x)$. 
\begin{defn}\label{def:chebyshev}
    Given an integer $d \geq 0$, define the \emph{degree $d$ monic Chebyshev Polynomial $Q_d(x)$} by the relation 
    \begin{equation*}
        Q_d(x) = 2^{1-d}\cos(d \theta) \qquad \qquad \text{when $x = \cos \theta$}. 
    \end{equation*}
\end{defn}

\begin{restatable}{thm}{chebyshevapprox}\label{thm:Chebyshevapprox}
Let $t\geq 2$, $\epsilon >0$, $d \geq 10 \epsilon t$, and let $$ P(x) = \sum_{v = 0}^{d}2^{v-1}Q_v(x)c_{v,\epsilon,t}$$ be the polynomial obtained by taking the first $d+1$ terms of the Chebyshev expansion of $(1-\epsilon x)^{-t}$ on $[-1,1].$ Then 
\begin{itemize}
    \item $\max_{ v\leq d} |c_{v,\epsilon,t}| \leq 2^{O(\epsilon t)}$, and 
    \item $\sup_{x \in [-1,1]} |P(x) - (1-\epsilon x)^{-t}| =  O(d e^{-d})$. 
    \end{itemize}
\end{restatable}
 
 We first prove \Cref{thm:standardbasis} assuming \Cref{thm:Chebyshevapprox}, and then proceed to prove \Cref{thm:Chebyshevapprox} in the rest of this section.

\begin{proof}[Proof of  \Cref{thm:standardbasis}]
Let $$d =  20 \epsilon t + 2\log(\varphi^2/\epsilon),$$ and let 
$$ P(x) \coloneqq \sum_{v = 0}^{d}2^{v-1}Q_v(x)c_{v,\epsilon,t}$$ be the polynomial obtained by taking the first $d+1$ terms of the Chebyshev expansion of $(1-\epsilon x)^{-t}$ on $[-1,1].$ 
Let
 $$p(x) \coloneqq x^{t} \cdot P\left(\frac{1-x}{\epsilon} \right).$$ 
We now prove that $p$ satisfies all of the required properties. 

Clearly, $p$ is of the form $x^t q(x)$ for $$q(x) \coloneqq P\left(\frac{1-x}{\epsilon} \right), $$
and $\deg(q) = d = 20 \epsilon t + 2\log(\varphi^2/\epsilon). $ To bound the absolute value of coefficients of $p$, we just need to bound the absolute value of coefficients of $q(x)$. This is accomplished by the following claim: 
\begin{restatable}{claim}{maxcoeff}\label{lemma:max_3oeff} The absolute value of coefficients of $q(x)$ is upper bounded by $\left( \frac{1}{\epsilon}\right)^{O(d)}.$
\end{restatable}
To prove \Cref{lemma:max_3oeff}, we relate the Chebyshev coefficients $c_{v,\epsilon, t}$ of $P$ to the coefficients with respect to the standard basis $\{1, x, \dots x^d\}$, and then perform the change of variables $y \coloneqq \frac{1-x}{\epsilon}$. The proof is presented in \Cref{sec:standardbasis}. 
The bound on the maximum coefficient on $p$ now follows from \Cref{lemma:max_3oeff} since $p(x) = x^t q(x)$. 

Next, we prove two stated properties of $p$, namely that 
 $|p(x)-1| \leq \epsilon/\varphi^2$ for $x\in [1-\epsilon, 1]$ and  $|p(x)|\leq  n^{-4}$  for $x \in [0, 1 - \varphi^2/4]$.

\paragraph{Bounding $|p(x)-1|$ on $[1-\epsilon, 1]$:} 
We will use the bound $\sum_{x \in [-1,1]}|P(x) - (1-\epsilon x)^{-t}| = O(d e^{-d})$ in \Cref{thm:Chebyshevapprox}.  More concretely, we have 
\begin{align*}
    \sup_{x \in [1-\epsilon,1]}|p(x) - 1| & \leq  \sup_{x \in [1-\epsilon,1]}\left|x^t\right|\left|P\left( \frac{1-x}{\epsilon}\right)-x^{-t}\right| && \\
    & \leq  \sup_{x \in [1-\epsilon,1]}\left|P\left( \frac{1-x}{\epsilon}\right)-x^{-t}\right|&&  \\
    & = \sup_{y \in [0,1]}\left|P(y) - (1-\epsilon y)^{-t} \right|&&  \\
    & \leq  \sup_{y \in [-1,1]}\left|P(y) - (1-\epsilon y)^{-t} \right| && \\
     & \leq O(de^{-d}) && \text{by \Cref{thm:Chebyshevapprox}}  \\
     & \leq O\left(\left(\epsilon t + \log(\varphi^2/\epsilon)\right) e^{-(\epsilon t + 2\log(\varphi^2/\epsilon))}\right)&& \\
     & \leq \epsilon t\cdot   e^{-\epsilon t}\cdot \frac{\epsilon}{\varphi^2}  + \log(\varphi^2/\epsilon) \cdot \left(\frac{\epsilon}{\varphi^2} \right)^2 &&\\
      & \leq \frac{\epsilon}{\varphi^2}. 
\end{align*}
Here the last transition uses the inequalities $xe^{-x} \leq 1/2$ for $x \in \mathbb R$ and $\log(x)/x \leq 1/2 $ for $x > 0$. 

\paragraph{Bounding $|p(x)|$ on $[0,1-\varphi^2/4]$:}
We have 
\begin{equation}\label{eq:px_norm}
\begin{aligned}
        \max_{x \in [0, 1 - \varphi^2/4]}|p(x)| & \leq \max_{x \in [0, 1 - \varphi^2/4]}|x^t|\cdot \max_{x \in [0, 1 - \varphi^2/4]}|q(x)|\\
    & \leq \left(1 - \frac{\varphi^2}{4}\right)^t \cdot\deg(q)\cdot \max|\coeff(q)|  \\
    &\leq  \left(1 - \frac{\varphi^2}{4}\right)^t\cdot O(\epsilon t + \log(\varphi^2/\epsilon))\cdot \left( \frac{1}{\epsilon}\right)^{O(\epsilon t + \log(\varphi^2/\epsilon))}. \\
\end{aligned}
\end{equation}
To continue, we observe that the final expression above is decreasing in $t$. To see this, we use the assumption $\epsilon/\varphi^2 \log(1/\epsilon)< c_1 $ for a sufficiently small constant $c_1$, which gives 
$$\left(1-\frac{\varphi^2}{4}\right)\left(\frac{1}{\epsilon}\right)^{O(\epsilon)} = \left(1-\frac{\varphi^2}{4}\right) e^{O(\epsilon \log(1/\epsilon))} \leq \left(1-\frac{\varphi^2}{4}\right)e^{\varphi^2/4} <1. $$ 
So the expression $ \left(1 - \frac{\varphi^2}{4}\right)^t\cdot O(\epsilon t + \log(\varphi^2/\epsilon))\cdot \left( \frac{1}{\epsilon}\right)^{O(\epsilon t + \log(\varphi^2/\epsilon))}$ is decreasing in $t$. 

Therefore, we can upper-bound the expression in \Cref{eq:px_norm} by setting $t = 20 \log n/\varphi^2$.  For $x \in [0, 1 - \varphi^2/4]$ it holds that 
\[|x^t| \leq \left(1 - \frac{\varphi^2}{4}\right)^{20\log(n)/\varphi^2} < n^{-5}.\] 
Setting $t = 20 \log n/\varphi^2$ and combining with  \Cref{eq:px_norm}, we get 
\begin{align*}
    \max_{x \in [0, 1 - \varphi^2/4]}|p(x)| &\leq \left(1 - \frac{\varphi^2}{4}\right)^t\cdot O(\epsilon t + \log(\varphi^2/\epsilon))\cdot \left( \frac{1}{\epsilon}\right)^{O(\epsilon t + \log(\varphi^2/\epsilon))} \\
    &\leq  n^{-5}\cdot O(\epsilon/\varphi^2\log n + \log(\varphi^2/\epsilon))\cdot \left( \frac{1}{\epsilon}\right)^{O(\epsilon/\varphi^2 \log n + \log(\varphi^2/\epsilon))} \\
     & \leq n^{-5} \cdot \left( 2^{O(\epsilon/\varphi^2 \log(1/\epsilon)\log n } + 2^{O(\log(\varphi^2/\epsilon) \cdot \log(1/\epsilon))} \right) \\
     & \leq n^{-4}\text{,}
\end{align*}
where the last inequality follows from the assumptions $\epsilon/\varphi^2 \log(1/\epsilon) \leq c_1$ and $\log(1/\epsilon)\log(\varphi^2/\epsilon) \leq c_2 \log n$ for sufficiently small constants $c_1, c_2$.  Thus, $p$ has all of the required properties, which completes the proof. 

\end{proof}

\subsection{Chebyshev approximation of $(1-\epsilon x)^{-t}$ on $[-1,1]$ (Proof of \Cref{thm:Chebyshevapprox})}\label{sec:Chebyshev_technical}

In this section we prove \Cref{thm:Chebyshevapprox}, restated here for the convenience of the reader: 
\chebyshevapprox*

It is well known that smooth functions can be approximated well by truncating their Chebyshev expansion. The following Lemma provides a standard bound for the error. The result is already known and follows directly from the \Cref{def:chebyshev}. See for example Proposition 2.2 in \cite{AA22} for a proof. 
\begin{lemma}\label{lemma:chebyshev_expansion}
Let $d \geq 1$ be an integer and let $a_0, a_1, \dots \in \mathbb{R}$ satisfy $\sum_{j =0}^{\infty} |a_j| < \infty$. Let $f \colon [-1,1] \to \mathbb{R}$ be defined by the absolutely convergent series $f(x) = \sum_{j = 0}^{\infty} 2^{j-1}a_j Q_j(x)$, and let $\widehat{f}_d \coloneqq \sum_{j = 0}^{d}2^{j-1}a_j Q_j(x)$ be the polynomial obtained by taking the first $d+1$ terms of the Chebyshev expansion of $f$. Then 
\begin{equation*}
    \sup_{x \in [-1,1]}|\widehat{f}_d(x) - f(x)| \leq \sum_{j = d+1}^{\infty}|a_j|. 
\end{equation*}
\end{lemma}

Using the above fact, we will show that the function $(1-\epsilon x)^{-t}$ on $[-1,1]$ can be approximated well by truncating the Chebyshev expansion. More concretely, we will show that the coefficients $c_{v,\epsilon,t}$ in the Chebyshev expansion of $(1-\epsilon x)^{-t}$ decay like a geometric sequence as $v \rightarrow \infty$, which will allow us to upper bound the error $\sum_{v=1}^{\infty}|c_{v,\epsilon,t}|$. First, we need to compute the coefficients $c_{v,\epsilon,t}$.  The following lemma will allow us to relate the Taylor expansion to the Chebyshev expansion. 

\begin{lemma}[(2.14) in \cite{MH_chebyshev}]\label{lemma:chebyshev_x}
    \begin{equation*} x^n = \sum_{k =0}^{\lfloor n/2 \rfloor}2^{-2k}{n \choose k}Q_{n-2k}(x), \end{equation*}
\end{lemma}
We now explicitly compute each coefficient $c_{v,\epsilon,t}$. 
\begin{lemma}\label{lemma:coef_explicit}
     For every $\epsilon >0, t > 0$ and $x \in [-1,1]$, we have  
    \begin{equation*}
        (1-\epsilon x)^{-t} =  \sum_{v = 0}^{\infty}2^{v-1} Q_{v}(x) c_{v,\epsilon,t}, 
    \end{equation*}
        where 
\begin{equation*}c_{v,\epsilon,t}= 2 \sum_{n -v \in 2\mathbb{Z}_{\geq 0}}\epsilon^n 2^{-n}{n-1+t \choose n}{n \choose \frac{n-v}{2}}. 
\end{equation*}
\end{lemma}
\begin{proof}
    The Taylor expansion of $(1-x)^{-t}$ is given by
\begin{equation*}
(1-x)^{-t} = \sum_{n=0}^{\infty}{-t \choose n}(-1)^n x^{n}= \sum_{n=0}^{\infty}{n-1+t \choose n} x^{n}, 
\end{equation*}
where the last equality holds because
\begin{equation*}
(-1)^n {-t \choose n}=  (-1)^n \frac{(-t)(-t-1) \cdots (-t-(n-1))}{n!} = \frac{t (t+1) \cdots (t+n-1)}{n!} = {n-1+t \choose n}.
\end{equation*}
This gives
\begin{equation*} (1-\epsilon x)^{-t} = \sum_{n=0}^{\infty}{n-1+t \choose n}\epsilon^n x^n.
\end{equation*}
By \Cref{lemma:chebyshev_x}, we have
\begin{equation*} x^n = \sum_{k =0}^{\lfloor n/2 \rfloor}2^{-2k}{n \choose k}Q_{n-2k}(x). \end{equation*}
Substituting in, we obtain
\begin{align*}
    (1-\epsilon x)^{-t} &= \sum_{n=0}^{\infty}{n-1+t \choose n}\epsilon^n  \sum_{k =0}^{\lfloor n/2 \rfloor}2^{-2k}{n \choose k}Q_{n-2k}(x) \\
    &= \sum_{v = 0}^{\infty} 2^{v-1} Q_{v}(x)\sum_{n-v \in 2\mathbb{Z}_{\geq 0}}{n-1+t \choose n}\epsilon^n 2^{1-n}{n \choose \frac{n-v}{2}} \\
    &=   \sum_{v = 0}^{\infty} 2^{v-1}Q_{v}(x) c_{v,\epsilon,t}.
\end{align*}
\end{proof}

\Cref{lemma:coef_explicit} gives us the coefficients $c_{v, \epsilon,t}$ in terms of an infinite sum. In order to understand the behavior of the coefficients, it will be useful to bound them as follows. 
\begin{lemma}\label{lemma:C-stirling-bound}
    We have that 
    \begin{equation*}
    c_{v,\epsilon,t} = \Theta \left(\sum_{n -v \in 2\mathbb{Z}_{\geq 0} } \sqrt{\frac{(n+t)}{t((n+1)^2-v^2)}}\exp(F_{v,\epsilon,t}(n)) \right), 
    \end{equation*}
    where
    \begin{align*}
    F_{v, \epsilon,t}(n) &\coloneqq -n \log \frac{1}{\epsilon} -n \log 2 +(n+t-1)\log(n+t-1)  \\
    & -(t-1)\log(t-1) -\left( \frac{n-v}{2}\right) \log \left( \frac{n-v}{2}\right)-\left( \frac{n+v}{2}\right) \log\left( \frac{n+v}{2}\right).
    \end{align*}
\end{lemma}
\begin{proof}
    We will use the following formulation of Stirling's formula, which follows from \cite{Robbins55}:   \\

    There exists constants $c,C >0$, such that for all $n\geq 0$,
    \begin{equation*}
    c(n+1)^{1/2}n^n e^{-n} \leq n! \leq C(n+1)^{1/2}n^ne^{-n}.
    \end{equation*}
    The lemma now follows immediately from \Cref{lemma:coef_explicit}, since
    \[ c_{v,\epsilon,t}= 2 \sum_{n -v \in 2\mathbb{Z}_{\geq 0}}\epsilon^n 2^{-n} \frac{(n+t-1)!}{(t-1)!(\frac{n-v}{2})! (\frac{n+v}{2})!}.\]
\end{proof}

We want to bound the coefficients $c_{v,\epsilon,t}$ by a single expression rather than an infinite sum. We will achieve this by showing that the function $F_{v,\epsilon,t}(n)$ decays at least linearly in $n$, for $n$ sufficiently large. This will allow us to upper-bound $c_{v,\epsilon,t}$ in terms of a geometric series. 

For simplicity of notation, we perform the change of variables $t \mapsto t+1$. More concretely, 
given $v,t, \epsilon$,  we define the function $F(n)$ on $[v,\infty)$ by
    \begin{align*}
     F(n)& \coloneqq F_{v,\epsilon,t+1}(n) = -n \log \frac{1}{\epsilon} -n \log 2 +(n+t)\log(n+t)\\
     &-t\log t -\left( \frac{n-v}{2}\right) \log \left( \frac{n-v}{2}\right)-\left( \frac{n+v}{2}\right) \log\left( \frac{n+v}{2}\right)  
     \end{align*}
     Furthermore, define
     \begin{align*}
     n_0 &\coloneqq \frac{t\epsilon^2 + \sqrt{\epsilon^2 t^2+v^2(1-\epsilon^2)}}{1-\epsilon^2} \\
     \Phi_{v, \epsilon,t}& \coloneqq F(n_0)
     \end{align*}
We now analyze the behavior of the function $F(n)$ and show that it decays at least linearly for sufficiently large  $n$. 
\begin{lemma}\label{lemma:Fbounds} Fix $v, t, \epsilon$. We have
\begin{enumerate}[label=(\textbf{\arabic*})]
\item $F$ is maximized at $n_0$ and $F(n_0) = \Phi_{v, \epsilon,t}$
\item For $z \geq 2n_0$, we have $F(n_0 + z) \leq F(n_0) -\frac{\log 2}{4}\cdot z$. 
\end{enumerate}
\end{lemma}
\begin{proof}
First, we prove ${\bf (1)}$, and then we prove ${\bf (2)}$. 
    \paragraph{Proving (1):}
    
    The function $F$ is twice differentiable with 
     $$F'(n) = -\log(1/\epsilon) - \log(2) + \log(n+t) - \frac{1}{2}\log \left( \frac{n-v}{2}\right) -\frac{1}{2}\log \left( \frac{n+v}{2}\right) $$
     and
     $$F''(n) = -\frac{nt+v^2}{(n+t)(n^2-v^2)} < 0,\\$$ where the last inequality holds since $n \geq v$. 
     So $F$ is concave on its entire domain. 
    The equation $F'(n_0) = 0 $ has the two solutions 
   \begin{align*}
    n &=  \frac{t\epsilon^2 \pm \sqrt{\epsilon^2t^2+v^2(1-\epsilon^2)}}{1-\epsilon^2}, 
   \end{align*}
    of which only the positive solution $n = n_0 = \frac{t\epsilon^2 + \sqrt{\epsilon^2t^2+v^2(1-\epsilon^2)}}{1-\epsilon^2}$ lies in the domain of $F$. 
    So $F$ is maximized at $n_0$. By definition of  $\Phi_{v, \epsilon,t}$, we have $F(n_0) = \Phi_{v, \epsilon,t}$. 

\paragraph{Proving (2):}
    First, we need 
    \begin{claim}
        $F'(2n_0) \leq - \frac{\log 2}{2}$ 
    \end{claim}
    \begin{proof}
        We have 
        \begin{align*} F'(n) &= -\log(2/\epsilon) + \log(n+t) - \frac{1}{2}\log\left(\frac{n-v}{2}\right) - \frac{1}{2}\log\left(\frac{n+v}{2}\right) \\
        &=  \log\left(\frac{\epsilon(n+t)}{\sqrt{n^2-v^2}} \right) = \frac{1}{2} \log \left( \frac{\epsilon^2 (n+t)^2}{n^2-v^2}\right),
        \end{align*}
       so 
        $$F'(2n_0) = \frac{1}{2} \log \left( \frac{\epsilon^2 (2n_0^2+t)^2}{4n_0^2-v^2}\right) .$$
        So it suffices to show that 
        $$\frac{\epsilon^2 (2n_0^2+t)^2}{4n_0^2-v^2}\leq \frac{1}{2}.$$
        Note that $n_0 \geq v, t \epsilon$. 
        So for $\epsilon \leq 1/10$, we have 
        \begin{align*}
        4n_0^2 - v^2 \geq 3n_0^2 \geq 2\epsilon^2 t^2 + n_0^2 \geq  2\epsilon^2 t^2 + \epsilon n_0 \cdot 10 n_0 \geq  2\epsilon^2 t^2 +  \epsilon n_0 (8 \epsilon n_0 + 8 \epsilon t) = 2\epsilon^2(2n_0^2+t)^2,
        \end{align*}
        as required. 
    \end{proof}
   Using the above claim, for $z \geq n_0$, we have
    \begin{align*}
    F(n_0 + z) - F(n_0) & \leq   F(n_0 + z) - F(2n_0) &&  \text{by maximality of $F(n_0)$} \\
    & \leq (z-n_0)F'(2n_0) && \text{since $F''(n)>0$ for all $n\geq v$ } \\
    & \leq -\frac{\log 2}{2}(z-n_0) && \text{since $F(2n_0) \leq - \frac{\log 2}{2}$ } \\
    & \leq - \frac{\log 2}{4} z && \text{using $n_0 \leq z/2,$}
    \end{align*}
    which is exactly what we needed to show. 
\end{proof}
We can now upper bound the coefficients $c_{v,\epsilon,t}$ in terms of $\Phi_{v,\epsilon,t}$. 
\begin{lemma}\label{lemma:Cbound}For all $v$, it holds that 
\begin{equation*}
 c_{v,\epsilon,t} = O\left((\epsilon t + v)\exp(\Phi_{v, \epsilon,t})\right) 
\end{equation*}
\end{lemma}
\begin{proof}
Recall from \Cref{lemma:C-stirling-bound}, that 
   \begin{equation*}
    c_{v,\epsilon,t} = \Theta \left(\sum_{n -v \in 2\mathbb{Z}_{\geq 0} } \sqrt{\frac{(n+t)}{t((n+1)^2-v^2)}}\exp(F_{v,\epsilon,t}(n)) \right), 
    \end{equation*}
    Observe that for all $n \geq v$,  we have 
\begin{equation}\label{eq:boundingsq}
        \sqrt{\frac{(n+t)}{t((n+1)^2-v^2)}} \leq \sqrt{\frac{(n/t+1)}{(n^2+2n-1)}} \leq \sqrt{\frac{n}{2n}} <1.  
    \end{equation}
    So we have
    \begin{align*}
         c_{v,\epsilon,t} & \leq C \sum_{n \geq v}^{\infty} \exp{ F(n) } && \text{by \Cref{eq:boundingsq} and \Cref{lemma:C-stirling-bound}} \\  
        &=  C \sum_{n = v}^{3n_0-1}\exp{F(n_0)} + C\sum_{n \geq 3n_0} \exp{F(n)} \\
        & \leq C \cdot 3n_0F (n_0) + C\sum_{z \geq 2 n_0} \exp({F(n_0)-c z}) && \text{by \Cref{lemma:Fbounds}}\\
        & =  C \cdot 3n_0F (n_0) + CF(n_0) \sum_{z \geq 2 n_0} \exp(-\log{2} z/4) \\
        & \leq C \cdot 3n_0F (n_0) + C\cdot 10F(n_0) \\
        & = O( (\epsilon t + v )\exp(\Phi_{v, \epsilon,t})) && \text{since $n_0 \leq 2(\epsilon t +v)$ and since $F(n_0) = \Phi_{v,\epsilon,t}$}
    \end{align*}
\end{proof}

Next, in order to upper bound the coefficients $c_{v,\epsilon,t}$, we need to understand the behavior of  $\Phi_{v,\epsilon,t}$ as a function in $v$. The following lemma shows that $\Phi_{v,\epsilon,t}$ can be upper bounded by a function that decays at least linearly in $v$ for sufficiently large values of $v$. This, in turn, will be useful for upper bounding the error of the Chebyshev expansion in terms of a geometric series. 

\begin{lemma}\label{lemma:Fn_large_v}
  If $v \geq 10 \epsilon t$, then
  $ \Phi_{v, \epsilon,t} \leq -v $
\end{lemma}
\begin{proof}
Given any fixed $v,t,\epsilon$, we can re-write the function $F(n)$ as 
 \begin{equation} \label{eqn:Fn_rewriting}
 \begin{aligned}
F(n) &= -n \log \frac{2}{\epsilon} +(n+t)\log(n+t) -t\log t -
\left( \frac{n+v}{2}\right) \log\left( \frac{n+v}{2}\right)- \left( \frac{n-v}{2}\right) \log \left( \frac{n-v}{2}\right)\\
&=  -n \log \frac{2}{\epsilon}  + n \log(n+t) + t \log(1+n/t) - n  \log\left( \frac{n+v}{2}\right) + \left( \frac{n-v}{2}\right) \log \left( \frac{n+v}{2}\right) - \left( \frac{n-v}{2}\right) \log \left( \frac{n-v}{2}\right) \\
& =  -n \log \frac{2}{\epsilon} + n \log \left( \frac{2(n+t)}{n+v}\right) + t \log(1+n/t) +  \left( \frac{n-v}{2}\right) \log \left( \frac{n+v}{n-v}\right) 
\end{aligned}
\end{equation}
\noindent
We will show that for $n = n_0 = \frac{t\epsilon^2 + \sqrt{\epsilon^2 t^2+v^2(1-\epsilon^2)}}{1-\epsilon^2}$, all of the three terms $ n \log \left( \frac{2(n+t)}{n+v}\right)$, $t \log(1+n/t)$, $\left( \frac{n-v}{2}\right) \log \left( \frac{n+v}{n-v}\right)$ can offset by the $-n \log \frac{2}{\epsilon}$ term.

First, observe that we have $$v < n_0 \leq 2v$$ (the lower bound is immediate, and the upper bound holds since  $v \geq 10 \epsilon t$ and $\epsilon$ is sufficiently small). We now have
\begin{equation}\label{eqn:term1}
\begin{aligned}
 n_0 \log \left( \frac{2(n_0+t)}{n_0+v}\right) & \leq  n_0 \log \left( \frac{2(n_0+t)}{n_0+n_0/2}\right)&   \\
&= n_0 \left( \log {4/3} + \log(1+t/n_0) \right) & \\
& \leq n_0 \left( \log {4/3} + \log\left(1+\frac{1}{10\epsilon}\right) \right) & \text{ since $n_0 \geq v \geq 10\epsilon t$} 
\end{aligned}
\end{equation}
Next, by applying the inequality $\frac{1}{x}\log (1+x) \leq 1$ with $x = n_0/t$, we obtain 
\begin{equation}\label{eqn:term2}
t \log(1+n_0/t)  \leq n_0 .
\end{equation}
Finally, we bound the $\left( \frac{n-v}{2}\right) \log \left( \frac{n+v}{n-v}\right)$ term. Given $y >0$, the function $f(x) = x \log(y/x)$ is maximised at $x = y/e$, so for all $x,y > 0$, it holds that 
$$ x \log(y/x) \leq \frac{y}{e} \log(e) = \frac{y}{e}. $$
Applying this with $x = n-v$ and $y = n+v$ gives
\begin{equation}\label{eqn:term3}
 \frac{n_0-v}{2} \log \left( \frac{n_0+v}{n_0-v} \right) \leq \frac{(n_0+v)}{2e} \leq \frac {n_0}{e}.
\end{equation}
Before putting everything together, observe that 
\begin{equation}\label{eqn:calculation}
\log{4/3}+\log\left( 1 + \frac{1}{10\epsilon} \right) + \frac{1}{e}+2 \leq \log\left(\frac{2}{\epsilon}\right). 
\end{equation}
To see this, note that 
\begin{equation*}
\frac{4}{3}\frac{e^{1/e+2}}{10} \approx 1.423 <2, 
\end{equation*}
and therefore, for $\epsilon$ sufficiently small, we get
\begin{equation*}
\frac{4}{3} \cdot \left( 1+ \frac{1}{10\epsilon} \right) e^{1/e+2} \leq \frac{2}{\epsilon}
\end{equation*}
Taking log-s on both sides gives \Cref{eqn:calculation}.

We can now combine Equations \eqref{eqn:Fn_rewriting}, \eqref{eqn:term1}, \eqref{eqn:term2}, \eqref{eqn:term3} and \eqref{eqn:calculation} to obtain 

\begin{equation*}
\Phi_{v, \epsilon,t} =  F(n_0) \leq -n_0 \log{\frac{2}{\epsilon}} + n_0\log{4/3}+n_0\log\left( 1 + \frac{1}{10} \right) + \frac{n_0}{e}+n_0 \leq -n_0 \leq -v, 
\end{equation*}
as required. 
\end{proof}
It remains to control $\Phi_{v,\epsilon,t}$ for small values of $v$. This is accomplished by the following lemma.  
\begin{lemma}\label{lemma:F_small_v}
    If $v \leq 10 \epsilon t$, then $\Phi_{v, \epsilon,t} \leq 100\epsilon t $
\end{lemma}
\begin{proof}
Recall  that 
    $$F(n) = -n \log \frac{2}{\epsilon} + n \log \left( \frac{2(n+t)}{n+v}\right) + t \log(1+n/t) +  \left( \frac{n-v}{2}\right) \log \left( \frac{n+v}{n-v}\right) $$
    and 
    $$n_0 = \frac{t\epsilon^2 + \sqrt{\epsilon^2 t^2+v^2(1-\epsilon^2)}}{1-\epsilon^2}.$$ 
\noindent
    Note that  $$v, t\epsilon < n_0 \leq 20 \epsilon t$$ for $0 \leq v \leq 10\epsilon t$.  We will now bound each of the terms in $F(n)$ separately. 

   \noindent
   To bound the $t \log(1+n/t)$-term: As shown in \Cref{eqn:term2}, we have 
    $$ t \log(1+n/t) \leq n_0 $$

\noindent
    To bound the  $\left( \frac{n-v}{2}\right) \log \left( \frac{n+v}{n-v}\right)$-term: As show in \Cref{eqn:term3}, we have
    $$  \frac{n_0-v}{2} \log \left( \frac{n_0+v}{n_0-v} \right) \leq \frac{(n_0+v)}{2e} \leq \frac{n_0}{e}< n_0.$$
\noindent
To bound the  $n_0 \log \left( \frac{2(n_0+t)}{n_0+v}\right) $ term:  We have
    $$ n_0 \log \left( \frac{2(n_0+t)}{n_0+v}\right)\leq  n_0 \log \left( \frac{2(n_0+t)}{n_0}\right) =    n_0 \log(2(1+t/n_0)) \leq  n_0 \log(2(1+1/\epsilon))  \leq  n_0 \log(4) + n_0 \log(1/\epsilon).$$
    Putting it together, we get 
    $$F(n_0) \leq -n_0 \log(2/\epsilon) + n_0 \log(4) + n_0 \log(1/\epsilon) + n_0 + n_0 \leq 3 n_0 \leq 60 \epsilon t.$$
\end{proof}

We are now ready to prove \Cref{thm:Chebyshevapprox}.
\begin{proof}[Proof of \Cref{thm:Chebyshevapprox}]
The first bullet-point follows immediately from \Cref{lemma:Cbound}, \Cref{lemma:Fn_large_v} and \Cref{lemma:F_small_v}, together with the fact that all the coefficients are non-negative. 

To prove the second bullet point: By \Cref{lemma:chebyshev_expansion}, we have 

$$\sup_{x \in [-1,1]} |P(x) - (1-\epsilon x)^{t+1}| \leq \sum_{v = d}^{\infty}|c_{v,\epsilon,t}|.$$
\noindent
Using this we obtain 
\begin{align*}
    \sup_{x \in [-1,1]} |P(x) - (1-\epsilon x)^{t+1}| & \leq \sum_{v = d}^{\infty}|c_{v,\epsilon,t}|  \\
    & \leq C_1 \sum_{v \geq d} (\epsilon t + v) \exp(\Phi_{v, \epsilon,t}) && \text{by \Cref{lemma:Cbound} } \\
    & \leq C_2  \sum_{v \geq d}  v \exp(-v) && \text{by \Cref{lemma:Fn_large_v}, since $d \geq 10\epsilon t$} \\
    & \leq C_3 d e^{-d},
\end{align*}
for some universal constants $C_1$, $C_2$ and $C_3$. 
\end{proof}

\subsection{Change of polynomial basis and domain (Proof of \Cref{lemma:max_3oeff})}\label{sec:standardbasis}
In this section, we prove \Cref{lemma:max_3oeff}, restated below for the convenience of the reader. 
\maxcoeff*

First, we need to relate the coefficients with respect to the Chebyshev basis to the coefficients with respect to the standard basis $\{1, x, \dots, x^d\}$. This is accomplished by the following lemma. 

\begin{lemma}\label{lemma:changeofbasis}
For $l \geq 0$, let $C(l)$ denote the maximum absolute value of the coefficients of $2^{l-1}Q_l(x)$ with respect to the standard basis $\{1, x, \dots x^n\}$. Then 
$$C(l) \leq (1 + \sqrt{2})^l.$$
\end{lemma}

\begin{proof} Observe that polynomials $T_l(x) = 2^{l-1}Q_l(x)$ satisfy the following recurrence 

$$T_{l+1}(x) = 2xT_l(x) + T_{l - 1}(x),$$
\noindent
which means that 

$$C(l+1) \leq 2C(l)+C(l - 1).$$
\noindent
So to upper bound  $C(l)$, if suffices to solve the recurrence 

$$\widetilde{C}(l + 1) = 2\widetilde{C}(l) + \widetilde{C}(l - 1), \widetilde{C}(0) = 1, \widetilde{C}(1) = 1.$$
Here the initial conditions follow from the fact that $T_0(x) = 1$ and $T_1(x) = x$. 
It is easy to see that

$$C(l) = 0.5\left((1-\sqrt{2})^l + (1 + \sqrt{2})^l\right) \leq (1 + \sqrt{2})^l.$$

\end{proof}

As a corollary, we bound the coefficients of the polynomial $P$ defined in \Cref{thm:Chebyshevapprox} with respect to the standard basis. 
\begin{cor} \label{cor:standardbasis}
Let $d \geq 10 \epsilon t$ and let $P(x) = \sum_{v = 0}^{d}2^{v-1}Q_v(x)c_{v,\epsilon,t}$ be the 
the degree $d$ Chebyshev approximation of $(1-\epsilon x)^{-t}$. Then the absolute value of the coefficients of $P$ with respect to the standard basis is upper bounded by $2^{O(d)}$
\end{cor}
\begin{proof} Denote by $[a_m]Q(x)$ the coefficient of a polynomial $Q(x)$ at monomial $x^m$. Then, 
\begin{align*}
    \left|[a_m]P(x)\right| &= \sum_{v = 0}^d[a_m]2^{v - 1}Q_v(x)c_{v, \epsilon, t}\\
    & \leq d\cdot \max_{v \leq d}\left|[a_m]2^{v-1}Q_v(x)c_{v, \epsilon, t}\right| \\
    & \leq d\cdot \max_{i \leq d}\left|[a_m]2^{i-1}Q_i(x)\right|\cdot\max_{j \leq d}c_{j, \epsilon, t}
\end{align*}
\noindent
By \Cref{lemma:changeofbasis}, for all $i$ $\max_{m \leq d}\left|[a_m]2^{i-1}Q_i(x)\right| \leq C(i) \leq 2^{O(i)}$, so \[\max_{i \leq d}\left|[a_m]2^{i-1}Q_i(x)\right| \leq \max_{m, i \leq d}\left|[a_m]2^{i-1}Q_i(x)\right| \leq \max_{i}2^{O(i)} = 2^{O(d)}.\]
\noindent
By \Cref{lemma:Cbound}, $C_{j, \epsilon, t} \leq O\left((\epsilon t + j)\exp(\Phi_{j,t,\epsilon})\right)$ and, by Lemmas \ref{lemma:Fn_large_v} and \ref{lemma:F_small_v} $\Phi_{j,t,\epsilon} \leq \max\{-j, 100\epsilon t\}$, so
\[C_{j, \epsilon, t} \leq O\left((\epsilon t + j)\exp(100\epsilon t)\right)\leq O\left(d\cdot 2^{O(\epsilon t)}\right).\]
\noindent
Hence, 

$$[a_m]P(x) \leq d\cdot \max_{m, i \leq d}\left|[a_m]Q_i(x)2^{i-1}\right|\cdot\max_{j \leq d}C_{j, \epsilon, t} \leq O\left( d^2\cdot 2^{O(d)}\cdot2^{O(\epsilon t)} \right) = 2^{O(d)}.$$
\end{proof}

Finally, we are ready to prove \Cref{lemma:max_3oeff}.

\begin{proof}[Proof of \Cref{lemma:max_3oeff}] 
Recall that $$q(y) := P\left(\frac{y - 1}{\epsilon}\right).$$
 Let $a_0, \ldots, a_d$ denote the coefficients of $P$ with respect to the standard basis, and let $b_0, \ldots, b_d$ denote the coefficients of $q$ with respect to the standard basis. In other words, 
$$P(x) = a_d x^d + \ldots + a_0,$$
and 
$$q(x) = b_d x^d + \ldots + b_0.$$
We have 
$$b_i = \sum_{j = i}^da_j \left(-\frac{1}{\epsilon}\right)^{j - i}\binom{j}{i}.$$
By Corollary \ref{cor:standardbasis}, we have $|a_i| \leq 2^{O(d)}$ for all $i$, and so the can upper bound absolute value of the $i^{th}$ coefficient of $q$ as 
$$|b_i| \leq d\max_{j = i, \ldots, d}|a_j|\left(\frac{1}{\epsilon}\right)^{d-i}\binom{d}{i} \leq d\cdot2^{O(d)}\left(\frac{1}{\epsilon}\right)^{d - i}\binom{d}{i}.$$
Hence, the absolute value of the coefficients of $q$ is upper bounded by 
\begin{align*}
\max_i |b_i| \leq \max_i  d\cdot2^{O(d)}\left(\frac{1}{\epsilon}\right)^{d - i}\binom{d}{i} \leq d\cdot 2^{O(d)}\left(\frac{1}{\epsilon}\right)^{d}\binom{d}{d/2} \leq 2^{O(\log(1/\epsilon)d)} = \left(\frac{1}{\epsilon}\right)^{O(d)}, 
\end{align*}
as required. 
\end{proof}

\section{Bounding $\ell_2$-norm of $p^t_x$ for all $x \in V$}
Several of our proofs in \Cref{sec:well-spread_sets} require a good bound on $\|p^t_x\|_2 = \|M^t 1_x\|_2 $ for \emph{all} $x \in V$ when $t$ is sufficiently large. Note that the bound $\|M^t 1_x\|_2 = O\left(\sqrt{k/n}\right)$ for \emph{good} vertices $x \in V \setminus B_{\delta}$ follows from the fact that $\|f_x\|^2_2 =  O\left(k/n\right)$ for all $x \in V \setminus B_{\delta}$ (see \Cref{remark:norm}). Informally, this is because $\|f_x\|^2_2 = \sum_{i = 1}^k\langle \1_x, v_i\rangle^2$, and for sufficiently large $t$ the projections of $M^t\1_x$ on $v_i$, $i > k$, are negligibly small. However, we need a similar bound also for bad vertices $x \in B_{\delta}$. This is used in the proofs of \Cref{lemma:good_pts_wrt_S} and \Cref{lemma:random_good_cluster}, where we bound the number of $x \in V$ and $x \in S$ that that violate the conditions $\left \langle p_{x}^{t} , \left(p_S^{l}\right)^2 \right \rangle \leq O^*\left(\frac1{n^{3/2}} \cdot \frac{k^{3/2}}{|S|^2}\right)$ or $ \left \langle \left( p_x^t\right)^2, p_S^l\right \rangle \leq O^*\left(\frac1{n^{3/2}} \cdot \frac{k^{3/2}}{|S|}\right)$ in \Cref{def:good_pt_wrt_S}. 

Bounds that work simultaneously for all vertices in $V$ are not common in the literature. The only such bound we are aware of is Lemma 22 from \cite{GKLMS21}, which follows by a careful analysis of balance conditions that bottom $k$ eigenvectors of the Laplacian satisfy. They show that $\|M^t \1_x\|_2 = O\left(k \cdot n^{-1/2 + O(\epsilon/\varphi^2)} \right)$, a factor $\sqrt{k}$ more than the bound we need. Unfortunately, we cannot afford this additional loss. We prove the tight bound by different methods below.

\begin{lemma}\label{lem:bound-norms-1/2}Let $M$ be the lazy random walk matrix $M =  \frac{1}{2}I + \frac{1}{2d}A$. Then, for $t \geq \frac{10\log(n)}{\varphi^2}$ and all $x \in V$,

$$\|M^t\mathbbm{1}_x\|_2 = O\left(\sqrt{\frac{k}{n}}\right).$$
\end{lemma}
In fact, we prove the following stronger result: 
\begin{lemma}\label{lemma:bound-norms}
Let $k \geq 2$ be an integer, $\varphi \in (0,1)$, and $\epsilon \in (0,1)$. Let $G = (V,E)$ be a
$d$-regular graph that admits a $(k, \varphi, \epsilon)$-clustering $C_1, \ldots, C_k$. Let $\widetilde{M}$ be the transition matrix of lazy random walk that stays in the vertex  with probability $3/4$ (and crosses any edge with probability $1/4d$), i.e.
$$ \widetilde{M} = \frac{3}{4}I + \frac{1}{4d}A.$$

Then, for all $t \geq \frac{10\log(n)}{\varphi^2}$ and all $x \in V$, 

$$\|\widetilde{M}^t\mathbbm{1}_x\|_2 = O\left(\sqrt{\frac{k}{n}}\right).$$
\end{lemma}
\Cref{lem:bound-norms-1/2} follows as a corollary of Lemma 
\ref{lemma:bound-norms}. 
\begin{proof}[Proof of \Cref{lem:bound-norms-1/2} ]
Observe that $M = I - \frac{1}{2}\mathcal{L}$ and $\widetilde{M} = I - \frac{1}{4}\mathcal{L}$. Matrices $M$, $\widetilde{M}$, $\mathcal{L}$ have the same eigenvectors. Recall that we denote by $\lambda_1 \leq \ldots \leq \lambda_n$ the eigenvalues of $\mathcal{L}$ in ascending order. We can express the eigenvalues of $M$ and $\widetilde{M}$ in terms of $\lambda_1, \ldots, \lambda_n$ as $1 - \frac{\lambda_1}{2}, \ldots, 1 - \frac{\lambda_n}{2}$ and $1 - \frac{\lambda_1}{4}, \ldots, 1 - \frac{\lambda_n}{4}$ correspondingly.

Define $\widetilde{\Sigma}$ to be the diagonal matrix of eigenvalues of $\widetilde{M}$ arranged in non-increasing order, and recall that we denote by  $\Sigma$ the diagonal matrix of eigenvalues of $M$ arranged in non-increasing order.  Note that since $0 \leq \lambda_1 \leq \ldots \leq \lambda_n \leq 2$, we have that $0 \preceq \Sigma \preceq \widetilde{\Sigma}$, and therefore for every $t$ we have that $\Sigma^t \preceq \widetilde{\Sigma}^t$. Hence, $M^t = U\Sigma^tU^\top  \preceq U\widetilde{\Sigma}^tU^\top  = \widetilde{M}^t$ and for every $x$ and every $t \geq \frac{10 \log n}{\varphi^2}$, by \Cref{lemma:bound-norms},
\[\|M^t\1_x\|^2_2 = \1_x^\top M^{2t}\1_x \leq \1_x^\top \widetilde{M}^{2t}\1_x = \|\widetilde{M}^t\1_x\|^2_2 \leq O(k/n),\]
as desired.
\end{proof}

It remains to prove \Cref{lemma:bound-norms}.

\begin{proof}[Proof of \Cref{lemma:bound-norms} ]
Let $G_{cross}$ be the subgraph of $G$ consisting of all edges with endpoints in different clusters $C_1, \ldots, C_k$ with added self-loops so that the degree of every vertex is $d$. Let $G_{in}$ be the subgraph of $G$ consisting of all edges inside the clusters with added self-loops so that the degree of every vertex is $d$. Let $\widetilde{M}_{in}$ and $\widetilde{M}_{cross}$ be the matrices of lazy random walks on $G_{in}$ and $G_{cross}$ respectively with probability of staying in the vertex $1/2$. Then, 

$$\widetilde{M} = 1/2\widetilde{M}_{in} + 1/2\widetilde{M}_{cross}.$$

For every distribution $\pi$ on the vertices of $G$, by triangle inequality, 
$$\|\widetilde{M}\pi\|_2 \leq 1/2\|\widetilde{M}_{in}\pi\|_2 + 1/2\|\widetilde{M}_{cross}\pi\|_2,$$ and, moreover,
\begin{equation}\label{eq:M_cross}
\|\widetilde{M}_{cross}\pi\|_2 \leq \|\pi\|_2
\end{equation}
since the eigenvalues of $\widetilde{M}_{cross}$ take values in $[0,1]$, as $\widetilde{M}_{cross}$ is a matrix of a random walk on a $d$-regular graph. 

For a distribution $\pi$ and every $i\in [n]$, let $\pi_i$ denote the restriction of $\pi$ to $C_i$, i.e. $(\pi_i)_x$ equals $\pi_x$ if $x\in C_i$ and zero otherwise. Let us denote by \(\widetilde{M}_{\mathrm{in},i}\) the restriction of \(\widetilde{M}_{\mathrm{in}}\) to \(C_i \times C_i\); that is, \(\widetilde{M}_{\mathrm{in},i} \in \mathbb{R}^{n \times n}\) agrees with \(\widetilde{M}_{\mathrm{in}}\) on indices in \(C_i \times C_i\) and is zero elsewhere.
 Note that $(M_{in})_{x,y}=0$ whenever $x$ and $y$ belong to different clusters, so 

\begin{equation}\label{eq:Min_fistbd} 
\|\widetilde{M}_{in}\pi\|^2_2 = \left\|\sum_{i=1}^k \widetilde{M}_{in, i}\pi_i\right\|^2_2 =  \sum_{i = 1}^k \|\widetilde{M}_{in, i}\pi_i\|^2_2.
\end{equation}
Fix an orthonormal eigenbasis $v_{1, i}, \ldots, v_{n, i}$ of $\widetilde{M}_{in, i}$. We will denote by $\lambda_{l, i}$ the $l$-st biggest eigenvalue of $\widetilde{M}_{in, i}$ and by $v_{l, i}$~-- the corresponding unit norm eigenvector. 

Note that the $n - |C_i|$ smallest eigenvalues of $\widetilde{M}_{in, i}$ are 0. Since every $C_i$ is a $\varphi$-expander, by Cheeger's inequality, we have $\lambda_{2, C_i} \leq 1 - \frac{\varphi^2}{4}$. Since $G_{in}\cap C_i$ is a $d$-regular graph for every $i \in [k]$, $\lambda_{1, C_i} = 1$ and $v_{1, i}$ is the normalized indicator of cluster $C_i$ for every $i \in [k]$.

\begin{equation}\label{eq:M_in}
\begin{aligned}
\|\widetilde{M}_{in, i}\pi_i\|_2^2 = \|\pi_i^\top v_{1, i} + \sum_{l  = 2}^{|C_i|} \lambda_{l, i}\pi_i^\top v_{l, i}\|^2_2 \leq |\pi_i^\top v_{1, i}|^2 &+ \left(1 - \frac{\varphi^2}{4}\right)^2\left\|\sum_{l  = 2}^{|C_i|} \pi_i^\top v_{l, i}\right\|^2_2\\
\leq \left(1 - \frac{\varphi^2}{4}\right)^2\|\pi_i\|^2_2 &+ \left(1 - \left(1 - \frac{\varphi^2}{4}\right)^2\right)|\pi_i^\top v_{1, i}|^2 \\
\leq \left(1 - \frac{\varphi^2}{4}\right)^2\|\pi_i\|^2_2 &+ O(\varphi^2)\cdot |\pi_i^\top v_{1, i}|^2.
\end{aligned}
\end{equation}
Note that $|\pi_i^\top v_{1, i}| = \frac{1}{\sqrt{|C_i|}}\|\pi_i\|_1$, since $v_{1, i}$ is just the normalized indicator vector on cluster $C_i$. Hence,

\begin{equation}\label{eq:sum_pi_l1}
\sum_{i = 1}^k |\pi_i^\top v_{1, i}|^2 = \sum_{i=1}^k \frac{1}{|C_i|}\|\pi_i\|_1^2 \leq \frac{\eta k}{n}\sum_{i = 1}^k\left\|\pi_i\right\|^2_1  \leq \frac{\eta k}{n}\left\|\sum_{i = 1}^k\pi_i\right\|^2_1 \leq O(k/n),
\end{equation} where the second inequality holds since $\sum_{i = 1}^k\pi_i = \pi$ and $\pi$ is a probability distribution.

Summing \Cref{eq:M_in} over all $i \in [k]$ and combining with Equations\eqref{eq:Min_fistbd} and  \eqref{eq:sum_pi_l1} gives 

\begin{equation}\label{eq:indyp}
    \|\widetilde{M}_{in}\pi\|^2_2 \leq \left(1 - \frac{\varphi^2}{4}\right)^2\|\pi\|^2_2 + O(\varphi^2k/n).
\end{equation}
Let $C$ be a sufficiently large universal constant for which $\|\widetilde{M}_{in}\pi\|^2_2 \leq \left(1 - \frac{\varphi^2}{4}\right)^2\|\pi\|^2_2 + C\varphi^2k/n.$ We can represent $\widetilde{M}^t$ as the sum of $2^t$ matrices: $\widetilde{M}^t = (1/2\widetilde{M}_{in} + 1/2\widetilde{M}_{cross})^t.$ Recall from \Cref{eq:M_cross} that for \emph{any} distribution $\pi'$, it holds that $\|\widetilde{M}_{cross} \pi'\|^2_2 \leq \|\pi'\|^2_2.$ Thus, repeatedly using \Cref{eq:indyp} and \Cref{eq:M_cross},  we see that for any matrix $\widetilde{M}_l$, product of $l$ matrices $\widetilde{M}_{in}$ and $t - l$ matrices $\widetilde{M}_{cross}$, it is true that

\begin{equation}\label{eq:Mm}
    \begin{aligned}
         \|\widetilde{M}_l\pi\|^2_2 &\leq \left(1 - \frac{\varphi^2}{4}\right)^{2l}\|\pi\|^2_2 +  C\varphi^2k/n \cdot 
       \left(1 + \left(1 - \frac{\varphi^2}{4}\right)^2  + \ldots + \left(1 - \frac{\varphi^2}{4}\right)^{2(l - 1)}\right)\\
       &\leq \left(1 - \frac{\varphi^2}{4}\right)^{2l}\|\pi\|^2_2 + 4Ck/n,
    \end{aligned}
\end{equation}
where the last inequality follows since the sum $\left(1 + \left(1 - \frac{\varphi^2}{4}\right)^2  + \ldots + \left(1 - \frac{\varphi^2}{4}\right)^{2(l - 1)}\right)$ is a geometric sequence bounded by $\frac{1}{1 - \left(1 - \frac{\varphi^2}{4}\right)^2} \leq \frac{4}{\varphi^2}$ for any $l$, since $\varphi \leq 1$. 

Therefore, 
\begin{equation}
    \begin{aligned}
        \|\widetilde{M}^t\pi\|_2 \leq & 2^{-t}\sum_{l = 1}^t \binom{t}{l} \sqrt{\left(1 - \frac{\varphi^2}{4}\right)^{2l}\|\pi\|^2_2 + 4Ck/n} \leq  2^{-t}\sum_{l = 1}^t \binom{t}{l}\left(\left(1 - \frac{\varphi^2}{4}\right)^{l}\|\pi\|_2 + \sqrt{4Ck/n}\right)\\
        &\leq 2^{-t}\left(1 + 1 - \frac{\varphi^2}{4}\right)^t\|\pi\|_2 + \sqrt{4Ck/n} = \left(1 - \frac{\varphi^2}{8}\right)^t\|\pi\|_2 + \sqrt{4Ck/n},
    \end{aligned}
\end{equation}
where the third inequality uses the identity $(1+z)^t = \sum_{l = 0}^t {t \choose l}z^t$ for all $z \in \R$, applied to $z = \left(1 - \frac{\varphi^2}{4}\right)$ and $z = 1$. 
Since $\pi$ is a distribution, we have $\|\pi\|_2 \leq 1$, so for every  $t \geq \frac{10\log(n)}{\varphi^2}$,
\[\|\widetilde{M}^t\pi\|_2 \leq 1/n + \sqrt{4Ck/n} \leq O(\sqrt{k/n}).\]
\end{proof}
\section{Properties of the Spectral Embedding (Proof of \Cref{lemma:good_pts_properties} and \Cref{lemma:sketchx}) 
}\label{sec:spectral_facts}

We begin this section by proving \Cref{lemma:good_pts_properties}, restated below. 
\goodptsproperties*
\begin{proof}
Suppose $x,y \in C_i \setminus B_{\delta}$. Then 
\begin{equation}\label{eq:good_pts_props}
\begin{aligned}
 \left| \langle f_x, f_y \rangle - \|\mu_i\|_2^2\right|&= \left|\langle (f_x- \mu_{i}) + \mu_{i}, (f_y  - \mu_{i}) + \mu_{i}\rangle - \|\mu_{i}\|^2_2\right| \\
& \leq  | \langle f_x- \mu_{i}, f_y  - \mu_{i}\rangle | + | \langle f_x- \mu_{i},  \mu_{i}\rangle | + | \langle \mu_{i}, f_y  - \mu_{i}\rangle |  && \text{by triangle inequality}\\
& \leq \| f_x- \mu_{i}\|_2 \| f_y  - \mu_{i}\|_2 + \| f_x- \mu_{i},  \mu_{i}\|_2 +\| \mu_{i}\|_2 \|  f_y  - \mu_{i}\|_2 \\
&\leq   \frac{\delta}{|C_{i}|} + \frac{\sqrt{2\delta}}{|C_{i}|}+ \frac{\sqrt{2\delta}}{|C_{i}|} \quad \text{by \Cref{def:B_delta} and \Cref{lemma:clustermeans}} \\
& \leq \frac{4\sqrt{\delta}}{|C_{i}|}.
\end{aligned}
\end{equation}
Therefore,  
\[\left| \langle f_x, f_y \rangle - \frac{1}{|C_{i}|}\right| \leq \left| \langle f_x, f_y \rangle - \|\mu_{i}\|_2^2\right| + \left|\|\mu_{i}\|^2_2 - \frac{1}{|C_{i}|}\right| \leq \frac{4\sqrt{\delta}}{|C_{i}|} + \frac{4\sqrt{\epsilon}}{\varphi}\frac{1}{|C_{i}|} \qquad \text{by \eqref{eq:good_pts_props} and  \Cref{lemma:clustermeans}}.\]
\end{proof}

\noindent 
Next, we prove \Cref{lemma:sketchx},  restated below. 
\sketchxguarantee*
\begin{proof}
     Let  $\sigma_x \sum_{t} c_t \cdot \widehat{p}^t_x$ and $\sigma'_x \sum_{t} c_t \cdot \widehat{p}'^t_x $ be the output from the two independent calls to $\sketch(x)$, where we denote the two independent instantiations of random signs inside the two calls $\sigma_x$ and $\sigma_x’$. Given $t_1$, $t_2 \in [t_{\min}, t_{\min} + t_{\Delta}]$, by \Cref{lem:bound-norms-1/2} we have $p^{t_1}_{x}(x)= \|M^{t_1}\1_x\|^2_2 \leq O(\frac{k}{n})$, and $\left \langle p_x^{t_1}, \left(p_x^{t_2}\right)^2\right \rangle  \leq \|p_x^{t_1}\|_2 \|\left(p_x^{t_2}\right)^2\|_2 \leq \|p_x^{t_1}\|_2 \|p_x^{t_2}\|^2_2 \leq O\left( \frac{k^{3/2}}{n^{3/2}}\right).$ 
     Here the inequality $\|\left(p_x^{t_2}\right)^2\|_2 \leq \|p_x^{t_2}\|^2_2$ holds by noting that   $\|\left(p_x^{t_2}\right)^2\|^2_2 = \sum_{v \in V}p_x^{t_2}(v)^4 \leq \left(\sum_{v \in V}p_x^{t_2}(v)^2 \right)^2 = \|p_x^{t_2}\|^4_2$. 
     So we can apply \Cref{lemma:variance_calc}  with $\beta = O(1)$, $\gamma = O\left(\sqrt{\frac{n}{k}}\right)$,  $\rho = 10^{-6}\frac{1}{(t_{\Delta}+1)^2} = O^*(1)$ and $\xi = \frac{5 \cdot 10^{-4}}{\eta \cdot |c_{t_1}||c_{t_2}|(t_{\Delta}+1)^2} =O^*\left( \frac{1}{(1/\epsilon)^{t_{\Delta}}}\right)$ (where the last equality holds by bound on the coefficient size in \Cref{thm:standardbasis}). Setting $Q, R = 7\cdot\sqrt{\frac{n}{k}} \cdot  \frac{1}{\rho \cdot  \xi^2} =  O^*\left(\sqrt{\frac{n}{k}} \cdot \left(\frac{1}{\epsilon}\right)^{ O(t_{\Delta})}\right)$ in \Cref{lemma:variance_calc}, and using the fact that $\|f_x\|^2_2 = \Theta(\frac{k}{n})$ for all $x \in  V \setminus B_{\delta}$ (see \Cref{remark:norm}),
 we obtain \footnote{\Cref{eqn:sketchx} is bounding the difference between $\langle p^{t_1}_x,p^{t_2}_x\rangle , $ and its empirical version. However, since we are ultimately interested in approximating the polynomial $p(M)$, it makes sense to write the expression in terms of $M^{t_1}\1_x$ , $M^{t_1}\1_x$, rather than $ p^{t_1}_x$,$p^{t_2}_x$. }
 \begin{equation}\label{eqn:sketchx}
\Pr_{\text{random walks}}\left[\left| \left\langle \sigma_x\widehat{p}_x^{t_1}, \sigma'_x \widehat{p}_x^{t_2}\right \rangle  - \left\langle \sigma_x M^{t_1} \1_x, \sigma'_x M^{t_2} \1_x \right \rangle\right| \leq \frac{5 \cdot 10^{-4}}{|c_{t_1}| |c_{t_2}|(t_{\Delta}+1)^2} \cdot \|f_x\|^2_2 \right] \geq 1- 10^{-6}\frac{1}{(t_{\Delta}+1)^2}.
\end{equation}
By taking a union bound over the $(t_{\Delta} + 1)^2$ possible pairs of $t_1, t_2 \in [t_{\min}, t_{\min} + t_{\Delta}]$, we get that 

 \begin{equation}\label{eq:sketchx}
     \Pr_{\text{random walks}} \left[ \left| \langle \sketch(x) , \sketch(x) \rangle -   \left\langle \sigma_x \sum_t c_t \cdot  M^t \1_x,\sigma'_x \sum_{t}    c_t  \cdot M^t \1_x\right \rangle\right| > 0.0005 \|f_x\|^2_2 \right] \leq 10^{-6}.
 \end{equation}
 Finally, we have  
 \begin{restatable}{claim}{claimfx}\label{claim:fx_vs_chebyshev}
 For all $x \in V \setminus B_{\delta}$, 
 \begin{equation*} 
     \left | \left \|\sum_t c_t M^t\1_x  \right\|^2_2 - \|f_x\|^2_2 \right| \leq 0.0005 \|f_x\|^2_2.
\end{equation*}
\end{restatable}
\begin{proof}
The claim essentially follows by \Cref{thm:standardbasis}. 
     Writing $p(M) \1_x$ and $f_x$ with respect to the eigenvectors $u_1, \dots u_n$ of the Laplacian $\mathcal{L}$, we get
     \begin{align*}
     \left | \left \|\sum_t c_t M^t\1_x  \right\|^2_2 - \|f_x\|^2_2 \right| & =  \left | \left \|p(M)\1_x  \right\|^2_2 - \|f_x\|^2_2 \right| \\
     & = \left| \sum_{i =1}^n \langle u_i, \1_x\rangle^2 p((1-\lambda_i/2))^2 - \sum_{i=1}^k  \langle u_i, \1_x\rangle^2\right| \\
     & \leq \left| \sum_{i =1}^k \langle u_i, \1_x\rangle^2 p((1-\lambda_i/2))^2 - \sum_{i=1}^k  \langle u_i, \1_x\rangle^2\right| + \left| \sum_{i =k+1}^n \langle u_i, \1_x\rangle^2 p((1-\lambda_i/2))^2 \right| \\
     & \leq  \left| \sum_{i =1}^k \langle u_i, \1_x\rangle^2 (1-\epsilon/\varphi^2)^2 - \sum_{i=1}^k  \langle u_i, \1_x\rangle^2\right| + n^{-8} \left|  \sum_{i =k+1}^n \langle u_i, \1_x\rangle^2  \right| \\
     &\leq  \frac{2\epsilon}{\varphi^2} \left| \sum_{i =1}^k \langle u_i, \1_x\rangle^2 \right| +  n^{-8} \|\1_x\|^2_2  \\
     & =  \frac{2\epsilon}{\varphi^2} \|f_x\|^2_2 + n^{-8}  \\
     & \leq  0.0005 \|f_x\|^2_2.
     \end{align*}
      Here, the fourth line follows from \Cref{remark:Meigengap} and \Cref{thm:standardbasis}, and the last line follows from the assumption that $\epsilon/\varphi^2$ is less than a constant (see \Cref{rem:param_assumptions}) and the fact that $\|f_x\|_2^2 = \Theta(\frac{n}{k})$ for $x \in V \setminus B_{\delta}$ (see \Cref{remark:norm}). 
\end{proof}
 
Combining \Cref{eq:sketchx} and \Cref{claim:fx_vs_chebyshev}, with probability at least $1-10^{-6}$, it holds that  \newline $| \langle \sketch(x), \sketch(x) \rangle | \in \left[0.99\|f_x\|^2_2, 1.01 |f_x\|^2_2\right]$, as required. 
\end{proof}

Finally, we prove two auxiliary lemmas, which will be used in the proofs of \Cref{lemma:random_good_cluster} and \Cref{lemma:random_good_set}. 
\begin{lemma}\label{lemma:sum_y_in_V}
    \[\sum_{y \in V}f_yf_y^\top  = I_{k}.\]
In particular, for all $\alpha \in \mathbb{R}^k$, it holds that  $\sum_{y \in V}\langle \alpha, f_y\rangle ^2 = \| \alpha\|^2$. 
\end{lemma}
\begin{proof}
We have
\begin{align*}
    \sum_{y \in V}f_yf_y^\top   
    & = \sum_{y \in V}U_{[k]}^\top \1_y\1_y^\top U_{[k]} && \text{by definition of $f_y$} \\
    & =  U_{[k]}^\top \left(\sum_{y \in V}\1_y\1_y^\top  \right)U_{[k]} \\
    & =  U_{[k]}^\top  I_nU_{[k]}\\
    & =  U_{[k]}^\top  U_{[k]}\\
    & = I_{k},  
\end{align*}
where the last equality follows by orthonormality of eigenvectors, since $U_{[k]}$ is the matrix whose columns are the bottom $k$ eigenvectors of $\mathcal L$ . 

The second part of the lemma now follows, since
\[ \sum_{y \in V}\langle \alpha, f_y\rangle ^2 = \alpha^\top  \left( \sum_{y \in V} f_y f_y ^\top \right)\alpha  = \|\alpha\|^2_2. \]
\end{proof}
\begin{lemma}\label{lemma:mu_ifysquared}
For every $i \in [k]$
    \[\mu_i\mu_i^\top \bullet\sum_{y \in V\setminus C_i}f_yf_y^\top  \leq 4\cdot \frac{\sqrt{\epsilon}}{\varphi} \cdot\|\mu_i\|^2_2\]
\end{lemma}

\begin{proof}
We have 
\begin{equation}\label{eq:fxfysquared}
\begin{aligned}
\mu_i\mu_i^\top \bullet\sum_{y \in V\setminus C_i}f_yf_y^\top  = \sum_{y \in V \setminus C_i} \langle \mu_i, f_y \rangle ^2 & = \sum_{y \in V} \langle \mu_i, f_y \rangle ^2 - \sum_{y \in C_i}\langle \mu_i, f_y \rangle ^2 \\
& = \|\mu_i\|^2 - \sum_{y \in C_i}\langle \mu_i, f_y \rangle ^2 && \text{by \Cref{lemma:sum_y_in_V}}.
\end{aligned}
\end{equation}
We will now lower bound the $\sum_{y \in C_i}\langle \mu_i, f_y \rangle ^2$ term. For every $y \in C_i$, we have
\begin{align}\label{eq:sum_clustermeansquared}
\begin{split}
    \langle \mu_i, f_y \rangle ^2 & = \left( \langle \mu_i, f_y - \mu_i \rangle   +\langle \mu_i,  \mu_i \rangle \right)^2 \\
    & = \langle \mu_i, f_y - \mu_i \rangle^2 + 2 \langle \mu_i, f_y - \mu_i \rangle \|\mu_i\|^2_2 + \|\mu_i\|^4_2
\end{split}
\end{align}
Furthermore, note that
\begin{equation}\label{eq:crossterm}
\begin{aligned}
    \sum_{y \in C_i }\left \langle \mu_i, f_y - \mu_i \right \rangle & =|C_i| \cdot \frac1{|C_i|}\sum_{y\in C_i} \langle \mu_i, f_y-\mu_i\rangle \\
    & =|C_i|  \left \langle \mu_i, \frac1{|C_i|}\sum_{y\in C_i}( f_y-\mu_i)\right \rangle  \\
    & =0,  && \text{by definition of $\mu_i$ \eqref{eq:muis-def}}. 
\end{aligned}
\end{equation}
Summing \Cref{eq:sum_clustermeansquared} over $y \in C_i$ and using \Cref{eq:crossterm}, we obtain 
\begin{equation}\label{eq:fxfyowncluster}
\begin{aligned}
     \sum_{y \in C_i}\langle \mu_i, f_y \rangle ^2 & =  \sum_{y \in C_i}\left(\langle\mu_i, f_y - \mu_i \rangle^2 + 2 \langle \mu_i, f_y - \mu_i \rangle \|\mu_i\|^2_2  + \|\mu_i\|^4_2\right) \\
     & \geq 2 \|\mu_i\|^2_2 \sum_{y \in C_i}\langle \mu_i, f_y - \mu_i \rangle   + |C_i| \cdot \|\mu_i\|^4_2  && \text{by dropping the $\langle \mu_i, f_y - \mu_i\rangle^2$ terms} \\
     & = |C_i| \cdot \|\mu_i\|^4_2&&  \text{by \Cref{eq:crossterm}} \\
     & \geq \|\mu_i\|^2_2 \cdot \left(1-\frac{4 \sqrt{\epsilon}}{\varphi} \right)&&  \text{by  \Cref{lemma:clustermeans}} \\
\end{aligned}
\end{equation}
Combining \Cref{eq:fxfysquared} and \Cref{eq:fxfyowncluster} gives 
\[\mu_i\mu_i^\top \bullet\sum_{y \in V\setminus C_i}f_yf_y^\top  \leq \|\mu_i\|^2_2 -  \left(1-4\cdot\frac{\sqrt{\epsilon}}{\varphi} \right) \|\mu_i\|^2_2  = 4\cdot \frac{\sqrt{\epsilon}}{\varphi} \|\mu_i\|^2_2.  \]
\end{proof}

\section{Properties of well-spread sets and typical vertices}\label{sec:well-spread_sets}

In this section, we prove \Cref{claim:good_set_subset}, \Cref{lemma:random_good_set}, \Cref{lemma:good_pts_wrt_S}, \Cref{lemma:bounding_var},  \Cref{claim:Equal} and Lemma  \ref{lemma:random_good_cluster}. 
\subsetremark*
\begin{proof}
Suppose that $S$ is well-spread (as per \Cref{def:good_S}), and suppose that $S' \subseteq S$. We now verify that each of the conditions in \Cref{def:good_S} hold for $S'$. 
\paragraph{Condition \ref{con:good_S1}.} Follows immediately from the assumption that $S' \subseteq S$. 

\paragraph{Condition \ref{con:good_S2}.} The low self-collision probability for the sets $S$ and $S'$ can be equivalently stated as 
    \[\left\|\sum_{y \in S}M^t\1_y\right\|^2_2  = O^*\left(\frac{k^2}{n}\right)\]
    and 
    \[\left\|\sum_{y \in S'}M^t\1_y\right\|^2_2  = O^*\left(\frac{k^2}{n}\right)\]
    respectively. Since all of the entries of all vectors $M^t\1_y$ are non-negative, we get
    \[\left\|\sum_{y \in S'}M^t\1_y\right\|^2_2 \leq  \left\|\sum_{y \in S}M^t\1_y\right\|^2_2 = O^*\left(\frac{k^2}{n}\right).\]
\paragraph{Condition \ref{con:good_S3}.} The first condition follow immediately from the assumption that $S' \subseteq S$.  To show the remaining two conditions, observe that for every $i$ and every $y$, it holds that
    \[\mu_i\mu_i^\top \bullet f_yf_y^\top  = \langle \mu_i, f_y\rangle^2 \geq 0,\]
    so
    $$ \mu_i \mu_i^\top  \bullet \sum_{y \in S'\setminus C_{i}} f_y f_y^\top  \leq \mu_i \mu_i^\top  \bullet \sum_{y \in S\setminus C_{i}} f_y f_y^\top \leq 10^{-10}\cdot \|\mu_i\|^4_2 $$  and $$\mu_i \mu_i^\top  \bullet \sum_{y \in S'}f_y f_y^\top   \leq \mu_i \mu_i^\top  \bullet \sum_{y \in S}f_y f_y^\top  \leq O\left(\left(\varphi^2/\epsilon\right)^{1/3}\right)\|\mu_i\|_2^4.$$

\end{proof}
\randomgoodset*

\begin{proof}
We show that each of the conditions hold with high constant probability.
\paragraph{Condition \ref{con:B_delta}: Small overlap with $B_{\delta}$.}
We show that Condition \ref{con:B_delta} holds with probability at least $0.99998$. By Lemma~\ref{lemma:close_to_clutermean}, we have $|B_{\delta}| = O\left(\frac{1}{\delta} \cdot \epsilon/\varphi^2 \cdot n\right)$. So we have  \[ \E[|S\cap B_{\delta}|] \leq O\left(\frac{1}{\delta} \cdot \frac{\epsilon}{\varphi^2} k \log(\varphi^2/\epsilon)\right).\] By Markov's inequality, we obtain that for some sufficiently large constant $C$, it holds that

\[ \Pr\left[|S \cap B_{\delta}| \geq  C \cdot \frac{ \epsilon/\varphi^2}{\delta} \cdot k \log(\varphi^2/\epsilon)\right] \leq 0.00001.\] 

\paragraph{Condition \ref{con:rw}: Low self-collision probability.}
We show that Condition \ref{con:rw} holds with probability at leat $0.99998$. 
Recall from \Cref{eq:p-s} that $p_S^t = \frac{1}{|S|}\sum_{y \in S}M^t\1_y$, and note that Condition \ref{con:rw} is equivalent to

\[ \left\| \sum_{y\in S}  M^t \1_y\right\|^2 \leq O^*\left( \frac{k^{2}}{n}\right) \quad \text{ for all } t \in \left[t_{\min}, t_{\min} + t_{\Delta}\right].\]

More precisely, we will show that with probability at least $0.99994$

\begin{equation}\label{eq:||p^t_S||^2}
   \left\| \sum_{y\in S}  M^l \1_y\right\|^2 \leq O\left(\frac{k^2}{n} \cdot \log^2(\varphi^2/\epsilon)\cdot(t_{\Delta}+1)\right)\quad \text{ for all } t \in \left[t_{\min}, t_{\min} + t_{\Delta}\right],
\end{equation}
where $t_{\Delta}= \tdelta = O^*(1)$. Write $S = \{y_1, \ldots, y_{L}\}$ as a random (multi) set of fixed size $L = O(k\cdot\log(\varphi^2/\epsilon))$, we where    $y_1, \ldots, y_{L}$ are independent random variables uniformly distributed over $V$.  By independence of the $y_l$s, we have 
\begin{equation}\label{eq:crossterms}
\begin{aligned}
    \E_{S} \left[\sum_{l, {l'} \in [L]: l \neq l'} \left \langle M^t \1_{y_l}, M^t \1_{y_{l'}}\right\rangle \right] &= 
    \E_{S}\left[\sum_{l, {l'} \in [L] : l \neq l'}\sum_{v \in V} M^t \1_{y_l}(v)   M^t \1_{y_{l'}}(v) \right]  \\
    & = L\cdot(L - 1)\cdot\E_{y, y'}\left[\sum_{v\in V}M^t\1_y(v)M^t\1_{y'}(v)\right] \\
    & = L\cdot(L-1) \cdot\frac{1}{n^2}\sum_{v \in V}\left(\sum_{y \in V}M^{t}\1_y(v)\right)^2\\
    & = \sum_{v \in V}L\cdot(L-1)\cdot\frac{1}{n^2} &&  \text{since $\sum_{y \in V} M^t \1_y = \1$}\\
    & \leq \frac{L^2}{n}.
\end{aligned}
\end{equation}
Next, for every $y \in V$ and $t \geq t_{\min}$, we have  $\|M^t\1_y\|^2_2 \leq O(k/n)$ by \Cref{lem:bound-norms-1/2}.
Combining this with \Cref{eq:crossterms}, we get
\[\E_S\left[\left\| \sum_{y \in S}  M^t \1_{y}\right\|_2^2 \right] = \E_S\left[\sum_{l \in [L]}\|M^t\1_{y_l}\|^2_2\right] +  \E_{S} \left[\sum_{l,l' \in [L] : l \neq l'\in S} \left \langle M^t \1_{y_l}, M^t \1_{y_{l'}}\right\rangle \right] \leq O\left(\frac{k^2\log^2(\varphi^2/\epsilon)}{n}\right). \]

By Markov's inequality, with probability at least $\frac{0.99998}{t_{\Delta}+1}$, it holds that 
 \[  \sum_{v \in V}\left(\sum_{y\in S}  M^l \1_y(v)\right)^2 \leq   O\left(\frac{k^2}{n} \cdot \log^2(\varphi^2/\epsilon)\cdot(t_{\Delta}+1)\right).\]
 The result now follows by taking a union bound over all $l \in [t_{\min},t_{\min}+t_{\Delta}].$ 
 \paragraph{Condition \ref{con:good_S3}: Nearly isotropic in the embedding space.} We show that Condition \ref{con:good_S3} holds with probability at least 0.99994.
 
 Observe that
\[\E_S\left[\sum_{y \in S}\|f_y\|^2_2\right] = |S|\cdot \E_{y \sim V}\left[\|f_y\|^2_2\right] = \frac{|S|}{n}\sum_{v \in V}\|f_v\|^2_2 = \frac{|S| \cdot k}{n} = O\left(\frac{k^2\log(\varphi^2/\epsilon)}{n}\right),\]
where the third inequality follows from $\sum_{v \in V} \|f_v\|^2_2 = \Tr\left( \sum_{v \in V} f_v f_v^\top  \right)=\Tr(I_k)  = k$, by \Cref{lemma:sum_y_in_V}. So by Markov's inequality, with probability at least 0.99998, it holds that
\[\sum_{y \in S}\|f_y\|^2_2 \leq O\left(\frac{k^2\log(\varphi^2/\epsilon)}{n}\right),\]
as required. Next, observe that
\begin{equation}\label{eq:exp_isotropic}
   \begin{aligned}
       \E_S\left[\sum_{i \in [k]}\left( \mu_i \mu_i^\top  \bullet \sum_{y \in S\setminus C_{i}} f_y f_y^\top \right)\right] & = 
       \sum_{i \in [k]}\left( \mu_i \mu_i^\top  \bullet \E_S\left[ \sum_{y \in S} f_y f_y^\top  \mathbbm{1}\{y \notin C_i\}\right)\right] \\
       & =  \frac{|S|}{n}\sum_{i\in [k]} \left( \mu_i \mu_i^\top  \bullet \sum_{y \in V}f_y f_y^\top  \mathbbm{1}\{y \notin C_i\} \right)\\
       & = \frac{|S|}{n} \left( \sum_{i \in [k]}\mu_i \mu_i^\top  \bullet \left(\sum_{y \in V\setminus C_i}f_y f_y^\top  \right) \right).
   \end{aligned}
\end{equation}
To upper bound the final expression in \eqref{eq:exp_isotropic}, note that by \Cref{lemma:mu_ifysquared}, we have
\begin{equation}\label{eq:isotropic2}
\frac{|S|}{n}\sum_{i \in [k]}\left( \mu_i \mu_i^\top  \bullet \sum_{y \in V\setminus C_i}f_yf_y^\top \right) = \frac{|S|}{n}\sum_{i \in [k]}\sum_{y \in V\setminus C_i}\langle \mu_i, f_y\rangle^2 \leq 4\cdot\frac{|S|}{n}\cdot \frac{\sqrt{\epsilon}}{\varphi}\sum_{i \in [k]}\|\mu_i\|^2_2 \leq O\left(\left(\frac{\epsilon}{\varphi^2}\right)^{1/2}\cdot\frac{|S|\cdot k^2}{n^2}\right),
\end{equation}
where the last inequality uses the fact that $\|\mu_i\|^2_2 = \Theta(k/n)$ for every $i$, by \Cref{rem:||mu_i||}. Combining \eqref{eq:exp_isotropic} and \eqref{eq:isotropic2}, and applying Markov's inequality, with probability at least 0.99998,  
\begin{equation}\label{eq:markov}
    \sum_{i \in [k]}\mu_i \mu_i^\top  \bullet \sum_{y \in S\setminus C_{i}} f_y f_y^\top  \leq \alpha\left(\frac{\epsilon}{\varphi^2}\right)^{1/2}\cdot\frac{|S|\cdot k^2}{n^2}
\end{equation}
for a sufficiently big constant $\alpha$. For all $i$ except at most $\frac{\alpha}{c}\cdot k\cdot (\epsilon/\varphi^2)^{1/3}\log(\varphi^2/\epsilon) = O\left(k\cdot (\epsilon/\varphi^2)^{1/3}\log(\varphi^2/\epsilon\right)$ for a sufficiently small constant $c$
\begin{equation}\label{eq:averaging}
    \mu_i \mu_i^\top  \bullet \sum_{y \in S\setminus C_{i}} f_y f_y^\top  \leq c\cdot\frac{|S|\cdot k^2 \cdot (\epsilon/\varphi^2)^{1/2}}{ k\cdot (\epsilon/\varphi^2)^{1/3}\log(\varphi^2/\epsilon)\cdot n^2} \leq c\cdot\left(\frac{\epsilon}{\varphi^2}\right)^{1/6}\cdot\|\mu_i\|^4_2 < 10^{-10}\|\mu_i\|^4_2.
\end{equation}
If the number of violating indices $i$ was higher, it would contradict \Cref{eq:markov}. Finally, observe that  
   \begin{align*}
       \E_S\left[\sum_{i \in [k]} \left( \mu_i \mu_i^\top  \bullet \sum_{y \in S} f_y f_y^\top \right) \right] & = 
       \sum_{i \in [k]}\mu_i \mu_i^\top  \bullet \E\left[ \sum_{y \in S} f_y f_y^\top  \right] \\
       & =  \frac{|S|}{n}\sum_{i \in [k]} \mu_i \mu_i^\top  \bullet \sum_{y \in V}f_y f_y^\top   \\
       & = \frac{|S|}{n} \sum_{i \in [k]} \|\mu_i\|^2_2 && \text{since $ \sum_{y \in V}f_y f_y^\top  = I_k$} \\
       & = O\left(\frac{|S|\cdot k^2}{n^2}\right) && \text{by \Cref{rem:||mu_i||}.}
   \end{align*}
Therefore, by Markov's inequality, with probability at least 0.99998,  it holds that
\[\sum_{i \in [k]}\mu_i \mu_i^\top  \bullet \sum_{y \in S} f_y f_y^\top  \leq  \kappa\frac{|S|\cdot k^2}{n^2}.\]
for a sufficiently large constant $\kappa$. 
For all except for at most $\frac{\kappa}{c'}\cdot k\cdot (\epsilon/\varphi^2)^{1/3}\log(\varphi^2/\epsilon)= O(k\cdot (\epsilon/\varphi^2)^{1/3}\log(\varphi^2/\epsilon))$ of values $i \in [k]$  for a sufficiently small $c'$
\[\mu_i \mu_i^\top  \bullet \sum_{y \in S} f_y f_y^\top  \leq c'\cdot \frac{ |S|\cdot k^2}{k\cdot (\epsilon/\varphi^2)^{1/3}\log(\varphi^2/\epsilon)\cdot n^2} \leq c'\cdot\frac{k^2\cdot(\varphi^2/\epsilon)^{1/3}}{n^2}\leq (\varphi^2/\epsilon)^{1/3}\|\mu_i\|_2^4,\] by the same argument which we used to obtain \Cref{eq:averaging}.

Finally, since each of the two conditions holds with probability at least 0.99998  for all except for at most $O(k\cdot (\epsilon/\varphi^2)^{1/3}\log(\varphi^2/\epsilon))$ of values $i \in [k]$ then both hold simultaneously with probability at least 0.99996 and for all except for at most $O(k\cdot (\epsilon/\varphi^2)^{1/3}\log(\varphi^2/\epsilon))$ of values $i \in [k]$.

By a union bound, all three conditions \ref{con:good_S1}, \ref{con:good_S2} and \ref{con:good_S3} hold simultaneously with probability at least $0.9999$. 
\end{proof}

Next, we prove \Cref{lemma:good_pts_wrt_S}, restated below for the convenience of the reader.  
\goodptswrtS*
\begin{proof}
    Let $S$ be a well-spread set (as per \Cref{def:good_S}). Recall from \Cref{def:good_pt_wrt_S} and \Cref{rem:strongly_typical} that a strongly typical vertex with respect to $S$ is defined by 5 conditions. For each of the conditions, we define the set of all points in $V$ which violate this condition, and we bound each of their sizes separately. Let
\begin{align*}
   F_1&  \coloneqq \{ x \in V: S \cap B_{\delta} \cap C_{i(x)} \neq  \emptyset\}  \\
    F_2  & \coloneqq \left\{ x \in V : |S \cap C_{i(x)}| \geq 10^{-5} \delta^{-1/2}\right\} \\
   F_3& \coloneqq \left\{ x \in V: p^{t}_S(x) \geq 10^{5} \cdot \left(\frac{\varphi^2 }{\epsilon}\right)^{1/3}\cdot \frac{t_{\Delta}+1}{|S|}\cdot \frac{k}{n}  \text{ for some $t \in \left[2t_{\min}, 2t_{\min} + 2t_{\Delta}\right]$}  \right\}\\
   F_4 & \coloneqq \left\{ 
   \begin{aligned}  x \in V : &  \left \langle p_{x}^{t} , \left(p_S^{l}\right)^2 \right \rangle \geq 10^{5} \cdot \left(\frac{\varphi^2}{\epsilon}\right)^{1/3}\cdot  \frac{(t_{\Delta}+1)^3}{|S|^2}\cdot \left(\frac{k}{n}\right)^{2}\log(\varphi^2/\epsilon) \\
    \text{ or }&  \left \langle \left( p_x^t\right)^2, p_S^l\right \rangle \geq 10^{5} \cdot \left(\frac{\varphi^2}{\epsilon}\right)^{1/3}\cdot  \frac{(t_{\Delta}+1)^3}{|S|}\cdot \left(\frac{k}{n}\right)^{2}\log(\varphi^2/\epsilon)  \text{ for some  $t, l \in \left[t_{\min}, t_{\min} + t_{\Delta}\right]$} 
   \end{aligned}
   \right\}\\
    F_5 & \coloneqq \left\{x \in V:  (f_x -\mu_{i(x)})(f_x -\mu_{i(x)})^\top  \bullet \sum_{y \in S }f_yf_y^\top  \geq 10^{-10}\|f_x\|^4_2\right\}\\
\end{align*}
    be the sets of vertices failing Conditions \ref{con:bad_cluster}, \ref{con:ScapC_i}, \ref{con:rw1}, \ref{con:rw2} and \ref{con:rw3} in \Cref{def:good_pt_wrt_S} respectively. 
    Then,  
    \begin{equation*}
        \left| \left\{ x \in V : \text{x is \emph{not} strongly typical with respect to  S}\right\}\right| \leq |F_1| + |F_2| + |F_3| + |F_4| + |F_5|,
    \end{equation*}
    so it suffices to show that $ |F_1|, |F_2|, |F_3|, |F_4|, |F_5|\leq  O \left(n \cdot \left(\frac{\epsilon}{\varphi^2}\right)^{1/3}\cdot \log(\varphi^2/\epsilon)\right)$. Let us now bound each of the five sets separately. 
\paragraph{Bounding $|F_1|$:} We have
\begin{align*}
    |F_1| & = \sum_{i = 1}^k \1\{S \cap B_{\delta}\cap C_i \neq \emptyset \}|C_i| \\
    & \leq \max_i |C_i| \cdot |S\cap B_{\delta}| \\ 
    & \leq \eta \cdot \frac{n}{k} \cdot O\left( (\epsilon/\varphi^2)^{1/3}\cdot  k\cdot \log(\varphi^2/\epsilon)\right) \qquad \qquad \text{by Condition {\bf \ref{con:B_delta}} in \Cref{def:good_S}}\\
    & = O \left(n \cdot (\epsilon/\varphi^2)^{1/3}\cdot \log(\varphi^2/\epsilon)\right).
\end{align*}
\paragraph{Bounding $|F_2|$: }
We have 
\begin{align*}
    |F_2| & = \sum_{i = 1}^{k}\1\{|S \cap C_i| \geq 10^{-5}\cdot \delta^{-1/2} \}|C_i|  \\
    & \leq \max_i |C_i| \cdot |\{ i \in [k] : 
    |S \cap C_i| \geq 10^{-5}\delta^{-1/2}\}| \\
    & \leq \eta \cdot \frac{n}{k} \frac{|S|}{10^{-5}\delta^{-1/2}} \\
     & = O\left(\delta^{1/2} \log(\varphi^2/\epsilon)\cdot n \right) && \text{since $|S| \leq O(k \log(\varphi^2/\epsilon))$ by \Cref{def:good_S}}\\
     & = O\left(n \cdot \left(\epsilon/\varphi^2\right)^{1/3}  \log(\varphi^2/\epsilon) \right)&& \text{ by choice of $\delta = O\left((\epsilon/\varphi^2)^{2/3}\right)$}.
\end{align*}

\paragraph{Bounding $|F_3|$:} For $t \in [2t_{\min}, 2t_{\min} + 2t_{\Delta}]$, let $F_{3, t} \coloneqq \left\{ x \in V: p^{t}_S(x) \geq 10^{5} \cdot \left(\frac{\varphi^2 }{\epsilon}\right)^{1/3}\cdot \frac{t_{\Delta}+1}{|S|}\cdot \frac{k}{n} \right\}$.

\noindent Recall from \eqref{eq:p-s} that $p^t_S(x) = \frac{1}{|S|}\sum_{y \in S}M^t\1_y(x) = \frac{1}{|S|}\langle \1_x, \sum_{y \in S} M^t\1_y\rangle $. Therefore, we have 
\begin{align*}
 \sum_{x \in V}  p^t_S(x) &=  \frac{1}{|S|}\sum_{x \in V} \left \langle \1_x , \sum_{y \in S} M^t \1_{y}\right\rangle \\
   &= \frac{1}{|S|} \sum_{y \in S} \left \langle \1  , M^{t_2}\1_y\right\rangle \\ 
    &  =1 && \text{since } \sum_{v \in V}M^{t_2}\1_y(v) = 1,  \text{ for all $y \in V$}
\end{align*}
We now have $0 \leq p^t_S(x)$ for all $x$, and $\sum_{x \in V}p^t_S(x) = 1$, 
which gives
\[ |F_{3,t}| \leq \frac{1}{10^5(\varphi^2/\epsilon)^{1/3}(t_{\Delta}+1)/|S| \cdot k/n} = O\left(\frac{1}{t_{\Delta}+1} \left(\frac{\epsilon}{\varphi^2}\right)^{1/3}\log(\varphi^2/\epsilon) \cdot n\right).\]
Recall that, by the definition of $F_3$, we have $F_3 = \cup_{t\in [2t_{\min}, 2t_{\min} + 2t_{\Delta}]} F_{3, t}$. Taking a union over $2(t_{\Delta}+1)$ values of $t$ in $[2t_{\min}, 2t_{\min} + 2t_{\Delta}]$ we obtain 

\[ |F_3| \leq \sum_{t \in [2t_{\min}, 2t_{\min} + 2t_{\Delta}]} |F_{3,t}| \leq O\left( \left(\frac{\epsilon}{\varphi^2}\right)^{1/3} \log(\varphi^2/\epsilon) \cdot n \right).\]

\paragraph{Bounding $|F_4|$:} Recall from \eqref{eq:p-s}, that $p^l_S(v) = \frac{1}{|S|}\sum_{y \in S}M^l\1_y(v)$, so $(p^l_S(v))^2 = \frac{1}{|S|^2}\sum_{y, y' \in S}M^l\1_y(v)M^l\1_{y'}(v)$. For every $t, l \in [ t_{\min}, t_{\min} + t_{\Delta}]$, define
\begin{align*}
F_{4, t, l} =&\left\{ x \in V :  \left \langle p_{x}^{t} , \left(p_S^{l}\right)^2 \right \rangle +  \left \langle \left( p_x^t\right)^2, p_S^l\right \rangle \geq 10^{5} \cdot \left(\frac{\varphi^2}{\epsilon}\right)^{1/3}\cdot  \frac{(t_{\Delta}+1)^3}{|S|}\cdot \left(\frac{k}{n}\right)^{2}\log(\varphi^2/\epsilon)\right\}\\
=& \left\{ \begin{aligned}
 x\in V :&   \sum_{v \in V}(M^t\1_x(v))^2\sum_{y \in S}M^l\1_y(v) + \sum_v M^{t} \1_x(v)  \sum_{y, y' \in S}     M^{l} \1_{y}(v)   M^{l} \1_{y'}(v)  \\
&  \geq 10^{5} \cdot \left(\frac{\varphi^2}{\epsilon}\right)^{1/3} (t_{\Delta}+1)^3 \left(\frac{k}{n}\right)^{2}\log(\varphi^2/\epsilon)
\end{aligned}
 \right\}. 
\end{align*}
Then $F_4 \subseteq \bigcup_{t,l \in [t_{\min}, t_{\min} + t_{\Delta}]} F_{4,t,l}$.  Condition {\bf \ref{con:rw}} in \Cref{def:good_S} can be equivalently written as

\[\sum_{v \in V}\left(\sum_{y\in S}  M^l \1_y(v)\right)^2 \leq O\left( \frac{k^{2}}{n}\cdot \log^2(\varphi^2/\epsilon) \cdot (t_{\Delta}+1)\right) \quad \text{ for all } l \in \left[t_{\min}, t_{\min} + t_{\Delta}\right].\]
\noindent
So for all $t,  l \in \left[t_{\min}, t_{\min} + t_{\Delta}\right]$, we have 
\begin{align*}
\sum_{v \in V}\sum_{x \in V}M^{t}\1_x(v)\sum_{y, y'\in S}  M^l \1_y(v) M^l \1_{y'}(v) & =\sum_{v \in V}\sum_{y, y'\in S}  M^l \1_y(v) M^l \1_{y'}(v) && \text{since } \sum_{x\in V}M^{t}\1_x = \1\\
    &  \leq  O\left( \frac{k^{2}}{n}\cdot \log^2(\varphi^2/\epsilon) \cdot (t_{\Delta}+1)\right) && \text{by \ref{con:rw} in \Cref{def:good_S}}.
\end{align*}
Next, note that $M^t\1_x(v) = \langle \1_v, M^t\1_x\rangle = \langle \1_x, M^t\1_v\rangle = M^t\1_v(x)$, since $M$ is a symmetric matrix. So
\begin{align*}
    \sum_{v\in V}\sum_{x \in V}(M^t\1_x(v))^2\sum_{y \in S}M^l\1_y(v)&  = \sum_{v \in V} \sum_{x \in V}(M^t\1_v(x))^2\sum_{y\in S}M^l\1_y(v) \\
   &  = \sum_{v \in V}\|M^t\1_v\|^2_2\sum_{y \in S}M^l\1_y(v)\\
    & \leq O(k/n)\cdot \sum_{v \in V}\sum_{y \in S}M^l\1_y(v) && \text{since }\|M^t\1_v\|^2_2 =  O(k/n) \text{ by \Cref{lem:bound-norms-1/2}}\\
    & =O\left(\frac{k}{n}\cdot|S|\right)&& \text{since }\|M^l\1_y\|_1 = 1 \text{ for all $y \in V$}\\
    &  =O\left(\frac{k^2\log(\varphi^2/\epsilon)}{n}\right).
\end{align*}
Therefore,
\begin{equation}\label{eq:averaging_2}
    \sum_{x \in V}\left( \sum_{v \in V}(M^t\1_x(v))^2\sum_{y \in S}M^l\1_y(v) + \sum_v M^{t} \1_x(v)  \sum_{y, y' \in S}     M^{l} \1_{y}(v)   M^{l} \1_{y'}(v)\right) \leq \alpha\cdot \frac{k^{2}}{n}\cdot \log^2(\varphi^2/\epsilon) \cdot (t_{\Delta}+1)
\end{equation}
for a sufficiently large constant $\alpha$. By the definition of $F_{4, t, l}$, for any $x \in F_{4, t, l}$ the corresponding summand in the LHS must be at least  $10^{5} \cdot \left(\frac{\varphi^2}{\epsilon}\right)^{1/3} (t_{\Delta}+1)^3 \left(\frac{k}{n}\right)^{2}\log(\varphi^2/\epsilon)$. At the same time, there can be no more than $\frac{\alpha}{10^5}\cdot \frac{1}{(t_{\Delta}+1)^2} \left(\frac{\epsilon}{\varphi^2}\right)^{1/3}\log(\varphi^2/\epsilon) \cdot n$ of such summands as otherwise \Cref{eq:averaging_2} does not hold. Therefore,
\begin{equation}\label{eq:averaging_3}
    |F_{4,t, l}| \leq O\left(\frac{1}{(t_{\Delta}+1)^2} \left(\frac{\epsilon}{\varphi^2}\right)^{1/3}\log(\varphi^2/\epsilon) \cdot n\right).
\end{equation} 

\noindent 
Summing over all $(t_{\Delta}+1)^2$ pairs of $t, l \in [t_{\min}, t_{\min} + t_{\Delta}]$ we obtain 

\[ |F_4| \leq \sum_{t, l \in [t_{\min}, t_{\min} + t_{\Delta}]} |F_{4,t, l}| \leq O\left( \left(\frac{\epsilon}{\varphi^2}\right)^{1/3} \log(\varphi^2/\epsilon) \cdot n \right).\]

\paragraph{Bounding $|F_5|$:} 
Recall from \Cref{lemma:variancebound},that for every vector $\alpha \in \mathbb{R}^k$, it holds that 
\begin{equation}\label{eq:preceq 4epsilon/varphi^2}
    \alpha^\top \left( \sum_{x \in V} (f_x - \mu_{i(x)})(f_x - \mu_{i(x)})^\top \right) \alpha \leq \frac{4 \epsilon}{\varphi^2}\|\alpha\|^2_2.
\end{equation}
Therefore, we have 
\begin{equation}\label{eq:averaging_4}
    \begin{aligned}
    \sum_{x \in V} (f_x - \mu_{i(x)})(f_x - \mu_{i(x)})^\top  \bullet \sum_{y \in S}f_y f_y^\top   &= \Tr\left(\sum_{x \in V} (f_x - \mu_{i(x)})(f_x - \mu_{i(x)})^\top  \sum_{y \in S}f_y f_y^\top  \right) \\
    & = \sum_{y \in S} \Tr\left(\sum_{x \in V} (f_x - \mu_{i(x)})(f_x - \mu_{i(x)})^\top  f_y f_y^\top  \right) \\
    & = \sum_{y \in S} f_y^\top  \left(\sum_{x \in V} (f_x - \mu_{i(x)})(f_x - \mu_{i(x)})^\top \right) f_y \\
    & \leq \frac{4 \epsilon}{\varphi^2} \sum_{y \in S}\|f_y\|^2_2 && \text{by \eqref{eq:preceq 4epsilon/varphi^2} }\\
    & \leq \alpha \frac{\epsilon}{\varphi^2}\cdot\frac{k^2}{n} \log(\varphi^2/\epsilon) && \text{by \Cref{remark:norm}}
\end{aligned}
\end{equation}

for a sufficiently big constant $\alpha$. Here, the last inequality also uses the assumption that $|S|  = O( k \log(\varphi^2/\epsilon)$, by \Cref{def:good_S}. 

Similarly to the averaging argument used to derive \Cref{eq:averaging_3}, we apply Markov’s inequality to the sum on the left-hand side, using the lower bound
\(\|f_x\|_2^2 \ge c \cdot \frac{k}{n}\) which holds for a sufficiently small constant \(c\) and all
\(x \in V \setminus B_{\delta}\) (see \Cref{remark:norm}). This implies that there can be no more than
\(\frac{\alpha \cdot 10^{10}}{c^2} \frac{\epsilon}{\varphi^2} \log(\varphi^2/\epsilon) \cdot n\)
vertices \(x \in V\) for which the corresponding summand in the LHS exceeds
\(10^{-10} \cdot \|f_x\|_2^4\). Otherwise, the sum over just these terms is greater than $\frac{\alpha}{c^2}\cdot\|f_x\|^4_2\frac{\epsilon}{\varphi^2}\log(\varphi^2/\epsilon)\cdot n \geq \alpha \frac{\epsilon}{\varphi^2}\cdot\frac{k^2}{n} \log(\varphi^2/\epsilon)$, which would contradict \Cref{eq:averaging_4}. Formally, we get that

$$ \left| \left\{ x \in V \setminus B_{\delta} :  (f_x -\mu_{i(x)})(f_x -\mu_{i(x)})^\top  \bullet \sum_{y \in S }f_yf_y^\top  \geq 10^{-10}\|f_x\|^4_2\right\}\right| \leq O\left(\frac{\epsilon}{\varphi^2}\log(\varphi^2/\epsilon)\cdot n \right).$$
And hence
\begin{align*}\left| \left \{ x \in V :  (f_x -\mu_{i(x)})(f_x -\mu_{i(x)})^\top  \bullet \sum_{y \in S }f_yf_y^\top  \geq 10^{-10}\|f_x\|^4_2\right \}\right | &\leq |B_{\delta}| +O\left(\frac{\epsilon}{\varphi^2}\log(\varphi^2/\epsilon)\cdot n \right)\\
&\leq O\left(\left( \frac{\epsilon}{\varphi^2}\right)^{1/3}\log(\varphi^2/\epsilon)\cdot n \right).
\end{align*}
\end{proof}

\noindent 
Next, we prove \Cref{lemma:bounding_var}, restated below for the convenience of the reader.  
\boundingvar*
\begin{proof}
The first  property follows immediately from \Cref{claim:D1} below,  together with the assumption that $x$ is typical with respect to $S$ (as per \Cref{def:good_pt_wrt_S}). 

\begin{claim}\label{claim:D1}
    If $x \in V \setminus B_{\delta}$ satisfies $ \mu_{i(x)}^\top  \mu_{i(x)}\bullet \sum_{y \in S \setminus C_{i(x)}}  f_y^\top  f_y \leq 10^{-10} \cdot \|\mu_{i(x)}\|^4_2 $ and  
    $ (f_x -\mu_{i(x)})(f_x -\mu_{i(x)})^\top  \bullet \sum_{y \in S }f_yf_y^\top  \leq 10^{-10}\|f_x\|^4_2$, 
    then 
    $$ f_x^\top  f_x \bullet \sum_{y \in S} f_y f_y^\top   \leq 10^{-9}\cdot \|f_x\|^4_2.$$
\end{claim}
\begin{proof}
     We have 
    \begin{equation}\label{eq:replacing_con5_a}
    \begin{aligned}
     f_x^\top  f_x \bullet \sum_{y \in S} f_y f_y^\top  & =   f_x^\top  \left(\sum_{y \in S\setminus C_{i(x)}}f_y f_y^\top  \right)f_x  \\
     & = (f_x-\mu_{i(x)})^\top  \left(\sum_{y \in S\setminus C_{i(x)}}f_y f_y^\top  \right)(f_x-\mu_{i(x)})^\top  +  (f_x - \mu_{i(x)})\left(\sum_{y \in S\setminus C_{i(x)}}f_y f_y^\top  \right)\mu_{i(x)}  \\
        & +  \mu_{i(x)}^\top \left(\sum_{y \in S\setminus C_{i(x)}}f_y f_y^\top  \right)(f_x - \mu_{i(x)}) + 
          \mu_{i(x)}^\top  \left(\sum_{y \in S\setminus C_{i(x)}}f_y f_y^\top  \right)\mu_{i(x)}.  
    \end{aligned}
    \end{equation}
To bound the first term, we use the fact that $\sum_{y \in C_{i(x)}}f_y f_y^\top $ is a PSD matrix, together with the assumption 
\begin{equation}\label{eq:replacing_con5_b}
 (f_x -\mu_{i(x)})(f_x -\mu_{i(x)})^\top  \bullet \sum_{y \in S \setminus C_{i(x)}}f_yf_y^\top  \leq 
 (f_x -\mu_{i(x)})(f_x -\mu_{i(x)})^\top  \bullet \sum_{y \in S }f_yf_y^\top  \leq 10^{-10}\|f_x\|^4_2
\end{equation}
To  bound the last term in \Cref{eq:replacing_con5_a}, we use the assumption 
\begin{equation}\label{eq:replacing_con5_c}
    \mu_{i(x)}^\top  \mu_{i(x)}\bullet \sum_{y \in S\setminus C_{i(x)}}  f_y^\top  f_y \leq 10^{-10} \cdot \|\mu_{i(x)}\|^4_2 \leq 4\cdot 10^{-10}\|f_x\|^4_2 \qquad \text{by \Cref{bulletpt:f_norm}.}
\end{equation}
To bound the second and third term in \Cref{eq:replacing_con5_a}, we use the fact that $\sum_{y \in S\setminus C_{i(x)}} f_yf_y^\top $ is a PSD matrix, so 
\begin{equation}\label{eq:replacing_con5_d}
\begin{aligned}
    & (f_x - \mu_{i(x)})^\top \left(\sum_{y \in S\setminus C_{i(x)}}f_y f_y^\top  \right)\mu_{i(x)}  +  \mu_{i(x)}^\top \left(\sum_{y \in S\setminus C_{i(x)}}f_y f_y^\top  \right)(f_x - \mu_{i(x)}) \\
    & \leq (f_x - \mu_{i(x)})^\top \left(\sum_{y \in S\setminus C_{i(x)}}f_y f_y^\top  \right)(f_x - \mu_{i(x)}) + \mu_{i(x)}^\top \left(\sum_{y \in S\setminus C_{i(x)}}f_y f_y^\top  \right)\mu_{i(x)} \\
    & \leq 3 \cdot 10^{-10}\|f_x\|^4_2 &&\text{by Equations \eqref{eq:replacing_con5_b} and  \eqref{eq:replacing_con5_c}}.
    \end{aligned}
\end{equation}
Combining Equations \eqref{eq:replacing_con5_a}, \eqref{eq:replacing_con5_b}, \eqref{eq:replacing_con5_c} and \eqref{eq:replacing_con5_d} gives the claim. 
\end{proof}

\noindent 
The second property follows immediately by \Cref{claim:D2} below, together with the assumptions that $x$ is typical with respect to $S$ (as per \Cref{def:good_pt_wrt_S}). 
\begin{claim}\label{claim:D2}
    If $x \in V \setminus B_{\delta}$ satisfies $ \mu_{i(x)}^\top  \mu_{i(x)}\bullet \sum_{y \in S}  f_y^\top  f_y \leq O(\varphi^2/\epsilon)^{1/3} \cdot \|\mu_{i(x)}\|^4_2 $ and  
    $ (f_x -\mu_{i(x)})(f_x -\mu_{i(x)})^\top  \bullet \sum_{y \in S }f_yf_y^\top  \leq 10^{-10}\|f_x\|^4_2$, 
    then 
    $$ f_x^\top  f_x \bullet \sum_{y \in S} f_y f_y^\top   \leq O(\varphi^2/\epsilon)^{1/3}\cdot \|f_x\|^4_2.$$
\end{claim}
\begin{proof}
     We have 
    \begin{equation}\label{eq:replacing_con6_a}
    \begin{aligned}
    f_x^\top  f_x \bullet \sum_{y \in S} f_y f_y^\top   & = f_x^\top  \left(\sum_{y \in S}f_y f_y^\top  \right)f_x \\
         & = (f_x-\mu_{i(x)})^\top  \left(\sum_{y \in S}f_y f_y^\top  \right)(f_x-\mu_{i(x)})^\top  +  (f_x - \mu_{i(x)})\left(\sum_{y \in S}f_y f_y^\top  \right)\mu_{i(x)}  \\
        & +  \mu_{i(x)}^\top \left(\sum_{y \in S}f_y f_y^\top  \right)(f_x - \mu_{i(x)}) + 
          \mu_{i(x)}^\top  \left(\sum_{y \in S}f_y f_y^\top  \right)\mu_{i(x)}.  
    \end{aligned}
    \end{equation}
\noindent 
To bound the first term in \Cref{eq:replacing_con6_a}, we use the assumption 
\begin{equation}\label{eq:replacing_con6_b}
 (f_x -\mu_{i(x)})(f_x -\mu_{i(x)})^\top  \bullet \sum_{y \in S }f_yf_y^\top  \leq 10^{-10}\|f_x\|^4_2. 
\end{equation}
\noindent 
To  bound the last term in \Cref{eq:replacing_con6_a}, we use the assumption 
\begin{align}\label{eq:replacing_con6_c}
    \mu_{i(x)}^\top  \mu_{i(x)}\bullet \sum_{y \in S}  f_y^\top  f_y \leq O(\varphi^2/\epsilon)^{1/3} \cdot \|\mu_{i(x)}\|^4_2 \leq O(\varphi^2/\epsilon)^{1/3} \cdot \|f_x\|^4_2 , && \text{by \Cref{bulletpt:f_norm}.}
\end{align}
\noindent 
To bound the second and third term in \Cref{eq:replacing_con6_a}, we use the fact that $\sum_{y \in S} f_yf_y^\top $ is a PSD matrix, so 
\begin{equation}\label{eq:replacing_con6_d}
\begin{aligned}
    & (f_x - \mu_{i(x)})^\top \left(\sum_{y \in S}f_y f_y^\top  \right)\mu_{i(x)}  +  \mu_{i(x)}^\top \left(\sum_{y \in S}f_y f_y^\top  \right)(f_x - \mu_{i(x)}) \\
    & \leq (f_x - \mu_{i(x)})^\top \left(\sum_{y \in S}f_y f_y^\top  \right)(f_x - \mu_{i(x)}) + \mu_{i(x)}^\top \left(\sum_{y \in S}f_y f_y^\top  \right)\mu_{i(x)} \\
    & \leq O\left(\left(\varphi^2/\epsilon)^{1/3}\right)\right)\|f_x\|^4_2 && \text{by Equations \eqref{eq:replacing_con6_b} and  \eqref{eq:replacing_con6_c}}.
    \end{aligned}
\end{equation}
Combining Equations \eqref{eq:replacing_con6_a}, \eqref{eq:replacing_con6_b}, \eqref{eq:replacing_con6_c} and \eqref{eq:replacing_con6_d} gives the claim. 
\end{proof}
\end{proof}

\subsection{Proof of \Cref{claim:Equal}}
In this section, we prove \Cref{claim:Equal}, restated below for the convenience of the reader. 
\Equal*
\begin{proof}
We have  
$$\Pr_{\sigma} \left[\mathcal{E}_{\mathrm{equal}} \right] = \1\{|S \cap C_i| \text{ is even}\} \cdot \binom{|S \cap C_{i(x)}|}{\frac{|S \cap C_{i(x)}|}{2}} 2^{-|S \cap C_{i(x)}|}.$$ 
To bound the right-hand side, we consider the two cases  $|S \cap C_i| = 2$ and $|S \cap C_i| >2 $ separately. \\
\noindent 
\textbf{Case a: $|S \cap C_i| = 2$.} Then $\binom{|S \cap C_{i(x)}|}{|S \cap C_{i(x)}|/2} \cdot 2^{-|S \cap C_{i(x)}|} = 2 \cdot 2^{-4} = 0.5$.  \\

\noindent 
\textbf{Case b: $|S \cap C_i| > 2$.} Since the probability is non-zero only when $|S \cap C_i|$ is even, we can assume that $|S \cap C_i| \geq 4$. Using the following formulation of Stirling's formula (due to \cite{Robbins55})
\begin{equation*}
    n! = \sqrt{2 \pi}n^{n+1}e^{-n}e^{r_n},
\end{equation*}
where $r_n \in \left[\frac{1}{12n+1},\frac{1}{12n}\right]$, we obtain $\binom{n}{n/2} \leq \frac{2^{n+1}\sqrt{2}}{\sqrt{\pi \cdot n}} < 0.5 \cdot 2^{n}$ for all $n \geq 4$.  Setting $n = |S \cap C_i|$ proves the claim.
\end{proof}

\subsection{Almost all clusters are well-represented by $S$ (Proof of \Cref{lemma:random_good_cluster})}\label{sec:randomgoodcluster}
In this section we prove \Cref{lemma:random_good_cluster}. We start by proving \Cref{lemma:random_good_pts_in_S} bounding the number of vertices which are \text{not} typical with respect to $S$. 
\begin{lemma}\label{lemma:random_good_pts_in_S}
If $S$ is a (multi) set of $O(k \log(\varphi^2/\epsilon))$ vertices sampled independently uniformly at random from $V$, then with probability at least $0.999$ (over the randomness of $S$), it holds that 
    \[ \left| \left\{ x \in S : \text{x is \emph{not} typical with respect to  S}\right\}\right| \leq  O \left(k \cdot \left(\frac{\epsilon}{\varphi^2}\right)^{1/3}\cdot \log^2(\varphi^2/\epsilon)\right). \]    
\end{lemma}
\begin{proof}
Recall from \Cref{def:good_pt_wrt_S} that a typical vertex with respect to $S$ is defined by 5 conditions. For each of the conditions, we define the set of all points in $S$ which violate this condition, and we bound each of their sizes separately. Let
\begin{align*}
   F_1&  \coloneqq \{ x \in S: S \cap B_{\delta} \cap C_{i(x)} \neq  \emptyset\}  \\
    F_2  & \coloneqq \left\{ x \in S : |S \cap C_{i(x)}| \geq 10^{-5} \delta^{-1/2}\right\} \\
   F_3& \coloneqq \left\{ x \in S: p^{t}_S(x) \geq 10^{5} \cdot \left(\frac{\varphi^2 }{\epsilon}\right)^{1/3}\cdot \frac{(t_{\Delta}+1)}{|S|}\cdot \frac{k}{n}  \text{ for some  $t \in \left[2t_{\min}, 2t_{\min} + 2t_{\Delta}\right]$}  \right\}\\
     F_4 & \coloneqq \left\{ x \in S : \left \langle p_{x}^{t} , \left(p_S^{\ell}\right)^2 \right \rangle \geq 10^{5} \cdot \left(\frac{\varphi^2}{\epsilon}\right)^{1/3}\cdot  \frac{(t_{\Delta}+1)^4}{|S|^2}\cdot \left(\frac{k}{n}\right)^{3/2}\log(\varphi^2/\epsilon) \right.\\
   & \qquad \quad \left. \text{ or } \left \langle \left( p_x^t\right)^2, p_S^{\ell}\right \rangle \geq 10^{5} \cdot \left(\frac{\varphi^2}{\epsilon}\right)^{1/3}\cdot  \frac{(t_{\Delta}+1)^4}{|S|}\cdot \left(\frac{k}{n}\right)^{3/2}\log(\varphi^2/\epsilon)  \text{ for some  $t, \ell \in \left[t_{\min}, t_{\min} + t_{\Delta}\right]$} \right\}\\
    F_5 & \coloneqq \left\{x \in S :  (f_x -\mu_{i(x)})(f_x -\mu_{i(x)})^\top  \bullet \sum_{y \in S }f_yf_y^\top  \geq 10^{-10}\|f_x\|^4_2 \right\}\\
\end{align*}

    be the sets of vertices failing Conditions {\bf \ref{con:bad_cluster}}, {\bf \ref{con:ScapC_i}}, {\bf \ref{con:rw1}}, {\bf \ref{con:rw2}} and {\bf \ref{con:rw3}} in \Cref{def:good_pt_wrt_S} respectively. 
    Then 
    \begin{equation*}
        \left| \left\{ x \in V : x \text{ is \emph{not} typical with respect to  S}\right\}\right| \leq |F_1| + |F_2| + |F_3| + |F_4| + |F_5|,
    \end{equation*}
    so it suffices to show that $ |F_1|, |F_2|, |F_3|, |F_4|, |F_5| \leq  O \left(k \cdot \left(\frac{\epsilon}{\varphi^2}\right)^{1/3}\cdot \log^2(\varphi^2/\epsilon)\right)$. 
    
    We will now show that for every $i \in [5]$, $|F_i| \leq O \left(k \cdot \left(\frac{\epsilon}{\varphi^2}\right)^{1/3}\cdot \log^2(\varphi^2/\epsilon)\right)$ with probability at least $0.9992$. The lemma then follows by a union bound. 

 Everywhere in the proofs below we think of $S = \{x_1, \ldots, x_{L}\}$ as a random (multi) set of fixed size $L = O(k\cdot\log(\varphi^2/\epsilon))$; we think of $x_1, \ldots, x_{L}$ as random variables uniformly distributed over $V$. 

\paragraph{Bounding $|F_1|$:}\label{par:bound_F_1}
Slightly abusing the notation, $F_1$ is the random variable which counts the number of elements $x_l \in$ violating constraint $S \cap B_{\delta} \cap C_{i(x_l)} \neq  \emptyset$. Our goal is to show that
\[\Pr_{S}\left[|F_1| \leq O\left(k\cdot \left(\frac{\epsilon}{\varphi^2}\right)^{1/3}\cdot\log^2(\varphi^2/\epsilon)\right)  \right] \geq 0.9992.\]

We achieve this by showing a constant lower bound on $\Pr_{S}\left[|F_1| \leq O\left(k\cdot \left(\frac{\epsilon}{\varphi^2}\right)^{1/3}\cdot\log^2(\varphi^2/\epsilon)\right)  \middle| \mathcal{E} \right]$ for a high constant probability event $\mathcal{E}$. 

\paragraph{High probability event $\mathcal{E}$.} Let $\alpha$ be a sufficiently small constant. Let
\[\mathcal{E} \coloneqq \left\{ |S\cap B_{\delta}| \leq \alpha \cdot k\cdot \left(\frac{\epsilon}{\varphi^2}\right)^{1/3}\cdot\log(\varphi^2/\epsilon)\right\}.\] 
Since, as shown in \Cref{lemma:close_to_clutermean}, $|B_{\delta}| \leq O(\frac{1}{\delta} \cdot \epsilon/\varphi^2 \cdot n)$, we have

\[ \E_{S}[|S\cap B_{\delta}|] \leq O\left(k \cdot \log(\varphi^2/\epsilon)\frac{\epsilon}{\delta\varphi^2}\right).\]
By the setting $\delta = \Omega((\epsilon/\varphi^2)^{2/3})$ (as per \Cref{def:params}), by Markov's inequality

\[1 - \Pr_{S}[\mathcal{E}] = \Pr_{S}\left[S \cap B_{\delta} \geq O\left(k\cdot \left(\frac{\epsilon}{\varphi^2}\right)^{1/3}\cdot\log(\varphi^2/\epsilon)\right)\right] \leq 0.0001.\]

\paragraph{Why conditioning on $\mathcal{E}$ helps.} We could try to bound $|F_1|$ with the following naive argument. If $\mathcal{E}$ holds then $|S \cap B_{\delta}| \leq \alpha \cdot k\cdot \left(\frac{\epsilon}{\varphi^2}\right)^{1/3}\cdot\log(\varphi^2/\epsilon)$, and therefore there can be no more than $\alpha \cdot k\cdot \left(\frac{\epsilon}{\varphi^2}\right)^{1/3}\cdot\log(\varphi^2/\epsilon)$ clusters $C_i$ for which $S \cap B_{\delta}\cap C_i \neq \emptyset$. Since $|C_i| \leq O(n/k)$ for each $i$, we have that
\[|F_1| \leq O\left(\left(\frac{\epsilon}{\varphi^2}\right)^{1/3}\cdot\log(\varphi^2/\epsilon)\cdot n\right).\]

While the argument is valid, the upper bound we get is too big. We will use the ideas developed in this argument to bound $\Pr_S\left[ S \cap B_{\delta} \cap C_{i(x_l)} \neq  \emptyset \middle| \mathcal{E} \right]$ for each $x_l \in S$~-- we will then get the desired bound on $|F_1|$ by Markov's inequality. 

\paragraph{Bounding $\Pr_S\left[ S \cap B_{\delta} \cap C_{i(x_l)} \neq  \emptyset \middle| \mathcal{E} \right]$ for all $l$.}
Fix $l \in [L]$ and define $S\setminus x_l \coloneqq \{x_1, \ldots, x_{l-1}, x_{l+1}, \ldots, x_L\}$.   Observe that
\begin{equation}\label{eq:x_l_bad}
    \begin{aligned}
    \Pr_S\left[\mathcal{E} \right]\cdot\Pr_S\left[ S \cap B_{\delta} \cap C_{i(x_l)} \neq  \emptyset \middle| \mathcal{E} \right] =  
    \Pr_S\left[ S \cap B_{\delta} \cap C_{i(x_l)} \neq  \emptyset \text{ and } \mathcal{E} \right] \leq  \\
    \Pr_{S}\left[ S \cap B_{\delta} \cap C_{i(x_l)} \neq  \emptyset \text{ and } (S\setminus x_l) \cap B_{\delta} \leq \alpha\cdot k\cdot \left(\frac{\epsilon}{\varphi^2}\right)^{1/3}\cdot\log(\varphi^2/\epsilon)  \right] \leq \\
    \Pr_{x_l}[x_l \in B_\delta] + \Pr_{S}\left[ (S\setminus x_l) \cap B_{\delta} \cap C_{i(x_l)} \neq  \emptyset \text{ and } (S\setminus x_l) \cap B_{\delta} \leq \alpha\cdot k\cdot \left(\frac{\epsilon}{\varphi^2}\right)^{1/3}\cdot\log(\varphi^2/\epsilon)   \right].
\end{aligned}
\end{equation}
Here the first inequality follows since $S\setminus x_l \subset S$.  
The last inequality follows from the fact that $S \cap B_{\delta} \cap C_{i(x_l)} \neq  \emptyset$ only if $x_l$ got sampled from $B_{\delta}$ or if $x_l$ got sampled from one of the clusters which have non-empty intersection with $(S\setminus x_l)\cap B_{\delta}$. The first probability can be bounded as
\[\Pr_{x_l}\left[x_l \in B_{\delta}\right] = \frac{|B_{\delta}|}{n} \leq O\left(\left(\frac{\epsilon}{\varphi^2}\right)^{1/3}\right).\]

Note that if $(S\setminus x_l) \cap B_{\delta} \leq \alpha\cdot k\cdot \left(\frac{\epsilon}{\varphi^2}\right)^{1/3}\cdot\log(\varphi^2/\epsilon)$, then there can be no more than $\alpha\cdot k\cdot \left(\frac{\epsilon}{\varphi^2}\right)^{1/3}\cdot\log(\varphi^2/\epsilon) \ll k$ of clusters which have non-empty intersection with $(S\setminus x_l)\cap B_{\delta}$. Therefore, in this case, $(S\setminus x_l) \cap B_{\delta} \cap C_{i(x_l)} \neq  \emptyset$ only if $x_l$ is sampled from this small subset of clusters:
\begin{align*}
\Pr_S\left[ (S\setminus x_l) \cap B_{\delta} \leq \alpha\cdot k\cdot \left(\frac{\epsilon}{\varphi^2}\right)^{1/3}\cdot\log(\varphi^2/\epsilon) \text{ and } (S\setminus x_l) \cap B_{\delta}\cap C_{i(x_l)} \neq \emptyset\right] \leq \\
\alpha\cdot k\cdot \left(\frac{\epsilon}{\varphi^2}\right)^{1/3}\cdot\log(\varphi^2/\epsilon)\cdot \max_{i}\frac{|C_i|}{n} \leq  O\left(\left(\frac{\epsilon}{\varphi^2}\right)^{1/3}\cdot\log(\varphi^2/\epsilon)\right).
\end{align*}

Finally, by plugging these bounds into \Cref{eq:x_l_bad} we get
\[ \Pr_S\left[ S \cap B_{\delta} \cap C_{i(x_l)} \neq  \emptyset \middle| \mathcal{E} \right]\leq \frac{O\left(\left(\frac{\epsilon}{\varphi^2}\right)^{1/3}\cdot\log(\varphi^2/\epsilon)\right)}{\Pr_S\left[ \mathcal{E} \right]}.\]

Since $\Pr_S[\mathcal{E}] \geq 0.9999$,
\[\Pr_S\left[ S \cap B_{\delta} \cap C_{i(x_l)} \neq  \emptyset \middle| \mathcal{E} \right] \leq O\left( \left(\frac{\epsilon}{\varphi^2}\right)^{1/3}\cdot\log(\varphi^2/\epsilon)\right).\] 

\paragraph{Deriving the bound on $|F_1|$.} Since the above bounds holds for every $x_l \in S$, and $F_1$ counts precisely the number of $x_l$ for which $S \cap B_{\delta} \cap C_{i(x_l)} \neq  \emptyset$, 
\[\E_{S}\left[|F_1| \middle|\mathcal{E}\right] \leq O\left( \left(\frac{\epsilon}{\varphi^2}\right)^{1/3}\cdot\log(\varphi^2/\epsilon)\cdot |S|\right) = O\left(k\cdot \left(\frac{\epsilon}{\varphi^2}\right)^{1/3}\cdot\log^2(\varphi^2/\epsilon)\right).\]
By Markov's inequality, for a large enough constant $C$
\[\Pr_{S}\left[|F_1| \leq C\cdot k\cdot \left(\frac{\epsilon}{\varphi^2}\right)^{1/3}\cdot\log^2(\varphi^2/\epsilon) \middle| \mathcal{E} \right] \geq 0.9999\]
Finally, 
\begin{align*}
    \Pr_S\left[|F_1| \leq O\left(k\cdot \left(\frac{\epsilon}{\varphi^2}\right)^{1/3}\cdot\log^2(\varphi^2/\epsilon)\right)\right]  \geq 
    \Pr_S\left[|F_1| \leq 
    O\left(k\cdot \left(\frac{\epsilon}{\varphi^2}\right)^{1/3}\cdot\log^2(\varphi^2/\epsilon)\right)  \middle| \mathcal{E} \right]\cdot\Pr_S\left[\mathcal{E} \right] \geq 0.9992
\end{align*}
as desired.

\paragraph{Bounding $|F_2|:$} We will show that with probability at least $0.9992$, it holds that $|F_2| \leq O\left(k\cdot \left(\epsilon/{\varphi^2}\right)^{1/3}\cdot\log^2(\varphi^2/\epsilon)\right).$

Fix $l \in [L]$. We use notation $S \setminus x_l $ to denote the set $\{x_1, \dots, x_{l-1}, x_{l+1}, \dots , x_L\}$.  Since $|S| = O(k\cdot\log(\varphi^2/\epsilon))$ (by Condition \ref{con:good_S1} in \Cref{def:good_S}), there can be at most $\frac{|S|}{10^{-6}\delta^{-1/2}} = O(k\cdot\delta^{1/2}\cdot \log(\varphi^2/\epsilon))$ different clusters $C_i$ such that $|(S \setminus x_l ) \cap C_i| \geq 10^{-6}\delta^{-1/2}$. Therefore, 
\[\Pr_{x_l}\left[(S\setminus x_l) \cap C_{i(x_l)} 
      \geq 10^{-6}\delta^{-1/2} \right] \leq O\left(k\cdot\delta^{1/2}\cdot\log(\varphi^2/\epsilon)\right)\cdot \max_{i}\frac{|C_i|}{n}\] for every realization of the random variable $S \setminus x_l$ . Hence,
\begin{align*}
     \Pr_S\left[  |S \cap C_{i(x_l)}| \geq 10^{-5} \delta^{-1/2}\right] & \leq \Pr_{S\setminus x_l} \left[ \Pr_{x_l}\left[(S\setminus x_l) \cap C_{i(x_l)} 
      \geq 10^{-6}\delta^{-1/2} \right] \mid S\setminus x_l \right] \\
& = O\left(k\cdot\delta^{1/2}\cdot\log(\varphi^2/\epsilon)\right)\cdot \max_{i}\frac{|C_i|}{n} \\ 
&=   O\left(\left(\frac{\epsilon}{\varphi^2}\right)^{1/3}\cdot\log(\varphi^2/\epsilon)\right) &&  \text{since } \delta = O((\epsilon/\varphi^2)^{2/3}),\\ 
\end{align*} where the first inequality holds because $ |S \cap C_{i(x_l)}| \geq 10^{-5} \delta^{-1/2}$ implies $(S\setminus x_l) \cap C_{i(x_l)} \geq 10^{-6}\delta^{-1/2}$. This is because $\delta = O((\epsilon/\varphi^2)^{2/3})$ (as per \Cref{def:params}) and, in particular, $\delta$ is smaller than a sufficiently small constant.

Since the above inequality holds for every $x_l \in S$, and $|F_2|$ counts the number of $x_l$ for which $|S \cap C_{i(x_l)}| \geq 10^{-5}\delta^{-1/2}$, we get

\[\E_{S}\left[|F_2|\right] \leq O\left(k\cdot \left(\frac{\epsilon}{\varphi^2}\right)^{1/3}\log^2(\varphi^2/\epsilon)\right).\]
By Markov's inequality, with probability at least $0.9992$, we get 
\[|F_2| \leq O\left(k\cdot \left(\frac{\epsilon}{\varphi^2}\right)^{1/3}\cdot\log^2(\varphi^2/\epsilon)\right)\]
as desired.

\paragraph{Bounding $|F_3|$:}  Given $t$,  let 
\[F_{3, t} \coloneqq \left\{ x \in S: p^{t}_S(x) \geq 10^{5} \cdot \left(\frac{\varphi^2 }{\epsilon}\right)^{1/3}\cdot \frac{t_{\Delta}+1}{|S|}\cdot \frac{k}{n}  \right\}.\]  
By definition of $F_3$, $F_3 = \cup_{t \in \left[2t_{\min}, 2t_{\min} + 2t_{\Delta}\right]}F_{3, t}$, and it would be easier for us to bound all $F_{3, t}$ separately. Recall from \eqref{eq:p-s} that \[p^t_S(x) = \frac{1}{|S|}\sum_{y \in S}M^t\1_y(x) = \frac{1}{|S|}\left \langle \1_x, \sum_{y \in S} M^t\1_y \right \rangle = \frac{1}{|S|}\left \langle M^{t_1}\1_x, \sum_{y \in S} M^{t_2}\1_y\right \rangle\]
where $t_1 = \lceil t/2 \rceil$, $t_2 = \lfloor t/2\rfloor $.

Define a random variable
\[Z \coloneqq \sum_{l = 1}^L\left \langle M^{t_1} \1_{x_l}, \sum_{j = 1}^L M^{t_2} \1_{x_j}\right\rangle.\]
Note that, if we think of $S$ as a random variable, $Z = |S|\cdot\sum_{x \in S}p^t_S(x)$. A good upper bound on $Z$ will allow us to argue that there cannot be too many $x \in S$ for which the corresponding $p^t_S(x)$ is large. Because the set $F_{3, t}$ consists of all $x \in S$ for which $p^t_S(x)$ exceeds a certain threshold, this will imply an upper bound on $|F_{3, t}|$.

Observe that
\begin{equation}\label{eq:EZ}
\begin{aligned}\E_S[Z] = \sum_{l = 1}^L\E_{x_l}\left[\left\langle M^{t_1} \1_{x_l}, M^{t_2} \1_{x_l}\right\rangle\right] + \sum_{j,l \in [L] : l \neq j }\E_{x_l, x_j}\left[\left\langle M^{t_1} \1_{x_l}, M^{t_2} \1_{x_j}\right\rangle\right],
\end{aligned}
\end{equation}
and all of the expectations in each of the sums are equal. Therefore, 
\[\E_S[Z] = |S|\cdot \E_{x_1}\left[\left\langle M^{t_1} \1_{x_1}, M^{t_2} \1_{x_1}\right\rangle\right] + |S|\cdot(|S|-1)\cdot \E_{x_1, x_2}\left[\left\langle M^{t_1} \1_{x_1}, M^{t_2} \1_{x_2}\right\rangle\right].\]
For all $x$ and all $t \geq t_{\min}$, we have $\|M^t\1_x\|^2_2 \leq O(k/n)$,  by \Cref{lem:bound-norms-1/2}, so
\[\E_{x_1}\left[\langle M^{t_1}\1_{x_1}, M^{t_2}\1_{x_1}\rangle\right] \leq \E_{x_1}\left[\|M^{t_1}\1_{x_1}\|^2_2 + \|M^{t_2}\1_{x_1}\|^2_2\right] \leq  O\left(\frac{k}{n}\right).\]

\noindent To bound the second term in \Cref{eq:EZ}, note that for every realization of $x_1$,  we have $\E_{x_2}[\langle M^{t_1} \1_{x_1}, M^{t_2} \1_{x_2}\rangle] = \frac{1}{n}\sum_{v \in V}\langle M^{t_1} \1_{x_1}, M^{t_2} \1_{v}\rangle = \frac{1}{n}\langle M^{t_1} \1_{x_1}, \1 \rangle = \frac{1}{n}$, so $\E_{x_1, x_2}[\langle M^{t_1} \1_{x_1}, M^{t_2} \1_{x_2}\rangle] = \frac{1}{n}$. Combining, we get

$$\E_S[Z] \leq O\left(|S|\cdot\frac{k}{n}\right) + \frac{|S|^2}{n} \leq O\left(\frac{k^2 \log^2(\varphi^2/\epsilon)}{n}\right).$$

Therefore, with probability at least $1- \frac{0.0004}{t_{\Delta} +1}$ over the sampling of $S$, by Markov's inequality, we get

$$Z =  O\left(\frac{k^2 \log^2(\varphi^2/\epsilon)\cdot (t_{\Delta}+1)}{n}\right) .$$
For every $l \in [L]$, let $Z_{l} \coloneqq \left \langle M^t \1_{x_l} , \sum_{j = 1}^L M^t \1_{x_j}\right \rangle$, and observe that $Z = \sum_{l = 1}^L Z_{l}$. Condition on the event that we have sampled a set $S$ for which $Z \leq  O\left(\frac{k^2\cdot (t_{\Delta}+1)\cdot\log^2(\varphi^2/\epsilon)}{n}\right)$. Then realizations of all except for at most $O\left(k \cdot \left(\frac{\epsilon}{\varphi^2}\right)^{1/3} \frac{\log^2(\varphi^2/\epsilon)}{(t_{\Delta}+1)}\right)$ random variables $x_l \in S$ satisfy
\[\left \langle M^{t_1} \1_{x_l} , \sum_{j = 1}^L M^{t_2} \1_{x_j}\right\rangle \leq 10^{5} \cdot \left(\frac{\varphi^2 }{\epsilon}\right)^{1/3}\cdot (t_{\Delta}+1)^2\cdot \frac{k}{n}, \]
\noindent 
which is equivalent to writing $|F_{3, t}| \leq O\left(k \cdot \left(\frac{\epsilon}{\varphi^2}\right)^{1/3} \frac{\log^2(\varphi^2/\epsilon)}{t_{\Delta}+1}\right)$. 

Finally, recall that $F_3 = \cup_{t \in \left[2t_{\min}, 2t_{\min} + 2t_{\Delta}\right]}F_{3, t}$. By a union bound, with probability at least \linebreak  ${1 - \frac{0.0004}{t_{\Delta}+1}\cdot {2(t_{\Delta}+1)} = 0.9992}$, it holds that $|F_{3, t}| \leq O\left(k \cdot \left(\frac{\epsilon}{\varphi^2}\right)^{1/3} \frac{\log^2(\varphi^2/\epsilon)}{t_{\Delta}+1}\right)$ for all $t \in [2t_{\min}, 2t_{\min}+2t_{\Delta}]$, and so

\[|F_3| \leq O\left(k \cdot \left(\frac{\epsilon}{\varphi^2}\right)^{1/3} \log^2(\varphi^2/\epsilon)\right).\]

\paragraph{Bounding $|F_4|$:} 

For all $t, u \in [t_{\min}, t_{\min} + t_{\Delta}]$ define $F_{4, t, \ell}$:
\begin{align*}
F_{4, t, \ell} = \left\{ x \in S : \sum_{v \in V}(M^t\1_x(v))^2\sum_{y \in S}M^{\ell}\1_y(v) + \sum_v M^{t} \1_x(v)  \sum_{y, y' \in S}     M^{\ell} \1_{y}(v)   M^{\ell} \1_{y'}(v) \geq \right. \\\left. 10^{5} \cdot \left(\frac{\varphi^2}{\epsilon}\right)^{1/3} (t_{\Delta}+1)^4 \left(\frac{k}{n}\right)^{3/2}\log(\varphi^2/\epsilon) \right\}.
\end{align*} Note that $F_4 \subseteq \cup_{t, \ell \in [t_{\min}, t_{\min} + t_{\Delta}]}F_{4, t, \ell}$. To verify this statement, first recall from \Cref{eq:p-s}, that for every $v \in V$,  $p^t_S(v) = \frac{1}{|S|}\sum_{y \in S}M^t\1_y(v)$, so $(p^t_S(v))^2 = \frac{1}{|S|^2}\sum_{y, y' \in S}M^t\1_y(v)M^t\1_{y'}(v)$. If $x \in F_{4}$ then there exist some $t, \ell \in [t_{\min}, t_{\min} + t_{\Delta}]$ for which either
\[\frac{1}{|S|^2}\sum_{v \in V}M^t\1_x(v)\sum_{y, y' \in S}M^{\ell}\1_y(v)M^{\ell}\1_{y'}(v) = \left\langle p_{x}^{t} , \left(p_S^{\ell}\right)^2 \right \rangle \geq 10^{5} \cdot \left(\frac{\varphi^2}{\epsilon}\right)^{1/3}\cdot  \frac{(t_{\Delta}+1)^4}{|S|^2}\cdot \left(\frac{k}{n}\right)^{3/2}\log(\varphi^2/\epsilon)\] or
\[\frac{1}{|S|}\sum_{v \in V}(M^t\1_x(v))^2\sum_{y \in S}M^{\ell}\1_y(v) = \left \langle \left( p_x^t\right)^2, p_S^{\ell}\right \rangle \geq 10^{5} \cdot \left(\frac{\varphi^2}{\epsilon}\right)^{1/3}\cdot  \frac{(t_{\Delta}+1)^4}{|S|}\cdot \left(\frac{k}{n}\right)^{3/2}\log(\varphi^2/\epsilon).\] Hence, by the definition of $F_{4, t, \ell}$ we get that $x \in F_{4, t, \ell}$.

 Define random variables
 \[Z \coloneqq \sum_{l \in [L]}\sum_v \sum_{j\neq j' \in [L]}   M^{t} \1_{x_l}(v)   M^{\ell} \1_{x_j}(v)   M^{\ell} \1_{x_{j'}}(v)\]
 and
 \[Z' \coloneqq \sum_{l, j \in [L]}\sum_{v \in V}\left[(M^t\1_{x_l}(v))^2M^{\ell}\1_{x_j}(v) + M^t\1_{x_l}(v)(M^{\ell}\1_{x_j}(v))^2\right].\]
 Note that, if we think of $S$ as a random variable, then $Z + Z'= \sum_{x \in S}\left(|S|^2\cdot\langle p^t_x, \left(p^{\ell}_S\right)^2\rangle + |S|\cdot \langle \left(p^{t}_x\right)^2, p^{\ell}_S \rangle\right)$. Good upper bounds on $Z$ and $Z'$ will allow us to argue that there cannot be too many $x \in S$ for which the corresponding inner product is large. Because the set $F_{4, t, \ell}$ consists of all $x \in S$ for which the corresponding inner product exceeds a certain threshold, this will imply an upper bound on $|F_{4, t, \ell}|$.

 We begin by showing an upper bound on $Z$. First, we split the terms which constitute $Z$ into two groups: where $l, j, j'$ are all different indices and where $l$ is allowed to coincide with either $j$ or $j'$.
\begin{align*}
    \E_S[Z] = 2\sum_{l \neq j \in [L]}&\E_{x_l, x_j}\left[\sum_v M^{t} \1_{x_l}(v)   M^{\ell} \1_{x_l}(v)   M^{\ell} \1_{x_j}(v)\right]\\
    + \sum_{\substack{l, j, j' \in [L]: \\  l \neq j \neq j' \neq l}}&\E_{x_l, x_j, x_{j'}}\left[\sum_v M^{t} \1_{x_l}(v)   M^{\ell} \1_{x_j}(v)   M^{\ell} \1_{x_{j'}}(v)\right]
\end{align*}
 Note that all of the expectations in each of the sums are equal. Therefore, 
\begin{align*}
    \E_S[Z] \leq 2|S|^2\cdot \E_{x_1, x_2}\left[\sum_v M^{t} \1_{x_1}(v)   M^{\ell} \1_{x_1}(v)   M^{\ell} \1_{x_2}(v)\right] + |S|^3\cdot\E_{x_1, x_2, x_3}\left[\sum_v M^{t} \1_{x_1}(v)   M^{\ell} \1_{x_2}(v)   M^{\ell} \1_{x_3}(v)\right]
\end{align*}

Observe that for every $v$

\[\E_{x_1}\left[\E_{x_2}\left[\E_{x_3}\left[M^{t} \1_{x_1}(v)   M^{\ell} \1_{x_2}(v) M^{\ell} \1_{x_3}(v)\right]\right]\right] = \frac{1}{n}\E_{x_1}\left[\E_{x_2}\left[M^{t} \1_{x_1}(v)  M^{\ell} \1_{x_2}(v)\right]\right] = \frac{1}{n^2}\E_{x_1}\left[M^{t} \1_{x_1}(v)\right] = \frac{1}{n^3},\]
so the second term equals $\frac{|S|^3}{n^3}$. All of the transitions above hold because $\sum_{x \in V}M^t\1_x = M^t\1 = \1$. Next,

\[\E_{x_2}\left[\sum_v M^{t} \1_{x_1}(v)   M^{\ell} \1_{x_1}(v)   M^{\ell} \1_{x_2}(v)\right] = \frac{1}{n}\sum_v M^{t} \1_{x_1}(v)   M^{\ell} \1_{x_1}(v) = \frac{1}{n}\left \langle M^{t} \1_{x_1}, M^{\ell} \1_{x_1}\right\rangle,\]
and $\E_{x_1}[\langle M^{t} \1_{x_1}, M^{\ell} \1_{x_1}\rangle]  \leq \E_{x_1}[\|M^t\1_{x_1}\|_2\cdot \|M^{\ell}\1_{x_1}\|_2]\leq O\left(\frac{k}{n}\right)$, as follows from \Cref{lem:bound-norms-1/2}. So, 

\[\E_S[Z] \leq \frac{2|S|^2\cdot k}{n^2} + \frac{|S|^3}{n^3} \leq O\left(\frac{k^3\log^3(\varphi^2/\epsilon)}{n^2}\right).\]

Similarly, we bound $Z'$. First, we split the terms which constitute $Z'$ into two groups: where $l = j$ and where $l \neq j$. 
\begin{align*}
\E_{S}[Z'] &= \sum_{l \in [L]}\E_{x_l}\left[\sum_{v \in V}(M^t\1_{x_l}(v))^2M^{\ell}\1_{x_l}(v) + \sum_{v \in V}M^t\1_{x_l}(v)(M^{\ell}\1_{x_l}(v))^2 \right]  \\ &+ 
\sum_{l \neq j \in [L]}\E_{x_l, x_j}\left[\sum_{v \in V}(M^t\1_{x_l}(v))^2M^{\ell}\1_{x_j}(v)  + \sum_{v \in V}M^t\1_{x_l}(v)(M^{\ell}\1_{x_j}(v))^2 \right] \\
&\leq  O\left(\frac{|S|\cdot\sqrt{k}}{\sqrt{n}}\right)\cdot \E_{x_l}\left[\|M^t\1_{x_l}\|^2_2 + \|M^{\ell}\1_{x_l}\|^2_2\right] + |S|^2\cdot\left(\frac{\E_{x_l}\left[\sum_v(M^t\1_{x_l}(v))^2\right]}{n} + \frac{\E_{x_l}\left[\sum_v(M^{\ell}\1_{x_l}(v))^2\right]}{n}\right)\\
&\leq O\left(\frac{|S|\cdot k^{3/2}}{n^{3/2}} \right) + O\left(\frac{|S|^2\cdot k}{n^2}\right),
\end{align*}
where the third line follows from $M^t\1_x(v) \leq \|M^t\1_x\|_2 \leq O(\sqrt{k/n})$ and the last line follows from $\E_{x_l}[\sum_v(M^{\ell}\1_{x_l}(v))^2] = \E_{x_l}[\|M^t\1_{x_l}\|^2_2]  \leq O(k/n)$ for all $x$ and for all $t \geq t_{\min}$ by \Cref{lem:bound-norms-1/2}. From here, since $|S| = O(k\log(\varphi^2/\epsilon))$,
\[\E_S[Z + Z'] \leq O\left(\frac{k^3\log^3(\varphi^2/\epsilon)}{n^2}\right) +  O\left(\frac{|S|\cdot k^{3/2}}{n^{3/2}} \right) \leq  O\left(\frac{k^{5/2}\log^3(\varphi^2/\epsilon)}{n^{3/2}} \right)\]
since $k/n \leq 1.$ With probability at least $1 - \frac{0.001}{(t_{\Delta}+1)^2}$ over sampling the set $S$, by Markov's inequality
\[Z + Z'\leq O\left(\frac{k^{5/2}\cdot(t_{\Delta}+1)^2\cdot\log^3(\varphi^2/\epsilon)}{n^{3/2}}\right).\] 

Observe that $Z + Z' = \sum_{l = 1}^L Z_{x_l} + Z'_{x_l}$ where $Z_{x_l} \coloneqq \sum_v \sum_{j, j' \in [L]}   M^{t} \1_{x_l}(v)   M^{\ell} \1_{x_j}(v)   M^{\ell} \1_{x_{j'}}(v)$ and $Z'_{x_l} \coloneqq \sum_{v \in V}\sum_{l \in [L]}\left[(M^t\1_{x_l}(v))^2M^{\ell}\1_{x_2}(v) + M^t\1_x(v)(M^{\ell}\1_y(v))^2\right]$. Suppose that we have sampled a set $S$ for which $Z + Z' \leq  C\cdot \frac{k^{5/2}\cdot (t_{\Delta}+1)^2\cdot\log^3(\varphi^2/\epsilon)}{n^{3/2}}$ for some large constant $C$. Then realizations of all except for at most $C\cdot 10^{-5}\cdot k \cdot \left(\frac{\epsilon}{\varphi^2}\right)^{1/3} \frac{\log^2(\varphi^2/\epsilon)}{(t_{\Delta}+1)^2}$ random variables $x_l \in S$ (which we also denote as $x_l$ here slightly abusing notation) satisfy

\begin{align*}
\sum_v (M^t\1_{x_l}(v))^2\sum_{j \in [L]}M^{\ell}\1_{x_j}(v) & + \sum_v M^{t} \1_{x_l}(v)  \sum_{j, j' \in [L]}     M^{\ell} \1_{x_j}(v)   M^{\ell} \1_{x_{j'}}(v)\\
&\leq 10^5\cdot\left(\frac{\varphi^2}{\epsilon}\right)^{1/3} (t_{\Delta}+1)^4 \left(\frac{k}{n}\right)^{3/2}\log(\varphi^2/\epsilon).
\end{align*}

Equivalently, $|F_{4, t, \ell}| \leq O\left(k \cdot \left(\frac{\epsilon}{\varphi^2}\right)^{1/3} \frac{\log^2(\varphi^2/\epsilon)}{(t_{\Delta}+1)^2}\right)$. Finally, recall that $F_4 = \cup_{t, \ell \in \left[t_{\min}, t_{\min} + t_{\Delta}\right]}F_{4, t, \ell}$. With probability at least $\left(1 - \frac{0.001}{(t_{\Delta}+1)^2}\right)^{t_{\Delta}^2} \geq 0.9992$ it holds that $|F_{4, t, \ell}| \leq O\left(k \cdot \left(\frac{\epsilon}{\varphi^2}\right)^{1/3} \frac{\log^2(\varphi^2/\epsilon)}{(t_{\Delta}+1)^2}\right)$ for all $t, \ell \in [t_{\min}, t_{\min}+t_{\Delta}]$, and so

\[|F_4| \leq O\left(k \cdot \left(\frac{\epsilon}{\varphi^2}\right)^{1/3} \log^2(\varphi^2/\epsilon)\right).\]

\paragraph{Bounding $|F_5|$:} 

Define a random variable $Z$:

\[Z =  \sum_{l \in [L]}(f_{x_l} -\mu_{i(x_l)})(f_{x_l} -\mu_{i(x_l)})^\top  \bullet \sum_{j \in [L]\setminus l }f_{x_j}f_{x_j}^\top \] Then, 

\begin{align*}
    \E_S[Z] = \sum_{l \neq j \in [L]} \E_{{x_l}, {x_j}}\left[(f_{x_l} -\mu_{i(x_{l})})(f_{x_l} -\mu_{i(x_l)})^\top  \bullet f_{x_j}f_{x_j}^\top \right]
\end{align*}
Note that all of the expectations in the sum are equal. Therefore, 
\begin{align*}\E_S[Z] & \leq |S|^2\cdot\E_{x_1, x_2}\left[(f_{x_1} -\mu_{i(x_{1})})(f_{x_1} -\mu_{i(x_1)})^\top  \bullet f_{x_2}f_{x_2}^\top \right]
\\ & = \frac{|S|^2}{n^2}\sum_{x \in V}(f_x -\mu_{i(x)})(f_x -\mu_{i(x)})^\top  \bullet\sum_{y \in V}f_yf_y^\top  \\
& = \frac{|S|^2}{n^2}\sum_{x \in V}\|f_x - \mu_{i(x)}\|^2_2 && \text{by \Cref{lemma:sum_y_in_V}}\\
    & \leq O\left(\frac{\epsilon k}{\varphi^2}\cdot\frac{k^2\log^2(\varphi^2/\epsilon)}{n^2}\right) &&  \text{ by \Cref{eq:Bdelta1} in \Cref{lemma:variancebound}.}
\end{align*} 
By Markov's inequality, with probability at least 0.9999 over the sampling of set $S$
\[Z \leq C\cdot \frac{\epsilon k}{\varphi^2}\cdot\frac{k^2\log^2(\varphi^2/\epsilon)}{n^2}\]

for some large constant $C$. Observe that $Z = \sum_{l = 1}^L Z_{x_l}$ where $Z_{x_l} \coloneqq (f_{x_l} -\mu_{i(x_l)})(f_{x_l} -\mu_{i(x_l)})^\top  \bullet \sum_{j \in [L]\setminus l }f_{x_j}f_{x_j}^\top $. Suppose that we have sampled a set $S$ for which $Z \leq C\cdot \frac{\epsilon k}{\varphi^2}\cdot\frac{k^2\log^2(\varphi^2/\epsilon)}{n^2}$. Then realizations of all except for at most $\frac{C}{c}\cdot k \cdot \left(\frac{\epsilon}{\varphi^2}\right)^{1/3} \log(\varphi^2/\epsilon)$ for a sufficiently small constant $c$ random variables $x_l \in S$, which we also denote as $x_l$ here slightly abusing notation, satisfy

\begin{equation}\label{eq:Z_x}
     (f_{x_l} -\mu_{i(x_l)})(f_{x_l} -\mu_{i(x_l)})^\top  \bullet \sum_{j \in [L]\setminus l }f_{x_j}f_{x_j}^\top \leq \frac{c\cdot k^2\cdot(\epsilon/\varphi^2)^{2/3}\cdot\log(\varphi^2/\epsilon)}{n^2}.
\end{equation}

By \Cref{lemma:random_good_set}, $S$ is well-spread with probability at least 0.9999, and so by Condition \ref{con:B_delta} of \Cref{def:good_S} 

\[|S \cap B_{\delta}| \leq O\left(k\cdot (\epsilon/\varphi^2)^{1/3}\cdot   \log(\varphi^2/\epsilon)\right).\]

Define \[D \coloneqq \left\{x_l \in S: x_l\in B_{\delta} \text{ or }  (f_{x_l} -\mu_{i(x_l)})(f_{x_l} -\mu_{i(x_l)})^\top  \bullet \sum_{j \in [L]\setminus l }f_{x_j}f_{x_j}^\top > \frac{c\cdot k^2\cdot(\epsilon/\varphi^2)^{2/3}\cdot\log(\varphi^2/\epsilon)}{n^2}\right\}.\]

\noindent 
With probability at least 0.9998,  $|D| \leq O\left(k \cdot \left(\frac{\epsilon}{\varphi^2}\right)^{1/3} \log^2(\varphi^2/\epsilon)\right).$
It remains to show that $F_5 \subseteq D$. Indeed, for every $x \notin D$
\begin{align*}
    (f_{x} -\mu_{i(x)})(f_{x} -\mu_{i(x)})^\top  \bullet \sum_{y \in S }f_yf_y^\top   \leq \|f_x - \mu_{i(x)}\|^2_2\cdot\|f_x\|^2_2 + c\cdot \frac{k^2\cdot(\epsilon/\varphi^2)^{2/3}\cdot\log(\varphi^2/\epsilon)}{n^2} \\
   < \frac{\sqrt{\delta} k}{n}\cdot\|f_x\|^2_2 + c\cdot\frac{k^2\cdot(\epsilon/\varphi^2)^{2/3}\cdot\log(\varphi^2/\epsilon)}{n^2} \leq 10^{-10}\|f_x\|^4_2,
\end{align*} 
so $x \notin F_5$. Here we used that for every $x \in V\setminus B_{\delta}$ we have $\|f_x - \mu_{i(x)}\|^2_2 \leq \frac{\delta k}{n}$ and $\|f_x\|^2_2 = \Theta(k/n)$ by \Cref{def:B_delta} and by \Cref{bulletpt:f_norm}. By the setting $\delta = \Theta((\epsilon/\varphi^2)^{2/3})$ (as per \Cref{def:params}), $\delta$ is smaller than any sufficiently small constant.

\end{proof}
Using \Cref{lemma:random_good_pts_in_S}, we can now prove \Cref{lemma:random_good_cluster}, restated below for the convenience of the reader. 
\randomgoodcluster*

\begin{proof}
Let  $\mathcal F \coloneqq \{ i \in [k] :  C_i \text{ is \emph{not} well-represented by $S$}  \}$ be the set of clusters that fail to meet \Cref{def:good_cluster_wrt_S}. We will now bound the number of clusters violating each of the two conditions in \Cref{def:good_cluster_wrt_S}. 
Let $$\mathcal{F}_1 \coloneqq \{i \in [k] : S \cap C_i = \emptyset\} \qquad \text{and} \qquad \mathcal{F}_2 \coloneqq \{i \in [k] : \exists x \in S \cap C_i \text{ such that $x$ is not typical wrt. $S$} \}$$ be the sets of clusters that violating Conditions \ref{con:nonempty} and \ref{con:alltypical}, respectively.  Then $\mathcal F = \mathcal F_1 \cup \mathcal F_2$. 

First, we bound $|\mathcal{F}_1|$. The following claim achieves this. 
\begin{claim}
\label{claim:B_small}
If  $|S| \geq 100 k \log(\varphi^2/\epsilon) \cdot \eta$, then with probability at least $0.9999$, it holds that $|\mathcal{F}_1| \leq O(k \cdot \epsilon/\varphi^2)$ 
\end{claim}
\begin{proof}
Let $S$ be a (multi) set of at least $100 k \log(\varphi^2/\epsilon) \cdot \eta$ vertices sampled independently uniformly at random from $V$. Fix a cluster $C_i$, and define the random variable $X  = |S \cap C_i| $ over the randomness of $S$. Then $X$ is the sum of $|S|$ independent Bernoulli random variables, each with mean $\frac{|C_i|}{n}$, so we have $\mu \coloneqq \E[X] = |S| \cdot \frac{|C_i|}{n} \geq 100k  \log(\varphi^2/\epsilon)\eta \cdot  \frac{1/\eta}{k} = 100\log(\varphi^2/\epsilon). $  Let $\delta$ be defined by the equation  $(1-\delta)\mu = 1$, i.e. $\delta  = 1 - \frac{1}{\mu}$. Note that $\delta \geq 1/2$, since $\mu \geq 2$ for $\epsilon/\varphi^2$ is sufficiently small. An application of Chernoff bounds gives 
\begin{equation}\label{eq:D4}
\begin{aligned}
    \Pr[|S \cap C_i| < 1 ] 
    & \leq \Pr[X \leq (1-\delta) \mu ] \\
    & \leq e^{-\delta^2 \mu/2} \\
    & \leq  e^{- \mu/8} && \text{since $\delta \ge \frac{1}{2}$} \\
    & \leq  e^{- 100 \log(\varphi^2/\epsilon)/8} \\
    & \leq 0.0001 \epsilon/\varphi^2. 
\end{aligned}
\end{equation}
\noindent 
For every $i \in [k]$, define the indicator random variable $Y_i \coloneqq \1\{|C_i \cap S|< 1\}$, and let $Y \coloneqq \sum_{i \in [k]}Y_i$. Note that $Y = |\mathcal{F}_1|$. By \eqref{eq:D4}, we have 
\begin{equation*}
    \E[Y] = \sum_{i=1}^k \Pr[Y_i = 1] \leq0.0001 k \cdot \epsilon/\varphi^2.
\end{equation*}
So by Markov's inequality, 
\begin{equation*}
      \Pr[Y \geq k \cdot \epsilon/\varphi^2] \leq 0.0001,
\end{equation*} 
which gives the required bound on $|\mathcal F_1|$. 
\end{proof}
\noindent
Next, we need to bound $|\mathcal{F}_2|$. By \Cref{lemma:random_good_pts_in_S}, with probability $0.9999$, it holds that  
 \[ \left| \left\{ x \in S : \text{x is \emph{not} typical with respect to  S}\right\}\right| \leq  O \left(k \cdot \left(\frac{\epsilon}{\varphi^2}\right)^{1/3}\cdot \log(\varphi^2/\epsilon)\right). \]    
 Conditioned on this event, we have
\begin{align*}
    |\mathcal{F}_2| \leq \left| \left\{ x \in S : \text{x is \emph{not} typical with respect to  S}\right\}\right| \leq  O \left(k \cdot \left(\frac{\epsilon}{\varphi^2}\right)^{1/3}\cdot \log(\varphi^2/\epsilon)\right). 
\end{align*}
To conclude the proof, take a union bound over the failure events for $\mathcal{F}_1$ and $\mathcal{F}_2$. This gives that with probability at least $0.9998$, it holds that 
\begin{equation*}
    |\mathcal{F}| \leq   |\mathcal{F}_1| +   |\mathcal{F}_2|  \leq O(k \cdot \epsilon/\varphi^2) + O \left(k \cdot \left(\frac{\epsilon}{\varphi^2}\right)^{1/3}\cdot \log(\varphi^2/\epsilon)\right) = O \left(k \cdot \left(\frac{\epsilon}{\varphi^2}\right)^{1/3}\cdot \log(\varphi^2/\epsilon)\right). 
\end{equation*}
\end{proof}

\section{Approximating $\left \langle f_x, \sum_{y}\sigma_y f_y \right\rangle $ by random walks (Proof of \Cref{lemma:variance_calc})}\label{section:collision_counting}
In this section, we prove our collision-counting lemma 
(\Cref{lemma:variance_calc} restated below), which shows that $\left\langle M^{t_1}\1_x, \sum_{y \in S}\sigma_u \1_y \right \rangle$ be estimated by running few random walks. The lemma assumes that the probability mass assigned to $x$ under $p_S^t$ is similar to that assigned by the uniform distribution (condition {\bf (1)}) and that the $t$ step distribution of $x$ is weekly correlated with $p_S^t$ (condition {\bf (2)}). Note that we consider two different lengths of random walks $t_1,$ $t_2$. This is because we will use the lemma to argue about the variance of $\left \langle p(M) \1_x, \sum_{y \in S} p(M)\sigma_y \1_y \right \rangle = \sum_{t_1, t_2} c_{t_1}c_{t_2}\left \langle M^{t_1}\1_x, M^{t_2}\1_y \right \rangle$ in \Cref{lemma:rw_to_embedding} in \Cref{sec:sketch}. 

Suppose that we run $Q$ lazy random walks of length $t_1$ from $x$ and $R$ lazy random walks of length $t_2$ from each $y \in S$. Let $\widehat{p}_x^{t_1}$ denote the empirical distribution of the random walks from $x$, and for each $y \in S$, let $\widehat{p}_y^{t_1}$ denote the empirical distribution of the random walks from $y$. Then we have the following: 
\collisioncounting*
\begin{proof}
Let $W^1_x, \dots,  W^Q_x \sim  M^{t_1} \1_{x}$  denote the endpoints of the $Q$ random walks started from $x$, and for each $y \in S$, let $W^1_y, \dots, W^R_y \sim  M^{t_2} \1_{y} $ denote the endpoints of the $R$ random walks started from $y$.  Then $ \widehat{p}^{t_1}_x = \frac{1}{R}\sum_{r = 1}^RW_x^r$ and $\widehat p_y^{t_2} =  \frac{1}{Q}\sum_{q = 1}^Q W_y^q$ for every $y \in S$. 

We have
$$\mathbb{E}[W^q_x(v)] = M^{t_1} \1_{x}(v),\qquad \mathbb{E}[W^r_y(v)] = M^{t_2} \1_{y}(v).$$

 Let $$ Z = \frac{1}{Q\cdot R}\sum_{q=1}^{Q}\sum_{r=1}^{R}\sum_{y \in S}\sum_{v \in V} \sigma_y W^q_x(v) W^r_y(v) =  \left\langle \widehat{p}_x^{t_1}, \sum_{y \in S}\sigma_y \widehat{p}_y^{t_2}\right\rangle$$ be the random variable denoting the signed collision rates. 
 We have 

$$\mathbb{E}[Z] =  \frac{1}{Q\cdot R}\sum_{v, q,r,y}\sigma_y \E[W_x^q(v)] \E[W_y^r(v)] = \sum_{v,y }\sigma_y M^{t_1}\1_x(v) M^{t_2}\1_{y}(v) = \left \langle M^{t_1}\1_x, \sum_y \sigma_y M^{t_2} \1_{y} \right\rangle,$$

\begin{equation}\label{eq:EZ^2}
\begin{aligned}
    \mathbb{E}[Z]^2& =  \frac{1}{Q^2\cdot R^2} \sum_{q , q', r,r',y,y',v,u}\sigma_y \sigma_{y'}\E[ W^q_x(v)]\E[ W^r_y(v)] \E[W^{q'}_x(u)] \E[ W^{r'}_{y'}(u)] \\
    & =   \left( \sum_y \sigma_y \left \langle M^{t_1} \1_x, M^{t_2} \1_{y} \right\rangle  \right)^2,  \\
\end{aligned}
\end{equation}
and
 \begin{equation}\label{eqn:secondmoment}
\begin{aligned}
\mathbb{E}[Z^2] &= \frac{1}{Q^2\cdot R^2}\E\left[  \left(\sum_{q,r,y,v} \sigma_y W^q_x(v) W^r_y(v) \right) \left(\sum_{q',r',y',u} \sigma_{y'} W^{q'}_x(u) W^{r'}_{y'}(u) \right)\right] \\
& = \frac{1}{Q^2\cdot R^2} \sum_{q , q', r,r',y,y',v,u}\sigma_y \sigma_{y'}\E\left[ W^q_x(v) W^r_y(v) W^{q'}_x(u) W^{r'}_{y'}(u) \right]. 
\end{aligned}
\end{equation}
We will now upper bound the variance $\Var[Z] = \E[Z^2] - \E[Z]^2$. We do this by carefully bounding each of the terms in the expression for $\E[Z^2]$ given by the right hand side of \Cref{eqn:secondmoment}.  It will be useful to split the terms into classes depending on the size of $\{ q,q',(r,y), (r',y')\}$. This is because in terms with $|\{ q,q',(r,y), (r',y')\}| = 4$, the four variables $W^q_x(v)$, $W^r_y(v)$, $W^{q'}_x(u)$  $W^{r'}_{y'}(u)$ are independent, so  $\E\left[ W^q_x(v) W^r_y(v) W^{q'}_x(u) W^{r'}_{y'}(u) \right] = \E[ W^q_x(v)]\E[ W^r_y(v)] \E[W^{q'}_x(u)] \E[ W^{r'}_{y'}(u)]$, and we can cancel those terms against $\E[Z]^2$. On the other hand, terms with $|\{ q,q',(r,y), (r',y')\}| < 4$ do not cancel against $\E[Z]^2$, so we will bound them using the conditions in the lemma statement. 

\paragraph{Bounding terms with $|\{ q,q',(r,y), (r',y')\}| = 2$:} Then $q=q'$, $r=r'$ and $y=y'$, so we just need to consider the cases $u=v$ and $u \neq v$.
We will bound these terms by using Condition \ref{con:p^{t1+t2}_S} in the lemma statement. Since we are working with expressions of the form $M^{t_1}\1_x$ and $M^{t_2}\1_y$, it will be convenient to rewrite the condition in the following equivalent form: $\left \langle M^{t_1} \1_x , \sum_{y \in S} M^{t_2} \1_{y}\right \rangle \leq \beta \cdot \frac{k}{n}$. We now bound the contribution from the terms in \Cref{eqn:secondmoment} with $|\{ q,q',(r,y), (r',y')\}| = 2$. 

\begin{itemize}
      \item Terms with $u =v , y = y'$, $q = q', r=r'$ contribute
   \begin{align*}
       \frac{1}{Q^2\cdot R^2}\sum_{v,y,q,r}\E\left[W^q_x(v)^2 W^r_y(v)^2 \right] & = \frac{1}{Q^2 \cdot R^2}\cdot Q R \sum_v \sum_y M^{t_1}\1_x(v)M^{t_2}\1_{y}(v) \\
       & = \frac{1}{Q R}\sum_y \left \langle M^{t_1}\1_x M^{t_2}\1_{y}\right\rangle\\
       &= \frac{1}{QR} \cdot \beta  \cdot \frac{k}{n}  && \text{by Condition \ref{con:p^{t1+t2}_S} } \\
       &\leq\frac{\rho \cdot \xi^2}{7} \cdot \frac{k^2}{n^2} && \text{by assumption on $Q,R$.}
   \end{align*}
     \item   Terms with $u \neq v , y = y'$, $q = q', r=r'$ contribute $0$ because the same random walk cannot end in two different endpoints. 
\end{itemize}

\paragraph{Bounding terms with $|\{ q,q',(r,y), (r',y')\}| = 3$:}
We will bound these terms by using Condition \ref{con:p^{t1}_x, p^{t2}_S} in the lemma statement. Since we are working with expressions of the form $M^{t_1}\1_x$ and $M^{t_2}\1_y$, it will be convenient to rewrite the condition in the following equivalent form: 
 $$\sum_v \sum_{y, y' \in S }   M^{t_1} \1_x(v)   M^{t_2} \1_{y}(v)   M^{t_2} \1_{y'}(v) \leq \gamma \cdot \frac{k^2}{n^2}  \text{ and }\sum_v \sum_{y\in S }   M^{t_1} \1_x(v)^2   M^{t_2} \1_{y}(v)   \leq  \gamma \cdot \frac{k^2}{n^2}.$$
We now bound the contribution from the terms in \Cref{eqn:secondmoment} with $|\{ q,q',(r,y), (r',y')\}| = 3$. 
\begin{itemize}
    \item   Terms with $u =v , y = y'$, $q \neq q', r=r'$ contribute
        \begin{align*}
        \frac{1}{Q^2 \cdot R^2} \sum_{v,y,q \neq q',r}\E \left[ W^q_x(v) W^{q'}_x(v) W^r_y(v)^2\right] & \leq 
       \frac{1}{Q^2 \cdot R^2}\cdot Q^2 R \sum_v \sum_y M^{t_1}\1_x(v)^2 M^{t_2}\1_{y}(v) \\
       & \leq \frac{1}{R} \cdot \gamma \cdot \frac{k^2}{n^2}  \qquad \qquad \qquad \qquad \qquad   \text{by Condition \ref{con:p^{t1}_x, p^{t2}_S}} \\
        & \leq \frac{\rho \cdot \xi^2}{7}\cdot \frac{k^2}{n^2}  \qquad \qquad \qquad \qquad \qquad  \text{by the assumption on $R$. }
    \end{align*}
    \item   Terms with $u =v , y=y'$, $q=q', r \neq r'$ contribute
        \begin{align*}
        \frac{1}{Q^2 \cdot R^2} \sum_{v,y,q,r \neq r'}\E \left[ W^q_x(v)^2  W^r_y(v)W^{r'}_y(v)\right] & \leq 
       \frac{1}{Q^2 \cdot R^2}\cdot Q R^2 \sum_v \sum_y  M^{t_1}\1_x(v) M^{t_2}\1_{y}(v)^2 \\
       & \leq \frac{1}{Q}\cdot  \gamma \cdot \frac{k^2}{n^2}  \qquad \qquad  \qquad \qquad \qquad   \text{by Condition \ref{con:p^{t1}_x, p^{t2}_S}} \\
        & \leq \frac{\rho \cdot \xi^2}{7} \cdot \frac{k^2}{n^2} \qquad \qquad  \qquad   \qquad \qquad  \text{by the assumption on $Q$. }
        \end{align*}
    \item   Terms with $u =v , y \neq y'$, $q=  q' $ contribute
    \begin{align*}
  \frac{1}{Q^2 \cdot R^2}\sum_{v,y\neq y',q, r,r'} \sigma_y \sigma_{y'}\E \left[W^q_x(v)^2 W^{r}_y(v) W^{r'}_{y'}(v)  \right] &=
        \frac{1}{Q} \sum_{y,y': y \neq y'} \sigma_y \sigma_{y'} \sum_v M^{t_1}\1_x(v)  M^{t_2}\1_{y}(v)  M^{t_2}\1_{y'}(v) \\
        & \leq \frac{1}{Q} \sum_{y, y': y \neq y'} \sum_v M^{t_1}\1_x(v)  M^{t_2}\1_{y}(v)  M^{t_2}\1_{y'}(v)  \\
        & \leq  \frac{1}{Q} \cdot \gamma \cdot \frac{k^2}{n^2} \qquad \qquad  \qquad \text{by Condition \ref{con:p^{t1}_x, p^{t2}_S}} \\
        & \leq \frac{\rho \cdot \xi^2}{7} \cdot \frac{k^2}{n^2} \qquad \qquad  \qquad  \text{by the assumption on $R$.} 
    \end{align*}
     \item   Terms with $u \neq v$ contribute $0$ if $q= q'$ or $(r,y) = (r',y')$ because the same random walk cannot end in two different endpoints. 
\end{itemize}
\paragraph{Bounding terms with $|\{ q,q',(r,y), (r',y')\}| = 4$:}
Most of the terms with $|\{ q,q',(r,y), (r',y')\}| = 4$ will cancel out against $\E[Z^2]$. However, there will be some small left over lower-order terms which do not cancel out against  $\E[Z^2]$. We will use the conditions from the lemma statement to show that these left over terms can be upper bounded by $O\left(\rho \cdot \xi^2 \cdot \frac{k^2}{n^2}\right)$. 

\begin{itemize}
    \item   Terms with $u =v , y=y'$, $q \neq q', r \neq r' $ contribute
     \begin{align*}
     \frac{1}{Q^2 \cdot R^2} \sum_{v,y,q,r \neq r'}\E \left[ W^q_x(v)W^{q'}_x(v)   W^r_y(v) W^{r'}_y(v)\right] & \leq 
     \sum_y \sum_v M^{t_1} \1_x(v)^2 M^{t_2} \1_{y}(v)^2
     \end{align*}
     \item    Terms with $u =v , y \neq y'$, $q \neq  q'  $ contribute
    \begin{align*}
   &  \frac{1}{Q^2 \cdot R^2}\sum_{v,y \neq y',q\neq q', r,r'} \sigma_y \sigma_{y'}\E \left[W^q_x(v)W^{q'}_x(v) W^r_y(v)W^{r'}_{y'}(v)\right]\\
   & = \frac{Q-1}{Q}
    \sum_{y,y': y \neq y'}\sigma_y \sigma_{y'}  \sum_v M^{t_1}\1_x(v)^2 M^{t_2} \1_{y}(v) M^{t_2}\1_{y'}(v) \\
    & \leq  \sum_{y,y': y \neq y'}\sigma_y \sigma_{y'}  \sum_v M^{t_1}\1_x(v)^2 M^{t_2} \1_{y}(v) M^{t_2}\1_{y'}(v) + \frac{1}{Q}  \sum_{y,y': y \neq y'}\sum_v M^{t_1}\1_x(v)^2 M^{t_2} \1_{y}(v) M^{t_2}\1_{y'}(v)  \\
    & \leq  \sum_{y,y': y \neq y'}\sigma_y \sigma_{y'}  \sum_v M^{t_1}\1_x(v)^2 M^{t_2} \1_{y}(v) M^{t_2}\1_{y'}(v) +  \frac{1}{Q} \sum_{y,y'}\sum_v M^{t_1}\1_x(v) M^{t_2} \1_{y}(v) M^{t_2}\1_{y'}(v)  \\
    & \leq \sum_{y,y': y\neq y'}\sigma_y \sigma_{y'}  \sum_v M^{t_1}\1_x(v)^2 M^{t_2} \1_{y}(v) M^{t_2}\1_{y'}(v)  + \frac{\rho \cdot \xi^2 }{7}\cdot \frac{k^2}{n^2}  \quad \text{by Condition \ref{con:p^{t1}_x, p^{t2}_S} and assumption on $Q$}.
    \end{align*}
      \item Terms with $u \neq v$ and $q \neq q'$ and $(r,y) \neq (r',q')$ contribute
     \begin{align*}
 & \frac{1}{Q^2 \cdot R^2}\sum_{u \neq v,q\neq q', (r,y) \neq (r',y')} \sigma_y \sigma_{y'}\E \left[W^q_x(v)W^{q'}_x(u) W^{r}_{y}(v) W^{r'}_{y'}(u)  \right] \\
 & = \frac{Q-1}{Q}\cdot \frac{R^2-1}{R^2}
     \sum_{y, y'} \sigma_y \sigma_{y'}  \sum_{u,v: u \neq v} M^{t_1}\1_x(v)M^{t_1}\1_x(u)   M^{t_2} \1_{y}(v) M^{t_2}\1_{y'}(u) \\
     & \leq \sum_{y, y'} \sigma_y \sigma_{y'}  \sum_{u,v: u \neq v} M^{t_1}\1_x(v)M^{t_1}\1_x(u)   M^{t_2} \1_{y}(v) M^{t_2}\1_{y'}(u)\\
     & + 2\left(  \frac{1}{Q} + \frac{1}{R^2}\right) \sum_{y, y'}  \sum_{u,v: u \neq v} M^{t_1}\1_x(v)M^{t_1}\1_x(u)   M^{t_2} \1_{y}(v) M^{t_2}\1_{y'}(u) \\
     & \leq 
     \sum_{y, y'} \sigma_y \sigma_{y'}  \sum_{u,v: u \neq v} M^{t_1}\1_x(v)M^{t_1}\1_x(u)   M^{t_2} \1_{y}(v) M^{t_2}\1_{y'}(u) +  2\left(  \frac{1}{Q} + \frac{1}{R^2}\right) \left \langle M^{t_1}\1_x, \sum_{y \in S}M^{t_2}\1_y \right \rangle^2 \\
     & \leq 
     \sum_{y, y'} \sigma_y \sigma_{y'}  \sum_{u,v: u \neq v} M^{t_1}\1_x(v)M^{t_1}\1_x(u)   M^{t_2} \1_{y}(v) M^{t_2}\1_{y'}(u) + 2\left(  \frac{1}{Q} + \frac{1}{R^2}\right) \cdot \beta ^2 \cdot \frac{k^2}{n^2}  \qquad \text{by Condition \ref{con:p^{t1+t2}_S}} \\
     & \leq 
     \sum_{y, y'} \sigma_y \sigma_{y'}  \sum_{u,v: u \neq v} M^{t_1}\1_x(v)M^{t_1}\1_x(u)  M^{t_2} \1_{y}(v) M^{t_2}\1_{y'}(u) +  \frac{2 \rho \cdot \xi^2}{7} \cdot \frac{k^2}{n^2} \qquad  \qquad \text{by assumption on $Q, R$. } 
     \end{align*}
\end{itemize}

\paragraph{Combining all the terms}
Combining all of the different types of terms, we obtain 
\begin{equation}\label{eq:2ndmoment}
\begin{aligned}
\E[Z^2] &= 7 \cdot \frac{\rho \cdot  \xi^2 }{7}\frac{k^2}{n^2} + \sum_y \sum_v M^{t_1} \1_x(v)^2 M^{t_2} \1_{y}(v)^2+  \sum_{y,y': y \neq y'} \sigma_y \sigma_{y'} \sum_v M^{t_1}\1_x(v)^2 M^{t_2} \1_{y}(v)M^{t_2} \1_{y'}(v) \\
+& \sum_{y, y'} \sigma_y \sigma_{y'}  \sum_{u,v: u \neq v} M^{t_1}\1_x(v)M^{t_1}\1_x(u)   M^{t_2} \1_{y}(v) M^{t_2}\1_{y'}(u) \\
& =\rho\cdot  \xi^2 \frac{k^2}{n^2}  + \sum_{y,y',u,v} \sigma_{y} \sigma_{y'}  M^{t_1}\1_x(v)M^{t_1}\1_x(u)   M^{t_2} \1_{y}(v)M^{t_2} \1_{y'}(u)  \\
& = \rho\cdot \xi^2 \frac{k^2}{n^2} +  \left \langle M^{t_1} \1_x, \sum_y \sigma_y M^{t_2} \1_{y}\right \rangle^2. 
\end{aligned}
\end{equation}
Combining Equations \eqref{eq:EZ^2} and  \eqref{eq:2ndmoment}, we get
\begin{align*}
\Var[Z] & = \E[Z^2] - \E[Z]^2 \\
& \leq \rho\cdot \xi^2 \frac{k^2}{n^2} +  \left \langle M^{t_1}\1_x, \sum_y \sigma_y M^{t_2} \1_{y}\right\rangle^2 - \left \langle M^{t_1}\1_x, \sum_y \sigma_y M^{t_2} \1_{y}\right\rangle^2  \\
& =  \rho\cdot  \xi^2 \frac{k^2}{n^2}. 
\end{align*}
By Chebyshev's inequality, we have

$$ \Pr \left[\left|Z -  \left \langle M^{t_1} \1_x, \sum_{l} \sigma_y M^{t_2} \1_{y} \right \rangle\right| > \xi \cdot \frac{k}{n}\right] \leq \frac{\rho \cdot  \xi^2 \frac{k^2}{n^2}}{\xi^2 \frac{k^2}{n^2}} = \rho,$$
as required. 
\end{proof}

\section{Proof of \Cref{thm:find_k} and \Cref{thm:lowerbnd}}\label{sec:appx_sqrtnk}
Throughout this section, we assume that $k$ is bounded away from $n$, namely that $k \leq n^{0.999}$. We start by proving \Cref{thm:lowerbnd}, restated below. 

\lowerbnd*

\Cref{thm:lowerbnd} is a statement about algorithms which output \textit{one-sided} $(2 - \Omega(1))$\textit{-approximation} to the number of clusters $k$:

\begin{defn}[One-sided $\alpha$-approximation] Let $\text{OBJ}$ denote the true value of a quantity of interest. We say that an algorithm is a \emph{one-sided $\alpha$-approximation algorithm} (for $\alpha \ge 1$) if it outputs $\hat{\text{OBJ}}$ satisfying
\[
\text{OBJ} \;\le\; \hat{\text{OBJ}} \;\le\; \alpha \cdot \text{OBJ}.
\]  
\end{defn}

\Cref{thm:lowerbnd} follows as an application of the following result in \cite{chiplunkar2018testing}. 
\begin{thm}[Theorem 5 in \cite{chiplunkar2018testing}]\label{thm:unclusterable}
 There exists an input distribution that is supported on
\begin{enumerate}
    \item The YES case: a union of two $\Omega(1)$-expanders on $n/2$ vertices each, each inducing a cut of sparsity $\e$.
    \item The NO case: a single $\Omega(1)$-expander on $n$ vertices 
\end{enumerate}
such that any algorithm that distinguishes between YES and NO cases with probability at least $2/3$ must make at least $n^{1/2+\Omega(\epsilon)}$ queries. 
\end{thm}

\begin{proof}[Proof of \Cref{thm:lowerbnd}]
   Let $G$ be a disjoint union of $k$ identical copies of an instance sampled from the hard distribution from \Cref{thm:unclusterable}, instantiated on $\frac{n}{k}$ vertices each. Then querying $G$ is in one-to-one correspondence with querying the hard instance.
   
   In the YES case, $G$ consists of $2k$ clusters on $\frac{n}{2k}$ vertices each, and in particular is a $(2k, \Omega(1), \epsilon)$-clusterable graph. In the NO case, $G$ consists of $k$ disjoint expanders on $\frac{n}{k}$ vertices each, and in particular, $G$ is $(k, \Omega(1), \epsilon)$-clusterable. So any algorithm  $\mathcal{A}$ that computes a $(2-\Omega(1))$-approximation to the number of clusters in $G$, distinguishes between the YES case and the NO case, and can therefore be used to distinguish between a single $\Omega(1)$-expander on $n/k$ vertices and two $\Omega(1)$-expanders on 
   $n/2k$ vertices each. So by \Cref{thm:unclusterable}, $\mathcal{A}$ must make at least $\left( n/k \right)^{1/2+\Omega(\epsilon)}$ queries.    
   \end{proof}

Next, we present Algorithm \ref{alg:find_k} and prove its performance guarantees formalized in \Cref{thm:find_k}.
\findK*

 \begin{algorithm}[H]
\caption{\findk$(\varepsilon)$}
\label{alg:find_k}
    \begin{algorithmic}[1]
    \State $L \gets 10^{-3}\cdot\left(\frac{\varphi^2}{\epsilon}\right)^{1/2}$
    \State $x_1, \ldots, x_{L} \sim \text{Unif}\left(V^{L}\right)$
\For{all $l \in [L]$}
\State $\sk_{l, 1}, \sk_{l, 2} \gets$ two independent trials of $\textsc{Sketch}(x_l)$ with more random walks, see Corollary \ref{cor:adasketch} 
\EndFor

\State  \Return{$\frac{1}{L}\sum_{l \in [L]} \langle\sk_{l, 1}, \sk_{l, 2}\rangle$}
\end{algorithmic}
\end{algorithm}

\begin{lemma}\label{lem:apx_k/n_||f_x||}
Let $F = x_1, x_2, \ldots, x_{L} \sim \operatorname{Unif}(V^L)$, where $\varepsilon \geq  C\cdot\left(\frac{\epsilon}{\varphi^2}\right)^{1/4}$ for a sufficiently big constant $C$ and $L = \frac{1}{2}\cdot 10^{-3}\cdot\left(\frac{\varphi^2}{\epsilon}\right)^{1/2}$. Then, with probability $0.999$, 
\[\left|\frac{1}{L}\sum_{l \in [L]}\|f_{x_l}\|^2_2 - \frac{k}{n}\right| \leq \varepsilon\cdot\frac{k}{n}.\]
\end{lemma}
\begin{proof}
Recall that by \Cref{lemma:close_to_clutermean}, we have $|B_{\delta}| \leq O\left(n\cdot \frac{1}{\delta}\cdot\frac{\epsilon}{\varphi^2}\right)$. Set $\delta$ to be $ \frac{(\epsilon/\varphi^2)^{1/2}}{\kappa}$ for a sufficiently large constant $\kappa$ just so that $|B_{\delta}| \leq n\cdot \left(\frac{\epsilon}{\varphi^2}\right)^{1/2}$. Define $\mathcal E$ to be the event  $\mathcal{E} \coloneqq \left\{F \cap B_{\delta} = \emptyset\right\}$. 

\[\Pr_{F}[\mathcal{E}] = \left(1 - \frac{|B_{\delta}|}{n}\right)^{L} \geq 1 - L\cdot\frac{|B_{\delta}|}{n} \geq  1 -L\cdot \left(\frac{\epsilon}{\varphi^2}\right)^{1/2} \geq 1-\frac{1}{2}\cdot 10^{-3},\] where the last inequality follows by the choice $L = \frac{1}{2}\cdot 10^{-3}\cdot\left(\frac{\varphi^2}{\epsilon}\right)^{1/2}$.

Conditioned on the event \(\mathcal{E}\), the random variables \(x_l \mid \mathcal{E}\) correspond to the original samples \(x_l\) restricted to the case where they lie outside the set \(B_\delta\); equivalently, they can be viewed as samples drawn from \(V \setminus B_\delta\). Under this conditioning, the collection \(\{x_l \mid \mathcal{E}\}_{l \in [L]}\) is independent and uniformly distributed over \(V \setminus B_\delta\). We have

\begin{equation}\label{eq:exp_upper}
    \begin{aligned}
        \E_{\{x_l\}|\mathcal{E}}\left[\frac{1}{L}\sum_{l = 1}^{L}\|f_{x_l}\|^2_2 \right] &\leq \frac{1}{|V \setminus B_{\delta}|}\sum_{v \in V\setminus B_{\delta}}\frac{1 + 4\frac{\sqrt{\epsilon}}{\varphi} + 4\sqrt{\delta}}{|C_{i(v)}|} \qquad \text{ by \Cref{remark:norm}}\\
&\leq \sum_{i \in [k]}\frac{|C_i \setminus B_{\delta}|}{|V \setminus B_{\delta}|}\cdot \frac{1 + 4\frac{\sqrt{\epsilon}}{\varphi} + 4\sqrt{\delta}}{|C_{i}|}. \\
    \end{aligned}
\end{equation}

Similarly, 
\begin{equation}\label{eq:exp_lower}
    \begin{aligned}\E_{\{x_l\}|\mathcal{E}}\left[\frac{1}{L}\sum_{l = 1}^{L}\|f_{x_l}\|^2_2 \right] &\geq \frac{1}{|V \setminus B_{\delta}|}\sum_{v \in V\setminus B_{\delta}}\frac{1 - 4\frac{\sqrt{\epsilon}}{\varphi} - 4\sqrt{\delta}}{|C_{i(v)}|} \qquad \text{ by \Cref{remark:norm}}\\
&\geq \sum_{i \in [k]}\frac{|C_i \setminus B_{\delta}|}{|V \setminus B_{\delta}|}\cdot \frac{1 - 4\frac{\sqrt{\epsilon}}{\varphi} - 4\sqrt{\delta}}{|C_{i}|}. \\
\end{aligned}
\end{equation}

To refine these upper and lower bounds, we must first bound $\sum_{i \in [k]}\frac{|C_i\setminus B_{\delta}|}{|C_i |}$. Let $C_{j_{min}}$ be the cluster of the smallest size. Note that 
\[\sum_{i \in [k]\setminus j_{min}}\frac{|C_i|}{|C_i|} - \frac{|B_{\delta}|}{|C_{j_{min}}|} \leq \sum_{i \in [k]}\frac{|C_i\setminus B_{\delta}|}{|C_i |} \leq \sum_{i \in [k]}\frac{|C_i|}{|C_i|}\leq k.\]

Since $|C_{j_{min}}| \geq n/(\eta k)$, where $\eta \in O(1)$ is the maximum ratio between cluster sizes, we have that  $\frac{|B_{\delta}|}{|C_{j_{min}}|} \leq \alpha\cdot k \cdot \frac{1}{\delta}\cdot\frac{\epsilon}{\varphi^2}$ and$|B_{\delta}| \leq \alpha\cdot n \cdot \frac{1}{\delta}\cdot \frac{\epsilon}{\varphi^2}$ for a sufficiently big constant $\alpha$. Consequently,

\[k\cdot \left(1 - \alpha \cdot \frac{1}{\delta}\cdot\frac{\epsilon}{\varphi^2}\right) \leq \sum_{i \in [k]}\frac{|C_i\setminus B_{\delta}|}{|C_i |} \leq k.\]
Furthermore, note that $n\cdot \left(1 - \alpha\cdot \frac{1}{\delta}\cdot \frac{\epsilon}{\varphi^2}\right)\leq |V\setminus B_{\delta}| \leq n$. Plugging these bounds into \Cref{eq:exp_upper} and \Cref{eq:exp_lower} we get

\[\frac{k}{n}\cdot\left(1 - \frac{4\sqrt{\epsilon}}{\varphi} - 4\sqrt{\delta}\right)\cdot\left(1 -  \alpha\cdot\frac{1}{\delta}\cdot\frac{\epsilon}{\varphi^2}\right)\leq\E_{\{x_l\}|\mathcal{E}}\left[\sum_{l = 1}^{L}\frac1{L}\|f_{x_l}\|^2_2 \right] \leq \frac{k}{n}\cdot \frac{1 + \frac{4\sqrt{\epsilon}}{\varphi} + 4\sqrt{\delta}}{1 -  \alpha\cdot \frac{1}{\delta}\cdot\frac{\epsilon}{\varphi^2}}.\]
Since $\delta = \frac{(\epsilon/\varphi^2)^{1/2}}{\kappa}$, other terms in LHS and RHS are negligible compared to $4\sqrt{\delta}$: 
\[\frac{k}{n}\cdot \left(1 - \left(\frac{\epsilon}{\varphi^2}\right)^{1/4}\right) \leq \E_{\{x_l\}|\mathcal{E}}\left[\sum_{l = 1}^{L}\frac1{L}\|f_{x_l}\|^2_2 \right] \leq \frac{k}{n}\cdot\left(1 + \left(\frac{\epsilon}{\varphi^2}\right)^{1/4}\right).\]

Before we can apply Hoeffding's inequality, it remains to note that for all $v \in V\setminus B_{\delta}$, by \Cref{remark:norm}, $\|f_v\|^2_2 = \Theta(k/n)$. Therefore, all of the independent random variables $\|f_{x_1}\|^2_2, \ldots, \|f_{x_l}\|^2_2$ take values in a bounded region. By Hoeffding's inequality, we now get

\begin{align*} \Pr\left[\left|\frac1{L}\sum_{l \in [L]}\|f_{x_l}\|^2_2 - \frac{k}{n}\right| \geq \varepsilon\frac{k}{n}\middle|\mathcal{E}\right]&\leq \Pr\left[\left|\frac1{L}\sum_{l \in [L]}\|f_{x_l}\|^2_2 - \E_{\{x_l\}|\mathcal{E}}\left[\sum_{l = 1}^{L}\frac1{L}\|f_{x_l}\|^2_2 \right]\right| \geq \left(\varepsilon - \left(\frac{\epsilon}{\varphi^2}\right)^{1/4}\right)\frac{k}{n}\middle|\mathcal{E}\right] \\
&\leq e^{-\Omega(\varepsilon'^2\cdot L),}
\end{align*}
where, for conciseness, we used notation $\varepsilon' = \varepsilon - \left(\frac{\epsilon}{\varphi^2}\right)^{1/4} \geq \frac{1}{2}\varepsilon$. In order to ensure  $e^{-\Omega(\varepsilon'^2\cdot L)} \leq \frac{1}{2}\cdot 10^{-3}$
it suffices to select $L \geq \frac{C'}{\varepsilon^2}$ for a sufficiently big constant $C'$. Note that this inequality holds for setting $L = \frac{1}{2}\cdot 10^{-3}\cdot\left(\frac{\varphi^2}{\epsilon}\right)^{1/2}.$

Finally, 

\[\Pr\left[\left|\frac1{L}\sum_{l \in [L]}\|f_{x_l}\|^2_2 - \frac{k}{n}\right| \geq \varepsilon\frac{k}{n}\right] \leq \Pr\left[\left|\frac1{L}\sum_{l \in [L]}\|f_{x_l}\|^2_2 - \frac{k}{n}\right| \geq \varepsilon\frac{k}{n}\middle|\mathcal{E}\right] + 1-\Pr[\mathcal{E}]\leq \frac{1}{2}\cdot10^{-3} + \frac{1}{2}\cdot10^{-3} \leq 10^{-3}.\]

\end{proof}

Now we would like to say that \sketch $(x_1), \ldots, $\sketch$(x_{L})$ with high probability provide good enough approximations to $\|f_{x_1}\|^2_2, \ldots, \|f_{x_{L}}\|^2_2$. Our guarantees on the performance of \sketch (Algorithm \ref{alg:spectralsketch}), unfortunately, are insufficient ~-- \sketch \ only provides a constant approximation to $\|f_x\|^2_2$ with constant probability. 

\newcommand{\adaptedsketch}{\textsc{Sketch}}
Recall, however, from the proof of \Cref{lemma:sketchx} in Appendix \ref{sec:spectral_facts} that the probability of success and the quality of  approximation of \sketch \ are functions of the number of the random walks which we run to generate \sketch. So, we modify \sketch \  by only updating the number of random walks:
\begin{cor}\label{cor:adasketch} Let $p > 0$ be the desired failure probability, and set $T = \frac{1}{p\cdot \varepsilon^2}$, $\varepsilon \geq \left(\frac{\epsilon}{\varphi^2}\right)^{1/4}$.
 The procedure \adaptedsketch \ (Algorithm \ref{alg:spectralsketch}) modified to run $T$ times more random walks has \newline  runtime  $O^*\left( \left(\frac{n}{k}\right)^{1/2 + O(\epsilon/\varphi^2 \log(1/\epsilon))} \cdot \left(1/\epsilon\right)^{O(\log(\varphi^2/\epsilon))}\cdot T\right)$, returns a vector of \newline   support at most   $O^*\left( \left(\frac{n}{k}\right)^{1/2 + O(\epsilon/\varphi^2 \log(1/\epsilon))} \cdot \left(1/\epsilon\right)^{O(\log(\varphi^2/\epsilon))}\cdot T\right)$ and has the following guarantee:  
 
Set $\delta$ to be $ \frac{(\epsilon/\varphi^2)^{1/2}}{\kappa}$, just as in \Cref{lem:apx_k/n_||f_x||}. For every $x \in V \setminus  B_{\delta}$, with probability at least $1 - p$ (over the internal randomness of $\sketch$) it holds that 
    $$\left|\langle \adaptedsketch(x), \adaptedsketch(x)\rangle  - \|f_x\|^2_2 \right| \leq \varepsilon\frac{k}{n}. $$
In the above, each \sketch$(x)$ is an independent instance. 
\end{cor}
\begin{proof}
The runtime of the \adaptedsketch$(x)$ is $T$ times higher than before, since the only adaptation we made is the number of random walks. Similarly, \adaptedsketch$(x)$ returns vectors of support no more than $T$ times bigger than before: $O^*\left( \left(\frac{n}{k}\right)^{1/2 + O(\epsilon/\varphi^2 \log(1/\epsilon))} \cdot \left(1/\epsilon\right)^{O(\log(\varphi^2/\epsilon))}\cdot T\right)$.

The proof of the performance guarantee follows the proof of \Cref{lemma:sketchx}. Just as there, we use \Cref{lemma:variance_calc} with parameter settings $\beta = O(1)$, $\gamma = O\left( \sqrt{\frac{n}{k}}\right)$ (the same as in \Cref{lemma:sketchx}), and different $\rho = \frac{p}{(t_{\Delta}+1)^2}$,  $\xi = \frac{\varepsilon}{|c_{t_1}||c_{t_2}|(t_{\Delta}+1)^2} =O^*\left( \frac{\varepsilon}{(1/\epsilon)^{t_{\Delta}}}\right)$, where $c_{t_1}$ and $c_{t_2}$ are the coefficients of $x^{t_1}$ and $x^{t_2}$ in the polynomial $p(x)$ from \Cref{thm:standardbasis}. Recall that in \Cref{lemma:variance_calc} that  $R, Q$ are the number of the random walks. Setting $R, Q = O \left(\sqrt{\frac{n}{k}} \cdot  \frac{1}{\rho \cdot  \xi^2} \right) =  O^*\left(\sqrt{\frac{n}{k}} \cdot \left(\frac{1}{\epsilon}\right)^{ O(t_{\Delta})}\cdot \frac{1}{p\cdot\varepsilon^2}\right)$, opposed to $R, Q =  O^*\left(\sqrt{\frac{n}{k}} \cdot \left(\frac{1}{\epsilon}\right)^{ O(t_{\Delta})}\right)$ used before, and using the fact that $\|f_x\|^2_2 = \Theta(\frac{k}{n})$ for all $x \in  V \setminus B_{\delta}$ (see \Cref{remark:norm}),
 we obtain

 \begin{equation}
\Pr_{\text{random walks}}\left[\left| \left\langle \sigma_x\widehat{p}_x^{t_1}, \sigma'_x \widehat{p}_x^{t_2}\right \rangle  - \left\langle \sigma_x M^{t_1} \1_x, \sigma'_x M^{t_2} \1_x \right \rangle\right| \leq \frac{\varepsilon}{|c_{t_1}| |c_{t_2}|(t_{\Delta}+1)^2} \cdot \frac{k}{n} \right] \geq 1- \frac{p}{(t_{\Delta}+1)^2}
\end{equation}
for every pair of lengths of random walks $t_1, t_2$ used by \adaptedsketch. By taking a union bound over the $(t_{\Delta} + 1)^2$ pairs of $t_1, t_2 \in [t_{\min}, t_{\min} + t_{\Delta}]$, we get that 

 \begin{align*}
     \Pr_{\text{random walks}} \left[ \left| \langle \adaptedsketch(x) , \adaptedsketch(x) \rangle -   \left\langle \sigma_x \sum_t c_t \cdot  M^t \1_x,\sigma'_x \sum_{t}    c_t  \cdot M^t \1_x\right \rangle\right| > \varepsilon\cdot\frac{k}{n}\right] \leq p.
 \end{align*}
Using a bound proven in \Cref{claim:fx_vs_chebyshev}, 
\[\left|\left\|\sum_t c_t M^t\1_x  \right\|^2_2 - \|f_x\|^2_2 \right| \leq \frac{2\epsilon}{\varphi^2} \|f_x\|^2_2 + n^{-8} \leq \varepsilon\frac{k}{n},\]
where the last inequality follows from setting $\varepsilon \geq \left(\frac{\epsilon}{\varphi^2}\right)^{1/4}$.
Combining the two and rescaling $\varepsilon \to \varepsilon/2$ we get
 \begin{align*}
     \Pr_{\text{random walks}} \left[ \left| \langle \adaptedsketch(x) , \adaptedsketch(x) \rangle - \|f_x\|^2_2\right| > \varepsilon\cdot\frac{k}{n}\right] \leq p,
 \end{align*} as desired.
 
\end{proof}

\begin{proof}[Proof of \Cref{thm:find_k}.] The proof of \Cref{thm:find_k} is a combination of \Cref{lem:apx_k/n_||f_x||} and  \Cref{cor:adasketch}. Since we will use both, let us denote the precision parameter with which we'll invoke \Cref{lem:apx_k/n_||f_x||} as $\varepsilon_1$, and the precision and probability of failure parameters of \Cref{cor:adasketch} as $\varepsilon_2, p_2$. 

Suppose that the statement of \Cref{lem:apx_k/n_||f_x||} holds and that for all pairs of trials of \adaptedsketch \,  \Cref{cor:adasketch} holds. Then, 

\begin{align*}
\left|\sum_{l \in [L]}\frac1{L} \langle\sk_{l, 1}, \sk_{l, 2}\rangle - \frac{k}{n}\right| 
\leq \left|\frac{1}{L}\sum_{l \in [L]}\|f_{x_l}\|^2_2 - \frac{k}{n}\right| + \frac{1}{L}\sum_{l \in [L]}\left|\langle \sk_{l,1}, \sk_{l,2}\rangle  - \|f_{x_l}\|^2_2\right| \leq
\varepsilon_1\frac{k}{n} + \varepsilon_2\frac{k}{n}.
\end{align*}
If we set $\varepsilon_1 = \varepsilon_2 = \varepsilon/2$, we get the desired approximation guarantee. Condition $\varepsilon_2 \geq \left(\frac{\epsilon}{\varphi^2}\right)^{1/4}$ necessary in Corollary \ref{cor:adasketch} follows from $\varepsilon_2 = \frac{\varepsilon}{2} \geq \frac{C}{2}\cdot \left(\left(\frac{\epsilon}{\varphi^2}\right)^{1/4}\right)$.

\paragraph{Probability of success.} Algorithm \findk \ runs $2L$ independent instances of \adaptedsketch. If each of them succeeds with probability $1 - p_2$, then all of them succeed with probability $(1 - p_2)^{2L}$. We also require that the statement of \Cref{lem:apx_k/n_||f_x||} holds. Therefore, the probability of success of Algorithm \ref{alg:find_k} is lower bounded by 
\[(1 - p_2)^{L} - 10^{-3} \geq 1 - 2p_2\cdot L - 10^{-3}.\]
In order to use \Cref{lem:apx_k/n_||f_x||} we need to verify $\varepsilon_1 \geq  C\cdot\left(\left(\frac{\epsilon}{\varphi^2}\right)^{1/4}\right)$~-- this follows from the analogous property for $\varepsilon$ . Then, $L = \frac{1}{2}\cdot 10^{-3}\cdot \left(\frac{\varphi^2}{\varepsilon}\right)^{1/2}$. Set $p_2 = 10^{-3}/(2L)$. Then $p_2 = \frac{1}{4}\left(\frac{\epsilon}{\varphi^2}\right)^{1/2}$, and \findk \ has the desired success probability at least $0.999$.

\paragraph{Runtime.} By Corollary \ref{cor:adasketch}, each trial of \adaptedsketch \ takes $O^*\left( \left(\frac{n}{k}\right)^{1/2 + O(\epsilon/\varphi^2 \log(1/\epsilon))} \cdot \left(1/\epsilon\right)^{O(\log(\varphi^2/\epsilon))}\cdot \frac{1}{p_2\cdot\varepsilon^2_2}\right)$. It takes  $O^*\left( \left(\frac{n}{k}\right)^{1/2 + O(\epsilon/\varphi^2 \log(1/\epsilon))} \cdot \left(1/\epsilon\right)^{O(\log(\varphi^2/\epsilon))}\cdot \frac{1}{p_2\cdot\varepsilon^2_2}\right)$ to compute the inner product of two trials of \adaptedsketch. Therefore, the total runtime of \findk \ is
 $$O^*\left( \left(\frac{n}{k}\right)^{1/2 + O(\epsilon/\varphi^2 \log(1/\epsilon))} \cdot \left(1/\epsilon\right)^{O(\log(\varphi^2/\epsilon))}\cdot \frac{L}{p_2\cdot\varepsilon^2_2}\right) = O^*\left( \left(\frac{n}{k}\right)^{1/2 + O(\epsilon/\varphi^2 \log(1/\epsilon))} \cdot \left(1/\epsilon\right)^{O(\log(\varphi^2/\epsilon))}\cdot \frac{1}{\varepsilon^2}\right)$$

\end{proof}

\section{Proof sketch of \Cref{thm:tradeoff}}\label{sec:tradeoff_sketch}
In this section, we discuss how to modify our main analysis in order to get the trade-offs, which are restated below. 
\tradeoffs*

At a high level, the trade-off holds because it is the \emph{product} of the number of random walks from a vertex $x$, and the number of random walks from each vertex $y \in S$ that determines how accurately we can estimate the quantity $\left\langle M^t \1_x, \sum_{y \in S} \sigma_u M^t \1_y \right \rangle $ (see \Cref{lemma:variance_calc}). Consequently, we can decrease the number of random walks from the query vertex by increasing the number of random walks from the representative vertices, and vice versa.

We now describe the required algorithmic modifications and then outline the corresponding changes in the analysis.
\paragraph{Algorithm description.}
Run $\findclustermeans(\widehat k)$ without any modifications to obtain a set of representative vertices $R$. For every $y \in R$, compute a tree of sketches (as per \Cref{def:tree}), but change the number of random walks in  line \ref{line:sketch_r} of \Cref{alg:spectralsketch} to
$$r =  O^*\!\left( (n/k)^{1-\delta} n^{O(\epsilon/\varphi^2 \log(1/\epsilon))}\cdot
(1/\epsilon)^{O(\log(\varphi^2/\epsilon))}\right)$$

Given a query vertex $x$, when computing $\sketch(x)$, the algorithm \dotproduct{} (\Cref{alg:dotproduct}) should use 
$$q=O^*\!\left( (n/k)^{\delta}\cdot n^{O(\epsilon/\varphi^2 \log(1/\epsilon))}\cdot
(1/\epsilon)^{O(\log(\varphi^2/\epsilon))} \right)$$ random walks line \ref{line:sketch_r} of \Cref{alg:spectralsketch}. 

Additionally, modify the test in \ref{line:dotproducttest} of \Cref{alg:dotproduct}. 
Instead of testing if \\ $|\langle \sketch(x), \sk_j\rangle| \geq 0.5\,|\langle \sketch(x), \sketch(x)\rangle|$, test if $|\langle \sketch(x), \sk_j\rangle| \geq 0.5 \frac{1}{\eta}\frac{k}{n}.$ This is necessary because when $\delta < \frac{1}{2}$, using only $ q=O^*\!\left( (n/k)^{\delta}\cdot n^{O(\epsilon/\varphi^2 \log(1/\epsilon))}\cdot
(1/\epsilon)^{O(\log(\varphi^2/\epsilon))} \right)$ random walks from $x$ is not enough to compute $\langle \sketch(x), \sketch(x)\rangle $ to a high enough precision. 

All other parts of the algorithms remain unchanged.

\paragraph{Sketch of analysis.}
By \Cref{rem:tradeoff_modific}, when we run the modified number of random walks, the guarantee of \Cref{lemma:rw_to_embedding} holds for every set $S$ and every vertex that is \emph{strongly typical} (as per \Cref{rem:strongly_typical}) with respect to $S$. This means that the guarantees of \sketch{} (\Cref{lemma:spectralsketch}) and \dotproduct{} (\Cref{lemma:dotproduct}) hold only for strongly typical vertices (as opposed to all typical vertices). Therefore, in the proof of \Cref{thm:preprocessing}, one should define the ``bad" set $B$ on which the algorithm may fail, to be 

$$ B \coloneqq \{ x \in V: i(x) \in \mathcal{F}\}  \cup \{x\in V : x \text{ is \emph{not strongly typical} with respect to $S$}\} \cup B_{\delta},$$ where $S$ is the set of vertices sampled in line \ref{line:preproc_S} of \Cref{alg:find_cluster_means}. 
In other words, the set $B$ should include all vertices that are not \emph{strongly typical} with respect to $S$, as opposed to only those that are not typical. By \Cref{lemma:good_pts_wrt_S}, 
$$|\{x \in V: x \text{ is \emph{not strongly typical} with respect to $S$}\}|\leq O(n \cdot (\epsilon/\varphi^2)^{1/3} \log(\varphi^2/\epsilon)),$$
which is the same bound as in \Cref{thm:preprocessing}. Therefore, the misclassification rate remains unchanged.

Finally, we discuss how changing the threshold to $0.5 \frac{1}{\eta}\frac{k}{n}$  in line \ref{line:dotproducttest} of \Cref{alg:dotproduct} affects the analysis. Intuitively, \( |\langle \sketch(x), \sk_j\rangle| \) is at least \( \|f_x\|_2^2 \) when \( S \) contains a vertex from the same cluster as \( x \), and is close to zero otherwise. The notion of ``closeness'' is determined by the noise contributed by vertices in other clusters, which is controlled by \Cref{lemma:grouptest_exact}. By redefining \ref{con:rw3} with a sufficiently small constant instead of $10^{-10}$, this noise can be made at most \( c\|f_x\|_2^2 \) for any sufficiently small constant \( c \). This change increases the set of non-typical vertices by only a constant factor. Hence, instead of comparing \( |\langle \sketch(x), \sk_j\rangle| \) to \( 0.5\,|\langle \sketch(x), \sketch(x)\rangle| \), it suffices to compare it to any constant-factor lower bound on \( \|f_x\|_2^2 \), such as \( 0.5\,\frac{1}{\eta}\frac{k}{n} \). 
\end{document}